\documentclass[english,fontsize=11pt,paper=a4,twoside,openright,titlepage,numbers=noenddot,headinclude,BCOR=5mm,footinclude=true,cleardoublepage=empty]{scrreprt}
\usepackage[T1]{fontenc}
    \usepackage[bookmarks=false]{hyperref}
\usepackage[latin9]{inputenc}
\pdfoutput=1
\usepackage{listings}
\usepackage{array}
\usepackage{verbatim}
\usepackage{longtable}
\usepackage{calc}
\usepackage{multirow}
\usepackage{amsthm}
\usepackage{amsmath}
\usepackage{amssymb}
\usepackage{graphicx}

\makeatletter

\providecommand{\tabularnewline}{\\}

\makeatother
\PassOptionsToPackage{eulerchapternumbers,listings,drafting,%
				 pdfspacing,
				 subfig,beramono,parts}{classicthesis}

\usepackage{ifthen}
\newboolean{enable-backrefs} 
\setboolean{enable-backrefs}{true} 

\newcommand{\myTitle}{A Classic Thesis Style\xspace}

\newcommand{\myName}{Andr\'e Miede\xspace}

\newcommand{\myFaculty}{Put data here\xspace}

\newcommand{\myUni}{Put data here\xspace}

\providecommand{\mLyX}{L\kern-.1667em\lower.25em\hbox{Y}\kern-.125emX\@}
\newcommand{\ie}{i.\,e.}

\PassOptionsToPackage{latin9}{inputenc}	
 \usepackage{inputenc}

 \usepackage[english]{babel}

\PassOptionsToPackage{square,numbers}{natbib}
 \usepackage{natbib}

\PassOptionsToPackage{fleqn}{amsmath}		
 \usepackage{amsmath}

\PassOptionsToPackage{T1}{fontenc} 
	\usepackage{fontenc}
\usepackage{textcomp} 
\usepackage{scrhack} 
\usepackage{xspace} 
\usepackage{mparhack} 
\usepackage{fixltx2e} 
\PassOptionsToPackage{printonlyused,smaller}{acronym}
	\usepackage{acronym} 

\usepackage{tabularx} 
\usepackage{caption}
\usepackage{subfig}

\usepackage{listings}
\PassOptionsToPackage{hyperfootnotes=false,pdfpagelabels}{hyperref}
	\usepackage{hyperref}  
	\usepackage{graphicx}

\newcommand{\backrefnotcitedstring}{\relax}
\newcommand{\backrefcitedsinglestring}[1]{(Cited on page~#1.)}
\newcommand{\backrefcitedmultistring}[1]{(Cited on pages~#1.)}
\ifthenelse{\boolean{enable-backrefs}}%
{%
		\PassOptionsToPackage{hyperpageref}{backref}
		\usepackage{backref} 
		   \renewcommand*{\backref}[1]{}  
		   \renewcommand*{\backrefalt}[4]{
		      \ifcase #1 %
		         \backrefnotcitedstring%
		      \or%
		         \backrefcitedsinglestring{#2}%
		      \else%
		         \backrefcitedmultistring{#2}%
		      \fi}%
}{\relax}

\hypersetup{%
    colorlinks=true, linktocpage=true, pdfstartpage=3, pdfstartview=FitV,%
    breaklinks=true, pdfpagemode=UseNone, pageanchor=true, pdfpagemode=UseOutlines,%
    plainpages=false, bookmarksnumbered, bookmarksopen=true, bookmarksopenlevel=1,%
    hypertexnames=true, pdfhighlight=/O,
    urlcolor=webbrown, linkcolor=RoyalBlue, citecolor=webgreen, 
    pdftitle={\myTitle},%
    pdfauthor={\textcopyright\ \myName, \myUni, \myFaculty},%
    pdfsubject={},%
    pdfkeywords={},%
    pdfcreator={pdfLaTeX},%
    pdfproducer={LaTeX with hyperref and classicthesis}%
}

\makeatletter
\@ifpackageloaded{babel}%
    {%
       \addto\extrasamerican{%
				}%
       \addto\extrasngerman{%
				}%
    }{\relax}
\makeatother

\usepackage{classicthesis}

\makeatletter

\theoremstyle{plain}
\ifx\thechapter\undefined
\newtheorem{thm}{\protect\theoremname}
\else
\newtheorem{thm}{\protect\theoremname}[chapter]
\fi
  \theoremstyle{definition}
  \newtheorem{defn}[thm]{\protect\definitionname}
  \theoremstyle{definition}
  \newtheorem{example}[thm]{\protect\examplename}
  \theoremstyle{remark}
  \newtheorem{rem}[thm]{\protect\remarkname}
  \theoremstyle{plain}
  \newtheorem{lem}[thm]{\protect\lemmaname}
  \theoremstyle{plain}
  \newtheorem{cor}[thm]{\protect\corollaryname}
  \theoremstyle{plain}
  \newtheorem{prop}[thm]{\protect\propositionname}
  \theoremstyle{definition}
  \newtheorem{problem}[thm]{\protect\problemname}

\usepackage{amsmath}
\usepackage{pxfonts}
\usepackage[defblank]{paralist}
\usepackage{listings}

\numberwithin{figure}{chapter}
\numberwithin{table}{chapter}

\usepackage{amssymb}

\usepackage[pass]{geometry} 
\usepackage{flafter} 

\usepackage[nottoc,numbib]{tocbibind}

\makeatother

\usepackage{babel}
  \providecommand{\corollaryname}{Corollary}
  \providecommand{\definitionname}{Definition}
  \providecommand{\examplename}{Example}
  \providecommand{\lemmaname}{Lemma}
  \providecommand{\problemname}{Problem}
  \providecommand{\propositionname}{Proposition}
  \providecommand{\remarkname}{Remark}
\providecommand{\theoremname}{Theorem}

\begin{document}
\newgeometry{margin=2cm}

\begin{titlepage}

 \newcommand{\HRule}{\rule{\linewidth}{0.5mm}}\pagenumbering{gobble}

\begin{center}
~\\[1cm]
\par\end{center}

\begin{center}
\includegraphics[width=0.66\textwidth]{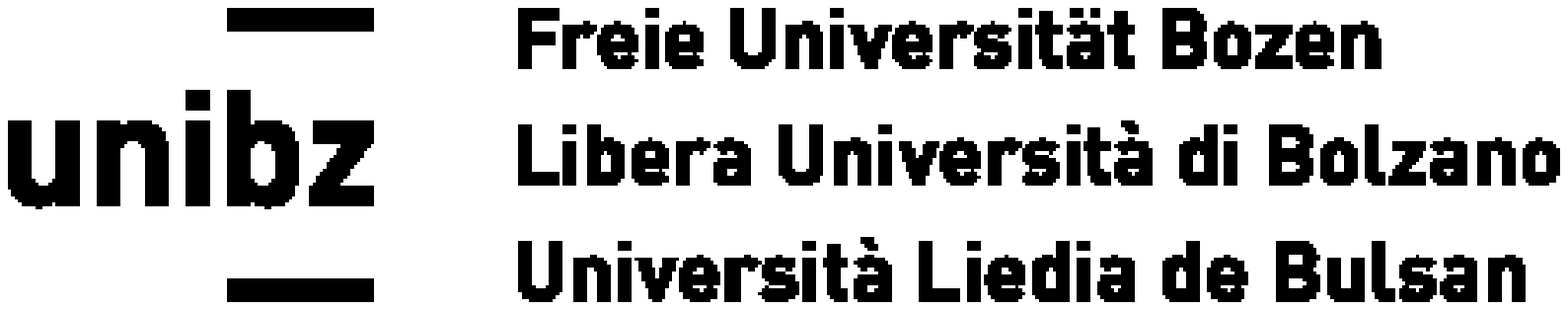}~\\[2.4cm]
\par\end{center}

\begin{center}
\textsc{\Large{}PhD Thesis}{\Large{}\\[1.2cm]}
\par\end{center}{\Large \par}

\begin{center}
\HRule \\[0.4cm] { \textbf{\huge{}Query-driven Data Completeness
Management \\[0.4cm] }}
\par\end{center}{\huge \par}

\begin{center}
\HRule \\[1.5cm]
\par\end{center}

\begin{center}
\begin{minipage}[c]{0.4\textwidth}%
\begin{flushleft}
\emph{\large{}Author:}{\large{}}\\
{\large{} Simon Razniewski }
\par\end{flushleft}%
\end{minipage}%
\begin{minipage}[c]{0.4\textwidth}%

\begin{flushright}
\emph{\large{}Supervisor:}{\large{} }\\
{\large{} Prof. Werner Nutt }
\par\end{flushright}%
\end{minipage}
\par\end{center}

\vspace{5.5cm}

\hspace{1.15cm}\emph{Thesis in part reviewed by:}

\hspace{1.3cm}- Prof. Leopoldo Bertossi, Carleton University, Canada

\hspace{1.3cm}- Prof. Jan van den Bussche, Universiteit Hasselt,
Belgium

\hspace{1.3cm}- Prof. Floris Geerts, University of Antwerp, Belgium

\hspace{1.3cm}- Prof. Francesco Ricci, Free University of Bozen-Bolzano,
Italy

\hspace{1.3cm}- Prof. Fabian M. Suchanek, Télécom ParisTech University,
France

\begin{center}
\vfill{}

\par\end{center}

\begin{center}
{\large{}October 2014}
\par\end{center}{\large \par}

\end{titlepage}

\restoregeometry

\newpage{}

 \ \newpage{}

 \ \vspace{3cm}

\begin{flushleft}
\emph{\hspace{0.7cm}There are two kinds of methodologies: \vspace{-0.2cm}
}
\par\end{flushleft}

\begin{flushleft}
\emph{\hspace{1cm}1. Those that cannot be used to reason about incomplete
information}
\par\end{flushleft}

\newpage{}

 \ \newpage{}

\chapter*{Acknowledgement}

First of all I would like to thank my advisor for his advice, guidance
and patience. I would also like to thank to my close colleagues and
my coauthors, in particular Ognjen Savkovic, Marco Montali, Fariz
Darari and Giuseppe Pirró. I would like to thank my collaborators,
in particular Divesh Srivastava, Flip Korn, Marios Hadjieleftheriou
and Alin Deutsch. I am thankful for the good surrounding and support
from my colleagues in the KRDB group and in the Faculty of Computer
Science. Thanks to my friends. Thanks to my family. Thanks to Mita.

\global\long\def\eat{}

\global\long\def\Wlog{W.l.o.g.\ }

\global\long\def\ie{i.\ e.\ ,}

\global\long\def\set#1{\{\,#1\,\}}

\global\long\def\const{\mathrm{const}}

\global\long\def\union{\cup}

\global\long\def\dom{\mathit{dom}}

\global\long\def\var{\mathit{Var}}

\global\long\def\query#1#2{#1 \qif#2}

\global\long\def\pdb{{\mathcal{D}\xspace}\xspace}

\global\long\def\D{{\pdb}}

\global\long\def\id#1{#1^{i}}

\global\long\def\av#1{#1^{a}}

\global\long\def\can{\mathrm{can}}

\global\long\def\fc{f_{\C}}

\global\long\def\complb#1{\mathit{Compl^{b}}(#1)}

\global\long\def\compls#1{\mathit{Compl^{s}}(#1)}

\global\long\def\complstar#1#2{\mathit{Compl}^{#2}(#1)}

\global\long\def\compl#1{\textit{Compl}($#1$)}

\global\long\def\da{\av D}

\global\long\def\di{\id D}

\global\long\def\lcstmt{C}

\global\long\def\cplstmt#1#2{#1 {\, \dot{\subseteq} \, } #2}

\global\long\def\Compl#1{\mathit{Compl}{(}#1{)}}

\global\long\def\Complb#1{\mathit{Compl}^{b}($#1$)}

\global\long\def\Compls#1{\mathit{Compl}^{s}($#1$)}

\global\long\def\qcompl#1{\mathit{Compl}{(}#1{)}}

\global\long\def\cplset{{\mathcal{C}}}
 \global\long\def\C{{\mathcal{C}}}

\global\long\def\uk{\mathrm{uk}}
 %

\global\long\def\na{\mathrm{N/A}}
 %

\global\long\def\codd{\mathrm{Codd}}

\global\long\def\sql{\mathrm{SQL}}

\global\long\def\Null{\textsl{null}\xspace}

\global\long\def\Nulls{\textsl{null}s\xspace}

\global\long\def\epos#1#2{\textit{EPos}(#1,#2)}

\global\long\def\student{\texttt{student}\xspace}

\global\long\def\person{\texttt{person}\xspace}

\global\long\def\it#1{\mathit{#1}}

\global\long\def\sid{\texttt{sid}\xspace}

\global\long\def\name{\texttt{name}\xspace}

\global\long\def\level{\texttt{level}\xspace}

\global\long\def\code{\texttt{code}\xspace}

\global\long\def\hometown{\texttt{hometown}\xspace}

\global\long\def\class{\texttt{class}\xspace}

\global\long\def\formTeacher{\texttt{formTeacher}\xspace}

\global\long\def\viceFormTeacher{\texttt{viceFormTeacher}\xspace}

\global\long\def\profile{\texttt{profile}\xspace}

\global\long\def\customer{\texttt{customer}}

\global\long\def\street{\texttt{street}}

\global\long\def\city{\texttt{city}}

\global\long\def\supplier{\texttt{supplier}}

\global\long\def\partner{\texttt{partner}}

\global\long\def\longTerm{\texttt{long\_term}}

\global\long\def\contact{\texttt{contact}}

\global\long\def\iD{\texttt{id}}

\global\long\def\John{\mathit{John}}

\global\long\def\Mary{\mathit{Mary}}

\global\long\def\pupil{\mathit{pupil}}

\global\long\def\req{\mathit{request}}

\global\long\def\enr{\mathit{enrolled}}

\global\long\def\test{\mathit{test}}

\global\long\def\Hoferschool{\mathit{HoferSchool}}

\global\long\def\DaVincischool{\mathit{DaVinciSchool}}

\global\long\def\livesin{\mathit{livesIn}}

\global\long\def\Bolzano{\mathit{Bolzano}}

\global\long\def\Bob{\mathit{Bob}}
 \global\long\def\Alice{\mathit{Alice}}

\global\long\def\Merano{\mathit{Merano}}

\global\long\def\etal{et al.}

\global\long\def\satisfies{\models}

\global\long\def\RQ{\ensuremath{\L_{{\rm {RQ}}}}}

\global\long\def\CQ{\ensuremath{\L_{{\rm {CQ}}}}}

\global\long\def\LRQ{\ensuremath{\L_{{\rm {LRQ}}}}}

\global\long\def\LCQ{\ensuremath{\L_{{\rm {LCQ}}}}}

\global\long\def\LC{TC\xspace}

\global\long\def\dd#1#2{#1_{1},\ldots,#1_{#2}}

\global\long\def\quotes#1{\lq\lq#1\rq\rq}

\global\long\def\wrt{w.r.t.\ }

\global\long\def\WLOG{wlog\xspace}

\global\long\def\PTIME{\ensuremath{\mathsf{PTIME}}}

\global\long\def\CONP{\ensuremath{\mathsf{coNP}}}

\global\long\def\NP{\ensuremath{\mathsf{NP}}}

\global\long\def\piptwo{\Pi_{2}^{P}}

\global\long\def\sigptwo{\Sigma_{2}^{P}}

\global\long\def\Cont{\ensuremath{\mathsf{Cont}}}

\global\long\def\ContU{\ensuremath{\mathsf{ContU}}}

\global\long\def\LCQC{\ensuremath{\mathsf{TC\text{-}QC}}}

\global\long\def\LCLC{\ensuremath{\mathsf{TC\text{-}TC}}}

\global\long\def\QCQC{\ensuremath{\mathsf{QC\text{-}QC}}}

\global\long\def\smixed{\mathrm{s+}}

\global\long\def\complpi#1{\textit{Compl}_{\pi}\ensuremath{(#1)}}

\global\long\def\cplhat{\hat{C}}

\global\long\def\chk{\check{}}

\global\long\def\localComp{table completeness\xspace}

\global\long\def\LocalComp{Table completeness\xspace}

\global\long\def\true{\textit{true}}

\global\long\def\false{\textit{false}}

\global\long\def\And{\wedge}

\global\long\def\Or{\vee}

\global\long\def\eset{\emptyset}

\global\long\def\bigset#1{ \Bigl\{ #1 \Bigr\} }

\global\long\def\bigmid{\ \Big|\ }

\global\long\def\bag#1{\{\hspace{-0.13em}|\, #1 \,|\hspace{-0.13em}\}}

\global\long\def\incl{\subseteq}

\global\long\def\incls{\supseteq}

\global\long\def\col{\colon}

\global\long\def\angles#1{\langle#1\rangle}

\global\long\def\Sum{\ensuremath{{\sf sum}}}

\global\long\def\Count{\ensuremath{{\sf count}}}

\global\long\def\Max{\ensuremath{{\sf max}}}

\global\long\def\Min{\ensuremath{{\sf min}}}

\global\long\def\L{\mathcal{L}}

\global\long\def\core#1{\mathring{#1}}

\global\long\def\qif{\,{:}{-}\,}

\global\long\def\lit#1{L_{#1}}

\global\long\def\cpred#1{V_{#1}}

\global\long\def\concond{G}

\global\long\def\Q{{\mathcal{Q}}}

\global\long\def\F{{\mathcal{F}}}

\global\long\def\S{{\mathcal{S}}}

\global\long\def\implies{\rightarrow}

\global\long\def\findom#1#2#3{\mathit{Dom}(#1,#2,#3)}

\global\long\def\determines{\rightarrow\! \! \! \! \! \rightarrow}

\global\long\def\dotequiv{\ \dot{\equiv} \ }

\global\long\def\queryset{{\mathcal{Q}}}

\global\long\def\schemaconstraintset{{\mathcal{F}}}

\global\long\def\val{\upsilon}

\global\long\def\viewset{{\mathcal{V}}}

\global\long\def\tpl#1{\bar{#1}}

\global\long\def\tplsub#1#2{{\bar{#1}}_{#2}}

\global\long\def\aufz#1#2{#1_{1},\ldots,#1_{#2}}

\global\long\def\modelsms{\models\! \! ^{\mathit{M\!\! S}}}

\global\long\def\joininquotes{\mbox{ ``\ensuremath{\bowtie}"}}

\global\long\def\domby{\preceq}

\global\long\def\query#1#2{#1\qif#2}

\global\long\def\pdb{\mathcal{D}}

\global\long\def\domby{\preceq}

\global\long\def\dominates{\succeq}

\global\long\def\query#1#2{#1\qif#2}

\global\long\def\compl#1{\mathrm{Compl}(#1)}

\global\long\def\tpl#1{\bar{#1}}

\global\long\def\ourStyle#1{{\small\texttt{#1}}}

\global\long\def\tuple#1{\langle\texttt{#1}\rangle}

\global\long\def\ur{\texttt{u}}

\global\long\def\G{{\mathcal{G}}}

\global\long\def\UU{{\mathcal{U}}}

\global\long\def\LL{{\mathcal{L}}}

\global\long\def\BB{{\mathcal{B}}}

\global\long\def\TC{T_{\gcs}}

\global\long\def\frozen#1{\tilde{#1}}

\global\long\def\widefrozen#1{\widetilde{#1}}

\global\long\def\ID{\mathit{id}}
 \global\long\def\T{{\cal T}}

\global\long\def\frozenID{\frozen{\ID}}

\global\long\def\mathSet#1{\ensuremath{\{\,#1\,\}}}

\global\long\def\Ga{\ensuremath{G^{a}}}

\global\long\def\Gi{\ensuremath{G^{i}}}

\global\long\def\QC{\ensuremath{\mathit{Compl}(Q)}}

\global\long\def\buildQC#1{\ensuremath{\mathit{Compl}(\mathit{#1})}}

\global\long\def\buildQCset#1{\ensuremath{\mathit{Compl}^{s}(\mathit{#1})}}

\global\long\def\buildQCbag#1{\ensuremath{\mathit{Compl}^{b}(\mathit{#1})}}

\global\long\def\defineGC#1#2{\ensuremath{\mathit{Compl}(#1\mid#2)}}

\global\long\def\defineGCone#1{\ensuremath{\mathit{Compl}(#1)}}

\global\long\def\buildGC#1#2{\ensuremath{\mathit{Compl}(\set{#1}\mid\set{#2})}}

\global\long\def\buildGCone#1{\ensuremath{\mathit{Compl}(\set{#1})}}

\global\long\def\calC{\ensuremath{\mathcal{C}}}

\global\long\def\GCC{\calC}

\global\long\def\term#1{\ensuremath{\mathtt{#1}}}

\global\long\def\mytabs{\hspace{1em}*{2\parindent}\=\hspace{1cm}\=\hspace{1cm}\=\hspace{1cm}\=\hspace{1cm}}

\global\long\def\Graph{\ensuremath{G}}
 \global\long\def\Pattern{\ensuremath{P}}


\global\long\def\buildDBQuery#1#2{\ensuremath{{#1}\;\text{:-}\;{#2}}}

\global\long\def\calF{\ensuremath{\mathcal{F}}}
\global\long\def\fed#1{\ensuremath{\mathit{fed}_{\calF}(#1)}}
 \global\long\def\evalfed#1#2{\ensuremath{\llbracket\fed{#1}\rrbracket_{#2}^{s}}}
 \global\long\def\evalfedset#1{\ensuremath{\llbracket\fed{#1}\rrbracket_{\emptyset}^{s}}}

\global\long\def\defineSERVICE#1#2{\ensuremath{(\ourStyle{SERVICE}\;#1\;#2)}}


\global\long\def\subsumedBy{\sqsubseteq}

\global\long\def\lessInformativeThan{\subsumedBy}

\global\long\def\buildSELECTnew#1#2{\ensuremath{(\mathit{SELECT}\;\mathSet{#1}\;#2)}}

\global\long\def\newvars#1{\ensuremath{\mathit{newvars}(#1)}}

\global\long\def\ext#1{\ensuremath{\mathit{ext}(#1)}}

\global\long\def\calT{\ensuremath{\mathcal{T}}}

\global\long\def\evaltopdown#1#2{\ensuremath{\llbracket#1\rrbracket_{#2}^{td}}}

\global\long\def\queryMINUS#1#2{(#1 \ourStyle{ MINUS } #2)}

\global\long\def\queryUNION#1#2{(#1 \ourStyle{ UNION } #2)}

\global\long\def\queryAND#1#2{(#1 \ourStyle{ AND } #2)}

\global\long\def\queryOPT#1#2{(#1 \ourStyle{ OPT } #2)}

\global\long\def\queryFILTER#1#2{(#1 \ourStyle{ FILTER } #2)}

\global\long\def\querySERVICE#1#2{(\ourStyle{SERVICE} \; #1 \; #2)}

\global\long\def\querySERVICEBraces#1#2{(\ourStyle{SERVICE} \; #1 \; \{#2\}}

\global\long\def\yes{\texttt{yes}}
 \global\long\def\notexists#1{\ensuremath{\ourStyle{NOT}\;\ourStyle{EXISTS}}(#1)}

\global\long\def\iri#1{\ensuremath{\mathit{iri}(#1)}}

\global\long\def\vterm#1{\ensuremath{\mathit{vterm}(#1)}}

\global\long\def\fedC{\ensuremath{\bar{\calC}}}

\global\long\def\fedG{\ensuremath{\bar{\calG}}}

\global\long\def\calS{\ensuremath{\mathcal{S}}}
 \global\long\def\calG{\ensuremath{\mathcal{G}}}

\global\long\def\modelsfed{\models_{\mathit{fed}}}

\global\long\def\fednew#1#2{\ensuremath{\mathit{fed}_{#1}(#2)}}

\global\long\def\SERVICE{\ourStyle{SERVICE}\xspace}

\global\long\def\AND{\ourStyle{AND}\xspace}

\global\long\def\g{\ensuremath{G}}

\global\long\def\pg{\ensuremath{\mathcal{G}}}
\global\long\def\gs{\ensuremath{\bar{G}}}

\global\long\def\pgs{\ensuremath{\bar{\mathcal{G}}}}

\global\long\def\gc{\ensuremath{C}}

\global\long\def\gcs{\ensuremath{\mathcal{C}}}

\global\long\def\fgcs{\ensuremath{\dot{\mathcal{C}}}}

\global\long\def\pgsIdeal{\ensuremath{\dot{\mathcal{G}}^{i}}}

\newcounter{mycount} 
\global\long\def\myprob#1#2{%
\stepcounter{mycount} \par\noindent Problem\ \themycount.}

\global\long\def\evalwithfed#1#2{\ensuremath{\llbracket#1\rrbracket_{#2}^{\mathit{}}}}

\global\long\def\Vcompl{\ourStyle{Completeness}}

\global\long\def\VhasComplStmt{\ourStyle{hasComplStmt}}

\global\long\def\VhasPattern{\ourStyle{hasPattern}}

\global\long\def\VhasCondition{\ourStyle{hasCondition}}

\global\long\def\Vsubject{\ourStyle{subject}}

\global\long\def\Vpredicate{\ourStyle{predicate}}

\global\long\def\Vobject{\ourStyle{object}}

\global\long\def\VspinVarName{\ourStyle{spin:varName}}

\global\long\def\VvarName{\ourStyle{varName}}

\global\long\def\dbp{\mathit{dbp}}


\global\long\def\mytpl#1{\langle#1\rangle}

\global\long\def\defineASK#1{\ensuremath{(\ourStyle{ASK}\;#1)}}

\global\long\def\buildASK#1{\ensuremath{(\ourStyle{ASK}\;\mathSet{#1})}}

\global\long\def\defineSELECT#1#2{\ensuremath{(\ourStyle{SELECT}\;#1\;#2)}}

\global\long\def\buildSELECT#1#2{\ensuremath{(\ourStyle{SELECT}\;\mathSet{#1}\;\mathSet{#2})}}

\global\long\def\defineCONSTRUCT#1#2{\ensuremath{(\ourStyle{CONSTRUCT}\;#1\;#2)}}

\global\long\def\buildCONSTRUCT#1#2{\ensuremath{(\ourStyle{CONSTRUCT}\;\{\mathSet{#1}\}\;\{\mathSet{#2}\})}}

\global\long\def\ASK{\ourStyle{ASK}\xspace}

\global\long\def\SELECT{\ourStyle{SELECT}\xspace}

\global\long\def\CONSTRUCT{\ourStyle{CONSTRUCT}\xspace}

\global\long\def\DISTINCT{\ourStyle{DISTINCT}\xspace}

\global\long\def\UNION{\ourStyle{UNION}\xspace}

\global\long\def\FILTER{\ourStyle{FILTER}\xspace}

\global\long\def\OPT{\operatorname{\mathrm{OPT}}}

\global\long\def\OR{\ourStyle{OR}\xspace}

\global\long\def\MINUS{\ourStyle{MINUS}\xspace}

\global\long\def\map#1#2{\ensuremath{\text{#1}\mapsto\text{#2}}}

\global\long\def\triple#1#2#3{(\mathit{#1}, \mathit{#2}, \mathit{#3})}

\global\long\def\buildIFF#1#2{\ensuremath{#1\;\;\;\;\text{iff}\;\;\;\;#2}}

\global\long\def\evalset#1#2{\ensuremath{\llbracket#1\rrbracket_{#2}^{s}}}

\global\long\def\evalbag#1#2{\ensuremath{\llbracket#1\rrbracket_{#2}^{b}}}

\global\long\def\evalsetbag#1#2{\ensuremath{\llbracket#1\rrbracket_{#2}}}

\global\long\def\eval#1#2{\ensuremath{\llbracket#1\rrbracket_{#2}}}

\global\long\def\setquery#1#2{(\set{#1},\set{#2})}

\global\long\def\sparqlquery#1#2{(#1,#2)}

\global\long\def\cardMapping#1#2{\ensuremath{\mathit{card}_{#1}(#2)}}

\global\long\def\vars#1{\ensuremath{\mathit{vars}(#1)}}

\global\long\def\emptymap{\ensuremath{\mu_{\emptyset}}}

\global\long\def\mappingset{\ensuremath{\Omega}}

\global\long\def\rhodf{\ensuremath{\rho\mathit{DF}}}

\global\long\def\cl{\ensuremath{\mathit{cl}}}

\global\long\def\corner{{\tt {CoRNER}}}

\global\long\def\void{{\tt {VoID}}}

\global\long\def\leftouterjoin{\; \includegraphics[scale=0.8]{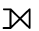} \;}

\global\long\def\TC{T_{\gcs}}

\global\long\def\frozen#1{\tilde{#1}}

\global\long\def\widefrozen#1{\widetilde{#1}}

\global\long\def\ID{\mathit{id}}
 \global\long\def\T{{\cal T}}

\global\long\def\frozenID{\frozen{\ID}}

\global\long\def\C{{\cal C}}

\global\long\def\T{\mathcal{T}}
\global\long\def\P{\mathcal{P}}

\global\long\def\film{\mathit{film}}
 \global\long\def\movie{\mathit{movie}}
\global\long\def\dir{\mathrm{dir}}
\global\long\def\tn{\mathrm{tn}}

\global\long\def\ids{\G}

\global\long\def\setCS{\C}

\global\long\def\aseq{\mathit{Aseq}}

\global\long\def\qats{\bar{T}}

\global\long\def\actrw{\mathrm{RA}}

\global\long\def\actdb{\mathrm{CA}}
 \global\long\def\tgd#1#2{#1\rightarrow#2}

\global\long\def\re{\mathit{re}}
 \global\long\def\ce{\mathit{ce}}

\global\long\def\CE{\mathrm{CE}}
 \global\long\def\start{start}
 \global\long\def\resident{\mathit{resident}}
 \global\long\def\copyy{\mathit{copy}}
 \global\long\def\Davinci{\mathit{DaVinci}}

\renewcommand\thepart{\Roman{part}}

\global\long\def\C{{\cal C}}

\global\long\def\ext#1{\mathrm{ext}(#1)}

\global\long\def\John{\mathit{John}}

\global\long\def\Mary{\mathit{Mary}}

\global\long\def\pupil{\mathit{pupil}}

\global\long\def\req{\mathit{request}}

\global\long\def\enr{\mathit{enrolled}}

\global\long\def\test{\mathit{test}}

\global\long\def\Hoferschool{\mathit{HoferSchool}}

\global\long\def\DaVincischool{\mathit{DaVinciSchool}}

\global\long\def\livesin{\mathit{livesIn}}

\global\long\def\Bolzano{\mathit{Bolzano}}

\global\long\def\Bob{\mathit{Bob}}

\global\long\def\Alice{\mathit{Alice}}

\global\long\def\Merano{\mathit{Merano}}

\global\long\def\coNP{\mathrm{coNP}}

\global\long\def\NP{\mathrm{NP}}

\global\long\def\actrw{\mathrm{RA}}

\global\long\def\actdb{\mathrm{CA}}

\global\long\def\tgd#1#2{#1\rightarrow#2}

\global\long\def\re{\mathit{re}}
 \global\long\def\ce{\mathit{ce}}

\global\long\def\CE{\mathrm{CE}}

\global\long\def\start{start}

\global\long\def\resident{\mathit{resident}}

\global\long\def\copyy{\mathit{copy}}

\global\long\def\Davinci{\mathit{DaVinci}}

\global\long\def\qats{\bar{T}}

\global\long\def\ucont#1{\mathit{ContU}(#1)}
\global\long\def\ent#1{\mathit{EntC}(#1)}

\global\long\def\reach{\mathit{reach}}

\global\long\def\resident{\mathit{resident}}

\global\long\def\theenumi{(\roman{enumi})}
 \global\long\def\labelenumi{\theenumi}

\global\long\def\domain{\ensuremath{\mathit{dom}}}

\global\long\def\mytpl#1{\langle#1\rangle}

\global\long\def\fc{f_{\C}}
 \global\long\def\uk{\mathrm{uk}}
 \global\long\def\na{\mathrm{n/a}}

\global\long\def\const{\mathrm{const}}
 \global\long\def\codd{\mathrm{Codd}}
 \global\long\def\sql{\mathrm{SQL}}
 \global\long\def\can{\mathrm{can}}
 \global\long\def\complb#1{\mathit{Compl^{b}}(#1)}
 \global\long\def\compls#1{\mathit{Compl^{s}}(#1)}
 \global\long\def\complstar#1#2{\mathit{Compl}^{#2}(#1)}
 \global\long\def\smixed{\mathrm{s+}}
 \global\long\def\pdb{\mathcal{D}}
 \global\long\def\domby{\preceq}
 \global\long\def\dominates{\succeq}
 \global\long\def\query#1#2{#1\qif#2}
 \global\long\def\tpl#1{\bar{#1}}
 \global\long\def\ctle{\mathbf{E}}
 \global\long\def\ctlf{\mathbf{F}}
 \global\long\def\ctla{\mathbf{A}}
 \global\long\def\ctlg{\mathbf{G}}
 \global\long\def\muexists{<\negthickspace\negthinspace-\negthickspace\negthinspace>}
 \global\long\def\muall{[\negthinspace-\negthinspace]}
 \global\long\def\statedescriptor{\mathrm{stateDescriptor}}
 \global\long\def\mulp{\mu\mathcal{L_{\mathit{P}}}}
 \global\long\def\mula{\mu\mathcal{L_{\mathit{A}}}}

\global\long\def\live{\mathrm{Live}}
 \global\long\def\chase{\mathrm{chase}}

\global\long\def\OnlyIf{\lq\lq$\Rightarrow$\rq\rq\ \ }

\global\long\def\NULL{\bot}

\global\long\def\arity#1{\textit{arity}(#1)}

\global\long\def\student{\ensuremath{\mathtt{student}}\xspace}
 \global\long\def\class{\ensuremath{\mathtt{class}}\xspace}
 \global\long\def\classCode{\ensuremath{\mathtt{classCode}}\xspace}
 \global\long\def\name{\ensuremath{\mathtt{name}}\xspace}
 \global\long\def\homeTown{\ensuremath{\mathtt{homeTown}}\xspace}
 \global\long\def\code{\ensuremath{\mathtt{code}}\xspace}
 \global\long\def\hascode{\ensuremath{\mathtt{hasCode}}\xspace}
 \global\long\def\profile{\ensuremath{\mathtt{profile}}\xspace}
 \global\long\def\formTeacher{\ensuremath{\mathtt{formTeacher}}\xspace}
 \global\long\def\lncourse{\ensuremath{\mathtt{lnCourse}}\xspace}

\global\long\def\hasattr{\ensuremath{\mathtt{hasAttr}}\xspace}
 \global\long\def\attr{\ensuremath{\mathtt{attr}}\xspace}

\global\long\def\arts{\mathrm{'arts'}\xspace}
 \global\long\def\school{\textit{school}\xspace}

\global\long\def\Mary{\mathrm{Mary}}
 \global\long\def\Paul{\mathrm{Paul}}
 \global\long\def\Chester{\mathrm{Chester}}
 \global\long\def\Hampton{\mathrm{Hampton}}
 \global\long\def\twoA{\mathrm{2a}}

\global\long\def\nonulls#1{{#1}^{\downarrow}}
\global\long\def\nonullsmod#1{{#1}^{\Downarrow}}

\global\long\def\Rule#1{\rho_{#1}}
\global\long\def\tcruleEx#1#2#3{{#1}\rightarrow\exists\,{#2}{\mbox{{\bf .}}}\,{#3}}
 \global\long\def\tcrule#1#2{{#1}\rightarrow{#2}}

%
\global\long\def\IF{\mathrm{inc}}
\global\long\def\RF{\mathrm{res}}
\global\long\def\PF{\mathrm{2\bot}}
\global\long\def\TN{\mathrm{3\bot}}
\global\long\def\AN{\mathrm{ambg}}
\global\long\def\cert{\mathrm{cert}}
\global\long\def\chase{\ensuremath{\mathit{chase}}\xspace}

\global\long\def\SIG{\Sigma}
 \global\long\def\SIGi{\id{\SIG}}
 \global\long\def\SIGa{\av{\SIG}}

\global\long\def\Null{\textsl{null}\xspace}
 \global\long\def\Nulls{\textsl{null}s\xspace}
 \global\long\def\D{{\pdb}}
 \global\long\def\id#1{#1^{i}}
 \global\long\def\av#1{#1^{a}}

\global\long\def\epos#1#2{\textit{EPos}(#1,#2)}

\global\long\def\etal{et al.}
 \global\long\def\satisfies{\models}

\global\long\def\RQ{\ensuremath{\L_{{\rm {RQ}}}}}
 \global\long\def\CQ{\ensuremath{\L_{{\rm {CQ}}}}}

\global\long\def\LRQ{\ensuremath{\L_{{\rm {LRQ}}}}}
 \global\long\def\LCQ{\ensuremath{\L_{{\rm {LCQ}}}}}

\global\long\def\L{{\mathcal{L}}}

\global\long\def\compl#1{\textit{Compl}(#1)}

\global\long\def\complpi#1{\textit{Compl}_{\pi}\ensuremath{(#1)}}
\global\long\def\card#1{|#1|}

\global\long\def\localComp{table completeness\xspace}
 \global\long\def\LocalComp{Table completeness\xspace}
 \global\long\def\LC{TC\xspace}


\global\long\def\eset{\emptyset}
 \global\long\def\bigset#1{ \Bigl\{ #1 \Bigr\} }

\global\long\def\bigmid{\ \Big|\ }

\global\long\def\bag#1{\{\hspace{-0.13em}|\, #1 \,|\hspace{-0.13em}\}}

\global\long\def\incl{\subseteq}

\global\long\def\incls{\supseteq}

\global\long\def\col{\colon}

\global\long\def\angles#1{\langle#1\rangle}

\global\long\def\dd#1#2{#1_{1},\ldots,#1_{#2}}

\global\long\def\quotes#1{\lq\lq#1\rq\rq}

\global\long\def\wrt{w.r.t.\ }

\global\long\def\WLOG{wlog\xspace}

\global\long\def\core#1{\mathring{#1}}


\global\long\def\qif{\,{:}{-}\,}
 \global\long\def\union{\cup}
 \global\long\def\dom{\mathit{dom}}
 \global\long\def\var{\mathit{Var}}
 \global\long\def\lit#1{L_{#1}}
 \global\long\def\cpred#1{V_{#1}}
 \global\long\def\concond{G}
 \global\long\def\da{\av D}
 \global\long\def\di{\id D}
 \global\long\def\lcstmt{C}
 \global\long\def\cplstmt#1#2{#1 {\, \dot{\subseteq} \, } #2}
 \global\long\def\Compl#1{\mathit{Compl}{(}#1{)}}
 \global\long\def\Complb#1{\mathit{Compl}^{b}($#1$)}
 \global\long\def\Compls#1{\mathit{Compl}^{s}($#1$)}
 \global\long\def\Complsred#1{\mathit{Compl}^{sred}($#1$)}

\global\long\def\qcompl#1{\mathit{Compl}{(}#1{)}}
 \global\long\def\cplset{{\mathcal{C}}}
 \global\long\def\C{{\mathcal{C}}}
\global\long\def\Q{{\mathcal{Q}}}
\global\long\def\F{{\mathcal{F}}}
\global\long\def\S{{\mathcal{S}}}
\global\long\def\cplhat{\hat{C}}
 \global\long\def\chk{\check{}}
 \global\long\def\implies{\rightarrow}
 \global\long\def\And{\wedge}
\global\long\def\Or{\vee}
\global\long\def\findom#1#2#3{\mathit{Dom}(#1,#2,#3)}

\global\long\def\true{\textit{true}}

\global\long\def\false{\textit{false}}

\global\long\def\determines{\rightarrow\! \! \! \! \! \rightarrow}

\global\long\def\dotequiv{\ \dot{\equiv} \ }
\global\long\def\Wlog{W.l.o.g.\ }
 \global\long\def\ie{i.\ e.\ ,}
 \global\long\def\piptwo{\Pi_{2}^{P}}
 \global\long\def\sigptwo{\Sigma_{2}^{P}}

\global\long\def\queryset{{\mathcal{Q}}}

\global\long\def\schemaconstraintset{{\mathcal{F}}}
\global\long\def\scset{\schemaconstraintset}
\global\long\def\qset{\queryset}
 \global\long\def\val{\upsilon}
 \global\long\def\viewset{{\mathcal{V}}}

\global\long\def\tplsub#1#2{{\bar{#1}}_{#2}}
\global\long\def\aufz#1#2{#1_{1},\ldots,#1_{#2}}
\global\long\def\set#1{\{\,#1\,\}}
 \global\long\def\Onlyif{\lq\lq$\Rightarrow$\rq\rq\ \ }

\global\long\def\If{\lq\lq$\Leftarrow$\rq\rq\ \ }

\global\long\def\tc#1#2#3{\mathit{Compl}(#1;\,#2;\,#3)}




\global\long\def\customer{\texttt{customer}}
 \global\long\def\street{\texttt{street}}
 \global\long\def\city{\texttt{city}}
 \global\long\def\supplier{\texttt{supplier}}
 \global\long\def\partner{\texttt{partner}}
 \global\long\def\longTerm{\texttt{long\_term}}
 \global\long\def\contact{\texttt{contact}}
 \global\long\def\hometown{\texttt{hometown}}
 \global\long\def\iD{\texttt{id}}

\global\long\def\mrm#1{\mathrm{#1}}
 \global\long\def\aufz#1#2{#1_{1},\ldots,#1_{#2}}

\global\long\def\dist{\mathrm{dist}}
 \global\long\def\tpl#1{\bar{#1}}

\global\long\def\hotel{\mathrm{Hotel}}
 \global\long\def\F{\mathcal{F}}

\global\long\def\acompl{\mathit{CA}}

\global\long\def\di{D^{i}}

\global\long\def\da{D^{a}}

\global\long\def\restrict#1#2{#1\cap#2}

\global\long\def\shrink{\mathit{shrink}}

\global\long\def\complatom#1{\mathit{CA}_{#1}}

\global\long\def\pot{\mathrm{\mathit{Pot}}}

\global\long\def\comploutofrange{\textit{ComplOutOfRange}_{\F,\da}}

\global\long\def\buffer{\mathit{buffer}}

\global\long\def\set#1{\{#1\}}

\global\long\def\bigset#1{ \bigl\{ #1 \bigr\} }

\global\long\def\bigmid{\ \big|\ }

\global\long\def\ca{\mathrm{CA}}
\global\long\def\oor{\mathrm{OR}}
\global\long\def\coor{\mathrm{COOR}}
\global\long\def\cert{\mathit{cert}}
\global\long\def\poss{\mathit{poss}}
\global\long\def\imposs{\mathit{imposs}}
\global\long\def\notins{\mathit{notins}}
\global\long\def\complins{\mathit{complins}}

\global\long\def\query#1#2{#1:\!-\ #2}

\global\long\def\compl#1{\mathrm{Compl}(#1)}

\global\long\def\C{\mathcal{C}}

\global\long\def\S{\mathcal{S}}

\global\long\def\Onlyif{\lq\lq$\Rightarrow$\rq\rq\ \ }

\global\long\def\If{\lq\lq$\Leftarrow$\rq\rq\ \ }

\global\long\def\size#1{\mid\!#1\!\mid}

\global\long\def\V{\mathcal{V}}
\global\long\def\det{-\!\!-\!\!>\!\!>}

\global\long\def\L{\mathcal{L}}

\global\long\def\UCont{\mathsf{\mathit{UCont}}}
\global\long\def\Cont{\ensuremath{\mathit{Cont}}}
\global\long\def\ContU{\mathit{UCont}}
\global\long\def\ucont{\mathit{UCont}}
\global\long\def\ucont#1{\mathit{UCont}(#1)}
\global\long\def\ent#1{\mathit{EntC}(#1)}

\global\long\def\D{\mathcal{D}}

\global\long\def\Q{\mathcal{Q}}

\global\long\def\det{\determines}

\global\long\def\male{\mathit{male}}
\global\long\def\female{\mathit{female}}

\global\long\def\compl#1{\mathit{Compl}(#1)}

\global\long\def\tcqc{\mbox{\ensuremath{\mathit{TC-QC}}}}

\global\long\def\TCQC{\mbox{\ensuremath{\mathit{TC-QC}}}}
\global\long\def\LCQC{\mathit{TC-QC}}
\global\long\def\Cont{\mathit{Cont}}

\global\long\def\qcqc{\mathit{QC-QC}}

\global\long\def\tctc{\mathit{TC-TC}}

\renewcommand{\tcqc}{\mbox{\textit{TC-QC}}}
\renewcommand{\TCQC}{\mbox{\textit{TC-QC}}}
\renewcommand{\LCQC}{\mbox{\textit{TC-QC}}}
\renewcommand{\tctc}{\mbox{\textit{TC-TC}}}
\renewcommand{\LCLC}{\mbox{\textit{TC-TC}}}
\renewcommand{\qcqc}{\mbox{\textit{QC-QC}}}

\global\long\def\wrt{wrt.\ }

\global\long\def\sigmaptwo{\Sigma_{2}^{P}}

\global\long\def\Q{\mathcal{Q}}

\global\long\def\NP{\mathrm{NP}}
\global\long\def\coNP{\mathrm{coNP}}
\global\long\def\CONP{\coNP}
\global\long\def\PTIME{\mathrm{PTIME}}

\global\long\def\Count{\mathrm{COUNT}}
\global\long\def\Sum{\mathrm{SUM}}
\global\long\def\Max{\mathrm{MAX}}
\global\long\def\Min{\mathrm{MIN}}

\global\long\def\tcop{T_{\C}}

\global\long\def\bag#1{\{\!\{#1\}\!\}}

\global\long\def\A{\mathcal{A}}

\global\long\def\leftouterjoin{=\!\!\bowtie}

\global\long\def\proofr{"\negmedspace\!\Rightarrow\negmedspace\!"\!:}

\global\long\def\proofl{"\negmedspace\!\Leftarrow\negmedspace\!"\!:}

\global\long\def\query#1#2{#1{:\!-}\,#2}

\global\long\def\FILTER{\mathrm{FILTER}}

\global\long\def\UNION{\mathrm{UNION}}

\global\long\def\complstern#1{\complstar{#1}*}

\pagenumbering{arabic}

\cleardoublepage{}

$\ \ $

\cleardoublepage{}

\begin{center}
\begin{minipage}[t]{0.87\columnwidth}%
$\ $

$\ $

$\ $

\begin{center}
{\large{}Abstract}
\par\end{center}{\large \par}

$\ $

Knowledge about data completeness is essentially in data-supported
decision making. In this thesis we present a framework for metadata-based
assessment of database completeness. We discuss how to express information
about data completeness and how to use such information to draw conclusions
about the completeness of query answers. In particular, we introduce
formalisms for stating completeness for parts of relational databases.
We then present techniques for drawing inferences between such statements
and statements about the completeness of query answers, and show how
the techniques can be extended to databases that contain null values.
We show that the framework for relational databases can be transferred
to RDF data, and that a similar framework can also be applied to spatial
data. We also discuss how completeness information can be verified
over processes, and introduce a data-aware process model that allows
this verification.%
\end{minipage}
\par\end{center}

$\ \ $

\cleardoublepage{}

\chapter*{Publication Overview}

\section*{Conference Publications}
\begin{itemize}
\item Simon Razniewski and Werner Nutt. Adding completeness information
to query answers over spatial databases. International Conference
on Advances in Geographic Information Systems (SIGSPATIAL), 2014.
\item Simon Razniewski, Marco Montali, and Werner Nutt. Verification of
query completeness over processes. International Conference on Business
Process Management (BPM), pages 155\textendash{}170, 2013. Acceptance
rate 14,4\%.
\item Fariz Darari, Werner Nutt, Giuseppe Pirrò, and Simon Razniewski. Completeness
statements about RDF data sources and their use for query answering.
International Semantic Web Conference (ISWC), pages 66\textendash{}83,
2013. Acceptance rate 21,5\%.
\item Simon Razniewski and Werner Nutt. Assessing the completeness of geographical
data (short paper). British National Conference on Databases (BNCOD),
2013. Acceptance rate 47,6\%. 
\item Werner Nutt and Simon Razniewski. Completeness of queries over SQL
databases. Conference on Information and Knowledge Management (CIKM),
pages 902\textendash{}911, 2012. Acceptance rate 13,4\%.
\item Simon Razniewski and Werner Nutt. Completeness of queries over incomplete
databases. International Conference on Very Large Databases (VLDB),
2011. Acceptance rate 18,1\%. 
\end{itemize}

\section*{Other Publications}
\begin{itemize}
\item Simon Razniewski, Marco Montali, and Werner Nutt. Verification of
query completeness over processes {[}extended version{]}. CoRR, abs/1306.1689,
2013.
\item Werner Nutt, Simon Razniewski, and Gil Vegliach. Incomplete databases:
Missing records and missing values. DQDI workshop at DASFAA, 2012.%

\item Simon Razniewski and Werner Nutt. Checking query completeness over
incomplete data. Workshop on Logic and Databases (LID) at ICDT/EDBT,
2011. 
\end{itemize}
\tableofcontents{}

\chapter{Introduction}

Decision processes in businesses and organizations are becoming more
and more data-driven. To draw decisions based on data, it is crucial
to know about the reliability of the data, in order to correctly assess
the trustworthiness of conclusions. A core aspect of this assessment
is completeness: If data is incomplete, one may wrongly believe that
certain facts do not hold, or wrongly believe that a derived characteristics
are valid, while in fact the present data does not represent the complete
data set, which may have different characteristics. With the advent
of in-memory database systems that merge the traditionally separated
transaction processing (OLTP) and decision support (OLAP), data quality
and data completeness assessment are also topics that require more
timely treatment than in the traditional setting, where transaction
data and data warehouses are separate modules.

This work is motivated by a collaboration with the school department
of the Province of Bolzano, which faces data completeness problems
when monitoring the status of the school system. The administration
runs a central database into which all schools should regularly submit
core data about pupil enrollments, teacher employment, budgets and
similar. However, as there are numerous schools in the province and
as there are various paths to submit data (database clients, Excel-sheets,
phone calls, ...), data for some schools is usually late or data about
specific topics is missing. For instance, when assigning teachers
to schools for the next school year, it is often the case that the
data about the upcoming enrollments is not yet complete for some schools.
In practice, decisions are then based on estimates, for instance using
figures from the previous year. Completeness information would greatly
help the decision makers to know which figures are reliable, and which
need further checks and/or estimates.

In this thesis, we discuss a framework for metadata-based data completeness
assessment. In particular, we present:
\begin{enumerate}
\item an investigation into reasoning about the completeness of query answers
including decision procedures and analyses of the complexities of
the problems,
\item an extension of completeness reasoning to geographical databases and
to RDF data, 
\item a formalization of data-aware processes and methods to extract completeness
statements from such process descriptions.
\end{enumerate}
The rest of this chapter is structured as follows. In Section \ref{sec:intro-dq-and-data-completeness}
we give a general introduction to the area of Data Quality and to
the problem of data completeness. In Section \ref{sec:intro-motivating-example}
we illustrate the problem of data completeness management with the
example of school data management. Section \ref{sec:intro-contribution}
summarizes the contributions in this thesis and in Section \ref{sec:intro-overview}
we explain the outline of this thesis.

\section{Data Quality and Data Completeness}

\label{sec:intro-dq-and-data-completeness}

Quality is a vague term, and this also transfers to data quality.
A general definition that most people concerned with data quality
could agree with is that data are of high quality ``if they are fit
for their intended uses in operations, decision making and planning''
\cite{juranDQhandbook}.

Data quality has been a problem since long. With the emergence of
electronic databases in the 1960s, creation and storage of larger
volumes of data has become easier, leading also to more potential
data quality problems. Since the very beginning, data quality has
been an issue in relational databases, e.g., keys were introduced
in order to avoid duplicates \cite{codd-relational-model}. As an
independent research area, data quality has gained prominence in the
1990s. Three areas of data quality have received particular attention: 
\begin{enumerate}
\item The first area is \emph{duplicate detection}, which is also referred
to as entity resolution, and which is one of the most important operations
within data cleansing \cite{record-linkage-old,record-linkage-current-state}.
It seems that this is the most common practical problem that nearly
any business that manages customer relations will run into. 
\item The second area are \emph{guidelines and methodologies for assessing
and improving data quality}, with a prominent one being the TDQM methodology
\cite{wang1996-data-quality-core-paper,wang_dq}. 
\item The third area are approaches for dealing with \emph{data quality
in data integration} settings, which are particularly concerned with
integration techniques \cite{lenzerini-data-integration} or methods
for identifying data sources that best satisfy certain information
needs \cite{naumann_completeness}.
\end{enumerate}
Since the very beginning, relational databases have been designed
so that they are able to store incomplete data~\cite{codd_null}.
The theoretical foundations for representing and querying incomplete
information were laid by Imielinski and Lipski~\cite{imielinski_lipski_representation_systems}
who captured earlier work on \emph{Codd-,} \emph{c-} and \emph{v-tables}
with their conditional tables and introduced the notion of representation
system. Later work on incomplete information has focused on the concepts
of certain and possible answers, which formalize the facts that certainly
hold and that possibly hold over incomplete data \cite{fagin_data_exchange,lenzerini-data-integration,abiteboul1991representation}.
Still, most work on incompleteness focuses on querying incomplete
data, not on the assessment of the completeness. A possible reason
is that unlike consistency, completeness can hardly be checked by
looking at the data itself. If one does not have another complete
data source to compare with, then except for missing values, incompleteness
is not visible, as one cannot see what is not present. As well, incompleteness
can only be fixed if one has a more complete data source at hand that
can be used, which is usually not the case as then one could directly
use that more complete data source.

In turn, if metadata about completeness is present, an assessment
of the completeness of a data source is possible. As queries are the
common way to use data, we investigate in particular how such metadata
can be used to annotate query answers with completeness information.

In difference to data cleansing, we do not aim to improve data quality,
but instead aim to give a usage-specific information about data quality.
In difference to the guidelines and methodologies, we do not give
hints on how to improve data quality, but instead focus on the algorithmic
question of how to logically reason about completeness information.
In contrast to the approaches in the area of data integration, we
do not investigate source selection optimization or query semantics
over incomplete data.

There has been previous work on metadata-based completeness assessment
of relational databases. A first approach is by Motro \cite{motro_integrity},
who used information about complete query answers to assess the completeness
of other query answers. Later on, Halevy, introduced the idea of using
statements about the completeness of parts of a database to assess
the completeness of query answers \cite{levy_completeness}. In both
works, the problem of deciding whether a query answer is complete
based on completeness metadata could only be answered in a few trivial
cases. 

In the next section, we see a motivating story for this research.

\section{Motivation}

\label{sec:intro-motivating-example}

Consider the school district administrator Alice. Her job in the administration
is to monitor the impacts of new teaching methodologies, special integration
programs and socioeconomic situations on learning outcomes of students.

As last year a new, more interactive teaching methodology was introduced
for Math courses, Alice is interested to see whether that shows any
impact on the performance on the students. So, two weeks after the
end of the school year, she uses her cockpit software to find out
how many pupils have the grade A in math. 

The results show that at high schools, the number changed insignificantly
by +0.3\%, while at middle schools the tool reports a drop of 37\%
compared to the last year.

Alice is shocked about this figure, and quickly calls her assistant
Frank to investigate this drop.

Frank calls several middle schools and questions them about the performance
of their students in Math. All schools that he calls say that the
Math results of their students are as usual.

Confused from hearing that the schools report no problems, Frank suspects
that something must be wrong with the cockpit software. He therefore
sends an email to Tom, the database administrator.

Tom's answer is immediate:

\medskip{}

\texttt{$\ \ $Dude, forget those figures, we don't have the data
yet.\vspace{-2bp}
}

\texttt{$\ \ $-tom}

\medskip{}
As Frank tells this to Alice, she is relieved to hear that the new
teaching methodology is not likely to have wrecked the Math performance.
Nevertheless she is upset to not know which data in the cockpit she
can actually believe in and which not. Maybe the brilliant results
of last year's sport campaign (-80\% overweight students) were actually
also due to missing data?

Alice orders the IT department to find a solution for telling her
which numbers in the cockpit are reliable and which not.

A week later, at a focus group meeting organized by Tom, all participants
quickly agree that it is no problem to know which data in the database
is complete. They just have to keep track of the batches of data that
the schools submit. However, how can they turn this information into
something that Alice can interpret? They decide that they need some
kind of reasoner, which attaches to each number in Alice's cockpit
a green/red flag telling her whether the number is reliable. Developing
this reasoner becomes Tom's summer project (Chapter \ref{chap:general-reasoning}).

At the end of the summer break, the reasoner seems successfully implemented.
However just during the presentation to Alice, the reasoner crashes
with the following error message:

\medskip{}

\texttt{$\ \ $java.lang.NullpointerException(\textquotedbl{}Grade
is null\textquotedbl{})}

\medskip{}

As it turns out, a null value for a grade caused the reasoner to crash.
Thus back to coding, Tom gets stuck when thinking of whether to treat
such null values as incomplete attributes or as attributes that have
no value. As it turns out after consultations with the administration,
both cases happen: Some courses are just generally ungraded, while
in other cases the grade may not yet be decided. As the reasoner has
to know which case applies, Tom finds himself changing the database
schema to allow a disambiguation of the meaning of null values (Chapter
\ref{chap:nulls}).

In his free time, Tom is also a member of the OpenStreetMap project
for creating a free open map of the world. During some pub meeting
with other members of OpenStreetMap he mentions his work on database
completeness. The others get curious. Don't they have similar problems
when trying to track completeness information in OpenStreetMap? Tom
therefore invents reasoning methods for geographical data (Chapter
\ref{chap:osm}).

In the meantime, Alice is very satisfied with the new green and red
flags in her cockpit software. She has a chat about this with some
colleagues of the provincial administration, which are involved in
the ongoing data publishing projects as part of the Open Government
initiative in the province. They consult again Tom, who adapts his
reasoner to deal also with the RDF data format that is used for data
publishing, and the SPARQL query language (Chapter \ref{chap:lod}).

In their efforts to standardize processes at schools, the administration
introduces a workflow engine. It now becomes a question how information
about of the states of the workflows of the different schools can
be utilized to assess query completeness. Thus, they investigate how
business process state information can be used to automatically extract
information about completeness (Chapter \ref{chap:bpm}).

\section{Contribution}

\label{sec:intro-contribution}

The contributions of this thesis are threefold:

First, we introduce the reasoning problems of TC-TC entailment, TC-QC
entailment and QC-QC entailment and show that most variants of these
problems can be reduced to the well-studied problem of query containment,
thus enabling implementations that can make use of a broad set of
existing solutions.

Second, we show that completeness reasoning can also be done over
RDF data or over geographical data, and that the additional challenges
in this settings are manageable.

Third, we show that in settings where data is generated by formalized
and accessible processes, instead of just assuming that given completeness
statements are correct, one instead can verify the completeness of
query answers by looking at the status of the processes.

\section{Structure}

\label{sec:intro-overview}

This thesis is structured as follows:

In Chapter \ref{chap:prelims}, we introduce relational databases,
queries over such databases and formalisms for expressing completeness.
In Chapter \ref{chap:general-reasoning}, we introduce the core reasoning
problems and discuss their complexity. In Chapter \ref{chap:nulls},
we extend the core framework by allowing null values in databases.
In Chapter \ref{chap:osm}, we discuss completeness reasoning over
geographical databases. In Chapter \ref{chap:lod}, we discuss completeness
reasoning over RDF data. In Chapter \ref{chap:bpm}, we show how completeness
statements can be verified over data-centric business processes. In
Chapter \ref{chap:discussion}, we discuss implications of the presented
results, possible limitation, and future directions.

\chapter{Preliminaries}

\label{chap:prelims}

In this chapter we discuss concepts and notation that are essential
for the subsequent content. In Section \ref{sec:running-example},
we introduce the running example used throughout this thesis. We introduce
relational databases and their logical formalization in Section \ref{sec:prelim:relational-databases}.
In Section \ref{sec:prelim:queries}, we formalize queries over relational
databases, focusing on the positive fragment of SQL. In Section \ref{sec:prelim:incomplete-databases}
we introduce the model for incompleteness of databases, and in Sections
\ref{sec:prelim:query-completeness} and \ref{sec:prelim:table-completeness}
two important kinds of completeness statements about incomplete databases,
namely table completeness and query completeness statements. In Section
\ref{sec:prelims:query-containment}, we recall the problem of query
containment, onto which many later problems will be reduced, and review
its complexity.

The concepts presented in this chapter were already known in the literature,
though our presentation may be different. On the complexity of query
containment, we present three new hardness results.

\section{Running Example}

\label{sec:running-example}

For the examples throughout this thesis we consider a database about
schools. We assume that this database consists of the following\global\long\def\custstyle#1{\mathit{#1}}
 tables:\global\long\def\student{\custstyle{student}}
\global\long\def\name{\custstyle{name}}
\global\long\def\school{\custstyle{school}}
\global\long\def\class{\custstyle{class}}

\global\long\def\person{\custstyle{person}}
\global\long\def\gender{\custstyle{gender}}
\global\long\def\livesIn{\custstyle{livesIn}}
\global\long\def\town{\custstyle{town}}
\global\long\def\result{\custstyle{result}}
\global\long\def\grade{\custstyle{grade}}
\global\long\def\request{\custstyle{request}}
\global\long\def\subject{\custstyle{subject}}
\global\long\def\science{\custstyle{science}}
\global\long\def\code{\custstyle{code}}
\global\long\def\formteacher{\custstyle{formTeacher}}
\global\long\def\profile{\custstyle{profile}}

\begin{itemize}
\item $\student(\name,\class,\school)$
\item $\person(\name,\gender)$
\item $\livesin(\name,\town)$
\item $\class(\school,\code,\formteacher,\profile)$
\item $\result(\name,\subject,\grade)$
\item $\request(\name,\school)$
\end{itemize}
As this is just a toy example, we assume that persons are uniquely
identified by their name (in practice one would assign unique IDs
or use nearly unique combinations such as birth data and birth place).
The $\student$ table stores for each student the class and the school
that he/she is attending. The $\person$ table stores for persons
such as students and teachers their gender. The $\livesin$ table
stores for persons the town they are living in. The $\result$ table
stores for students the results they have obtained in different subjects.
The $\request$ table stores enrollment requests of current or upcoming
students at schools.\global\long\def\nameone{\mathit{John}}
\global\long\def\nametwo{\mathit{Mary}}
\global\long\def\namethree{\mathit{Bob}}
\global\long\def\schoolone{\mathit{HoferSchool}}
\global\long\def\schooltwo{\mathit{DaVinci}}

\section{Relational Databases}

\label{sec:prelim:relational-databases}

\emph{Relational databases} are a very widely used technology for
storing and managing structured data. The formal background of relational
databases is the relational data model. A \emph{database schema} consists
of a set of relations, where each relation consists of a relation
name and a set of attributes. A relation usually represents either
an entity type or a logical relation between entity types.

To model relational databases, we assume a set of relation symbols
$\Sigma$, each with a fixed arity. We call $\Sigma$ the \emph{signature}
or the database schema. We also assume a dense ordered domain of constants
$\dom$, that is, a domain like the rational numbers or like the set
of possible strings over some alphabet.
\begin{defn}
Given a fixed database schema $\Sigma$, a \emph{database instance}
$D$ is a finite set of ground atoms over $\dom$ with relation symbols
from $\Sigma$. 
\end{defn}
For a relation symbol $R\in\Sigma$ we write $R(D)$ to denote the
interpretation of $R$ in $D$, that is, the set of atoms in $D$
with relation symbol~$R$.
\begin{example}
\label{ex-relational-database}Consider that $\nameone$ is male and
a student in class 3a, $\nametwo$ is female and a student in class
5c, and $\namethree$ is male. One of the possible ways to store this
information would be to use two database tables, \emph{person} with
the attributes \emph{name} and \emph{gender,} and \emph{student} with
the attributes \emph{name}, \emph{class} and \emph{school}, as shown
in Figure \ref{table:database-representation-of-database}. Then this
database $D_{\school}$ would contain the following set of facts:
\begin{eqnarray*}
 & \{\!\!\!\!\!\! & \student(\nameone,3a,\schoolone),\student(\nametwo,5c,\schoolone),\\
 &  & \person(\namethree,\it{male}),\person(\nametwo,\it{female}),\person(\namethree,\it{male})\}
\end{eqnarray*}
\begin{table}
\centering{}%
\begin{tabular}{|ccc|>{\centering}p{1cm}|cc|}
\cline{1-3} \cline{5-6} 
\multicolumn{3}{|c|}{Student} &  & \multicolumn{2}{c|}{Person}\tabularnewline
name & class & school &  & name & gender\tabularnewline
\cline{1-3} \cline{5-6} 
$\nameone$ & \emph{3a} & $\schoolone$ &  & $\nameone$ & \emph{male}\tabularnewline
$\nametwo$ & \emph{5c} & $\schoolone$ &  & $\nametwo$ & \emph{female}\tabularnewline
\cline{1-3} 
\multicolumn{1}{c}{} &  & \multicolumn{1}{c}{} &  & $\namethree$ & \emph{male}\tabularnewline
\cline{5-6} 
\end{tabular}\caption{Database representation of the information from Example \ref{ex-relational-database}}
\label{table:database-representation-of-database}
\end{table}
There exist several extensions of the core relational model that cannot
be captured with the basic model described above. To mention here
are especially database constraints, data types, null values and temporal
data models:\end{example}
\begin{itemize}
\item Real-world databases almost always have keys and foreign keys defined,
which both are \emph{database constraints. }A discussion of database
constraints and their effects on completeness reasoning can be found
in \cite{razniewski:diplom:thesis:2010,savkovic:ICLP:2013}.
\item Attributes in relational databases are normally typed, which both
can make some techniques easier, because different types need not
be compared, or harder e.g., when reasoning about data types with
a nondense domain. We do not consider \emph{data types} in this work.
\item A special value for representing missing or nonexisting information,
the \emph{null value}, has, despite principled concerns about its
meaning, entered the standard relational model. A detailed analysis
of completeness reasoning with null values is contained in Chapter
\ref{chap:nulls}.
\item Facts in a database are often time-stamped with information about
their creation in the database or in the real-world or both, and there
exists a body of work on such temporal databases. Although some of
our results may be transferable, in this work, we do not consider
temporal databases.
\end{itemize}

\section{Database Queries}

\label{sec:prelim:queries}Queries are a structured way of accessing
data in databases. For relational databases, the SQL query language
is the standard. A basic SQL query specifies a set of attributes,
a set of referenced tables and selection conditions.
\begin{example}
\label{ex:conjunctive-query-sql}Consider again the database schema
from Example \ref{ex-relational-database}. An SQL query to find the
names of all male pupils can be written as:
\begin{eqnarray*}
 &  & \mbox{SELECT Student.name}\\
 &  & \mbox{FROM Student, Person}\\
 &  & \mbox{WHERE Student.name=Person.name AND}\\
 &  & \mbox{\ \ \ \ \ \ \ \ \ \ \ \ \ \ \ \ Person.gender='male';}
\end{eqnarray*}

While SQL queries may also contain negation and set difference, the
positive fragment of SQL, that is, the fragment without negation,
set difference, union and disjunction, has a correspondence in \emph{(positive)
conjunctive queries}. Conjunctive queries are a well established logical
query language. To formalize conjunctive queries, we need some definitions.
\end{example}
A \emph{condition} $G$ is a set of atoms using relations from $\Sigma$
and possibly the comparison predicates $=$, $<$ and $\leq$. As
common, we write a condition as a sequence of atoms, separated by
commas. 

A condition is \emph{safe} if each of its variables occurs in a relational
atom. A \emph{term} is either a constant or a variable.
\begin{defn}
[Conjunctive Query]A safe \emph{conjunctive query} is an expression
of the form $\query{Q(\tpl t_{1})}{B(\tpl t_{1},\tpl t_{2})}$, where
$B$ is a safe condition, and $\tpl t_{1}$ and $\tpl t_{2}$ are
vectors of terms such that every variable in $\tpl t_{1}$ also occurs
in some relational atom in $B$, or is equal to some constant. 
\end{defn}
We only consider safe queries and therefore omit this qualification
in the future. We often refer to the entire query by the symbol $Q$.
We call $Q(\tpl t_{1})$ the \emph{head}, $B$ the \emph{body}, the
variables in $\tpl t_{1}$ the \emph{distinguished variables}, and
the variables in $\tpl t_{2}$ the \emph{nondistinguished variables}
of $Q$. We generically use the symbol $L$ for the subcondition of
$B$ containing the relational atoms and $M$ for the subcondition
containing the comparisons.

A conjunctive query is called \emph{projection free}, if $\tpl t_{2}$
contains no variables. A conjunctive query is called \emph{boolean},
if $\tpl t_{1}$ contains no variables.
\begin{rem}
[Notation] For simplicity, in some following results we will use
conjunctive queries whose head contains only variables. We will write
such queries as $\query{Q(\tpl x)}{B(\tpl x,\tpl y)}$, where $\tpl y$
are the nondistinguished variables of $Q$. Any queries with constants
in the head can be transformed into a query with only variables in
the head, by adding equality atoms to the body. Therefore, this does
not introduce loss of generality.
\end{rem}
A conjunctive query is linear, if it contains every relation symbol
at most once. A conjunctive query is relational, if it does not contain
arithmetic comparisons besides $"="$.

\paragraph{Classes of Conjunctive Queries}

In the following, we will focus on four specific classes of conjunctive
queries:
\begin{enumerate}
\item linear relational queries ($\LRQ$): conjunctive queries without repeated
relation symbols and without comparisons, 
\item relational queries ($\RQ$): conjunctive queries without comparisons, 
\item linear conjunctive queries ($\LCQ$): conjunctive queries without
repeated relation symbols, 
\item conjunctive queries ($\CQ$). 
\end{enumerate}
Classes 1-3 as subclasses of class 4 are interesting, because they
capture special cases of conjunctive queries for which, as we will
show later, reasoning can be computationally easier.

For a query $\query{Q(\tpl x)}{B(\tpl x,\tpl y)}$, a \emph{valuation}
is a mapping from $\{\tpl x\cup\tpl y\}$ into $\dom$. Conjunctive
queries can be evaluated under set or under bag semantics, respectively
returning a set or a bag of tuples as answer. A valuation $v$ \emph{satisfies}
a query $\query QB$ over a database $D$, if $vB\subseteq D$, that
is, if the ground atoms in $vB$ are in the database $D$.

The result of evaluating a query $Q$ under bag semantics over a database
instance $D$ is denoted as $Q(D)$, and is defined as the following
bag of tuples:
\[
Q(D)=\bag{v\tpl x\mid\mbox{\ensuremath{v\ }is a valuation that satisfies \ensuremath{B}over \ensuremath{D}}}
\]

that is, every $v\tpl x$ appears as often as there are different
valuations $v$ satisfying $B$ over $D$.

If the query is evaluated under set semantics, all duplicate elements
are removed from the result, thus, the query answer contains each
tuple at most once and hence is a set of tuples.

Where necessary, we will mark the distinction between bag or set semantics
by appropriate superscripts $\cdot^{s}$ or $\cdot^{b}$.
\begin{example}
Consider again the query from Example \ref{ex:conjunctive-query-sql},
that asks for the names of all male students by joining the student
and the person table. As a conjunctive query, this query would be
written as follows:

\[
\query{Q(g)}{\student(n,c,s),\person(n,\male)}
\]

If we evaluate this query over the database $D_{\school}$ defined
in Example \ref{ex-relational-database}, the only valuation $v$
for which it holds that $B(v\tpl x,v\tpl y)\subseteq D$ is the valuation
$\{n\rightarrow\nameone,c\rightarrow3a,s\rightarrow\schoolone\}$.
Thus, the answer to $Q^{s}(D_{\school})$ and $Q^{b}(D_{\school})$
is $\{\nameone\}$.

Consider also another query $\query{Q(g)}{\person(n,g)}$ that asks
for all the genders of persons. Under bag semantics, the result $Q_{\gender}^{b}(D_{\school})$
would be $\set{\male,\male,\female}$. Under set semantics, the multiplicities
of the fact $\male$ would collapse and hence the answer to $Q_{\gender}^{s}(D_{\school})$
would be $\set{\male,\female}$.\end{example}
\begin{rem}
[Freezing] In many technical results that follow we will evaluate
queries over atoms that include variables. A technique called freezing
has been used in the literature for that purpose, which uses a freeze
mapping to replaces variables in atom with fresh constants. Where
it is clear from the context which atoms are the frozen ones we will
not make the freeze mapping explicit but allow the evaluation of queries
directly over atoms that include variables.

Formally, we extend the definition of valuations such that they may
also map variables into variables. Then, a valuation $v$ satisfies
a query $\query QB$ over a set of atoms $\A$, if $vB\subseteq\A$.
\end{rem}
Two queries are \emph{equivalent} under bag or set semantics, if they
return the same result over all possible database instances. A query
is \emph{minimal} under set semantics, if no relational atom can be
removed from its body without leading to a non-equivalent query.

A query $Q_{1}$ is \emph{contained} under set semantics in a query
$Q_{2}$, if for all database instances it holds that the result of
$Q_{1}$ is a subset of the result of $Q_{2}$. All containment techniques
used in this thesis are for queries under set semantics, therefore,
whenever in the following we talk about query containment, we refer
to containment under set semantics. More details on query containment
are in Section \ref{sec:prelims:query-containment}.

Conjunctive queries can also be extended to contain aggregate functions
such as COUNT, SUM, MIN or MAX, for which we discuss completeness
reasoning in Section \ref{sec:extensions:aggregate-queries}.

\section{Incomplete Databases}

\label{sec:prelim:incomplete-databases}

A core concept for the following theory is the partially complete
database or \emph{incomplete database}. The concept of incomplete
databases was first introduced by Motro in \cite{motro_integrity}.

Incompleteness needs a reference: If an available database is considered
to be incomplete, then, at least conceptually, some complete reference
must exist. Usually, the complete reference is the state of the real
world, of which available databases capture only parts and may therefore
be incomplete.

We model an incomplete database as a pair of database instances: one
instance that describes the complete state, and another instance that
describes the actual, possibly incomplete state. 
\begin{defn}
An \emph{incomplete database} is a pair $\pdb=(\di,\da)$ of two database
instances $\di$ and $\da$ such that $\da\incl\di$. 
\end{defn}
Following the notation introduced by Levy \cite{levy_completeness},
we call $\di$ the \emph{ideal} database, and $\da$ the \emph{available}
database. The requirement that $\da$ is included in $\di$ implies
that all facts in the available database are correct wrt.\ the ideal
database, however, some facts from the ideal database may be missing
in the available database.
\begin{example}
Consider a partial database $\pdb_{S}=(\di_{S},\da_{S})$ for a school
with two students, $\nameone$ and $\nametwo$, and one teacher, $\namethree$,
as follows:
\begin{align*}
\di_{S}=\  & \{\textit{student(\ensuremath{\nameone}, 3a, \ensuremath{\schoolone})},\textit{student(\ensuremath{\nametwo}, 5c, \ensuremath{\schoolone})},\\[0.7ex]
 & \ \textit{\ person(\ensuremath{\nameone}, male)},\textit{person(\ensuremath{\nametwo}, female)},\\[0.7ex]
 & \textit{\ \ person(\ensuremath{\namethree}, male)}\,\}\\
\da_{S}=\  & \di_{S}\setminus\set{\textit{person(\ensuremath{\namethree}, male)},\textit{student(\ensuremath{\nametwo}, 5c, \ensuremath{\schoolone})}},
\end{align*}
 that is, the available database misses the facts that $\nametwo$
is a student and that $\namethree$ is a person.
\end{example}
In the next two sections we define statements to express that parts
of the information in $\da$ are complete with regard to the ideal
database~ $\di$. We distinguish query completeness and \localComp
statements.

\section{Query Completeness}

\label{sec:prelim:query-completeness}

Because an available database may miss information wrt.\ the ideal
database, it is of interest to know whether a query over the available
database still gives the same answer as what holds in the ideal database.
Query completeness statements allow to express this fact:
\begin{defn}
[Query Completeness]Let $Q$ be a query. Then $\compl Q$ is a query
completeness statement.
\end{defn}
Query completeness statements refer to incomplete databases: 
\begin{defn}
A query completeness (QC) statement $\Compl Q$ for a query $Q$ is
\emph{satisfied} by an incomplete database $\pdb$, denoted as $\pdb\models\Compl Q$,
if $Q(\av D)=Q(\id D)$.
\end{defn}
Intuitively, a query completeness statement is satisfied if the available
database is complete enough to answer the query in the same way as
the ideal database would do. In the following chapters, query completeness
will be the key property for which we study satisfaction.
\begin{example}
Consider the above defined incomplete database $D_{S}$ and the query
\[
\query{Q_{1}(n)}{\textit{student}(n,c,s),\textit{person}(n,\male)},
\]
 asking for all male students. Over both, the available database $\da_{S}$
and the ideal database $\di_{S}$, this query returns exactly $\nameone$.
Thus, $D_{S}$ satisfies the query completeness statement for $Q_{1}$,
that is, 
\[
\pdb_{S}\models\Compl{Q_{1}}.\qquad
\]
\end{example}
\begin{rem}
[Terminology]In contrast to the terminology used by Motro \cite{motro_integrity},
our definition of completeness not only requires that the answer over
the ideal database is contained in the available one, also the converse
and thus the equivalence of the query answers. In the work of Motro,
the equivalence property was called query integrity, and consisted
of the query completeness and the symmetric property of query correctness. 

In our work, there is no need to separate completeness and integrity:
as we do not consider incorrect but only incomplete databases, and
as we consider only positive queries, the property of query correctness
always holds, and hence any positive query that satisfies the property
of query completeness in Motro's sense also satisfies the property
of query integrity in Motro's sense.
\end{rem}

\section{Table Completeness}

\label{sec:prelim:table-completeness}The second important statement
for talking about completeness are table completeness (TC) statements.
A \localComp statement allows one to say that a certain part of a
relation is complete, without requiring the completeness of other
parts of the database. Table completeness statements were first introduced
by Levy in \cite{levy_completeness}, where they were called local
completeness statements. 

A table completeness statement has two components, a relation $R$
and a condition $G$. Intuitively, it says that all tuples of the
ideal relation~ $R$ that satisfy the condition $G$ in the ideal
database are also present in the available relation $R$.
\begin{defn}
[Table Completeness] Let $\tpl t$ be a vector of terms, $R(\tpl t)$
be an $R$-atom and let $G$ be a condition such that $R(\tpl t),G$
is safe. Then $\Compl{R(\tpl t);G}$ is a \emph{\localComp statement}. 
\end{defn}
Observe that $G$ can contain relational and built-in atoms and that
we do not make any safety assumptions about $G$ alone.

Each table completeness statement has an \emph{associated query},
which is defined as $\query{Q_{R(\tpl t);G}(\tpl t)}{R(\tpl t),G}$.
We often refer to $R$ as the head of the statement and $G$ as the
condition.
\begin{defn}
\label{def:satisfaction-of-tc} Let $C=\Compl{R(\tpl t);G}$ be a
table completeness statement and $\pdb=(\id D,\av D)$ be an incomplete
database. Then $C$ is satisfied over $\D$, written $\pdb\models\Compl{R(\tpl t);G}$,
if 
\[
Q_{R(\tpl t);G}(\id D)\incl R(\av D).
\]

\end{defn}
That is, the statement is satisfied if all $R$-facts that satisfy
the condition $G$ over the ideal database are also contained in the
available database.

The ideal database instance $\id D$ is used to determine those tuples
in the ideal version $R(\di)$ that satisfy $G$. Then, for satisfaction
of the completeness statement, all these facts have to be present
also in the available version $R(\da)$. In the following, we will
denote a TC statement generically as $C$ and refer to the associated
query simply as~$Q_{C}$.

The semantics of TC statements can also be expressed using a rule
notation like the one that is used for instance for tuple-generating
dependencies (TGDs) (see~\cite{fagin_data_exchange}). As a preparation,
we introduce two copies of our signature $\SIG$, which we denote
as $\SIGi$ and $\SIGa$. The first contains a relation symbol $\id R$
for every $R\in\SIG$ and the second contains a symbol $\av R$. Now,
every incomplete database $(\di,\da)$ can naturally be seen as a
$\SIGi\cup\SIGa$-instance. We extend this notation also to conditions
$G$. By replacing every occurrence of a symbol $R$ by $\id R$ (resp.\ $\av R$),
we obtain $\id G$ (resp.\ $\av G$) from $G$. Similarly, we define
$\id Q$ and $\av Q$ for a query $Q$. With this notation, $(\di,\da)\models\compl Q$
iff $\id Q(\di)=\av Q(\da)$. Now, we can associate to each statement
$C=\compl{R(\tpl t);G}$, a corresponding TGD $\rho_{C}$ as 
\[
\rho_{C}\col\id R(\tpl t),\id G\rightarrow\av R(\tpl t)
\]
from the schema $\id{\Sigma}$ to the schema $\av{\Sigma}$. Clearly,
for every TC statement $C$, an incomplete database satisfies $C$
in the sense defined above if and only if it satisfies the rule $\Rule C$
in the classical sense of rule satisfaction.
\begin{example}
In the incomplete database $\pdb_{S}$ defined above, we can observe
that in the available relation \textit{person}, the teacher \textit{$\namethree$}
is missing, while all students are present. Thus, \textit{person}
is complete for all students. The available relation \textit{student}
contains \textit{Hans}, who is the only male student. Thus, \textit{student}
is complete for all male persons. Formally, these two observations
can be written as \localComp statements: 
\begin{align*}
 & C_{1}=\Compl{\textit{person}(n,g);\textit{student}(n,c,s)},\\
 & C_{2}=\Compl{\textit{student}(n,c,s);\textit{person}(n,\textit{male})},
\end{align*}
 which, as seen, are satisfied by the incomplete database $\pdb_{S}$.
The TGD $\rho_{C_{2}}$ corresponding to the statement $C_{2}$ would
be 
\[
\id{\student}(n,c,s),\id{\person}(n,\male)\longrightarrow\av{\student}(n,c,s).
\]
\end{example}
\begin{rem}
[Notation]Analogous to conjunctive queries, without loss of generality,
TC statements that contain constants in their head can be rewritten
such that they contain no constants in their head, using additional
equality atoms. Thus, whenever in the following we assume that TC
statements have only variables in the head, this neither introduces
loss of generality.\end{rem}
\begin{example}
Consider the TC statement $\compl{\person(n,\male);\emptyset}$. Then
this statement is equivalent to the statement $\compl{\person(n,g);\allowbreak g=\male}$.
\end{example}
Table completeness cannot be expressed by query completeness, because
the latter requires completeness of the relevant parts of all the
tables that appear in the statement, while the former only talks about
the completeness of a single table.
\begin{example}
As an illustration, consider the \localComp statement $C_{1}$ that
states that \textit{person} is complete for all students. The corresponding
query $Q_{C_{1}}$ that asks for all persons that are students is
\[
\query{Q_{C_{1}}(n,g)}{\textit{person}(n,g),\textit{student}(n,c,s)}.
\]
 Evaluating $Q_{C_{1}}$ over $\di_{S}$ gives the result $\set{\textit{\ensuremath{\nameone}},\,\nametwo}$.
However, evaluating it over $\da_{S}$ returns only $\set{\nameone}$.
Thus, $\pdb_{S}$ does not satisfy the completeness of the query $Q_{C_{1}}$
although it satisfies the \localComp statement $C_{1}$.

{} As we will discuss in Chapter \ref{chap:general-reasoning}, query
completeness for queries under bag semantics can be expressed using
table completeness, while under set semantics generally it cannot
be expressed.
\end{example}

\section{Complexity of Query Containment}

\label{sec:prelims:query-containment}

As many complexity results in this thesis will be found by reducing
query containment to completeness reasoning or vice versa, in this
section we review the problem in detail, list known complexity results
and complete the picture by giving hardness results for three asymmetric
containment problems.

Remember that a query $Q_{1}$ is contained (under set semantics)
in a query $Q_{2}$, if over all database instances $D$ the set of
answers $Q_{1}^{s}(D)$ is contained in the set of answers $Q_{2}^{s}(D)$.
\begin{defn}
Given conjunctive query languages $\L_{1}$ and $\L_{2}$, the problem
\begin{itemize}
\item $\Cont(\L_{1,}\L_{2})$ denotes the problem of deciding whether a
query from language $\L_{1}$ is contained in a query from language
$\L_{2}$,
\item $\UCont(\L_{1,}\L_{2})$ denotes the problem of deciding whether a
query from language $\L_{1}$ is contained in a \emph{union} of queries
from language $\L_{2}$.
\end{itemize}
\end{defn}
A commonly used technique for deciding containment between relational
queries is checking for the existence of homomorphisms. The NP-completeness
of query containment for relational conjunctive queries was first
shown by Chandra and Merlin in \cite{ChandraM77}. Results regarding
the $\piptwo$-completeness of containment with comparisons were first
published by van der Meyden in \cite{meyden-Complexity_querying_ordered_domains-pods}.

To complete the picture for the languages $\{\LRQ,\LCQ,\RQ,\CQ\}$
introduced in Section \ref{sec:prelim:queries}, we also need to consider
asymmetric containment problems, which have received little attention
in the literature so far. To the best of our knowledge, the results
that will follow have not been shown in the literature before \cite{Razniewski:Nutt-Compl:of:Queries-VLDB11}.

We show the hardness of $\ContU(\LRQ,\LCQ)$, $\Cont(\RQ,\LRQ)$ and
$\Cont(\RQ,\LCQ)$ by a reduction of \textit{(i)} 3-UNSAT, \textit{(ii)}
3-SAT, and \textit{(iii)} $\forall\exists$3-SAT, respectively.

\subsection{$\ContU(\LRQ,\LCQ)$ is $\CONP$-hard}

\label{coNP-hardness-ContU(LRQ,LCQ)}

Containment checking for a linear conjunctive query in a linear conjunctive
query is in PTIME, and the same holds also when considering a union
of linear conjunctive queries as container. Thus, the problem $\ContU(\LRQ,\LCQ)$
is the minimal combination that leads to a jump into coNP.

A well-known coNP-complete problem is 3-UNSAT. A 3-SAT formula is
unsatisfiable exactly if its negation is valid.

Let $\phi$ be a negated 3-SAT formula in disjunctive normal form
as follows: 
\[
\phi=\gamma_{1}\Or\ldots\Or\gamma_{k},
\]
 where each clause $\gamma_{i}$ is a conjunction of literals $l_{i1},l_{il2}$
and $l_{i3}$, and each literal is a positive or negated propositional
variable $p_{i1},p_{i2}$ or $p_{i3}$, respectively.

Using new relation symbols $C_{1}$ to $C_{k}$, we define queries
$Q$, $Q'_{1},\ldots,Q'_{k}$ as follows: 
\[
\query{Q()}{C_{1}(p_{11},p_{12},p_{13}),\ldots,C_{k}(p_{k1},p_{k2},p_{k3})},
\]
 
\[
\query{Q'_{i}()}{C_{i}(x_{1},x_{2},x_{3}),x_{1}\circ_{1}0,x_{2}\circ_{2}0,x_{3}\circ_{3}0},
\]
 where $\circ_{j}=``\geq"$ if $l_{ij}$ is a positive proposition
and $\circ_{j}=``<"$ otherwise.

Clearly, $Q$ is a linear relational query and the $Q'_{i}$ are linear
conjunctive queries.
\begin{lem}
Let $\phi$ be a propositional formula in disjunctive normal form
with exactly 3 literals per clause, and $Q$ and $Q_{1}$ to $Q_{k}$
be constructed as above. Then 
\[
\phi\mbox{ is valid \ \ iff \ \ }Q\subseteq\bigcup_{i=1..k}Q'_{i}.
\]
\end{lem}
\begin{proof}
Observe first that the comparisons in the $Q'_{i}$ correspond to
the disambiguation between positive and negated propositions, that
is, whenever a variable is interpreted as a constant greater or equal
zero, this corresponds to the truth value assignment $\true$, while
less zero corresponds to $\false$. 

$\proofr$ If $\phi$ is valid, then for every possible truth value
assignment of the propositional variables $p_{ij}$, one of the clauses
$C_{i}$ evaluates to true. Whenever $Q$ returns true over some database
instance, the query $Q'_{i}$ that corresponds to the clause $C_{i}$
that evaluates to true under that assignment, returns true as well. 

$\proofl$ If the containment holds, then for every instantiation
of $Q$ we find a $Q'_{i}$ that evaluates to true as well. This $Q'_{i}$
corresponds to the clause $C_{i}$ of $\phi$ that evaluates to true
under that variable assignment. 
\end{proof}
We therefore conclude the following hardness result.
\begin{cor}
The problem $\ContU(\LRQ,\LCQ)$ is $\CONP$-hard.\label{cor:hardness-lrq-lcq}
\end{cor}
The next problem is to find the step into NP.

\subsection{$\Cont(\RQ,\LRQ)$ is $\NP$-hard}

Deciding containment of a linear relational query in a linear relational
query is in PTIME, as there is only one possibility to find a homomorphism.
The same also holds when checking containment of a relational query
in a relational query. As we show below, the complexity becomes NP
as soon as repeated relation symbols may occur in the containee query.

\label{NP-hardness-Cont(RQ,LRQ)}

Let $\phi$ be a 3-SAT formula in conjunctive normal form as follows:
\[
\phi=\gamma_{1}\wedge\ldots\wedge\gamma_{k},
\]
 where each clause $\gamma_{i}$ is a conjunction of literals $l_{i1},l_{i2}$
and $l_{i3}$, and each literal is a positive or negated propositional
variable $p_{i1},p_{i2}$ or $p_{i3}$, respectively.

Using new relation symbols $F_{1}/3$ to $F_{k}/3$, we define queries
$Q$ and $Q'$ as follows: 
\[
\query{Q()}{F_{1}^{(7)},\ldots,F_{k}^{(7)}},
\]
 where $F_{i}^{(7)}$ is a conjunction of seven ground facts that
use the relation symbol $F_{i}$ and all those seven combinations
of $\set{0,1}$ as arguments, under which, when 0 is considered as
the truth value false and 1 as the truth value true, the clause $\gamma_{i}$
evaluates to true, and 
\[
\query{Q'()}{C_{1}(p_{11},p_{12},p_{13}),\ldots,C_{k}(p_{k1},p_{k2},p_{k3})}.
\]

Clearly, $Q$ is a relational query and $Q'$ a linear relational
query.
\begin{lem}
Let $\phi$ be a 3-SAT formula in conjunctive normal form and let
$Q$ and $Q'$ be constructed as shown above. Then 
\[
\phi\mbox{ is satisfiable \ \ iff \ \ }Q\subseteq Q'.
\]
\end{lem}
\begin{proof}
$\proofr$ If $\phi$ is satisfiable, there exists an assignment of
truth values to the propositions, such that each clause evaluates
to true. This assignment can be used to show that whenever $Q$ returns
a result, every $C_{i}$ in $Q'$ can be mapped to one ground instance
of that predicate in $Q$. 

$\proofl$ If the containment holds, $Q'$ must be satisfiable over
a database instance that contains only the ground facts in $Q$. The
mapping from the variables in $Q'$ to the constant $\set{0,1}$ gives
a satisfying assignment for the truth values of the propositions in
$\phi$. 
\end{proof}
We therefore conclude the following hardness result.
\begin{cor}
The problem $\Cont(\RQ,\LRQ)$ is $\NP$-hard.\label{cor:hardness-containment-rq-lrq}
\end{cor}
Next we will discuss when containment reasoning becomes $\piptwo$-hard.

\subsection{$\Cont(\RQ,\LCQ)$ is $\piptwo$-hard}

\label{PiP2-hardness-cont(RQ,LCQ)}

It is known that containment of conjunctive queries is $\piptwo$-hard
\cite{meyden-Complexity_querying_ordered_domains-pods}. As we show
below, it is indeed sufficient to have comparisons only in the container
query in order to obtain $\piptwo$-hardness.

Checking validity of a universally-quantified 3-SAT formula is a $\piptwo$-complete
problem. A universally-quantified 3-SAT formula $\phi$ is a formula
of the form 
\[
\forall\aufz xm\exists\aufz yn:\gamma_{1}\wedge\ldots\wedge\gamma_{k},
\]
 where each $\gamma_{i}$ is a disjunction of three literals over
propositions $p_{i1}$, $p_{i2}$ and $p_{i3}$, and $\set{\aufz xm}\union\set{\aufz yn}$
are propositions.

Let the $C_{i}$ be again ternary relations and let $R_{i}$ and $S_{i}$
be binary relations. We first define conjunctive conditions $G_{j}$
and $G'_{j}$ as follows: 
\begin{align*}
 & G_{j}=R_{j}(0,w_{j}),R_{j}(w_{j},1),S_{j}(w_{j},0),S_{j}(1,1),\\
 & G'_{j}=R_{j}(y_{j},z_{j}),S_{j}(z_{j},x_{j}),y_{j}\leq0,z_{j}>0.
\end{align*}

Now we define queries $Q$ and $Q'$ as follows: 
\[
\query{Q()}{G_{1},\ldots,G_{k},F_{1}^{(7)},\ldots,F_{m}^{(7)}},
\]
 where $F_{i}^{(7)}$ stands for the 7 ground instances of the predicate
$C_{i}$ over $\set{0,1}$, under which, when 0 is considered as the
truth value false and 1 as the truth value true, the clause $\gamma_{i}$
evaluates to true, and 
\[
\query{Q'()}{G'_{1},\ldots,G'_{m},C_{1}(p_{11},p_{12},p_{13}),\ldots,C_{k}(p_{k1},p_{k2},p_{k3})}.
\]

Clearly, $Q$ is a relational query and $Q'$ is a linear conjunctive
query.
\begin{lem}
Let $\phi$ be a universally quantified 3-SAT formula as shown above
and let $Q$ and $Q'$ be constructed as above. Then 
\[
\phi\mbox{ is valid \ \ iff \ \ }Q\subseteq Q'.
\]
\end{lem}
\begin{proof}
Observe first the function of the conditions $G$ and $G'$: Each
condition $G_{j}$ is contained in the condition $G'_{j}$, as whenever
a structure corresponding to $G_{j}$ is found in a database instance,
$G'_{j}$ is also found there. However, there is no homomorphism from
$G'_{j}$ to $G_{j}$ as $x_{j}$ will either be mapped to 0 or 1,
depending on the instantiation of $w_{j}$ (see also Figure~\ref{fig:reduction_pip2}).
\begin{figure}
\begin{centering}
\includegraphics[scale=0.42]{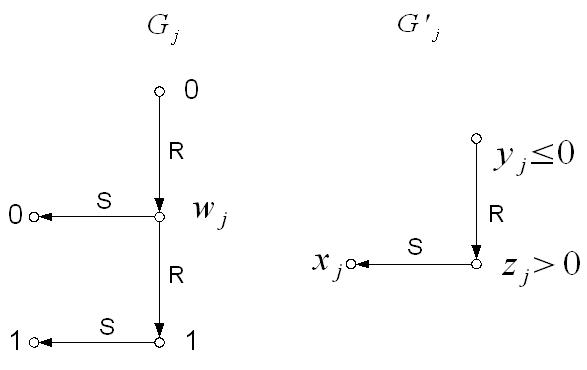}
\par\end{centering}

\caption{Structure of $G_{j}$ and $G_{j}'$. Depending on the value assigned
to $w_{j}$, the value of $x_{j}$ is either $0$ or $1$.}
\label{fig:reduction_pip2} 
\end{figure}

$\proofr$ If $\phi$ is valid, then for every possible assignment
of truth values to the universally quantified propositions, a satisfying
assignment for the existentially quantified ones exists.

Whenever a database instance $D$ satisfies $Q$, each condition $G_{j}$
must be satisfied there, and $w_{j}$ will have a concrete value,
that determines which value $x_{j}$ in $G'_{j}$ can take. As $\phi$
is valid, however, it does not matter which values the universally
quantified variables $x$ take, there always exists a satisfying assignment
for the other variables, such that each atom $C_{j}$ can be mapped
to one of the ground instances $F_{j}^{(7)}$ that are in $D$ since
$Q$ is satisfied over $D$. Then, $Q'$ will be satisfied over $D$
as well and hence $Q\subseteq Q'$ holds.

$\proofl$ If $Q$ is contained in $Q'$, for every database $D$
that instantiates $Q$, we find that $Q'$ is satisfied over it. Especially,
no matter whether we instantiate the $w_{j}$ by a positive or a negative
number, and hence whether the $x_{j}$ will be mapped to 0 or 1, there
exists an assignment for the existentially quantified variables such
that each $C_{j}$ is mapped to a ground instance from $F_{j}^{(7)}$.
This directly corresponds to the validity of $\phi$, where for every
possible assignment of truth values to the universally quantified
variables, a satisfying assignment for the existential quantified
variables exists.
\end{proof}
We therefore conclude the following hardness result.
\begin{cor}
The problem $\Cont(\RQ,\LCQ)$ is $\piptwo$-hard.\label{corr:hardness-cont-rq-lcq}
\end{cor}
We summarize the complexity for containment of unions of conjunctive
queries in Table \ref{table:complexity_querycontainment}. Regarding
the upper bounds for the cases $\UCont(\LCQ,\RQ)$, $\UCont(\LCQ,\RQ)$,
$\UCont(\LCQ,\CQ)$ and $\UCont(\LCQ,\LCQ)$, which imply all other
upper bounds, observe the following:
\begin{itemize}
\item $\UCont(\LCQ,\RQ)$ is in PTIME, because due to the linearity of the
containee query, there exists only one homomorphism that needs to
be checked.
\item $\UCont(\CQ,\RQ)$ is in NP because containment of a relational conjunctive
query in a positive relational query (an extension of relational conjunctive
queries that allows disjunction) is in NP (see \cite{Sagiv:Yannakakis-Containment-VLDB}),
and comparisons only for the containee query do not change the techniques,
besides a check for unsatisfiability of the containee query.
\item $\UCont(\LCQ,\CQ)$ is in coNP, because, in order to show non-containment,
it suffices to guess some valuation for the containee query such that
no homomorphism from the container exists.
\item $\UCont(\LCQ,\LCQ)$ is in $\piptwo$ as shown by van der Meyden \cite{meyden-Complexity_querying_ordered_domains-pods}.
\end{itemize}
\begin{table}[t]
\begin{centering}
{\scriptsize{}}%
\begin{tabular}{|>{\raggedright}p{1.4cm}c|cccc|}
\hline 
\multicolumn{2}{|c|}{} & \multicolumn{4}{c|}{{\scriptsize{}Containee Query Language $\L_{1}$}}\tabularnewline
\multicolumn{2}{|c|}{} & {\scriptsize{}LRQ } & {\scriptsize{}LCQ } & {\scriptsize{}RQ } & {\scriptsize{}CQ }\tabularnewline
\hline 
\multirow{4}{1.4cm}{{\scriptsize{}Container Query Language $\L_{2}$ }} & {\scriptsize{}LRQ } & {\scriptsize{}polynomial } & {\scriptsize{}polynomial } & {\scriptsize{}NP-complete } & {\scriptsize{}NP-complete }\tabularnewline
 & {\scriptsize{}RQ } & {\scriptsize{}polynomial } & {\scriptsize{}polynomial } & {\scriptsize{}NP-complete } & {\scriptsize{}NP-complete }\tabularnewline
 & {\scriptsize{}LCQ } & {\scriptsize{}coNP-complete } & {\scriptsize{}coNP-complete } & {\scriptsize{}$\piptwo$-complete } & {\scriptsize{}$\piptwo$-complete}\tabularnewline
 & {\scriptsize{}CQ } & {\scriptsize{}coNP-complete } & {\scriptsize{}coNP-complete } & {\scriptsize{}$\piptwo$-complete } & {\scriptsize{}$\piptwo$-complete}\tabularnewline
\hline 
\end{tabular}
\par\end{centering}{\scriptsize \par}

\caption{Complexity of checking containment of a query from language $\L_{1}$
in a union of queries from language $\L_{2}$. Observe the asymmetry
of the axes, as the step into coNP appears when allowing comparisons
in the container queries, while the step into NP appears when allowing
repeated relation symbols in the containee query.}

\label{table:complexity_querycontainment} 
\end{table}

\section{Entailment}

The next chapter will focus on entailment between completeness statements.
We therefore review the notion of entailment here:

A set of table completeness or query completeness statements $S_{1}$
entails a set $S_{2}$ of table completeness or query completeness
statements, written $S_{1}\models S_{2}$, if whenever the statements
in $S_{1}$ are satisfied, then also the statements in $S_{2}$ are
satisfied. Formally: 
\[
S_{1}\models S_{2}\ \ \ \ \ \mbox{iff for all partial databases \ensuremath{\D:}}\ \ \D\models S_{1}\ \mbox{implies}\ \D\models S_{2}
\]

To make a statement about a single incomplete database is hard, as
the content of the ideal database is usually hardly or not at all
accessible. Because of the universal quantification in this notion
of entailment, the content of the ideal database need not be accessed.
The entailment guarantees that no matter what the available and the
ideal database look like, as long as they satisfy $S_{1}$, they also
satisfy $S_{2}$. 
\begin{example}
Consider the table completeness statement $C=$\linebreak{}
 $\compl{\person(n,g);\emptyset}$, which states that the available
$\person$ table contains all facts from the ideal $\person$ table,
and consider the query\linebreak{}
 $\query{Q(x)}{\person(x,y)}$ asking for the names of all persons.
Clearly, $C$ entails $\compl Q$, because whenever all person tuples
are complete, then also the query asking for their names will be complete.
This entailment means that no matter how many person tuples are there
in the ideal database, as long as all of them are also in the available
database, the query $Q$ will return a complete answer.
\end{example}
In the next chapter we look into entailment reasoning between table
completeness and table completeness (TC-TC), table completeness and
query completeness (TC-QC) and query completeness and query completeness
(QC-QC) statements.

\chapter{Completeness Reasoning}

\label{chap:general-reasoning}

In the previous chapter we have introduced incomplete databases and
table and query completeness statements. In this chapter we focus
on the reasoning about the latter two. 

As in our view, table completeness statements are a natural way of
expressing that parts of a database are complete, and queries are
the common means to access data in a database, we will particularly
focus on the problem of \emph{entailment of query completeness by
table completeness} \emph{statements} (TC-QC entailment). 

\emph{Entailment of table completeness by table completeness} is useful
when managing sets of completeness statements, and in important cases
also for solving TC-QC entailment. 

\emph{Entailment of query completeness by query completeness} (QC-QC
entailment) plays a role when completeness guarantees are given in
form of query completeness statements, which may be the case for views
over databases.

The results in this chapter are as follows: For TC-QC entailment,
we develop decision procedures and assess the complexity of TC-QC
inferences depending on the languages of the TC and QC statements.
We show that for queries under bag semantics and for minimal queries
under set semantics, weakest preconditions for query completeness
can be expressed in terms of table completeness statements, which
allow to reduce TC-QC entailment to TC-TC entailment.

For the problem of TC-TC entailment, we show that it is equivalent
to query containment.

For QC-QC entailment, we show that the problem is decidable for queries
under bag semantics. For queries under set semantics, we give sufficient
conditions in terms of query determinacy.

For aggregate queries, we show that for the aggregate functions SUM
and COUNT, TC-QC has the same complexity as TC-QC for nonaggregate
queries under bag semantics. For the aggregate functions MIN and MAX,
we show that TC-QC has the same complexity as TC-QC for nonaggregate
queries under set semantics.

For reasoning wrt.\ a database instance, we show that TC-QC becomes
computationally harder than without an instance, while QC-QC surprisingly
becomes solvable, whereas without an instance, decidability is open.\medskip{}

The results on the equivalence between TC-TC entailment and query
containment (Section \ref{sec:general:TC-TC}), the upper bound for
TC-QC entailment for queries under bag semantics (Theorem \ref{theorem_characterizing_query_completeness})
and the combined complexity of TC-QC reasoning wrt.\ database instances
(Theorem \ref{thm:instancereasoning-tcqc}(i)) were already shown
in the Diplomarbeit (master thesis) of Razniewski \cite{razniewski:diplom:thesis:2010}.
Also Theorem \ref{thm:weakest-precond} was contained there, although
it was erroneously claimed to hold for conjunctive queries, while
so far it is only proven to hold for relational queries. As these
results are essential foundations for a complete picture of completeness
reasoning, we include them in this chapter. 

All results besides Section \ref{sec:general:QC-QC} and Theorem \ref{prop:hardness-tc-qc-bag-semantics}
were published at the VLDB 2011 conference \cite{Razniewski:Nutt-Compl:of:Queries-VLDB11}.

The results on QC-QC instance reasoning and on TC-QC instance reasoning
under bag semantics are unpublished.\medskip{}

This chapter is structured as follows: In Section \ref{sec:general:TC-TC},
we discuss TC-TC entailment and in particular its equivalence with
query containment. In Section \ref{sec:general:TC-QC} we discuss
TC-QC entailment and in Section \ref{sec:general:QC-QC} QC-QC entailment.
In Section \ref{sec:extensions:aggregate-queries}, we discuss completeness
reasoning for aggregate queries, in Section \ref{sec:reasoning_with_instances}
reasoning wrt. database instances, and in Section \ref{sec:general:intro}
we review related work on completeness entailment.

\section{Table Completeness Entailing Table Completeness}

\label{sec:general:TC-TC}

\LocalComp statements describe parts of relations, which are stated
to be complete. Therefore, one set of such statements entails another
statement if the part described by the latter is contained in the
parts described by the former. Therefore, as we will show, TC-TC entailment
checking can be done by checking containment of the parts described
by the statements, which in turn can be straightforwardly reduced
to query containment.
\begin{example}
Consider the TC statements $C_{1}$ and $C_{2}$, stating that the
\textit{person} table is complete for all persons or for all female
persons, respectively: 
\begin{align*}
 & C_{1}=\Compl{\textit{person}(n,g);\emptyset},\\
 & C_{2}=\Compl{\textit{person}(n,g);\, g=\textit{female}}.
\end{align*}
 Obviously, $C_{1}$ entails $C_{2}$, because having all persons
implies also having all female persons. To show that formally, consider
the associated queries $Q_{C_{1}}$ and $Q_{C_{2}}$, that describe
the parts that are stated to be complete, which thus ask for all persons
or for all female persons, respectively: 
\begin{align*}
 & \query{Q_{C_{1}}(n,g)}{\textit{person}(n,g)},\\
 & \query{Q_{C_{2}}(n,g)}{\textit{person}(n,g),g=\textit{female}}.
\end{align*}
 Again it is clear that $Q_{C_{2}}$ is contained in $Q_{C_{1}}$,
because retrievable female persons are always a subset of retrievable
persons. In summary, we can say that $C_{1}$ entails $C_{2}$ because
$Q_{C_{2}}$ is contained in $Q_{C_{1}}$. 
\end{example}
The example can be generalized to a linear time reduction under which
entailment of a TC statement by other TC statements is translated
into containment of a conjunctive query in a union of conjunctive
queries. Furthermore, one can also reduce containment of unions of
conjunctive queries to TC-TC entailment, as the next theorem states.
Recall the four classes of conjunctive queries introduced in Section
\ref{sec:prelim:queries}: linear relational queries ($\LRQ$), relational
queries ($\RQ$), linear conjunctive queries ($\LCQ$) and conjunctive
queries ($\CQ$).
\begin{thm}
\label{theo-equivalence:TC-TC:UCont} Let $\L_{1}$ and $\L_{2}$
be classes of conjunctive queries among $\{\LRQ,\RQ,\LCQ,\CQ\}$.
Then the problems of $\tctc(\L_{1},\L_{2}$) and \linebreak{}
$\UCont(\L_{2},\L_{1}$) can be reduced to each other in linear time. \end{thm}
\begin{proof}
\emph{Reducing TC-TC to $\UCont$:} Consider a TC-TC entailment problem
$"\set{C_{1},\ldots,C_{n}}\stackrel{?}{\models}\compl{C_{0}}"$, where
$C_{0}$ is a statement for a relation $R$. Since statements for
relations different from $R$ do not influence the entailment, we
assume that $C_{1}$ to $C_{n}$ are statements for $R$ as well.
Recall that for a TC statement $C=\compl{R(\tpl x);G}$, the query
$Q_{C}$ is defined as $\query{Q_{C}(\tpl x)}{R(\tpl x),G}$.

\medskip{}

Claim: $\set{C_{1},\ldots,C_{n}}\stackrel{}{\models}\compl{C_{0}}$
if and only if $Q_{C_{0}}\subseteq Q_{C_{1}}\cup\ldots\cup Q_{C_{n}}$

\medskip{}

$"\Rightarrow"$: Suppose the containment does not hold. Then, by
definition there exists a database $D$ and a tuple $\tpl c$ such
that $\tpl c$ is in $Q_{C_{0}}(D)$ but not in the answer of any
of the queries $Q_{C_{1}}$ to $Q_{C_{n}}$ over $D$. Thus, we can
construct an incomplete database $\D=(\di,\da)$ with $\di=D$ and
$\da=D\setminus\set{R(\tpl c)}$. Then, $\D$ satisfies $C_{1}$ to
$C_{n}$, because $R(\tpl c)$ is by assumption not in $Q_{C_{1}}(\di)\cup\cdots\cup Q_{C_{n}}(\di)$,
and $R(\tpl c)$ is the only difference between $\di$ and $\da$.
But $\D$ does not satisfy $C_{0}$, because $R(\tpl c)$ is in $Q_{C_{0}}(\di)$
and thus $\D$ proves that $C_{1}$ to $C_{n}$ do not entail $C_{0}$.

$"\Leftarrow"$: Suppose the entailment does not hold. Then, by definition
there exists an incomplete database $\D=(\di,\da)$ such that there
is a tuple $t$ in $Q_{C_{0}}(\di)$ which is not in $\da$. Since
$C_{1}$ to $C_{n}$ are satisfied over $\D$, by definition $t$
may not be in $Q_{C_{i}}(\di)$ for $Q_{C_{1}}$ to $Q_{C_{n}}$ and
hence $\di$ is a database that shows that $Q_{C_{0}}$ is not contained
in the union of the $Q_{C_{i}}$.

\emph{Reducing $\UCont$ to TC-TC:} A union containment problem has
the form $"Q_{0}\stackrel{?}{\subseteq}Q_{1}\cup\cdots\cup Q_{n}"$,
where the $Q_{i}$ are queries of the same arity, and it shall be
decided whether over all database instances the answer of $Q_{0}$
is a subset of the union of the answers of $Q_{1}$ to $Q_{n}$.

Consider now a containment problem $"Q_{0}\stackrel{?}{\subseteq}Q_{1}\cup\cdots\cup Q_{n}"$,
where each $Q_{i}$ has the form $\query{Q_{i}(\tpl x)}{B_{i}}$.
We construct a TC-TC entailment problem as follows: We introduce a
new relation symbol $H$ with the same arity as the $Q_{i}$, and
construct completeness statements $C_{0}$ to $C_{n}$ as $C_{i}=\compl{H(\tpl x);B_{i}}$.
Analogous to the reduction in the opposite direction, by contradiction
it is now straightforward that $Q_{0}$ is contained in $Q_{1}\cup\cdots\cup Q_{n}$
if and only if $C_{1}$ to $C_{n}$ entail $C_{0}$.
\end{proof}

\section{Table Completeness Entailing Query Completeness}

\label{sec:general:TC-QC}

\label{sec:TC-QC-entailment}

In this section we discuss the problem of TC-QC entailment and its
complexity. We first show that query completeness for queries under
bag semantics can be characterized by so-called canonical TC statements.
We then show that for TC-QC entailment for queries under set semantics,
for minimal relational queries TC-QC can also be reduced to TC-TC.
We then give a characterization for general TC-QC entailment for queries
under set semantics, that is, queries that are nonminimal or contain
comparisons.

\subsection{TC-QC Entailment for Queries under Bag Semantics}

\label{characterizing}

In this section we discuss whether and how query completeness can
be characterized in terms of \localComp. Suppose we want the answers
for a query $Q$ to be complete. An immediate question is which \localComp
conditions our database should satisfy so that we can guarantee the
completeness of $Q$.

To answer this question, we introduce canonical completeness statements
for a query. Intuitively, the canonical statements require completeness
of all parts of relations where tuples can contribute to answers of
the query. Consider a query $\query{Q(\tpl s)}{\dd An,M}$, with relational
atoms $A_{i}$ and comparisons $M$. The \emph{canonical completeness
statement} for the atom $A_{i}$ is the TC statement 
\[
C_{i}=\Compl{A_{i};A_{1},\ldots,A_{i-1},A_{i+1},\ldots,A_{n},M}.
\]
 We denote by $\C_{Q}=\set{\dd Cn}$ the set of all canonical completeness
statements for $Q$.
\begin{example}
Consider the query 
\[
\query{Q_{2}(n)}{\textit{student}(n,c,s),\textit{class}(s,c,f,\science),}
\]
 asking for the names of all students that are in a class with \textit{Science}
profile. Its canonical completeness statements are the \localComp
statements 
\begin{align*}
C_{1} & =\Compl{\textit{student}(n,c,s);\textit{class}(s,c,f,\science)}\\
C_{2} & =\Compl{\textit{class}(s,c,f,\science);\textit{student}(n,s,c)}.
\end{align*}
 As was shown in \cite{razniewski:diplom:thesis:2010}, query completeness
can equivalently be expressed by the canonical completeness statements
in certain cases.\end{example}
\begin{thm}
\label{theorem_characterizing_query_completeness} Let $Q$ be a conjunctive
query. Then for all incomplete databases~$\pdb$, 
\[
\pdb\models\complstar Q*\mbox{\quad iff\quad}\pdb\models\cplset_{Q},
\]
 holds for
\begin{enumerate}
\item $*=b$,
\item $*=s$, if $Q$ is a projection-free query.
\end{enumerate}
\end{thm}
\begin{proof}
\textit{(i)} \Onlyif 
Indirect proof: Suppose, one of the completeness assertions in $\cplset_{Q}$
does not hold over $\pdb$, for instance, assertion $C_{1}$ for atom
$A_{1}$. Suppose, $R_{1}$ is the relation symbol of $A_{1}$. Let
$C_{1}$ stand for the \LC statement 
$\Compl{A_{1};\, B_{1}}$ where $B_{1}=B\setminus\set{A_{1}}$ and
$B$ is the body of $Q$. Let $Q_{1}$ be the query associated to
$C_{1}$.

Then $Q_{1}(\id D)\not\subseteq R_{1}(\av D)$. Let $\tpl c$ be a
tuple that is in $Q_{1}(\id D)$, and therefore in $R_{1}(\id D)$,
but not in $R_{1}(\av D)$. By the fact that $Q_{1}$ has the same
body as $Q$, the valuation $\upsilon$ of $Q_{1}$ over $\id D$
that yields $\tpl c$ is also a satisfying valuation for $Q$ over
$\id D$. So we find one occurrence of some tuple $\tpl c'\in Q(\id D)$,
where $\tpl c'=v\tpl x_{1}$, with $\tpl x_{1}$ being the distinguished
variables of $Q$.

However, $\upsilon$ does not satisfy $Q$ over $\av D$ because $\tpl c$
is not in $R_{1}(\av D)$. By the monotonicity of conjunctive queries,
we cannot have another valuation yielding $\tpl c'$ over $\av D$
but not over $\id D$. Therefore, $Q(\av D)$ contains at least one
occurrence of $\tpl c'$ less than $Q(\id D)$, and hence $Q$ is
not complete over $D$.

\textit{(i)} \If Direct proof: We have to show that if $t$ is $n$
times in $Q(\di)$ then $\tpl c$ is also $n$ times in $Q(\da)$.

For every occurrence of $\tpl c$ in $Q(\di)$ we have a valuation
of the variables of $Q$ that is satisfying over $\di$. We show that
if a valuation is satisfying for $Q$ over $\di$, then it is also
satisfying for $Q$ over $\da$. A valuation $v$ for a conjunctive
condition $\concond$ is satisfying over a database instance if we
find all elements of the instantiation $v\concond$ in that instance.
If a valuation satisfies $Q$ over $\di$, then we will find all instantiated
atoms of $v\concond$ also in $\da$, because the canonical completeness
conditions hold in $D$ by assumption. Satisfaction of the canonical
completeness conditions requires that for every satisfying valuation
of $v$ of $Q$, for every atom $A$ in the body of $Q$, the instantiation
atom $vA$ is in $\da$. Therefore, each satisfying valuation for
$Q$ over $\di$ yielding a result tuple $\tpl c\in Q(\di)$ is also
a satisfying valuation over $\da$ and hence $Q$ is complete over
$\pdb$.

\textit{(ii)} Follows from \textit{(i)}. Since the query is complete
under bag semantics, it is also complete under set semantics, because
whenever two bags are equal, also the corresponding sets are equal.
\end{proof}
While the lower bounds were left open in \cite{razniewski:diplom:thesis:2010},
with the following theorem we show that $\UCont(\L_{1},\L_{2})$ can
also be reduced to $\tcqc$($\L_{2},\L_{1}$), both under bag or set
semantics.
\begin{thm}
\label{prop:hardness-tc-qc-bag-semantics}Let $\L_{1}$ and $\L_{2}$
be classes of conjunctive queries among $\{\LRQ,\RQ,\LCQ,\CQ\}$.
Then the problem of $\UCont(\L_{1},\L_{2})$ can be reduced to $\tcqc^{*}(\L_{2},\L_{1}$)
for $*\in\{s,b\}$ in linear time.\end{thm}
\begin{proof}
We show how the reduction works in principle. Consider a $\UCont$
problem $"Q_{0}\stackrel{?}{\models}Q_{1}\cup Q_{2}"$ for three queries,
each of the form $\query{Q_{i}(\tpl d_{i})}{B_{i}}$. We define a
set of TC statements $\C$ and a query $Q$ such that $\C\models\complstar Q*$
if and only if $Q_{0}\subseteq Q_{1}\union Q_{2}$.

Using a new relation symbol $S$ with the same arity as the $Q_{i}$,
we define the new query as $\query{Q(\tpl d_{0})}{S(\tpl d_{0}),B_{0}}$.
For every relation symbol $R$ in the signature $\Sigma$ of the $Q_{i}$
we introduce the statement $C_{R}=\Compl{R(\tpl x_{R});\true}$, where
$\tpl x_{R}$ is a vector of distinct variables. Furthermore, for
each of $Q_{i}$, $i=1,2$, we introduce the statement $C_{i}=\Compl{S(\tpl d_{i});B_{i}}$.
Let $\C=\set{C_{1},C_{2}}\union\set{C_{R}\mid R\in\Sigma}$. Then
it $\C$ entails $\complstar Q*$ if and only if $Q_{0}\subseteq Q_{1}\union Q_{2}$,
because on the one hand, the containment implies that any tuple needed
for the completeness of $Q$ is also constrained by the TC statements
in $\C$, and on the other hand, because, if the containment would
not hold, there would exist a database instance $D$ and a tuple $\tpl c$
such that $\tpl c$ would be in the answer of $Q_{0}$ over $D$ but
not in the answer of $Q_{1}$ to $Q_{n}$, and thus, an incomplete
database $(D,D\setminus\{S(\tpl d)\})$ would satisfy $\C$ but not
$\complstar Q*$, thus showing that the former does not entail the
latter.
\end{proof}
From the theorem above and Theorem \ref{theo-equivalence:TC-TC:UCont}
we conclude that $\UCont$ and $\tcqc$ for queries under bag semantics
can be reduced to each other and therefore have the same complexity.
We summarize the complexity results for $\tcqc$ entailment under
bag semantics in Table \ref{table:complexity:tc-qc-bag-semantics}.

\begin{table}
\begin{centering}
{\scriptsize{}}%
\begin{tabular}{|>{\raggedright}p{1.2cm}c|cccc|}
\hline 
\multicolumn{2}{|c|}{} & \multicolumn{4}{c|}{{\scriptsize{}Query Language}}\tabularnewline
\multicolumn{2}{|c|}{} & {\scriptsize{}LRQ } & {\scriptsize{}LCQ } & {\scriptsize{}RQ } & {\scriptsize{}CQ }\tabularnewline
\hline 
\multirow{4}{1.2cm}{{\scriptsize{}TC Statement Language }} & {\scriptsize{}LRQ } & {\scriptsize{}polynomial } & {\scriptsize{}polynomial } & {\scriptsize{}NP-complete } & {\scriptsize{}NP-complete }\tabularnewline
 & {\scriptsize{}RQ } & {\scriptsize{}polynomial } & {\scriptsize{}polynomial } & {\scriptsize{}NP-complete } & {\scriptsize{}NP-complete }\tabularnewline
 & {\scriptsize{}LCQ } & {\scriptsize{}coNP-complete } & {\scriptsize{}coNP-complete } & {\scriptsize{}$\piptwo$-complete } & {\scriptsize{}$\piptwo$-complete}\tabularnewline
 & {\scriptsize{}CQ } & {\scriptsize{}coNP-complete } & {\scriptsize{}coNP-complete } & {\scriptsize{}$\piptwo$-complete } & {\scriptsize{}$\piptwo$-complete}\tabularnewline
\hline 
\end{tabular}
\par\end{centering}{\scriptsize \par}

\caption{Complexity of deciding TC-QC entailment under bag semantics. The entries
are equivalent to those for query containment as shown in Table \ref{table:complexity_querycontainment}.}

\label{table:complexity:tc-qc-bag-semantics} 
\end{table}

\subsection{Characterizations of Query Completeness under Set Semantics}

In the previous section we have seen that for queries under bag semantics
and for queries under set semantics without projections, query completeness
can be characterized by \localComp. For queries under set semantics
with projections, this is not generally possible:
\begin{example}
Consider the query $\query{Q(c)}{\pupil(n,c,s)}$ asking for all the
classes of pupils. Its canonical completeness statements are
\[
\{\compl{\student(n,c,s);\true}\}
\]
Now, consider an incomplete database $\D=(\di,\da)$ with 
\begin{eqnarray*}
 &  & \di=\{\student(\mathit{John},3a,\schoolone),\student(\nametwo,3a,\schoolone)\}\\
 &  & \di=\{\student(\nameone,3a,\schoolone)\}
\end{eqnarray*}

Clearly, the canonical statement is violated because Mary is missing
in the available database. But the query is still complete, as it
returns the set $\{3a\}$ over both the ideal and the available database.
\end{example}
As it is easy to see that for any query $Q$ and any incomplete database
$\D$ it holds that $\D\models\complb Q$ implies $\D\models\compls Q$,
we can conclude the following corollary.
\begin{cor}
\label{c_q-models-compl_q} Let $Q$ be a conjunctive query. Then
for $*\in\{s,b\}$ 
\[
\C_{Q}\models\complstar Q*.
\]
\end{cor}
\begin{proof}
The claim for bag semantics is shown in Theorem~\ref{theorem_characterizing_query_completeness}.
For set semantics, we consider the projection-free variant $Q'$ of
$Q$. Note that $\C_{Q}=\C_{Q'}$. Thus, by the preceding theorem,
if $\pdb\models\C_{Q}$, then $\pdb\models\Compl{Q'}$, and hence,
$Q'(\da)=Q'(\di)$. Since the answers to $Q$ are obtained from the
answers to $Q'$ by projection, it follows that $Q(\da)=Q(\di)$ and
hence, $\pdb\models\Compl Q$.
\end{proof}
Let $Q$ be a conjunctive query. We say that a set $\C$ of \LC statements
is \emph{characterizing} for $Q$ if for all incomplete databases
$\pdb$ it holds that $\pdb\models\cplset$ if and only if $\pdb\models\Compl Q$.

From Corollary~\ref{c_q-models-compl_q} we know that the canonical
completeness statements are a sufficient condition for query completeness
under set semantics. However, one can show that they fail to be a
necessary condition for queries with projection. One may wonder whether
there exist other sets of characterizing TC statements for such que\-ries.
The next theorem tells us that this is not the case.
\begin{prop}
\label{theo-LCs:Not:Characterizing:For:Qs:W:Projection} Let $Q$
be a conjunctive query with at least one non-distinguished variable.
Then no set of \localComp statements is characterizing for $\Compl Q$
under set semantics. \end{prop}
\begin{proof}
Let $\query QB$ be a query with at least one nondistinguished variable
$y$. We construct three incomplete databases as follows:

\begin{align*}
 & \id{D_{1}}=\set{B[y/a],\, B[y/b]} & \av{D_{3}} & =\set{B[y/a]}\\
 & \id{D_{2}}=\set{B[y/a],\, B[y/b]} & \av{D_{3}} & =\set{B[y/b]}\\
 & \id{D_{3}}=\set{B[y/a],\, B[y/b]} & \av{D_{3}} & =\set{}.
\end{align*}

Suppose now there exists some set $\C$ of TC statements such that
$\C$ characterizes $\compl Q$, that is, for any incomplete database
$\D$ it holds that $\D\models\C$ iff $\D\models\compl Q$. The incomplete
databases $\D_{1}$ and $\D_{2}$ constructed above satisfy $\compl Q$,
because $B[y/a]$ and $B[y/b]$ are isomorphic, and the nondistinguished
variable does not appear in the output in the output of $Q(\D_{1})$
or $Q(\D_{2})$. Since the available databases $\da_{1}$ and $\da_{2}$
are however missing the facts $B[y/b]$ and $B[y/a]$, respectively,
it cannot be the case that any tuple from $B[y/b]$ or $B[y/a]$ in
the ideal databases $\di_{1}$ and $\di_{2}$ is constrained by a
TC statement in $\C$.

Thus, also the incomplete database $\D_{3}$ satisfies $\C$, because
$\di_{3}$ is the same as $\di_{1}$ and $\di_{2}$, and the available
database misses only the facts $B[y/a]\cup B[y/b]$ that are not constrained
by $\C$. But clearly $\D_{3}$ does not satisfy $\compl Q$, as $Q(\di_{3})$
contains the distinguished variables of $Q$ but $Q(\da_{3})$ is
empty.
\end{proof}
By Theorem~\ref{theo-LCs:Not:Characterizing:For:Qs:W:Projection},
for a projection query $Q$ the statement $\Compl Q$ is not equivalent
to any set of TC statements. Thus, if we want to perform arbitrary
reasoning tasks, no set of TC statements can replace $\Compl Q$. 

However, if we are interested in TC-QC inferences, that is, in finding
out whether $\Compl Q$ \emph{follows} from a set of TC statements
$\C$, then, as the next result shows, $\C_{Q}$ can take over the
role of $\Compl Q$ provided $Q$ is a minimal relational query and
the statements in $\C$ are relational:
\begin{thm}
Let $Q$ be a minimal relational conjunctive query and $\C$ be a
set of \localComp statements containing no comparisons. Then\label{thm:weakest-precond}
\[
\C\models\Compl Q\mbox{\quad implies\quad}\C\models\C_{Q}.
\]

\end{thm}
For completeness we list below the proof that is already contained
in \cite{razniewski:diplom:thesis:2010}. Note however that there
erroneously it is claimed that the theorem holds for conjunctive queries,
while the proof deals only with relational queries. Whether the theorem
holds also for conjunctive queries is an open question.
\begin{proof}
By contradiction. Assume $Q$ is minimal and $\C$ is such that $\C\models\qcompl Q$,
but $\C\not\models\C_{Q}$. Then, because $\C\not\models\C_{Q}$,
there exists some incomplete database $\pdb$ such that $\pdb\models\C$,
but $\pdb\not\models\C_{Q}$. Since $\pdb\not\models\C_{Q}$, we find
that one of the canonical completeness statements in $\C_{Q}$ does
not hold in $\pdb$. Let~$B$ be the body of $Q$. 

Without loss of generality, assume that $\pdb\not\models C_{1}$,
where $C_{1}$ is the canonical statement for $A_{1}=R_{1}(\tpl d_{1})$,
the first atom in~$B$. Let $Q_{1}$ be the query associated to $C_{1}$.
Thus, there exists some tuple $\tpl u_{1}$ such that $\tpl u_{1}\in Q_{1}(\di)$,
but $\tpl u_{1}\not\in R_{1}(\da)$. Now we construct a second incomplete
database $\pdb_{0}$. To this end let $B'$ be the frozen version
of $B$, that is, each variable in $B$ is replaced by a fresh constant,
and let $A'_{1}=R_{1}(\tpl d_{1}')$ be the frozen version of $A_{1}$.
Now, we define $\pdb_{0}=(B',B'\setminus\set{A'_{1}})$.

\medskip{}

\emph{Claim:} $\pdb_{0}$ satisfies $\C$ as well 

\medskip{}

\noindent To prove the claim, we note that the only difference between
$\di_{0}$ and $\da_{0}$ is that $A'_{1}\notin\da_{0}$, therefore
all \LC statements in $\C$ that describe \localComp of relations
other than $R_{1}$ are satisfied immediately. To show that $\pdb_{0}$
satisfies also all statements in $\C$ that describe \localComp of
$R_{1}$, we assume the contrary and show that this leads to a contradiction. 

\noindent Assume $\pdb_{0}$ does not satisfy some statement $C\in\C$.
Then $Q_{C}(\di_{0})\setminus R_{1}(\da_{0})\neq\eset$, where $Q_{C}(\bar{x}')$
is the query associated with $C$. Since $Q_{C}(\di_{0})\incl R_{1}(\da_{0})$,
it must be the case that $\tpl d'_{1}\in Q_{C}(\di_{0})\setminus R_{1}(\da_{0})$.
Let $B_{C}$ be the body of $Q_{C}$.Then, $\tpl d'_{1}\in Q_{C}(\di_{0})$
implies that there is a valuation $\delta$ such that $\delta B_{C}\incl B'$
and $\delta\tpl x'=\tpl d'_{1}$, where $\tpl x'$ are the distinguished
variables of $C$. As $\tpl u_{1}\in Q_{1}(\di)$, and $Q_{1}$ has
the same body as $Q$, there exists another valuation $\theta$ such
that $\theta B\subseteq\di$ and $\theta\tpl d_{1}=\tpl u_{1}$, where
$\tpl d_{1}$ are the arguments of the atom $A_{1}$.

Composing $\theta$ and $\delta$, while ignoring the difference between
$B$ and its frozen version $B'$, we find that $\theta\delta B_{C}\incl\theta B'=\theta B\incl\di$
and $\theta\delta\tpl x'=\theta\tpl d'_{1}=\theta\tpl d_{1}=\tpl u_{1}$.
In other words, $\theta\delta$ is a satisfying valuation for $Q_{C}$
over $\di$ and thus $\tpl u_{1}=\theta\delta\tpl x'\in Q_{C}(\di)$.
However, $\tpl u_{1}\notin R_{1}(\da)$, hence, $D$ would not satisfy
$C$. This contradicts our initial assumption. Hence, we conclude
that also $\pdb_{0}$ satisfies $\C$.

\smallskip{}

Since $\pdb_{0}$ satisfies $\C$ and $\C\models\Compl Q$, it follows
that $Q$ is complete over $\pdb_{0}$. As $\di_{0}=B'$, the frozen
body of $Q$, we find that $\tpl x''\in\id Q(\pdb_{0})$, with $\tpl x''$
being the frozen version of the distinguished variables $\tpl x$
of $Q$. As $Q$ is complete over $\pdb_{0}$, we should also have
that $\tpl x''\in Q(\da_{0})$. However, as $\di_{0}=B'\setminus\set{A'_{1}}$,
this would require a satisfying valuation from $B$ to $B'\setminus\set{A'_{1}}$
that maps $\tpl x$ to $\tpl x''$. This valuation would correspond
to a non-surjective homomorphism from $Q$ to $Q$ and hence $Q$
would not be minimal.
\end{proof}
By the previous theorems, we have seen that for queries without projection
and for minimal relational queries, the satisfaction of the canonical
completeness statements is a necessary condition for the entailment
of query completeness from table completeness for queries under set
semantics. 

As a consequence, in these cases the question of whether TC statements
imply completeness of a query $Q$ can be reduced to the question
of whether these TC statements imply the canonical completeness statements
of $Q$.

\subsection{TC-QC Entailment for Queries under Set Semantics}

\label{sub:general-tc-qc-entailment}

We have seen that for queries under bag semantics, a TC-QC entailment
problem can be translated into a TC-TC entailment problem by using
the canonical completeness statements. Furthermore, the TC-TC entailment
problem can be translated into a query containment problem. 

For queries under set semantics, we have seen that for queries containing
projections, no characterizing set of TC statements exists. For queries
without comparisons, we have seen that nevertheless for TC-QC entailment,
weakest preconditions in terms of TC exist, which allow to reduce
TC-QC to TC-TC. For queries with comparisons however, it is not known
whether such characterizing TC statements exist. In the following
we will show that there is a translation of TC-QC into query containment
directly.

Recall that we distinguish between four languages of conjunctive queries: 
\begin{itemize}
\item linear relational queries ($\LRQ$): conjunctive queries without repeated
relation symbols and without comparisons, 
\item relational queries ($\RQ$): conjunctive queries without comparisons, 
\item linear conjunctive queries ($\LCQ$): conjunctive queries without
repeated relation symbols, 
\item conjunctive queries ($\CQ$). 
\end{itemize}
We say that a TC statement is in one of these languages if its associated
query is in it. For $\L_{1}$, $\L_{2}$ ranging over the above languages,
we denote by $\LCQC(\L_{1},\L_{2})$ the problem to decide whether
a set of TC statements in $\L_{1}$ entails completeness of a query
in $\L_{2}$. As a first result, we show that TC-QC entailment can
be reduced to a certain kind of query containment. It also corresponds
to a simple containment problem \wrt tuple-generating dependencies.
From this reduction we obtain upper bounds for the complexity of TC-QC
entailment.

To present the reduction, we define the \emph{unfolding} of a query
\wrt to a set of TC statements. Let $\query{Q(\tpl s)}{\aufz An,M}$
be a conjunctive query where $M$ is a set of comparisons and the
relational atoms are of the form $A_{i}=R_{i}(\tpl s_{i})$, and let
$\C$ be a set of TC statements, where each $C_{j}\in\C$ is of the
form $\Compl{R_{j}(\tpl d_{j});G_{j}}$. Then the unfolding of $Q$
\wrt $\C$, written $Q^{\C}$, is defined as follows: 
\[
Q^{\C}(\tpl s)=\bigwedge_{i=1,..,n}\Big(R_{i}(\tpl s_{i})\wedge\bigvee_{C_{j}\in\C,}(G_{j}\wedge\tpl s_{i}=\tpl d_{j})\Big)\wedge M.
\]
 Intuitively, $Q^{\C}$ is a modified version of $Q$ that uses only
those parts of tables that are asserted to be complete by $\C$.
\begin{thm}
\label{thm:reduction_lc-qc_to_containment} Let $\C$ be a set of
TC statements and $Q$ be a conjunctive query. Then 
\[
\C\models\Compl Q\mbox{\quad iff\quad}Q\subseteq Q^{\C}.
\]

\end{thm}
The proof follows after the next Lemma.

Intuitively, this theorem says that a query is complete \wrt a set
of TC statements, iff its results are already returned by the modified
version that uses only the complete parts of the database. This gives
the upper complexity bounds of TC-QC entailment for several combinations
of languages for TC statements and queries. 

The containment problems arising are more complicated than the ones
commonly investigated. The first reason is that queries and TC statements
can belong to different classes of queries, thus giving rise to asymmetric
containment problems with different languages for container and containee.
The second reason is that in general $Q^{\C}$ is not a conjunctive
query but a conjunction of unions of conjunctive queries.

To prove Theorem~\ref{thm:reduction_lc-qc_to_containment}, we need
a definition and a lemma. 
\begin{defn}
\label{def:tc-operator} Let $C$ be a \LC-statement for a relation
$R$. We define the function $T_{C}$ that maps database instances
to $R$-facts as $T_{C}(D)=\set{R(\tpl d)\mid\tpl d\in Q_{C}(D)}$.
That is, if $\di$ is an ideal database, then $T_{C}(\di)$ returns
those $R$-facts that must be in $\da$, if $(\di,\da)$ is to satisfy
$C$. We define $T_{\C}(D)=\bigcup_{C\in\C}T_{C}(D)$ if $\C$ is
a set of TC-statements. \end{defn}
\begin{lem}
\label{lemma_reasoning_1} Let $\C$ be a set of TC statements. Then 
\begin{enumerate}
\item \label{lemma_reasoning1_1} $T_{\C}(D)\incl D$, \ for all database
instances $D$; 
\item \label{lemma_reasoning1_2} $\pdb\models\C\mbox{\ \ iff \ }T_{\C}(\id D)\incl\av D$,
\ for all incomplete databases $\pdb=(\di,\da)$ with $\da\incl\di$; 
\item \label{lemma_reasoning1_3} $Q^{\C}(D)=Q(T_{\C}(D))$, \ for all
conjunctive queries $Q$ and database instances $D$. 
\end{enumerate}
\end{lem}
\begin{proof}
\textit{(1)} Holds because of the specific form of the queries associated
with $\C$. 

\textit{(2)} Follows from the definition of when an incomplete database
satisfies a set of \LC statements. 

\textit{(3)} Holds because unfolding $Q$ using the queries in $\C$
and evaluating the unfolding over the original database $D$ amounts
to the same as computing a new database $T_{\C}(D)$ using the queries
in $\C$ and evaluating $Q$ over the result. 
\end{proof}

\begin{proof}
[Proof of Theorem \ref{thm:reduction_lc-qc_to_containment}]\Onlyif:
Suppose $\C\models\Compl Q$. We want to show that $Q\subseteq Q^{\C}$.
Let $D$ be a database instance. Define $\id D=D$ and $\av D=T_{C}(D)$.
Then $\pdb=(\id D,\av D)$ is a incomplete database, due to Lemma~\ref{lemma_reasoning_1}(i),
which satisfies $\C$, due to Lemma~\ref{lemma_reasoning_1}(iii).
Exploiting that $\pdb\models\Compl Q$, we infer that $Q(D)=Q(\id D)=Q(\av D)=Q(T_{\C}(D))=Q^{\C}(D)$.

\If: Suppose $Q\subseteq Q^{\C}$. Let $\pdb=(\id D,\av D)$ be an
incomplete database such that $\pdb\models\C$. Then we have $Q(\id D)\incl Q^{\C}(\id D)=Q(T_{\C}(\id D))\incl Q(\av D)$,
where the first inclusion holds because of the assumption, the equality
holds because of Lemma~\ref{lemma_reasoning_1}(iii), and the last
inclusion holds because of Lemma~\ref{lemma_reasoning_1}(ii), since
$\pdb\models\C$.
\end{proof}
We show that for linear queries $Q$ the entailment $\C\models\Compl Q$
can be checked by evaluating the function $\tcop$ over test databases
derived from $Q$. If $\C$ does not contain comparisons, one test
database is enough, otherwise exponentially many are needed. We use
the fact that containment of queries with comparisons can be checked
using test databases obtained by instantiating the body of the containee
with representative valuations (see~\cite{Klug-Containment:And:Comparisons-JACM}).
A set of valuations $\Theta$ is \emph{representative} for a set of
variables $\tpl x$ and constants $\tpl c$ relative to $M$, if the
$\theta\in\Theta$ correspond to the different ways to linearly order
the terms in $\tpl x\cup\tpl c$ while conforming to the constraints
imposed by $M$.
\begin{lem}
\label{lemma_reasoning2} Let $\query{Q(\tpl s)}{L,M}$ be a conjunctive
query, let $\C$ be a set of TC statements, and let $\Theta$ be a
set of valuations that is representative for the variables in $Q$
and the constants in $L$ and $\C$ relative to $M$. Then: \end{lem}
\begin{itemize}
\item \label{lemma_reasoning2_1} If $Q\in$ $\LCQ$, and $\C\subseteq\RQ$,
then 
\[
Q\subseteq Q^{\C}\mbox{\quad iff\quad}L=T_{C}(L).
\]

\item \label{lemma_reasoning2_2} If $Q\in\LCQ$ and $\C\subseteq\CQ$,
then 
\[
Q\subseteq Q^{\C}\mbox{\quad iff\quad}\theta L=T_{\C}(\theta L)\mbox{\ \ for all\ }\theta\in\Theta.
\]
\end{itemize}
\begin{proof}
\textit{(i)} \Onlyif Suppose $T_{\C}(L)\not\subseteq L$. Then there
is an atom $A$ such that $A\in L\setminus T_{C}(L)$. We consider
a valuation $\theta$ for $Q$ and create the database $D=\theta L$.
Then $Q(D)\neq\eset$ and, due to containment, $Q^{\C}(D)\neq\eset$.
At the same time, $Q^{\C}(D)=Q(T_{\C}(D))=Q(T_{\C}(\theta L))$. However,
since $A\not\in t_{C}(L)$, there is no atom in $T_{C}(D)$ with the
same relation symbol as $A$ and therefore $Q(T_{C}(D))=\eset$.

\If Let $\tpl c\in Q(D)$. We show that $\tpl c\in Q^{\C}(D)$. There
exists a valuation $\theta$ such that $\theta\models M$, $\theta L\incl D$,
and $\theta\tpl s=\tpl c$. Since $L=T_{\C}(L)$, we conclude that
$\theta L=T_{\C}(\theta L)\subseteq T_{C}(D)$. Hence, $\theta$ satisfies
$Q$ over $T_{\C}(D)$. Thus $\tpl c=\theta\tpl s\in Q(T_{\C}(D))=Q^{\C}(D)$. 

\textit{(ii)} $"\Rightarrow":$ Same argument as for \textit{(i)}.
Suppose $T_{\C}(\theta L)\not\subseteq\theta L$ for some $\theta\in\Theta$.
Then as for (i), there must exist an atom $A\in\theta L$ which is
not in $T_{\C}(\theta L)$ and which allows to construct a database
where $Q$ returns some answer using $A$, which $Q^{\C}$ does not
return.

$"\Leftarrow":$ Suppose $\tpl c\in Q(D)$. Then there must exist
some valuation $\theta$ such that $\theta\models M$, $\theta L\incl D$,
and $\theta\tpl s=\tpl c$, and $\theta$ must order the terms in
$Q$ in the same way as some valuation $\theta'\in\Theta$. Since
$\theta'L=T_{\C}(\theta'L)$, we conclude that also $\theta L=T_{\C}(\theta L)\subseteq T_{C}(D)$,
which again shows that $\tpl c$ is also returned over $Q^{\C}(D)$.
\end{proof}
The above lemma says that when checking \LC-QC entailment for relational
\LC statements and a query in $\LCQ$, we can ignore the comparisons
in the query and decide the containment problem by applying the function
$t_{\C}$ to the relational atoms of the query.
\begin{thm}
We have the following upper bounds: \label{theo-lcqc:upper:bounds} 
\begin{enumerate}
\item $\LCQC(\RQ,\LCQ)$ is in $\PTIME$. 
\item $\LCQC(\CQ,\LCQ)$ is in $\CONP$. 
\item $\LCQC(\RQ,\RQ)$ is in $\NP$. 
\item $\LCQC(\CQ,\CQ)$ is in $\piptwo$. 
\end{enumerate}
\end{thm}
\begin{proof}
\textit{(i)} By Lemma~\ref{lemma_reasoning2}(i), the containment
test requires to check whether whether $L=t_{\C}(L)$ for a linear
relational condition $L$ and a set $\C$ of relational \LC statements.
Due to the linearity of $L$, this can be done in polynomial time.

\textit{(ii)} By Lemma~\ref{lemma_reasoning2}(ii), non-containment
is in $\NP$, because it suffices to guess a valuation $\theta\in\Theta$
and check that $\theta L\setminus T_{C}(\theta L)\neq\eset$, which
can be done in polynomial time, since $L$ is linear.

\textit{(iii)} Holds because containment of a relational conjunctive
query in a positive relational query (an extension of relational conjunctive
queries that allows disjunction) is in $\NP$ (see \cite{Sagiv:Yannakakis-Containment-VLDB}).

\textit{(iv)} Holds because containment of a conjunctive query in
a positive query with comparisons is in $\piptwo$ \cite{meyden-Complexity_querying_ordered_domains-pods}.
\end{proof}
In Lemma \ref{prop:hardness-tc-qc-bag-semantics} we have seen that
$\UCont(\L_{2},\L_{1})$ can be reduced also to $\tcqc^{s}(\L_{1},\L_{2})$.
To use this lemma for showing the hardness of TC-QC entailment, we
have to consider the complexities of the asymmetric containment that
were discussed in Section \ref{sec:prelims:query-containment}. 

For one problem, namely $\LCQC(\LRQ,\CQ)$, the upper bound that was
shown in Theorem \ref{theo-lcqc:upper:bounds} ($\piptwo)$ and the
complexity of the corresponding query containment problem $\UCont(\CQ,\LRQ)$
(NP) do not match. Indeed, using the same technique as was used to
show the hardness of $\Cont(\RQ,\LCQ)$ in Corollary \ref{corr:hardness-cont-rq-lcq},
we are able to prove that $\LCQC(\LRQ,\CQ)$ is $\piptwo$-hard:
\begin{lem}
\label{hardness-1-case} There is a $\PTIME$ many-one reduction from
$\forall\exists$3-SAT to $\LCQC(\LRQ,\CQ)$. 
\end{lem}
In Corollary \ref{corr:hardness-cont-rq-lcq}, we have seen that $Cont(\RQ,\LCQ)$
is $\piptwo$-hard, because validity of $\forall\exists$3-SAT formulas
can be translated into a \linebreak{}
$Cont(\RQ,\LCQ)$ instance.

We now show $\piptwo$-hardness of $\LCQC(\LRQ,\CQ)$ by translating
those $\Cont(\RQ,\LCQ)$ instances into $\LCQC(\LRQ,\CQ)$ instances.

Recall that the $\Cont(\RQ,\LCQ)$ problems were of the form ``$Q\stackrel{?}{\subseteq}Q'$?\textquotedbl{},
where $Q$ and $Q'$ were 
\[
\query{Q()}{G_{1},\ldots,G_{m},F_{1}^{(7)},\ldots,F_{k}^{(7)}},
\]
\vspace{-15bp}
\[
\query{Q'()}{G'_{1},\ldots,G'_{m},C_{1}(p_{11},p_{12},p_{13}),\ldots,C_{k}(p_{k1},p_{k2},p_{k3})},
\]
 and $G_{j}$ and $G'_{j}$ were 
\begin{align*}
 & G_{j}=R_{j}(0,w_{j}),R_{j}(w_{j},1),S_{j}(w_{j},0),S_{j}(1,1),\\
 & G'_{j}=R_{j}(y_{j},z_{j}),S_{j}(z_{j},x_{j}),y_{j}\leq0,z_{j}>0.
\end{align*}

Now consider the set $\C$ of completeness statements containing for
every $1\leq j\leq m$ the statements 
\begin{align*}
\Compl{R_{j}(0,\_);\true},\\
\Compl{R_{j}(\_,1);\true},\\
\Compl{S_{j}(\_,0);\true},\\
\Compl{S_{j}(\_,1);\true},
\end{align*}
 and containing for every $1\leq i\leq k$ the statements 
\begin{align*}
\Compl{C_{i}(1,\_,\_);\true},\\
\Compl{C_{i}(0,\_,\_);\true},
\end{align*}
 where, for readability, $"\_$\textquotedbl{} stands for arbitrary
unique variables.

Clearly, $\C$ contains only statements that are in $\LRQ$ and $Q\cap Q'$
is in $\CQ$.
\begin{lem}
Let $Q$ and $Q'$ be queries constructed from the reduction of a
$\forall\exists$ 3-SAT instance, and let $\C$ be constructed as
above. Then 
\[
\C\models\Compl{Q\cap Q'}\mbox{\ \ iff \ \ }Q\subseteq Q'.
\]
\end{lem}
\begin{proof}
\If Assume $Q\subseteq Q'$. We have to show that $\C\models\Compl{Q\cap Q'}$.
Because of the containment, $Q\cap Q'$ is equivalent to $Q$, and
hence it suffices to show that $\C\models\Compl Q$.

Consider an incomplete database $\pdb$ such that $\pdb\models\C$
and $\di\models Q$. Because of the way in which $\C$ is constructed,
all tuples in $\di$ that made $Q$ satisfied are also in $\da$,
and hence $\da\models Q$ as well.

\Onlyif Assume $Q\not\subseteq Q'$. We have to show that $\C\not\models\Compl{Q\cap Q'}$.

Since the containment does not hold, there exists a database $D_{0}$
that satisfies $Q$ but not $Q'$. We construct an incomplete database
$\pdb$ with 
\begin{align*}
 & \di=D_{0}\union vB_{Q'}\\
 & \da=D_{0},
\end{align*}
 where $v$ is a valuation for $Q'$ that maps any variable either
to the constant -3 or 3.

By that, the tuples from $vB_{Q'}$, missing in $\da$ do not violate
$\C$, that always has constants 0 or 1 in the heads of its statements,
so $\C$ is satisfied by $\pdb$. But as $\di$ satisfies $Q\cap Q'$
and $\da$ does not, this shows that $\C\not\models\Compl{Q\cap Q'}$.
\end{proof}
We summarize our results for the lower complexity bounds of TC-QC
entailment:
\begin{thm}
\label{theo-lcqc:lower:bounds} We have the following lower bounds: \end{thm}
\begin{enumerate}
\item $\LCQC(\LCQ,\LRQ)$ is $\CONP$-hard. 
\item $\LCQC(\LRQ,\RQ)$ is $\NP$-hard. 
\item $\LCQC(\LCQ,\RQ)$ is $\piptwo$-hard. 
\item $\LCQC(\LRQ,\CQ)$ is $\piptwo$-hard. \end{enumerate}
\begin{proof}
Follows from Corollaries \ref{cor:hardness-lrq-lcq}, \ref{cor:hardness-containment-rq-lrq},
\ref{corr:hardness-cont-rq-lcq} and Lemma~\ref{hardness-1-case}.
\end{proof}
We find that the upper bounds shown by the reduction to query containment
(Theorem \ref{theo-lcqc:upper:bounds}) and the lower bounds shown
in the Theorem above match. The complexity of TC-QC entailment is
also summarized in Table~\ref{table:complexity_lc-qc}.

\begin{table}[t]
\begin{centering}
{\scriptsize{}}%
\begin{tabular}{|>{\raggedright}p{1.2cm}c|cccc|}
\hline 
\multicolumn{2}{|c|}{} & \multicolumn{4}{c|}{{\scriptsize{}Query Language}}\tabularnewline
\multicolumn{2}{|c|}{} & {\scriptsize{}LRQ } & {\scriptsize{}LCQ } & {\scriptsize{}RQ } & {\scriptsize{}CQ }\tabularnewline
\hline 
\multirow{4}{1.2cm}{{\scriptsize{}TC Statement Language }} & {\scriptsize{}LRQ } & {\scriptsize{}polynomial } & {\scriptsize{}polynomial } & {\scriptsize{}NP-complete } & {\scriptsize{}$\piptwo$-complete }\tabularnewline
 & {\scriptsize{}RQ } & {\scriptsize{}polynomial } & {\scriptsize{}polynomial } & {\scriptsize{}NP-complete } & {\scriptsize{}$\piptwo$-complete }\tabularnewline
 & {\scriptsize{}LCQ } & {\scriptsize{}coNP-complete } & {\scriptsize{}coNP-complete } & {\scriptsize{}$\piptwo$-complete } & {\scriptsize{}$\piptwo$-complete}\tabularnewline
 & {\scriptsize{}CQ } & {\scriptsize{}coNP-complete } & {\scriptsize{}coNP-complete } & {\scriptsize{}$\piptwo$-complete } & {\scriptsize{}$\piptwo$-complete}\tabularnewline
\hline 
\end{tabular}
\par\end{centering}{\scriptsize \par}

\caption{Complexity of deciding TC-QC entailment under set semantics. Compared
with the reasoning under bag semantics (Table \ref{table:complexity:tc-qc-bag-semantics}),
the reasoning becomes harder for TC-statements without comparisons
and repeated relation symbols in combination with conjunctive queries.}

\label{table:complexity_lc-qc} 
\end{table}

\subsection{Alternative Treatment}

Instead of treating set and bag semantics completely separated, one
can show that set-reasoning can easily be reduced to bag reasoning.

We remind the reader that a conjunctive query is minimal, if no atom
can be dropped from the body of the query without leading to a query
that is not equivalent under set semantics. 

Given a query $Q$ and a set of TC statements $\C$, it was shown
that $\C\models\complb Q$ entails $\C\models\compls Q$. Regarding
the contrary, observe the following:
\begin{prop}
[Characterization] \label{prop:bag-and-set-normally-the-same}Let
$\query QB$ be a satisfiable conjunctive query. Then the following
two are equivalent:
\begin{enumerate}
\item $Q$ is minimal
\item For every set $\C$ of TC statements, it holds that $\C\models\compls Q$
implies \textup{$\C\models\complb Q$.}
\end{enumerate}
\end{prop}
\begin{proof}
$\proofr$ Suppose $Q$ is minimal, and suppose some set $\C$ of
TC statements entails $\compls Q$. We have to show that $\C$ entails
also $\complb Q$.

Let $\D$ be an incomplete database that satisfies $\C$. We have
to show that $\D$ satisfies also $\complb Q$. Let $v$ be a satisfying
valuation for $Q$ over $\di$. We claim that every atom in $vB$
is also in $\da$:

Suppose some atom $A\in vB$ was not in $\da$. Then, $A$ cannot
be constrained by $\C$. But then, we could construct an incomplete
database as $(vB,vB\setminus A)$ which would satisfy $\C$ but would
not satisfy $\compls Q$ because of the minimality of $Q$. This incomplete
database would contradict the assumption that $\C\models\compls Q$
and hence we conclude that any atom in $vB$ is also in $\da$. But
this implies that $\D$ satisfies $\complb Q$, which concludes the
argument.

$\proofl$ Suppose $Q$ is not minimal. We have to show that there
exist a set of TC statements that entails $\compls Q$ but does not
entail $\complb Q$. Consider the set $\can(Q_{\min})$, where $Q_{\min}$
is a minimal version of $Q$. By Proposition \ref{c_q-models-compl_q},
$\can(Q_{\min})$ entails $\compls{Q_{\min}}$ and hence since $Q$
and $Q_{\min}$ are equivalent under set semantics, $\can(Q_{\min})$
entails also $\compls Q$. Since $Q$ is not minimal, at least one
atom $A_{i}$ can be dropped from the body of $Q$ without changing
its results under set semantics. Furthermore, since $Q$ is satisfiable,
there must exist some satisfying valuation $v$ for $Q$ over $\di$.
Thus, an incomplete database $(vB,vB\setminus vA_{i})$ constructed
from the body $B$ of $Q$ shows that $\compls Q\not\models\complb Q$,
because the valuation $v$ is not satisfying for $Q$ over $\da$.\end{proof}
\begin{thm}
[Reduction] Let $Q$ be a conjunctive query, $\C$ be a set of TC
statements, and let $Q_{\min}$ be a minimal version of $Q$ under
set semantics. Then 
\[
\C\models\compls Q\ \ \mbox{iff}\ \ \C\models\complb{Q_{\min}}
\]
\end{thm}
\begin{proof}
$"\Rightarrow":$ Consider an incomplete database $\D$ with $\D\models\compls Q$.
The query $Q_{\min}$ is equivalent to $Q$ under set semantics, and
since $Q$ is complete under set semantics, also $\min(Q)$ is complete
under set semantics, so by Proposition 1, since $Q_{\min}$ is minimal,
$Q_{\min}$ is also complete under bag semantics.

$"\Leftarrow":$ Since bag-completeness implies set-completeness entailment,
it holds that $\complb{Q_{\min}}$ entails $\compls{Q_{\min}}$, and
since the query $Q_{\min}$ is equivalent to $Q$ under set semantics,
therefore also $\compls Q$ holds.
\end{proof}
As a consequence, we can reduce completeness reasoning under set semantics
to query minimization and completeness entailment under bag semantics,
which is equivalent to query containment. 

Note that for \emph{asymmetric entailment problems} (more complex
language for the query than for the completeness statements), the
minimization may be harder than the query containment used to solve
bag-completeness reasoning. The complexity results shown before tell
that this is the case when the queries are conjunctive queries and
the TC statements linear relational queries, as this is the only reported
case where completeness entailment under set semantics $(\Pi_{2}^{P})$
is harder than under bag semantics (NP). 

For \emph{symmetric reasoning problems}, that is, problems where the
language of the TC statements is the same as that of the query, query
minimization and query containment have the same complexity and therefore
also TC-QC reasoning for queries under set and under bag semantics
have the same complexity.

\section{Query Completeness Entailing Query Completeness}

\label{sec:general:QC-QC}

In this section we discuss the entailment of query completeness statements
by query completeness statements. We first review the notion of query
completeness (QC-QC) entailment and the relation of the problem to
the problem of query determinacy. We then show that query determinacy
is a sufficient but not a necessary condition for QC-QC entailment,
that QC-QC entailment and determinacy are sensitive to set/bags semantics,
and that QC-QC entailment is decidable under bag semantics. All considerations
here are for conjunctive queries without comparisons. Proposition
\ref{prop:determinacy-sufficient-for-QC-QC} is already contained
in \cite{razniewski:diplom:thesis:2010} and \cite{Razniewski:Nutt-Compl:of:Queries-VLDB11},
the other results are new.

Query completeness entailment is the problem of deciding whether the
completeness of a set of queries entails the completeness of another
query. For a set of queries $\Q=\set{Q_{1},\ldots,Q_{n}},$ we write
$\complstar{\Q}*$ as shorthand for $\complstar{Q_{1}}*\wedge\cdots\wedge\complstar{Q_{n}}*$.
\begin{defn}
[Query Completeness Entailment] Let $Q$ be a query and $\Q$ be
a set of queries. Then $\complstern{\Q}$ entails $\complstern Q$
for $*\in\set{s,b}$, if it holds for all incomplete databases $\D=(\di,\da)$
that 
\[
\mbox{If for all queries }Q_{i}\in\Q:\ Q_{i}^{*}(\di)=Q_{i}^{*}(\da)\ \ \ \ \mbox{then also}\ \ \ \ Q^{*}(\di)=Q^{*}(\da).
\]
Given $\Q$ and $Q$, we also write $\qcqc^{*}(\Q,Q)$ as shorthand
for $\complstern{\Q}\models\complstern Q$. A variant of this problem
is when $\da$ is fixed, this is investigated in Section \ref{sec:reasoning_with_instances}.
\end{defn}
To solve QC-QC entailment, Motro proposed to look for \emph{rewritings}
of the query $Q$ in terms of the queries in $\Q$ \cite{motro_integrity}. 

If the query $Q$ can be rewritten in terms of the complete queries
$\Q$, then $Q$ can also be concluded to be complete, because the
answer to $Q$ can be computed from the complete answers of $\Q$.
\begin{example}
Consider the following three queries:
\begin{eqnarray*}
 &  & \query{Q_{1}(x)}{R(x),S(x),T(x)}\\
 &  & \query{Q_{2}(x)}{R(x),S(x)}\\
 &  & \query{Q_{3}(x)}{T(x)}
\end{eqnarray*}
 Assume now that the queries $Q_{2}$ and $Q_{3}$ are asserted to
be complete. Clearly, the query $Q_{1}$ can be rewritten in terms
of $Q_{2}$ and $Q_{3}$ as 
\[
\query{Q_{1}(x)}{Q_{2}(x),Q_{3}(x)}.
\]
 Thus, the result of $Q_{1}$ can be computed as the intersection
of the results of the queries $Q_{2}$ and $Q_{3}$, and therefore,
when $Q_{2}$ and $Q_{3}$ are complete, the intersection between
them is complete as well and therefore completeness of $Q_{2}$ and
$Q_{3}$ entails completeness of $Q_{1}$.
\end{example}
An information-theoretic definition of rewritability was given by
Segoufin and Vianu \cite{segoufin_vianu_determinacy} using the notion
of query determinacy: 
\begin{defn}
[Determinacy] Let $Q$ be a query and $\Q$ be a set of queries.
Then $\Q$ \emph{determines} $Q$ under $*$ semantics with $*\in\set{b,s}$,
if for all pairs of databases $(D_{1},D_{2})$ it holds that 
\[
\mbox{If for all queries }Q_{i}\in\Q:\ Q_{i}^{*}(D_{1})=Q_{i}^{*}(D_{2})\ \ \ \ \mbox{then also}\ \ \ \ Q^{*}(D_{1})=Q^{*}(D_{2}).
\]
If $\Q$ determines $Q$, we write $\Q\det Q$. 

So far the problem has received attention only under set semantics.
There, if a query $Q$ is determined by a set of queries $\Q$, the
answer to $Q$ can be computed from the answers of $\Q$, and $Q$
can be rewritten in terms of $\Q$ in second-order logic \cite{segoufin_vianu_determinacy}.
However, the rewriting need not be a conjunctive query itself. In
fact, Segoufin and Vianu showed that there exist queries $\Q$ and
$Q$ such that $Q$ can be rewritten in terms of $\Q$ as a first-order
query, while there is no rewriting as conjunctive query. A good example
for this case was given by Afrati in \cite{afrati-determinacy}: \end{defn}
\begin{example}
Consider the following queries $P_{3},P_{4}$ and $P_{5}$, asking
for paths of length 3, 4 and 5, respectively:
\begin{eqnarray*}
 &  & \query{P_{3}(x,y)}{R(x,z_{1}),R(z_{1},z_{2}),R(z_{2},y)}\\
 &  & \query{P_{4}(x,y)}{R(x,z_{1}),R(z_{1},z_{2}),R(z_{2},z_{3})},R(z_{3},y)\\
 &  & \query{P_{5}(x,y)}{R(x,z_{1}),R(z_{1},z_{2}),R(z_{2},z_{3}),R(z_{3},z_{4}),R(z_{4},y)}
\end{eqnarray*}
It is easy to see that $P_{5}$ cannot be rewritten as conjunctive
query in terms of $P_{3}$ and $P_{4}$. However, there exists a first-order
rewriting for $P_{5}$ as follows:
\[
\query{P_{5}(x,y)}{P_{4}(x,z)\wedge\forall w:\ P_{3}(w,z)\rightarrow P_{4}(w,y)}
\]

Whether determinacy for conjunctive queries under set semantics is
decidable, remains an open question to date. Various works have shown
decidability for sublanguages of conjunctive queries \cite{pasaila-11-decidability-determinacy-sublanguages,fan-geerts-determinacy-sublanguages-2012}.

It is easy to see that query determinacy is a sufficient condition
for QC-QC entailment, as expressed by the following proposition:\end{example}
\begin{prop}
[Sufficiency of Determinacy for QC-QC]\label{prop:determinacy-sufficient-for-QC-QC}Let
$\Q\cup\set Q$ be a set of queries and $*\in\set{b,s}$. Then 
\[
\qcqc^{*}(\Q,Q)\mbox{\quad if \quad}\Q\determines^{*}Q.
\]
\end{prop}
\begin{proof}
The definitions of query determinacy entails the definition of QC-QC
entailment, as query determinacy holds if $Q$ returns the same answer
over all pairs of databases $D_{1}$ and $D_{2}$, while QC-QC entailment
requires only to check those pairs where $D_{1}$ is a subset of $D_{2}$.
\end{proof}
For the special case of boolean queries, in \cite{razniewski:diplom:thesis:2010}
it was shown that query determinacy and $\qcqc$-entailment coincide.

In general, determinacy however is not a necessary condition for query
completeness entailment:
\begin{prop}
[Non-necessity of Determinacy for QC-QC] Let $\Q\cup\set Q$ be a
set of queries, and $*\in\set{b,s}$. Then there exist queries $\Q\cup\set Q$
such that $\qcqc*(\Q,Q)$ holds and $"\Q\determines^{*}Q"$ does not
hold.\label{prop:determinacy-not-necessary-for-qcqc}\end{prop}
\begin{proof}
\emph{(Set semantics)}: Consider the following two queries
\begin{eqnarray*}
 &  & \query{P_{2}(x,y)}{R(x,z),R(z,y)}\\
 &  & \query{P_{3}(x,y)}{R(x,z),R(z,w),R(w,y)}
\end{eqnarray*}
that ask for paths of length 2 and 3, respectively.

Then, $\compls{P_{2}}$ entails $\compls{P_{3}}$ for the following
reason:

Consider an incomplete database $\D=(\di,\da)$ and assume that in
$\di$ there is a path from a node 1 via nodes 2 and 3 to a node 4,
and assume that $P_{2}$ is complete over $\D$. Thus, since $P_{2}(\di)=\{(1,3)(2,4)\}$
there must be a path from 1 via some node x to node 3 and from 2 via
some node y to node 4 in the available database. But since $\da\subseteq\di$
those paths must also be in the ideal database. But then there is
a path x-3-4 in the ideal database and hence by the same reasoning
a path from x via some z to 4 in the available database and hence
the path 1-x-z-4 is in the available database and thus $P_{3}(\di)=P_{3}(\da)=\{(1,4)\}$
and hence $P_{3}$ is complete over $\D$.

However, the following pair of databases $D_{1}$ and $D_{2}$ shows
that $P_{2}$ does not determine $P_{3}$: Let $D_{1}=\{R(1,2),R(2,3),R(3,4)\}$
and $D_{2}=\{R(1,x),R(x,3),R(2,y),R(y,4)\}.$ Then $P_{2}(D_{1})=P_{2}(D_{2})=\{(1,3),(2,4)\}$
but $P_{3}(D_{1})=(1,3)$ is not the same as $P_{3}(D_{2})=\emptyset$
and hence determinacy does not hold.\smallskip{}

\emph{(Bag semantics)}: Consider queries $\query{Q_{1}()}{R(x)}$
and $\query{Q_{2}(x)}{R(x)}$. Then $\complb{Q_{1}}$ entails $\complb{Q_{2}}$,
because $\complb{Q_{1}}$ ensures that the same number of tuples is
in $R(\di)$ and $R(\da)$, and because of the condition $\da\subseteq\di$
that incomplete databases have to satisfy, this implies that $R(\di)$
and $R(\da)$ must contain also the same tuples.

But $Q_{1}$ does not determine $Q_{2}$ under bag semantics, because
having merely the same number of tuples in $R$ does not imply to
have also the same tuples, as e.g.\ a pair of databases $D_{1}=\set{R(a)}$
and $D_{2}=\set{R(b)}$ shows. 
\end{proof}
We show next that both regarding query determinacy and completeness
entailment, it is important to distinguish between set and bag semantics.
As the following theorem shows, both problems are sensitive to this
distinction:
\begin{prop}
[Set/bag sensitivity of QC-QC and Determinacy]There exist sets $\Q\cup\set Q$
of queries such that
\begin{enumerate}
\item $\Q\determines^{s}Q$ does not entail $\Q\determines^{b}Q$,
\item $\qcqc^{s}(\Q,Q)$ does not entail $\qcqc^{b}(\Q,Q)$,
\item $\qcqc^{b}(\Q,Q)$ does not entail $\qcqc^{s}(\Q,Q)$.
\end{enumerate}
\end{prop}
\begin{proof}
\emph{(Claims 1 and 2)}: Consider the following query $Q_{\mathrm{nr\text{\_french}}}$
that asks for the names of people that attended a French language
course and some other language course. Observe that under bag semantics,
this query returns for each person that takes French the name as often
as that person takes language courses, possibly also in other languages:
\[
\query{Q_{\mathrm{nr\text{\_french}}}(n)}{\result(n,French,g),\result(n,x,g')}.
\]
Consider now a second query $Q_{\mathrm{french}}$ which only asks
for the names of persons that took a French language course

\[
\query{Q_{\mathrm{french}}(n)}{\result(n,French,g)}.
\]
If both queries are evaluated under set semantics, then completeness
of $Q_{\mathrm{french}}$ implies completeness of $Q_{\mathrm{nr\text{\_french}}}$,
because under set semantics, both queries are equivalent, as $Q_{\mathrm{nr\text{\_french}}}$
is not minimal. But under bag semantics, completeness of $Q_{\mathrm{french}}$
does not entail completeness of $Q_{\mathrm{nr\text{\_french}}}$
as for instance the following incomplete database shows:

Consider the incomplete database $\D=(\di,\da)$ where $\result(\di)$
contains $\{(\nameone,French,A),(\nameone,Dutch,B)\}$ and $\result(\da)$
contains\linebreak{}
 $\{(\nameone,French,A)\}$. Then, $Q_{\mathrm{french}}$ returns
over both the ideal and the available database the answer $\{(\nameone)\}$
and hence is complete, however, $Q_{\mathrm{nr\text{\_french}}}$
returns $\{(\nameone),(\nameone)\}$ over the ideal database and $\{(\nameone)\}$
over the available database and hence is not complete.

While because of the equivalence under set semantics, it is also clear
that under set semantics $Q_{\mathrm{french}}$ determines $Q_{\mathrm{nr\_french}}$,
the incomplete database $\D$ from above shows that under bag semantics
that is not the case.\smallskip{}

\emph{(Claim 3)}: Consider the queries $\query{Q_{1}()}{R(x)}$ and
$\query{Q_{2}(x)}{R(x)}$ as used in the proof of Prop.\ \ref{prop:determinacy-not-necessary-for-qcqc}. 

Clearly, under bag semantics, completeness of $Q_{1}$ entails completeness
of $Q_{2}$, because $Q_{1}$ ensures that the same number of tuples
are present in $R(\da)$ as in $R(\di)$, and by the condition $\da\subseteq\di$,
this implies that those are exactly the same tuples.

But under set semantics completeness of $Q_{1}$ does not entail completeness
of $Q_{2}$, as an incomplete database with $\da=\set{R(a)}$ and
$\di=\set{R(a),R(b)}$ shows.
\end{proof}
The observations stated in the previous theorem are encouraging, as
they imply that QC-QC entailment under bag semantics may be not as
hard as under set semantics. That this is indeed the case, shows the
following theorem. Let $\L_{1}$ and $\L_{2}$ be conjunctive query
languages. We then denote with $\qcqc^{b}(\L_{1},\L_{2})$ the problem
of deciding whether completeness of a query in $\L_{1}$ under bag
semantics is entailed by completeness of a set of queries in $\L_{2}$
under bag semantics.
\begin{thm}
[Decidability of $\qcqc^b$] For all conjunctive query languages $\L_{1}$,
$\L_{2}$ in $\{\LRQ,\LCQ,$$\ $ $\RQ,\CQ\}$ there is a polynomial-time
reduction from $\qcqc^{b}(\L_{1},\L_{2})$ to $\UCont(\L_{2},\L_{1})$.\label{thm:qc-qc-bag-is-decidable}\end{thm}
\begin{proof}
This follows from Theorems \ref{theorem_characterizing_query_completeness}
and \ref{theo-equivalence:TC-TC:UCont}. The former theorem shows
that a query under bag semantics is complete over an incomplete database,
exactly if its canonical completeness statements are satisfied, while
preserving the languages. Thus, QC-QC entailment can be reduced to
the entailment of the canonical completeness statements, which is
a TC-TC entailment problem. 

The latter theorem shows that $\tctc(\L_{1},\L_{2})$ can be reduced
to $\UCont(\L_{2},\L_{1})$. Thus, QC-QC entailment under bag semantics
can be reduced to containment of unions of queries, while interchanging
languages.
\end{proof}
An interesting related decidable problem is query determinacy, with
the determining queries $\Q$ being evaluated under bag semantics
and the determined query $Q$ under set semantics. Decidability of
this problem follows from results by Fan et al.\ \cite{fan-geerts-determinacy-sublanguages-2012},
who showed that query determinacy is decidable when the determining
queries contain no projections. As queries under set semantics without
projection directly correspond to queries under bag semantics, this
implies decidability of query determinacy when the determining queries
are evaluated under bag and the determined query under set semantics.

Nevertheless, important questions remain open.
\begin{problem}
[Open Questions] Let $\Q\cup\set Q$ be a set of queries. Then the
following are open problems:
\begin{enumerate}
\item Does $\Q\determines^{b}Q$ imply $\Q\determines^{s}Q$?
\item Is $\Q\determines^{s}Q$ decidable?
\item Is $\Q\determines^{b}Q$ decidable?
\item Is $\qcqc^{s}(\Q,Q)$ decidable?
\end{enumerate}
\end{problem}
In Section \ref{sec:reasoning_with_instances}, we discuss completeness
reasoning with instances, and show that both QC-QC entailment and
query determinacy for queries under set semantics are decidable, when
one database instance is fixed.

\begin{figure}
\includegraphics[width=1\columnwidth]{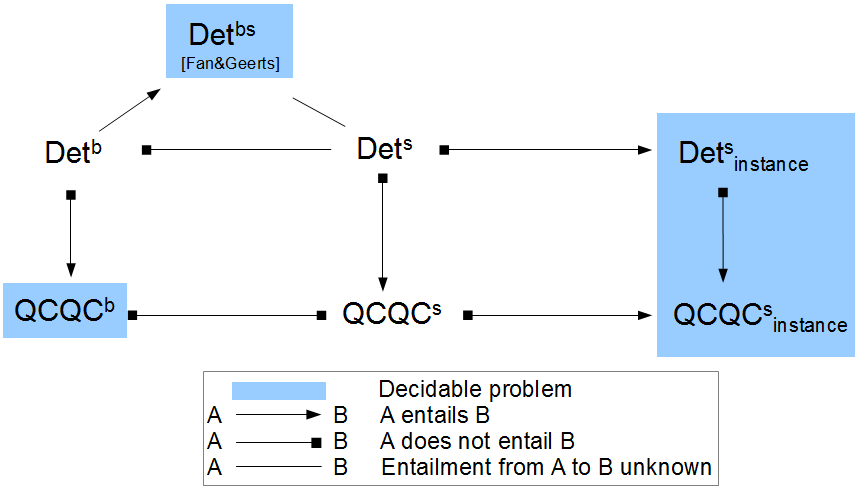}

\caption{Relation of different instances of QC-QC entailment and query determinacy
as discussed in Sections \ref{sec:general:QC-QC} and \ref{sec:reasoning_with_instances}.}

\label{fig:results-qcqc-and-determinacy}
\end{figure}

We summarize the results of this section in Figure \ref{fig:results-qcqc-and-determinacy}.
For completeness, we also include the results on instance reasoning
from Section~\ref{sec:reasoning_with_instances}.

\section{Aggregate Queries}

\label{sec:extensions:aggregate-queries}

Aggregate queries are important for data processing in many applications,
especially in decision support. In contrast to normal queries, aggregate
queries do not only ask for tuples but also allow to compute results
of aggregate functions such as SUM, COUNT, MIN or MAX over results.
The school administration for instance is mostly interested in knowing
how many students or teachers are there that satisfy a certain property,
not who those teachers or students are. Completeness reasoning for
aggregate queries may be different depending on the aggregation function.
\begin{example}
Consider a query $Q_{\mathrm{nr}}$ that asks for the number of pupils
in the class 4A. In SQL, this aggregate query would be written as
follows:
\begin{eqnarray*}
 &  & \mbox{SELECT count(*)}\\
 &  & \mbox{FROM pupil}\\
 &  & \mbox{WHERE class=4a AND school=\ensuremath{\schoolone}}
\end{eqnarray*}
 Furthermore, consider a query $Q_{\mathrm{best\_pt}}$ for the best
grade that a pupil of class 4A obtained in Pottery, and consider a
completeness statement that says that the database is complete for
all pupils in level 4A. Then, the query $Q_{\mathrm{nr}}$ will also
return a correct answer, because all pupils are there. Whether the
answer to $Q_{\mathrm{best\_pt}}$ is correct is however unknown,
because while the pupils are complete, nothing is asserted about their
grades.
\end{example}
Query equivalence for aggregate queries has already been studied by
Cohen, Nutt and Sagiv in \cite{CohenEtAl-Aggregate:Equivalence-JACM}.
We will leverage on those results. We will also draw upon the results
for non-aggregate queries as presented in Section \ref{sec:general:TC-QC}
to investigate when TC-statements imply completeness of aggregate
queries.

We consider queries with the aggregate functions $\Count$, $\Sum$,
and $\Max$. Results for $\Max$ can easily be reformulated for $\Min$.
Note that $\Count$ is a nullary function while $\Sum$ and $\Max$
are unary.

\paragraph{Formalization}

An \emph{aggregate term} is an expression of the form $\alpha(\tpl y)$,
where $\tpl y$ is a tuple of variables, having length 0 or~1. Examples
of aggregate terms are $\Count()$ or $\Sum(y)$. If $\query{Q(\tpl x,\tpl y)}{L,M}$
is a conjunctive query, and $\alpha$ an aggregate function, then
we denote by $Q^{\alpha}$ the aggregate query $\query{Q^{\alpha}(\tpl x,\alpha(\tpl y))}{L,M}$.
We say that $Q^{\alpha}$ is a \emph{conjunctive aggregate query}
and that $Q$ is the \emph{core} of $Q^{\alpha}$. 

Over a database instance, $Q^{\alpha}$ is evaluated by first computing
the answers of its core $Q$ under bag semantics, then forming groups
of answer tuples that agree on their values for $\tpl x$, and finally
applying for each group the aggregate function $\alpha$ to the bag
of $y$-values of the tuples in that group.

A sufficient condition for an aggregate query to be complete over
$\pdb$ is that its core is complete over~$\pdb$ under bag semantics.
Hence, Corollary~\ref{c_q-models-compl_q} gives us immediately a
sufficient condition for TC-QC entailment. Recall that $\C_{Q}$ is
the set of canonical completeness statements of a query $Q$.
\begin{prop}
Let $Q^{\alpha}$ be an aggregate query and $\C$ be a set of TC statements.
Then 
\[
\C\models\C_{Q}\mbox{\ \ \ \ implies\ \ \ \ }\C\models\Compl{Q^{\alpha}}.
\]

\end{prop}
For $\Count$-queries, completeness of $Q^{\Count}$ is the same as
completeness of the core $Q$ under bag semantics. Thus, we can reformulate
Theorem~\ref{theorem_characterizing_query_completeness} for $\Count$-queries:
\begin{thm}
\label{theo-count:query:completeness} Let $Q^{\Count}$ be a $\Count$-query
and $\C$ be a set of TC statements. Then
\[
\C\models\Compl{Q^{\Count}}\mbox{\ \ \ \ if and only if \ \ \ \ }\C\models\C_{Q}
\]
\end{thm}
\begin{proof}
Follows from the fact that a count is correct, if and only if the
nonaggregate query retrieves all tuples from the database.
\end{proof}
In contrast to $\Count$-queries, a $\Sum$-query can be complete
over an incomplete database $(\id D,\av D)$ although its core is
incomplete. The reason is that it does not hurt if some tuples from
$\id D$ that only contribute 0 to the overall sum are missing in
$\av D$. Nonetheless, we can prove an analogue of Theorem~\ref{theo-count:query:completeness}
if there are some restrictions on TC statements and query.

We say that a set of comparisons $M$ is \emph{reduced}, if for all
terms $s$, $t$ it holds that $M\models s=t$ only if $s$ and $t$
are syntactically equal. A conjunctive query is \emph{reduced} if
its comparisons are reduced. Every satisfiable query can be equivalently
rewritten as a reduced query in polynomial time. We say that a $\Sum$-query
is \emph{nonnegative} if the summation variable $y$ can only be bound
to nonnegative values, that is, if $M\models y\geq0$.
\begin{thm}
\label{theo-sum:query:completeness} Let $Q^{\Sum}$ be a reduced
nonnegative $\Sum$-query and $\C$ be a set of relational TC statements.
Then
\[
\C\models\Compl{Q^{\Sum}}\mbox{\ \ \ \ if and only if\ \ \ \ }\C\models\C_{Q}
\]
\end{thm}
\begin{proof}
The direction $\C\models\C_{Q}$ implies $\C\models\Compl{Q^{\Sum}}$
holds trivially. It remains to show that $\C\models\Compl{Q^{\Sum}}$
implies $\C\models\C_{Q}$.

Assume this does not hold. Then $\C\models\Compl{Q^{\Sum}}$ and there
exists some $\pdb=(\id D,\av D)$ such that $\pdb\models\C$, but
$\pdb\not\models\C_{Q}$. Without loss of generality, assume that
condition $C_{1}$ of $\C_{Q}$, which corresponds to the first relational
atom, say $A_{1}$, of the body of $Q$, is not satisfied by~$\pdb$.
Then there is a valuation $\theta$ such that $M\models\theta$ and
$\theta L\incl\id D$, but $\theta A_{1}\notin\av D$. If $\theta y\neq0$,
then we are done, because $\theta$ contributes a positive value to
the overall sum for the group $\theta\tpl x$. Otherwise, we can find
a valuation $\theta'$ such that \textit{(i)} $\theta'\models M$,
\textit{(ii)} $\theta'y>0$, \textit{(iii)} if $\theta'z\neq\theta z$,
then $\theta'z$ is a fresh constant not occurring in $\pdb$, and
\textit{(iv)} for all terms $s$, $t$, it holds that $\theta's=\theta't$
only if $\theta s=\theta t$. Such a $\theta'$ exists because $M$
is reduced and the order over which our comparisons range is dense.
Due to \textit{(iii)}, in general we do not have that $\theta'L\incl\id D$.

We now define a new incomplete database $\pdb'=(\id{D'},\av{D'})$
by adding $\theta'L\setminus\set{\theta'A}$ both to $\id D$ and
$\av D$. Thus, we have that \textit{(i)} $\theta'L\incl\id{D'}$,
\textit{(ii)} $\theta'L\not\incl\av{D'}$, and \textit{(iii)} $\pdb'\models\C$.
The latter claim holds because any violation of $\C$ by $\pdb'$
could be translated into a violation of $\C$ by $\pdb$, using the
fact that $\C$ is relational. Hence, $\theta'$ contributes the positive
value $\theta'y$ to the sum for the group $\theta'\tpl x$ over $\pdb'$,
but not over $\pdb$. Consequently, the sums for $\theta'\tpl x$
over $\id{D'}$ and $\av{D'}$ are different (or there is no such
sum over $\av{D'}$), which contradicts our assumption that $\C\models\Compl{Q^{\Sum}}$.
\end{proof}
In the settings of Theorems~\ref{theo-count:query:completeness}
and~\ref{theo-sum:query:completeness}, to decide TC-QC entailment,
it suffices to decide the corresponding TC-TC entailment problem with
the canonical statements of the query core. By Theorem~\ref{theo-equivalence:TC-TC:UCont},
these entailment problems can be reduced in $\PTIME$ to containment
of unions of conjunctive queries.

We remind the reader that for the query languages considered in this
work, TC-TC entailment has the same complexity as TC-QC entailment
(cf.\ Table~\ref{table:complexity_lc-qc}), with the exception of
$\LCLC(\LRQ,\CQ)$ and $\LCLC(\RQ,\CQ)$. The TC-QC problems for these
combinations are $\piptwo$-complete, while the corresponding TC-TC
problems are in $\NP$.

While for $\Count$ and $\Sum$-queries the multiplicity of answers
to the core query is crucial, this has no influence on the result
of a $\Max$-query. Cohen et al.\ have characterized equivalence
of $\Max$-queries in terms of \emph{dominance} of the cores~\cite{CohenEtAl-Aggregate:Equivalence-JACM}.
A query $Q(\tpl s,y)$ is dominated by a query $Q'(\tpl s',y')$ if
for every database instance $D$ and every tuple $(\tpl d,d)\in Q(D)$
there is a tuple $(\tpl d,d')\in Q'(D)$ such that $d\leq d'$. For
$\Max$-queries it holds that $Q_{1}^{\Max}$ and $Q_{2}^{\Max}$
are equivalent if and only if $Q_{1}$ dominates $Q_{2}$ and vice
versa. In analogy to Theorem~\ref{thm:reduction_lc-qc_to_containment},
we can characterize query completeness of $\Max$-queries in terms
of dominance.
\begin{thm}
\label{thm:max-queries-characterization}Let $\C$ be a set of TC-statements
and $Q^{\Max}$ be a $\Max$-query. Then 
\[
\C\models\Compl{Q^{\Max}}\mbox{\ \ \ \ iff \ \ \ \ensuremath{Q\ }is dominated by\ }{Q}^{\C}
\]
\end{thm}
\begin{proof}
$"\Rightarrow":$ By counterposition. Assume $Q$ is not dominated
by $Q^{\C}$. Then there exists a database instance $D$ such that
there is a tuple $t_{1}=(\tpl d,d)\in Q(D)$ but there is no tuple
$(\tpl d,d')\in Q^{\C}(D)$ with $d\leq d'$. By Lemma \ref{lemma_reasoning_1},
the incomplete database $(D,T_{\C}(D))$ satisfies $\C$, and $Q^{\C}(D)=Q(T_{\C}(D))$
and thus there is no tuple $(\tpl d,d')\in Q^{\C}(D)$ with $d\leq d'$.
Thus it is shown that $\C$ does not entail $\compl{Q^{\Max}}$.

$"\Leftarrow"$: Analogous to the proof of Theorem \ref{thm:reduction_lc-qc_to_containment}.
Suppose $Q$ is dominated by $Q^{\C}$. Let $\pdb=(\id D,\av D)$
be an incomplete database such that $\pdb\models\C$. Then we have
that $Q(\id D)$ is dominated by $Q^{\C}(\id D)$ because of the assumption,
and $Q^{\C}(\id D)=Q(T_{\C}(\id D))$ because of Lemma~\ref{lemma_reasoning_1}(iii),
and $Q(T_{\C}(\id D))\incl Q(\av D)$ because of Lemma~\ref{lemma_reasoning_1}(ii),
since $\pdb\models\C$. Thus, $Q(\di)$ is dominated by $Q(\da)$
and hence the $\Max$-query is complete.
\end{proof}
Dominance is a property that bears great similarity to containment.
For queries without comparisons it is even equivalent to containment
while for queries with comparisons it is characterized by the existence
of \emph{dominance mappings}, which resemble the well-known containment
mappings (see~\cite{CohenEtAl-Aggregate:Equivalence-JACM}). This
allows to conclude that the upper and lower bounds of Theorems~\ref{theo-lcqc:upper:bounds}
and~\ref{theo-lcqc:lower:bounds} hold also for $\Max$-queries. 

If $\L$ is a class of conjunctive queries, we denote by $\L^{\Max}$
the class of $\Max$-queries whose core is in $\L$. For languages
$\L_{1}$, $\L_{2}^{\Max}$, the problem $\LCQC(\L_{1},\L_{2}^{\Max})$
is defined as one would expect. With this notation, we can conclude
the following:
\begin{thm}
Let $\L_{1}$ and $\L_{2}$ be languages among \LRQ, \LCQ, \RQ
\ and \CQ. Then the complexity of $\LCQC(\L_{1},\L_{2}^{\Max})$
is the same as the one of $\LCQC(\L_{1},\L_{2})$.\end{thm}
\begin{proof}
That $\LCQC(\L_{1},\L_{2}^{\Max})$ is at most as hard as $\LCQC(\L_{1},\L_{2})$
follows from Theorem \ref{thm:max-queries-characterization} and the
complexity results for query dominance in \cite{CohenEtAl-Aggregate:Equivalence-JACM}.

That $\LCQC(\L_{1},\L_{2}^{\Max})$ is at least as hard as $\LCQC(\L_{1},\L_{2})$
follows from the fact that $\LCQC(\L_{1},\L_{2})$ can trivially be
reduced to $\LCQC(\L_{1},\L_{2}^{\Max})$ by introducing a new unary
relation symbol U with a new variable $x$, of which the maximum is
calculated, into a query and by adding the assertion that U is complete.
\end{proof}

\section{Instance Reasoning}

\label{sec:reasoning_with_instances}In many cases one has access
to the current state of the database, which may be exploited for completeness
reasoning. Already Halevy~\cite{levy_completeness} observed that
taking into account both a database instance and the functional dependencies
holding over the ideal database, additional QC statements can be derived.
Denecker \etal~\cite{Denecker:Calabuig-logical_theory_partial_databases:tods:10}
showed that for first order queries and TC statements, TC-QC entailment
with respect to a database instance is in $\CONP$, and $\CONP$-hard
for some queries and statements. They then focused on approximations
for certain and possible answers over incomplete databases.
\begin{example}
\label{example:instance-reasoning}As a very simple example, consider
the query 
\[
\query{Q(n)}{\mathit{student}(n,c,s),\result(n,\textit{'Greek'},g)},
\]
 asking for the names of students that attended Greek language courses.
Suppose that the \textit{language\_attendance} table is known to be
complete. Then this alone does not imply the completeness of $Q$,
because records in the \textit{student} table might be missing.

Now, assume that we additionally find that in our database that the
table \textit{$\result$} contains no record about Greek.

As the \textit{$\result$} table is known to be complete, it does
not matter which tuples are missing in the \textit{student} table.
No student can have taken Greek anyway. The result of $Q$ must always
be empty, and hence we can conclude that $Q$ is complete in this
case. 
\end{example}
In this section we will discuss TC-QC and QC-QC reasoning wrt. a concrete
database instance. We show that for queries under set semantics, TC-QC
reasoning becomes harder whereas QC-QC reasoning becomes easier.

\subsection{Entailment of Query Completeness by Table Completeness}

Formally, the question of \emph{\LC-QC entailment \wrt a database
instance} is formulated as follows: given an available data\-base
instance $\da$, a set of \localComp statements $\C$, and a query
$Q$, is it the case that for all ideal database instances $\di$
such that $(\di,\da)\models\C$, we have that $Q(\da)=Q(\di)$? If
this holds, we write 
\[
\da,\C\models\Compl Q.
\]

Interestingly, TC-QC entailment \wrt a concrete database is $\piptwo$-complete
even for linear relational queries:
\begin{thm}
$\tcqc$ entailment \wrt a database instance has (i) polynomial data
complexity and is (ii) $\piptwo$-complete in combined complexity
for all combinations of languages among \LRQ, $\L_{\mathrm{LCQ}}$,
\RQ, and \CQ.\label{thm:instancereasoning-tcqc}
\end{thm}
To show the $\piptwo$-hardness of $\LCQC(\LRQ,\LRQ)$ entailment
w.r.t.\ a concrete database instance, which implies the hardness
of all other combinations, we give a reduction of the previously seen
problem of validity of an universally quantified 3-SAT formula.

Consider $\phi$ to be an allquantified 3-SAT formula of the form
\[
\forall\aufz xm\exists\aufz yn:\gamma_{1}\wedge\ldots\wedge\gamma_{k}.
\]
 where each $\gamma_{i}$ is a disjunction of three literals over
propositions $p_{i1}$, $p_{i2}$ and $p_{i3}$, and where $\set{\aufz xm}\union\set{\aufz yn}$
are propositions.

We define the query completeness problem 
\[
\Gamma_{\phi}=(\ \da,\C\stackrel{?}{\models}\Compl Q\ )
\]
 as follows. Let the relation schema $\Sigma$ be $\set{B_{1}/1,\ldots,B_{m}/1,R_{1}/1,\break\ldots,R_{m}/1,C_{1}/3,\ldots,C_{k}/3}$.
Let $Q$ be a query defined as 
\[
\query{Q()}{B_{1}(x_{1}),R_{1}(x_{1}),\ldots,B_{m}(x_{m}),R_{m}(x_{m})}.
\]
 Let $\da$ be such that for all $B_{i}$, $B_{i}(\da)=\set{0,1}$,
and for all $i=1,\ldots,m$ let $R_{i}(\da)=\{\}$ and let $C_{i}(\da)$
contain all the 7 triples over $\set{0,1}$ such that $\gamma_{i}$
is mapped to true if the variables in $\gamma_{i}$ become the truth
values $\true$ for 1 and $\false$ for 0 assigned.

Let $\C$ be the set containing the following \LC statements 
\begin{align*}
 & \Compl{B_{1}(x),\true},\ldots,\Compl{B_{m}(x),\true}\\[0.3ex]
 & \textit{Compl}(R_{1}(x_{1});\, R_{2}(x_{2}),\ldots,R_{m}(x_{m}),\\
 & \qquad\qquad\qquad\qquad\ C_{1}(p_{11},p_{12},p_{13}),\ldots,C_{k}(p_{k1},p_{k2},p_{k3})),
\end{align*}
 where the $p_{ij}$ are either $x$ or $y$ variables as defined
in $\phi$.
\begin{lem}
\label{lower_bound_ext_inf} Let $\phi$ be a $\forall\exists$3-SAT
formula as shown above and let $Q$, $C$ and $\da$ be constructed
as above. Then 
\[
\phi\mbox{ is valid}\mbox{\quad iff\quad}\da,\C\models\Compl Q.
\]
\end{lem}
\begin{proof}
[Proof (of the lemma)]Observe first, that validity of $\phi$ implies
that for every possible instantiation of the $x$ variables, there
exist an instantiation of the $y$ variables such that $C_{1}$ to
$C_{k}$ in the second \LC statement in $\cplset$ evaluate to true.

Completeness of $Q$ follows from $\cplset$ and $\da$, if $Q$ returns
the same result over $\da$ and any ideal database instance $\di$
that subsumes $\da$ and $\cplset$ holds over $(\di,\da)$.

$Q$ returns nothing over $\da$. To make $Q$ return the empty tuple
over $\di$, one value from $\set{0,1}$ has to be inserted into each
ideal relation instance $\hat{R}_{i}$, because every predicate $R_{i}$
appears in $Q$, and every extension is empty in $\da$. This step
of adding any value from $\set{0,1}$ to the extensions of the $R$-predicates
in $\di$ corresponds to the universal quantification of the variables
$X$.

Now observe, that for the query to be complete, none of these combinations
of additions may be allowed. That is, every such addition has to violate
the \localComp constraint $\cplset$. As the extension of $R_{1}$
is empty in $\da$ as well, $\cplset$ becomes violated whenever adding
the values for the $R$-predicates leads to the existence of a satisfying
valuation of the body of $\cplset$. For the existence of a satisfying
valuation, the mapping of the variables $y$ is not restricted, which
corresponds to the existential quantification of the $y$-variables.

The reduction is correct, because whenever $\cplset,\da\models\Compl Q$
holds, for all possible additions of $\set{0,1}$ values to the extensions
of the $R$-predicates in $\di$ (all combinations of $x$), there
existed a valuation of the $y$-variables which yielded a mapping
from the $C$-atoms in $\cplset$ to the ground atoms of $C$ in $\da$,
that satisfied the existential quantified formula in $\phi$.

It is complete, because whenever $\phi$ is valid, then for all valuations
of the $x$-variables, there exists an valuation for the $y$-variables
that satisfies the formula $\phi$, and hence for all such extensions
of the $R$-predicates in $\di$, the same valuation satisfied the
body of the complex completeness statement, thus disallowing the extension.
\end{proof}

\begin{proof}
[Proof (of the theorem)] For $\piptwo$-membership, consider the
following naive algorithm for showing nonentailment: Given a query
$\query{Q(\tpl x)}B$, completeness statements $\C$ and an available
database $\da$, one has to guess a tuple $\tpl d$ and an ideal database
$\di$ such that $(\di,\da)$ satisfy $\C$ but do not satisfy $\compl Q$,
because $\tpl d$ is in $Q(\di)$ but not in $Q(\da)$. Verifying
that $(\di,\da)$ satisfies $\C$ is a coNP problem, as one has to
find all tuples in $\di$ that are constrained by some statement in
$\C$. Verifying that $(\di,\da)$ does not satisfy $\compl Q$ via
$\tpl d$ is a coNP problem as well, because one needs to show that
$t$ is not returned over $\da$. If one can guess a $\di$ and a
$\tpl d$ that satisfy these two properties, the completeness of $Q$
is not entailed by $\C$ and $\da$.

Now observe that for the guesses for $\tpl d$, it suffices to use
the constants in $\tpl d$ plus as many new constants as the arity
of $Q$. Also for the guesses for $\di$, one needs only minimally
larger databases that allow to retrieve new tuples. Therefore, it
is sufficient to guess ideal databases of the form $(\da\cup vB)$,
where $v$ is some valuation using only constants in $\da$ plus a
fixed set of additional constants.

As the range of possible ideal databases is finite, and given a guess
for an ideal database, the verification that $(\di,\da)$ satisfy
$\C$ and do not satisfy $\compl Q$ are coNP problems, the problem
is in $\piptwo$ \wrt combined complexity.There are only finitely
many databases $\di$ to consider, as it suffices to consider those
that are the result of adding instantiations of the body of $Q$ to
$\da$. Furthermore, since for the valuations $v$ it suffices to
only use the constants already present in the database plus one fresh
constant for every variable in $Q$, the obtained data complexity
is polynomial.
\end{proof}
This result shows that reasoning with respect to a database instance
is considerably harder, as $\tcqc(\LRQ,\LRQ)$ was in PTIME before.

\subsection{Entailment of Query Completeness by Query Completeness}

Entailment of query completeness by query completeness has already
been discussed in Section \ref{sec:general:QC-QC}. For queries under
set semantics, the close connection to the open problem of conjunctive
query determinacy was shown. For queries under bag semantics, the
equivalence to query containment was shown.

In the following, we show that when reasoning \wrt a database instance,
both QC-QC entailment under set semantics and determinacy become decidable,
and describe an algorithm in $\Pi_{3}^{P}$ for both.

QC-QC entailment \wrt a database instance is defined as follows:
\begin{defn}
[QC-QC Instance Entailment]Let $\Q=\{Q_{1},\ldots,Q_{n}$\} be a
set of queries, $Q$ be a query and $\da$ be a database instance.
We say that completeness of $\Q$ entails completeness of $Q$ wrt.
$\da$, written 

\[
\compl{\Q}\models_{\da}\compl Q
\]
if and only if for all ideal databases $\di$ with $\da\subseteq\di$
it holds that if $Q_{1}(\da)=Q_{1}(\di),\ldots,Q_{n}(\da)=Q_{n}(\di)$,
then $Q(\da)=Q(\di)$.
\end{defn}
QC-QC Instance entailment can be decided as follows: To show that
the entailment does not hold, one has to guess a tuple $\tpl d$ and
an ideal database $\di$, such that the incomplete database $(\di,\da)$
satisfies $\compl{\Q}$ but does not satisfy $\compl Q$, because
$\tpl d\in Q(\di)$ but $\tpl d\not\in Q(\da)$. As in the proof of
Theorem \ref{thm:instancereasoning-tcqc}, for the guesses for $\di$,
it suffices to consider minimal extensions of $\da$ using some valuation
$v$ for $B$. Also for the range of the valuation $v$, one has to
consider only the constants in $\da$ plus a as many new constants
as there are variables in $Q$. With this algorithm, we obtain an
upper bound for the complexity of $\qcqc$ entailment \wrt database
instances as follows.
\begin{prop}
QC-QC instance entailment for relational conjunctive queries is in
$\Pi_{3}^{P}$ \wrt combined complexity.\end{prop}
\begin{proof}
Consider the algorithm from above. To show that the entailment does
not hold, it suffices to guess one valuation $v$ for the body $B$
of $Q$, such that the incomplete database $\D=(\da\cup vB,\da)$
satisfies $\compl{\Q}$ but $v\tpl x\not\in Q(\da)$. Verifying the
latter is a coNP problem. Verifying that $\D$ satisfies $\compl{\Q}$
is a $\piptwo$-problem, as, in order to show that $\D$ does not
satisfy $\compl{\Q}$, it suffices to guess one $Q_{i}\in\Q$ and
one tuple $\tpl c\in Q_{i}(\di)$, for which one then needs to show
that there is no valuation $v'$ for $Q$ that allows to retrieve
$\tpl c$ over $\da$.
\end{proof}
Interestingly, also query determinacy wrt. an instance can be solved
analogously. In the following definition, notice the similarity to
the definition of $\qcqc$ instance entailment above. The only difference
are the considered models, which, for $\qcqc$ are incomplete databases,
while for determinacy are arbitrary pairs of databases.
\begin{defn}
[Instance Query Determinacy] Given a set $\Q$ of queries $Q_{1}$
to $Q_{n}$, a query $Q$ and a database $D_{1}$, we say that $\Q$
determines $Q$ wrt. $D_{1}$, written 
\[
\Q\det_{D_{1}}Q
\]
if and only if for all databases $D_{2}$ it holds that if $Q_{1}(D_{1})=Q_{1}(D_{2})$\linebreak{}
$\wedge\cdots\wedge Q_{n}(D_{1})=Q_{n}(D_{2})$, then $Q(D_{1})=Q(D_{2})$.

Again, to show that the entailment does not hold, one has to guess
a tuple $\tpl d$ and a database $D_{2}$, such that the $\Q(D_{1})=\Q(D_{2})$
and $\tpl d\in Q(D_{2})$ but $\tpl d\not\in Q(D_{1})$. Now for $D_{2}$
we have to consider all minimal extensions not of $D_{1}$ itself
but of $\Q(D_{1})$. That is, given the result of the queries $\Q$
over $D_{1}$, we construct a v-table $T$ such that $\Q(D_{1})=\Q(T)$.
This construction can be done by choosing for each tuple $\tpl c'$
in $Q_{i}(D_{1})$ some valuation $v'$ that computed $\tpl c'$,
replacing the images of nondistinguished variables of $Q_{i}$ in
$v$ with new variables, and then taking the union of all the v-tables
for all the tuples in $Q_{i}(D_{1})$ and then the union over all
queries in $\Q$.

Having this v-table $T$, a minimal extension of $\Q(D_{1})$ is any
database $D_{2}=(\sigma T\cup\theta B)$, where $\sigma$ is an instantiation
for the v-table $T$, and $\theta$ is a valuation for the body $B$
of $Q$.

As before, for both valuations one has to consider only the constants
in $D_{1}$ plus a as many new constants as there are variables in
$Q$ and $T$. With this algorithm, we obtain an upper bound for the
complexity of $\qcqc$ entailment wrt.\  database instances as follows.\end{defn}
\begin{prop}
Instance query determinacy for relational conjunctive queries is in
$\Pi_{3}^{P}$.\end{prop}
\begin{proof}
Consider the algorithm from above. To show that determinacy does not
hold, it suffices to guess $\sigma$ for $T$ and $\theta$ for $B$,
such that $\Q(D_{1})=\Q(D_{2})$ but $\theta\tpl x\not\in Q(D_{1})$.
Verifying the latter is a coNP problem. Verifying that $\Q(D_{1})=\Q(D_{2})$
is a $\piptwo$-problem, as, in order to show that $\Q(D_{1})\neq\Q(D_{2})$,
one has to guess a valuation for some $Q_{i}\in Q$ that yields a
tuple $\tpl c'$ over $D_{2}$, and then has to show that $\tpl c'$
is not returned by $Q_{i}$ over $D_{1}$.
\end{proof}

\section{Related Work}

\label{sec:general:intro}

Open- and closed world semantics were first discussed by Reiter in
\cite{reiter1978closed}, where he formalized earlier work on negation
as failure \cite{clark1978negation} from a database point of view.
The closed-world assumption corresponds to the assumption that the
whole database is complete, while the open-world assumption corresponds
to the assumption that nothing is known about the completeness of
the database.

Abiteboul et al.\ \cite{abiteboul1991representation} introduced
the notion of certain and possible answers over incomplete databases.
Certain answers are those tuples that are in the query answer over
all possible completions the incomplete database, while possible answers
are those tuples that are in at least one such completion. The notions
can also be used over partially complete databases. Then, query completeness
can be seen as the following relation between certain and possible
answers: A query over a partially complete database is complete, if
the certain and the possible answers coincide.

Motro~\cite{motro_integrity} 
introduced the notion of partially incomplete and incorrect databases
as databases that can both miss facts that hold in the real world
or contain facts that do not hold there. He described partial completeness
in terms of \emph{query completeness} (QC) statements, which express
that the answer of a query is complete. The query completeness statements
express that to some parts of the database the closed-world assumption
applies, while for the rest of the database, the open-world assumption
applies. He studied how the completeness of a given query can be deduced
from the completeness of other queries. His solution was based on
rewriting queries using views: to infer that a given query is complete
whenever a set of other queries are complete, he would search for
a conjunctive rewriting in terms of the complete queries. This solution
is correct, but not complete, as later results on query determinacy
show: the given query may be complete although no conjunctive rewriting
exists

While Levy \etal\ could show that rewritability of conjunctive queries
as conjunctive queries is decidable \cite{levy:1995:answering:queries:using:views:pods},
general rewritability of conjunctive queries by conjunctive queries
is still open: An extensive discussion on that issue was published
in 2005 by Segoufin and Vianu where it is shown that it is possible
that conjunctive queries can be rewritten using other conjunctive
queries, but the rewriting is not a conjunctive query \cite{segoufin_vianu_determinacy}.
They also introduced the notion of query determinacy, which for conjunctive
queries implies second order rewritability. The decidability of query
determinacy for conjunctive queries is an open problem to date.

Halevy~\cite{levy_completeness} suggested \emph{local completeness}
statements, which we, for a better distinction from the QC statements,
call table completeness (TC) statements, as an alternate formalism
for expressing partial completeness of an incomplete database. These
statements allow one to express completeness of parts of relations
independent from the completeness of other parts of the database.
The main problem he addressed was how to derive query completeness
from table completeness (TC-QC). He reduced TC-QC to the problem of
queries independent of updates (QIU)~\cite{elkan_qiu}. However,
this reduction introduces negation, and thus, except for trivial cases,
generates QIU instances for which no decision procedures are known.
As a consequence, the decidability of TC-QC remained largely open.
Moreover, he demonstrated that by taking into account the concrete
database instance and exploiting the key constraints over it, additional
queries can be shown to be complete.

Etzioni \etal~\cite{EtzioniEtAl-Sound:and:efficient:closed:world:reasoning:for:planning-AI}
discussed completeness statements in the context of planning and presented
an algorithm for querying partially complete data. Doherty \etal~\cite{DohertyEtAl-Efficient:reasoning:using:the:LCWA-AIMSA}
generalized this approach and presented a sound and complete query
procedure. Furthermore, they showed that for a particular class of
completeness statements, expressed using semi-Horn formulas, querying
can be done efficiently in $\PTIME$ \wrt data complexity.

Demolombe~\cite{demo1,demo2} captured Motro's definition of completeness
in epistemic logic and showed that in principle this encoding allows
for automated inferences about completeness.

Denecker~\etal~\cite{Denecker:Calabuig-logical_theory_partial_databases:tods:10}
studied how to compute possible and certain answers over a database
instance that is partially complete. They showed that for first-order
TC statements and queries, the data complexity of TC-QC entailment
wrt.\ a database instance is in $\CONP$ and $\CONP$-hard for some
TC statements and queries. Then they focused on approximations for
certain and possible answers and proved that under certain conditions
their approximations are exact.

In the Diplomarbeit (master thesis) of Razniewski \cite{razniewski:diplom:thesis:2010}
it was shown that TC-TC entailment and query containment are equivalent
(Section \ref{sec:general:TC-TC}), and that TC-QC entailment for
queries under bag semantics can be reduced to query containment (Theorem
\ref{theorem_characterizing_query_completeness} (i)). Also, reasoning
wrt.\ database instance was discussed, and the combined complexity
of TC-QC reasoning was shown, and Theorem \ref{thm:weakest-precond}
was contained there, although it was erroneously claimed to hold for
conjunctive queries, while so far it is only proven to hold for relational
queries. Furthermore, it was shown that TC-QC reasoning for databases
that satisfy finite domain constraints is $\piptwo$-complete.

Fan and Geerts \cite{Fan:Geerts-relative_information_completeness:pods:09}
discussed the problem of query completeness in the presence of master
data. In this setting, at least two databases exist: one master database
that contains complete information in its tables, and other, possibly
incomplete periphery databases that must satisfy certain inclusion
constraints wrt.\ the master data. Then, in the case that one detects
that a query over a periphery database contains already all tuples
that are maximally possible due to the inclusion constraints, one
can conclude that the query is complete. The work is not comparable
because completeness is not deduced from metadata but from an existing
data source, the master data, which gives an upper bound for the data
that other databases can contain. 

Abiteboul \etal~\cite{AbiteboulEtAl-Representing:and:querying:incomplete:XML-TODS}
discussed representation and querying of incomplete semistructured
data. They showed that the problem of deciding query completeness
from stored complete query answers, which corresponds to the QC-QC
problem raised in~\cite{motro_integrity} for relational data, can
be solved in PTIME\ \wrt data complexity.

Other work about completeness focused on completeness in sensor networks~\cite{naumann_biswas_completeness}.

\section{Summary}

\label{conclusion} In this chapter we have discussed three main inference
problems: The entailment of table completeness by table completeness
(TC-TC entailment), the entailment of query completeness by table
completeness (TC-QC entailment) and the entailment of query completeness
by query completeness (QC-QC entailment).

For the first problem of TC-TC entailment, we have shown that it naturally
corresponds to query containment and also has the same complexity.

For the second problem of TC-QC entailment, we have shown that for
queries under bag semantics, query completeness can be characterized
by table completeness and thus TC-QC entailment can be reduced to
TC-TC entailment. We have also shown the hardness of TC-QC under bag
semantics and that for queries under set semantics without projections
the same holds.

For queries under set semantics, we have shown that for minimal queries
without comparisons, weakest preconditions in terms of TC statements
can be found, thus again allowing to reduce TC-QC to TC-TC. For other
queries, we have given a direct reduction to query containment, and
also shown that the complexities achieved by this reduction are tight.
Whereas TC-QC under bag and set semantics mostly have the same complexity,
we have shown that for TC statements without comparisons or selfjoins,
but queries with both, the problem for queries under set semantics
is harder than under bag semantics.

For the third problem of QC-QC entailment, we have shown its close
correspondence to the problem of query determinacy, and that QC-QC
entailment for queries under bag semantics is decidable.

A surprising insight of this chapter may be that while query containment
for queries under bag semantics is usually harder than for queries
under set semantics, both the TC-QC and also the QC-QC entailment
reasoning for queries under bag semantics is easier than for queries
under set semantics.

The existence of weakest preconditions also for queries under set
semantics that contain comparisons remain open. In the following chapter
we discuss several extensions to the core framework by either extending
the formalism or by taking into account the actual database instance.

We have also discussed two extensions of the core relational model
that can be taken into account in completeness reasoning: Aggregate
queries and instance reasoning.

For aggregate queries, we have shown how the reasoning can be performed
and that for the aggregate functions COUNT and SUM, completeness reasoning
has the same complexity as for nonaggregate queries under bag semantics,
while for the functions MIN and MAX, it has the same complexity as
for queries under set semantics.

For the instance reasoning, we have shown that TC-QC reasoning becomes
harder, again jumping from NP to $\piptwo$, while for QC-QC entailment,
we have shown that the problem becomes decidable.

In the next chapter, we look into another interesting extension, namely
into databases with null values.

\chapter{Databases with Null Values}

\label{chap:nulls}In this section we extend the previous results
for relational queries to databases that contain null values. As arithmetic
comparisons can be seen as orthogonal to null values, we consider
only relational queries in this chapter.

Null values as used in SQL are ambiguous. They can indicate either
that no attribute value exists or that a value exists, but is unknown.
We study completeness reasoning for the different interpretations.
We show that when allowing both interpretations at the same time,
it becomes necessary to syntactically distinguish between different
kinds of null values. We present an encoding for doing that in standard
SQL databases. With this technique, any SQL DBMS evaluates complete
queries correctly with respect to the different meanings that null
values can carry.

The results in this section have been published at the CIKM 2012 conference
\cite{razniewski:nutt-CIKM2012}.

In Section \ref{nulls:sub:framework} we extend the previous formalisms
for incomplete databases and table completeness to databases with
null values. Section \ref{nulls:sub:specific:reasoning} presents
the reasoning for simple, uniform meanings of null value. Section
\ref{nulls:sub:make_nulls_explicit} shows how the different meanings
of nulls can be made explicit in standard SQL databases, while Sections
\ref{nulls:sub:reasoning_different_nulls} shows that reasoning is
possible in that case. Section \ref{nulls:sub:bag:semantics} discusses
the reasoning for queries under bag semantics, and in Section \ref{nulls:sub:complexity}
we summarize the complexity results and compare them with the results
for databases without null values.

\section{Introduction}

Practical SQL databases may contain null values. These null values
are semantically ambiguous, as they may mean that a value is missing,
non existing, or it is unknown which of the two applies. The different
meanings have different implications on completeness reasoning:
\begin{example}
Consider the table \emph{result(name,subject,grade)}, where name and
subject are the key of the table, and the table contains only the
record \emph{(John,Pottery,null)}. Then, if the null value means that
no grade was given to John, the database is not incomplete for a query
for all pottery grades of John. If the null means that the grade is
unknown then the query is incomplete, while if it is unknown which
of the two applies, the query may or may not be incomplete.
\end{example}
Classic work on null values by Codd introduced them for missing values
\cite{codd_null}. In recent work, Franconi and Tessaris~\cite{Franconi:Et:Al-SQL:Nulls-AMW}
have shown that the SQL way to evaluate queries over instances with
nulls captures exactly the semantics of attributes that are not applicable.

In this chapter, we show that it is important to disambiguate the
meaning of null values, and will present a practical way to do so.

\label{sec:extensions:null-values}

\section{Framework for Databases with Null Values}

\label{nulls:sub:framework}

In the following, we adapt the notion of incomplete database to allow
null values, and table completeness statements to allow to specify
the completeness only of projections of tables.

A problem with nulls as used in standard SQL databases is their ambiguity,
as those nulls may mean both that an attribute value exists but is
unknown, or that no value applies to that attribute. The established
models of null values, such as Codd, v-, and c-tables~\cite{imielinski_lipski_representation_systems},
avoid this ambiguity by concentrating on the aspect of unknown values.
In this work, we consider the ambiguous standard SQL null values~\cite{codd_null},
because those are the ones used in practice. Null values mainly have
two meanings: 
\begin{itemize}
\item an attribute value exists, but is \emph{unknown}; 
\item an attribute value does not exist, the attribute is \emph{not applicable}. 
\end{itemize}
In database theory, unknown values are represented by so-called \emph{Codd
nulls,} which are essentially existentially quantified first-order
variables. A relation instance with Codd nulls, called a \emph{Codd
table}, represents the set of all regular instances that can be obtained
by instantiating those variables with non-null values~\cite{foundations_of_dbs}. 

For a conjunctive query $Q$ over an instance with Codd nulls, say
$D_{\codd}$, one usually considers \emph{certain answer} semantics~\cite{foundations_of_dbs}:
the result set $Q_{\cert}(D_{\codd})$ consists of those tuples that
are in $Q(D')$ for every instantiation $D'$ of $D_{\codd}$. The
set $Q_{\cert}(D_{\codd})$ can be computed by evaluating $Q$ over
$D_{\codd}$ while treating each occurrence of a null like a different
constant and then dropping tuples with nulls from the result. Formally,
using the notation 
\begin{equation}
\nonulls{Q(D)}:=\set{\tpl d\in Q(D)\mid\tpl d\mbox{\ does not contain nulls}},\label{eq-nonulls:operator}
\end{equation}
 this means 
$Q_{\cert}(D_{\codd})=\nonulls{Q(D_{\codd})}$.

The null values supported by SQL (\quotes{SQL nulls} in short) have
a different semantics than Codd nulls. Evaluation of first order queries
follows a three-valued semantics with the additional truth value \textit{unknown}.
For a conjunctive query $Q$, we say that $y$ is a \emph{join variable}
if $y$ occurs at least twice in the body of $Q$ and a \emph{singleton
variable} otherwise. If $D_{\sql}$ contains facts with null values,
then under SQL's semantics the result of evaluating $Q(\tpl x)$ over
$D_{\sql}$ is 
\begin{equation}
Q_{\sql}(D_{\sql})=\set{v\tpl x\mid v\mbox{\ maps no join variable to nulls}}.\label{eq-sql:semantics}
\end{equation}
 To see this, note that a twofold occurrence of a variable $y$ is
expressed in an SQL query by an equality between two attributes, which
evaluates to \textit{unknown} if a null is involved.

Franconi and Tessaris~\cite{Franconi:Et:Al-SQL:Nulls-AMW} have shown
that the SQL way to evaluate queries over instances with nulls captures
exactly the semantics of attributes that are not applicable. To make
this more precise, suppose that $R$ is an $n$-ary relation with
attribute set $X:=\set{\dd An}$. If each attribute in an $R$-tuple
can be null, then $R$ can be seen as representing for each $Y\incl X$
a relation $R_{Y}$ with attribute set $Y$. In this perspective,
an instance of $R$ with tuples containing nulls represents a collection
of $2^{n}$ instances of the relations $R_{Y}$, where a tuple $\tpl d$
belongs to the instance $R_{Y}$ iff the entries in $\tpl d$ for
the attributes in $Y$ are not null. In other words, null values are
padding the positions that do not correspond to attributes of $R_{Y}$.
\begin{example}
\label{ex:difference:of:certain:answer:semantics} Consider the query
$Q$ that asks for all classes whose form teacher is also form teacher
of a class with arts as profile, which we write as 
\[
\query{Q(c_{1})}{\class(s_{1},c_{1},t,p),\class(s_{2},c_{2},t,\mathrm{'arts'})}
\]
 and consider the instance $D=\set{\class(\schoolone,1a,\NULL,\mathrm{'arts'})}$.
If we interpret $\NULL$ as Codd-null, then $(1a)\in Q_{\codd}(D)$.
If we evaluate $Q$ under the standard SQL semantics, we have that
$(1a)\notin Q_{\sql}(D)$.

Suppose we know that class 1a has a form teacher. Then whoever the
teacher of that class really is, the class has a teacher who teaches
a class with arts as profile and the interpretation of the null value
as Codd-null is correct. If the null however means that the class
has no form teacher, the SQL interpretation is correct.
\end{example}
Note that certain answer semantics and SQL semantics are not comparable
in that the former admits more joins, while the latter allows for
nulls in the query result. Later on we will show how for complete
queries we can compute certain answers from SQL answers by simply
dropping tuples with nulls.

We will say that a tuple with nulls representing an unknown but existing
value is an \emph{incomplete tuple}, since this nulls indicate the
absence of existing values. We say that a tuple where nulls represent
that no value exists is a \emph{restricted tuple}, because only the
not-null values in the tuple are related to each other. When modeling
databases with null values, we will initially not syntactically distinguish
between different kinds of null values and assume that some atoms
in an instance contain the symbol $\NULL$.

\subsection{Incomplete Databases with Nulls}

In Section \ref{sec:prelim:incomplete-databases}, incomplete databases
were modeled as pairs $(\di,\da)$, where $\da$ is contained in $\di$.
When allowing null values in databases, we have to modify this definition.

To formalize that the available database contains less information
than the ideal one we use the concept of fact dominance (not be mixed
with query dominance):
\begin{defn}
Let $R(\tpl s)$ and $R(\tpl d)$ be atoms that possibly contain nulls.
Then the fact $R(\tpl s)$ is \emph{dominated} by the fact $R(\tpl d)$,
written $R(\tpl s)\domby R(\tpl d)$, if $R(\tpl s)$ is the same
as $R(\tpl d)$, except that $R(\tpl s)$ may have nulls where $R(\tpl d)$
does not. An instance $D$ is \emph{dominated\/} by an instance $D'$,
written $D\domby D'$, if each fact in $D$ is dominated by some fact
in $D'$. \end{defn}
\begin{example}
Consider the two facts $\student(\John,\Null,\schoolone)$ and $\student(\John,3a,\schoolone).$
Then the former is dominated by the latter, because the null value
of the first fact is replaced by the constant 'A'.
\end{example}
By monotonicity of conjunctive queries we can immediately state the
following observation:
\begin{prop}
[Monotonicity] \label{prop-monotonicity} Let $Q$ be a conjunctive
query and $D$, $D'$ be database instances with nulls. Suppose that
$D$ is dominated by $D'$. Then $Q_{\cert}(D)\incl Q_{\cert}(D')$
and $Q_{\sql}(D)\domby Q_{\sql}(D')$. \end{prop}
\begin{defn}
An \emph{incomplete database} is a pair of database instances $(\di,\da)$
such that $\da$ is dominated by $\di$. Based on the previous discussion
of the possible semantics of null values, we distinguish two special
cases of incomplete databases: 
\begin{enumerate}
\item We say that $\pdb$ is an incomplete database \emph{with restricted
facts} if $\da\incl\di$. Note that in this case the ideal state may
contain nulls and that every fact in the available state must appear
in the same form in the ideal state. Thus, a null in the position
of an attribute means that the attribute is not applicable and nulls
are interpreted the way SQL does. 
\item The pair $(\di,\da)$ is an incomplete database \emph{with incomplete
facts} if $\di$ does not contain any nulls and $\da$ is dominated
by $\di$. In this case, there are no nulls in the ideal state, which
means that all attributes are applicable, while the nulls in the available
state indicate that attribute values are unknown. Therefore, those
nulls have the same semantics as Codd nulls.
\end{enumerate}
\end{defn}
\begin{table}
\begin{centering}
{\small{}}%
\begin{tabular}{|cccc|c|ccc|}
\multicolumn{8}{c}{{\small{}$\id D$}}\tabularnewline
\cline{1-4} \cline{6-8} 
\multicolumn{4}{|c|}{{\small{}class }} &  & \multicolumn{3}{c|}{{\small{}student}}\tabularnewline
{\small{}school} & {\small{}code} & {\small{}formTeacher} & {\small{}profile} &  & {\small{}name} & {\small{}class} & {\small{}school}\tabularnewline
\cline{1-4} \cline{6-8} 
{\small{}$\schoolone$} & \emph{\small{}1a} & \emph{\small{}Smith} & \emph{\small{}arts} &  & \emph{\small{}John} & \emph{\small{}1a} & {\small{}$\schoolone$}\tabularnewline
{\small{}$\schoolone$} & \emph{\small{}2b} & \emph{\small{}Rossi} & {\small{}$\bot$} &  & \emph{\small{}Mary} & {\small{}$\bot$} & {\small{}$\schoolone$}\tabularnewline
\cline{1-4} 
\multicolumn{1}{c}{} &  &  & \multicolumn{1}{c}{} &  & \textit{\small{}Paul} & \textit{\small{}2b} & \textit{\small{}$\schooltwo$}\tabularnewline
\cline{6-8} 
\multicolumn{5}{c}{} & \multicolumn{3}{c}{}\tabularnewline
\multicolumn{8}{c}{{\small{}$\av D$}}\tabularnewline
\cline{1-4} \cline{6-8} 
\multicolumn{4}{|c|}{{\small{}class }} &  & \multicolumn{3}{c|}{{\small{}student}}\tabularnewline
{\small{}school} & {\small{}code} & {\small{}formTeacher} & {\small{}profile} &  & {\small{}name} & {\small{}class} & {\small{}school}\tabularnewline
\cline{1-4} \cline{6-8} 
{\small{}$\schoolone$} & \emph{\small{}1a} & \emph{\small{}Smith} & \emph{\small{}arts} &  & \emph{\small{}John} & \emph{\small{}1a} & {\small{}$\schoolone$}\tabularnewline
{\small{}$\schoolone$} & \emph{\small{}2b} & \emph{\small{}Rossi} & {\small{}$\bot$} &  & \emph{\small{}Mary} & {\small{}$\bot$} & {\small{}$\schoolone$}\tabularnewline
\cline{1-4} \cline{6-8} 
\end{tabular}
\par\end{centering}{\small \par}

\caption{Incomplete database with restricted facts}

\label{fig:pdb:example:resf} 
\end{table}

\begin{example}
\label{ex:IDB} Recall the school database from our running example,
defined in Section~\ref{sec:intro-motivating-example}. In Table~\ref{fig:pdb:example:resf}
we see an incomplete database with restricted facts for this scenario.
The null values appearing in the available database mean that no value
exists for the corresponding attributes. The $\class$ table shows
that no profile has been assigned to class 2b and that Mary is an
external student not belonging to any class.

In contrast, Table~\ref{fig:pdb:example:incf} shows an incomplete
database with incomplete facts. Here, null values in the available
database mean that a value exists but is unknown. So, class 1a has
a form teacher, but we do not know who. Class 2b has a profile, but
we do not now which. John is in some class, but we do not know which
one.

Observe that in both kinds of incomplete databases, some facts, such
as the one about Paul being a student, can be missing completely.

\begin{table}[t]
\begin{centering}
{\small{}}%
\begin{tabular}{|cccc|c|ccc|}
\multicolumn{8}{c}{{\small{}$\id D$}}\tabularnewline
\cline{1-4} \cline{6-8} 
\multicolumn{4}{|c|}{{\small{}class }} &  & \multicolumn{3}{c|}{{\small{}student}}\tabularnewline
{\small{}school} & {\small{}code} & {\small{}formTeacher} & {\small{}profile} &  & {\small{}name} & {\small{}class} & {\small{}school}\tabularnewline
\cline{1-4} \cline{6-8} 
{\small{}$\schoolone$} & \emph{\small{}1a} & \emph{\small{}Smith} & \emph{\small{}arts} &  & \emph{\small{}John} & \emph{\small{}1a} & {\small{}$\schoolone$}\tabularnewline
{\small{}$\schoolone$} & \emph{\small{}2b} & \emph{\small{}Rossi} & \emph{\small{}science} &  & \emph{\small{}Mary} & \emph{\small{}2b} & {\small{}$\schoolone$}\tabularnewline
\cline{1-4} 
\multicolumn{1}{c}{} &  &  & \multicolumn{1}{c}{} &  & \textit{\small{}Paul} & \textit{\small{}2b} & \textit{\small{}$\schooltwo$}\tabularnewline
\cline{6-8} 
\multicolumn{5}{c}{} & \multicolumn{3}{c}{}\tabularnewline
\multicolumn{8}{c}{{\small{}$\av D$}}\tabularnewline
\cline{1-4} \cline{6-8} 
\multicolumn{4}{|c|}{{\small{}class }} &  & \multicolumn{3}{c|}{{\small{}student}}\tabularnewline
{\small{}school} & {\small{}code} & {\small{}formTeacher} & {\small{}profile} &  & {\small{}name} & {\small{}class} & {\small{}school}\tabularnewline
\cline{1-4} \cline{6-8} 
{\small{}$\schoolone$} & \emph{\small{}1a} & \emph{\small{}Smith} & \emph{\small{}arts} &  & \emph{\small{}John} & {\small{}$\bot$} & {\small{}$\schoolone$}\tabularnewline
{\small{}$\schoolone$} & \emph{\small{}2b} & \emph{\small{}Rossi} & {\small{}$\bot$} &  & \emph{\small{}Mary} & \emph{2b} & {\small{}$\schoolone$}\tabularnewline
\cline{1-4} \cline{6-8} 
\end{tabular}
\par\end{centering}{\small \par}

\caption{Incomplete database with incomplete facts}

\label{fig:pdb:example:incf} 
\end{table}

\end{example}
In practice, null values of both meanings will occur at the same time,
which may lead to difficulties if they cannot be distinguished.

\subsection{Query Completeness}

The result of query evaluation over databases with null values may
vary depending on whether the nulls are interpreted as Codd or as
SQL nulls.

Consider databases with incomplete facts. Then nulls are interpreted
as Codd nulls and queries are evaluated under certain answer semantics.
While $\da$ may contain nulls, $\di$ does not and $Q_{\cert}(\di)=Q(\di)$. 
\begin{defn}
Let $Q$ be a query and $\D$ be an IDB with incomplete facts. Then
for $\ast\in\set{s,b}$ 
\begin{equation}
\pdb\models_{\IF}\textit{Compl}^{\ast}(Q)\quad\mbox{iff}\quad Q^{\ast}(\di)=Q_{\cert}^{\ast}(\da)\label{eq-query:completeness:for:pdbs:with:incomplete:facts}
\end{equation}

\end{defn}
That is, the tuples returned by $Q$ over $\di$ are also returned
over $\da$ if nulls are treated according to certain answer semantics.
Conversely, that every null-free tuple returned over $\da$ is also
returned over $\di$ follows by monotonicity from the fact that $\da\domby\di$
(Proposition~\ref{prop-monotonicity}).

Consider databases with partial facts. Then nulls are interpreted
as SQL nulls and queries are evaluated under SQL semantics. 
\begin{defn}
Let $Q$ be a query and $\D$ be an incomplete database with partial
facts. Then for $\ast\in\set{s,b}$
\end{defn}
\begin{equation}
\pdb\models_{\RF}\textit{Compl}^{\ast}(Q)\quad\mbox{iff}\quad Q_{\sql}^{\ast}(\di)=Q_{\sql}^{\ast}(\da).\label{eq-query:completeness:for:pdbs:with:restricted:facts}
\end{equation}
Again, the crucial part is that tuples returned by $Q$ over $\di$
are also returned over $\da$ if nulls are treated according to SQL
semantics, while the converse inclusion holds due to monotonicity.
\begin{example}
\label{ex:query:completeness:arts:students} The query $\query{Q_{\mathrm{art\_students}}(n)}{\student(n,c,s),}$
\linebreak{}
$\class(s,c,f,\mathrm{'arts'})$ asks for the names of students in
classes with arts as profile. Over $\di$ in Table~\ref{fig:pdb:example:incf}
it returns the singleton set $\{(\mathrm{John})\}$ and over $\da$
as well. Therefore, $Q_{\mathrm{art\_students}}$ is complete over
that partial database.

In contrast, $\query{Q_{\mathrm{schools}}(s)}{\student(n,c,s)}$ is
not complete over this database, because it returns $\{(\mathrm{\schoolone}),(\mathrm{\schooltwo})\}$
over the ideal database but only $\{(\schoolone),(\bot)\}$ over the
available one.
\end{example}

\subsection{Table Completeness Statements with Projection}

Null values may lead to tuples becoming incomplete at certain positions.
Therefore, it can now happen that tables are complete for some columns
while for others they are not complete. To be able to describe such
completeness, we extend table completeness statements to talk about
the completeness of projections of tables.
\begin{defn}
[TC Statements] A \emph{table completeness} statement, written $\tc{R(\tpl s)}PG$,
consists of three components: (i)~a relational atom $R(\tpl s)$,
(ii)~a set of numbers $P\incl\set{1,\ldots,\arity R}$, and (iii)~a
condition $G$. The numbers in $P$ are interpreted as attribute positions
of~$R$. 
\end{defn}
For instance, if $R$ is the relation $\student$, then $\set{1,\,3}$
refers to the attributes $\name$ and $\school$.
\begin{defn}
[Satisfaction of TC Statements] Let $C=\tc{R(\tpl s)}PG$ be a TC
statement and $\pdb=(\di,\da)$ an incomplete database. An atom $R(\tpl u)\in\di$
is \emph{constrained by $C$} if there is a valuation $v$ such that
$\tpl u=v\tpl s$, $R(v\tpl s)\in\di$, and $vG\incl\di$. An atom
$R(\tpl u')\in\da$ is an \emph{indicator for $R(\tpl u)$ wrt $C$}
if $\tpl u[P]=\tpl{u'}[P]$, where $\tpl u[P]$ is the projection
of $\tpl u$ onto the positions in $P$. We say that \emph{$C$ is
satisfied by $\pdb$} if for every atom $R(\tpl u)\in\di$ that is
constrained by $C$ there is an indicator $R(\tpl u')\in\da$.\end{defn}
\begin{example}
\label{ex:tc:statement:students:from:art:classes} In our school scenario,
the TC statement 
\begin{equation}
\tc{\student(n,c,s)}{\set{1,3}}{\class(s,c,f,\arts)}\label{eqn-tc:students:from:arts:classes}
\end{equation}
 states, intuitively, that the available database contains for all
students of classes with arts as profile the name and the class. However,
the student's hometown need not be present. Over the ideal database
in Example~\ref{ex:IDB}, the fact $\student(\mathrm{John},\mathrm{1a},\schoolone)$
is constrained by the statement~(\ref{eqn-tc:students:from:arts:classes}).
Any fact $\student(\mathrm{John},\NULL,\schoolone)$, $\student(\mathrm{John},\mathrm{1a},\schoolone)$
or $\student(\mathrm{John},\mathrm{1a},\schooltwo)$ in $\da$ would
be an indicator. In the database in Table~\ref{fig:pdb:example:incf}
the first fact is present, and therefore Statement~(\ref{eqn-tc:students:from:arts:classes})
is satisfied over it.
\end{example}
As seen in \ref{sec:prelim:table-completeness}, the semantics of
TC statements can also be expressed using tuple-generating dependencies.
The TGDs are more complex now, as they contain an existentially quantified
variable in place of each attribute that is projected out: 

For instance, Statement~(\ref{eqn-tc:students:from:arts:classes})
would have the following TGD associated: 
\[
\tcruleEx{\id{\class}(s,c,f,\arts),\,\id{\student}(n,c,s)}{c'}{\av{\student}(n,c',s)}.
\]
 To simplify our notation, we assume that the projection positions
$P$ are the first $k$ positions of $R$ and that $\tpl s$ has the
form $(\tpl s',\tpl s'')$, where $\tpl s'$ has length $k$ and $\tpl s''$
has length $\arity R-k$. Then, for a completeness statement $C=\tc{R(\tpl s)}PG$
its corresponding TGD $\Rule C$ is 
\[
\tcruleEx{\id G,\id R(\tpl s',\tpl s'')}{\tpl z}{\av R(\tpl s',\tpl z)},
\]
 where $\tpl z$ is a tuple of distinct fresh variables that has the
same length as $\tpl s''$. Again, for every TC statement $C$, an
incomplete database satisfies $C$ in the sense defined above if and
only if it satisfies the rule $\Rule C$ in the classical sense of
rule satisfaction.

Note that our definition of when a TC statement is satisfied takes
into account null values. Regarding nulls in $\da$, we treat nulls
like non-null values and consider their presence sufficient to satisfy
an existential quantification in the head of a TC rule.

Nulls in $\di$, however, have to be taken into account when evaluating
the body of a rule. Since nulls in the ideal database always represent
the absence of a value, we always interpret the rules that we associated
with TC statements under SQL semantics.

\section{Reasoning for Specific Nulls}

\label{nulls:sub:specific:reasoning}

In this section we discuss reasoning for databases where the meaning
of nulls is unambiguous. In~\ref{subsection:reasoning:incomplete:facts},
we assume that nulls always mean that a value is missing but exists,
while in~\ref{subsec:reasoning:restricted:facts}, we assume that
nulls mean that a value is inapplicable. For both cases we give decidable
characterizations of TC-QC entailment. Moreover, we show that evaluation
under certain answer and under SQL semantics lead to the same results
for minimal complete queries.

\subsection{Incomplete Facts}

\label{subsection:reasoning:incomplete:facts} 

We suppose we are given a set of TC statements $\C$ and a conjunctive
query~$Q$, which is to be evaluated under set semantics. We say
that $\C$ entails $\compls Q$ over IDBs with \emph{incomplete facts},
written 
\begin{equation}
\C\models_{\IF}\compls Q,\label{eqn-TC:QC:Inc:Facts:sets}
\end{equation}
 iff for every such IDB $\pdb$ we have that 
\[
\pdb\models\C\quad\mbox{implies}\quad\pdb\models_{\IF}\compls Q.
\]

To decide the entailment of query completeness by table completeness,
we extend the $\tcop$ operator from Definition \ref{def:tc-operator},
which for every TC statement $\C$ maps an instance $D$ to the least
informative instance $T_{\C}(D)$ such that $(D,T_{\C}(D))\models\C$.
Let $C=\tc{\, R(\tpl s',\tpl s'')}P{G\,}$ be a TC statement, where
without loss of generality, $\tpl s'$ consists of the terms in the
positions $P$. We define the query $Q_{C}$ by the rule 
\begin{equation}
\query{Q_{C}(\tpl s',\tpl{\NULL})}{R(\tpl s',\tpl s''),\, G}.\label{eqn-aux:query:for:T:C}
\end{equation}
 This means, given an instance $D$, the query $Q_{C}$ returns for
every~$\alpha$ satisfying the condition ${R(\tpl s',\tpl s''),G}$,
a tuple $(\alpha\tpl s',\tpl{\NULL})$ that consists of the projected
part $\alpha\tpl s'$ and is padded with nulls $(\NULL)$ for the
positions projected out. We then define 
\begin{equation}
T_{C}(D):=\set{R(\tpl d)\mid\tpl d\in Q_{C}(D)}\label{eq:T:C:Operator}
\end{equation}
 and $T_{\C}(D):=\bigcup_{C\in\C}T_{C}(D)$.

Intuitively, for a database instance $\di$ and a TC statement $C$,
the function $T_{C}$ calculates the minimal information that any
available database $\da$ must contain in order that $(\di,\da)$
together satisfy $C$. Observe that every atom in $T_{C}(D)$ is an
indicator for some $R(\tpl u)$ in $D$ wrt $C$. This is the case
because every fact in $T_{C}(D)$ is created as an indicator for some
fact in $D$ constrained by $C$. Observe also that in general, $T_{\C}(D)$
may contain more facts than $D$, because several TC statments may
constrain the same atom and therefore several indicators are produced.
\begin{example}
Consider the TC statement $C$ defined in Example~\ref{ex:tc:statement:students:from:art:classes}
as $\tc{\student(n,c,s)}{\set{1,3}}{\class(s,c,f,\arts)}$. The corresponding
query is $\query{Q_{C}(n,\NULL,s)}{\student(n,c,s),\class(s,c,f,\arts)}$.
For the partial database in Table~\ref{fig:pdb:example:incf}, $Q_{C}(\di)$
is $\set{(John,\NULL,\schoolone)}$ and hence $T_{C}(\di)=\set{\student(John,\NULL,\schoolone)}$,
which is the minimal information that any available database must
contain to satisfy together with $\di$ the TC statement $C$. 
\end{example}
Similarly to the properties of $T_{C}$ over databases without nulls,
as stated in Lemma \ref{lemma_reasoning_1}, the following properties
now hold for the function $T_{\C}$:
\begin{prop}
\label{prop-properties:of:TC:transformation-incF} Let $D$ be a database
instance without nulls and let $\D_{0}$ be the incomplete database
$(D,T_{\C}(D))$. Then 
\begin{enumerate}
\item $T_{\C}(D)$ is dominated by $D$, 
\item $\D_{0}$ is an IDB with incomplete facts, and 
\item $\D_{0}\models\C$. 
\end{enumerate}
Moreover, if $D'$ is another instance such that $(D,D')$ is an IDB
with incomplete facts that satisfies $\C$, then $D'$ dominates $T_{\C}(D)$. 
\end{prop}
The following characterization of TC-QC-entailment over IDBs with
incomplete facts says that completeness of $Q$ wrt.\ $\C$ can be
checked by evaluating $Q$ over $T_{\C}(L)$.
\begin{thm}
\label{thm-characterisation:TC:QC:inc} \label{th:characterisation:TC-QC:incf}
Let $\query{Q(\tpl x)}L$ be a conjunctive query and $\C$ be a set
of table completeness statements. Then 
\[
\C\ensuremath{\models_{\IF}\compls Q}\quad\mbox{iff}\quad\tpl x\in Q_{\cert}(T_{\C}(L)).
\]
\end{thm}
\begin{proof}
\OnlyIf By Proposition~\ref{prop-properties:of:TC:transformation-incF},
$(L,T_{\C}(L))$ is an IDB with incomplete facts that satisfies $\C$.
Thus, by assumption, $(L,T_{\C}(L))\models_{\IF}\compls Q$, which
implies $Q^{s}(L)=Q_{\cert}^{s}(T_{\C}(L))$. The identity from $L$
to $L$ is a satisfying assignment for $Q$ over $L$, from which
it follows that $\tpl x\in Q(L)$, and hence $\tpl x\in Q_{\cert}(T_{\C}(L))$.

\If Suppose that $\tpl x\in Q_{\cert}(T_{\C}(L))$. We show that
$\C\ensuremath{\models_{\IF}\compls Q}$. Let $\D=(\di,\da)$ be an
IDB with incomplete facts that satisfies $\C$. We show that $Q^{s}(\di)=Q_{\cert}^{s}(\da)$.
Note that we only have to show $Q^{s}(\di)\incl{Q_{\cert}^{s}(\da)}$,
since the other inclusion holds by monotonicity (Proposition~\ref{prop-monotonicity}).
Let $\tpl d\in Q^{s}(\di)$. We show that $\tpl d\in{Q_{\cert}^{s}(\da)}$.

There is a valuation $\delta$ such that $\delta L\incl\di$ and $\delta\tpl x=\tpl d$.
We will construct a valuation $\delta'$ such that $\delta'L\incl\da$
and $\delta'\tpl x=\tpl d$. To define $\delta'$, we specify how
it maps atoms of $L$ to $\da$.

Let $A$ be an atom in $L$. Since $\tpl x\in Q_{\cert}(T_{\C}(L))$,
there is a homomorphism $\theta$ from $L$ to $T_{\C}(L)$ such that
$\theta\tpl x=\tpl x$. Let $B'=\theta A\in T_{\C}(L)$. By construction
of $T_{\C}(L)$, there is a TC-statement $C=\tc BPG$ such that $(B,G)\incl L$
and $B'$ has been constructed as indicator for $B$ wrt $C$. Since
$\delta L\incl\di$, we have $(\delta B,\delta G)\incl\di$. Clearly,
$\delta B$ is constrained by $C$ over $\D$. Since $\D\models\C$,
there is an indicator atom $\tilde{B}$ for $\delta B$ in $\da$.
We now define $\delta'A:=\tilde{B}$.

For $\delta'$ to be well-defined, we have to show that $\delta'$
induces a mapping on the terms of $L$, that is, (i)~if $A$ contains
a constant $c$ at position $p$, then $\delta'A$ contains $c$ at
$p$, (ii)~if $A_{1}$ contains variable $y$ at position $p_{1}$,
and $A_{2}$ contains $y$ at $p_{2}$, then $\delta'A_{1}$ and $\delta'A_{2}$
have the same term at position $p_{1}$ and $p_{2}$, respectively.

To see this, let $c$ be in $A$ at $p$, which we denote as $c=A[p]$.
Since $\theta$ is a homomorphism, $B'[p]=(\theta A)[p]=c$. By construction
of $T_{\C}(L)$, we have a statement $C\in\C$ such that $p\in P$,
the set of projected positions of $C$, and $B[p]=B'[p]$. Moreover,
since $\delta$ is a homomorphism, we have that $(\delta B)[p]=B[p]$.
As $\tilde{B}$ is an indicator for $\delta B$ wrt $C$, and $p\in P$,
it follows that $\tilde{B}[p]=(\delta B)[p]$. In summary, $(\delta'A)[p]=\tilde{B}[p]=A[p]$.

Next, suppose that $A_{1}[p_{1}]=A_{2}[p_{2}]=y$. We will show that
$\tilde{B}_{1}[p_{1}]=\tilde{B}_{2}[p_{2}]$. Since $\theta$ is a
homomorphism, it holds that $B'_{1}[p_{1}]=B'_{2}[p_{2}]$. By construction
of $T_{\C}(L)$, we have a statement $C_{1}$ such that $p_{1}\in P_{1}$,
the set of projected positions of $C_{1}$, and $B_{1}[p_{1}]=B'_{1}[p_{1}]$.
An analogous argument holds for $B'_{2}$, so $B_{1}[p_{1}]=B_{2}[p_{2}]$.
Moreover, since $\delta$ is a homomorphism, we have that $(\delta B_{1})[p_{1}]=(\delta B_{2})[p_{2}]$.
As $\tilde{B}_{1}$ is an indicator for $\delta B_{1}$ wrt $C_{1}$,
and $p_{1}\in P_{1}$, it follows that $\tilde{B}_{1}[p_{1}]=(\delta B_{1})[p_{1}]$
An analogous statement holds for $\delta B_{2}$. Therefore, it also
holds that $\tilde{B}_{1}[p_{1}]=\tilde{B}_{2}[p_{2}]$.
\end{proof}
The intuition of this theorem is the following: To check whether completeness
of a query $Q$ is entailed by a set of TC statements $\C$, we perform
a test over a prototypical database: Considering the body of the query
as an ideal database, we test whether the satisfaction of the TC statements
$\C$ implies that there is also enough information in any available
database to return the tuple of the distinguished variables $\tpl x$.
If that is the case, then also for any other tuple found over an ideal
database, there is enough information in the available database to
compute that tuple again.
\begin{example}
\label{ex:prototypical:reasoning} Consider again the query from Example~\ref{ex:query:completeness:arts:students},
which is $\query{Q_{\mathrm{art\_students}}(n)}{\mathrm{student}(n,c,s),\mathrm{class}(s,c,f,\mathrm{'arts'})}$.
Suppose we are given TC statements $C_{1}=\tc{\class(s,c,f,p)}{\{1,2,3,4\}}{\true}$
and $C_{2}=\tc{\student(n,c,s)}{\{1,3\}}{\class(s,c,f,p)}$, which
state that complete facts about all classes are in our database, and
that for all students from art classes the name and the school attribute
are in the database. When we want to find out whether $C_{1}$ and
$C_{2}$ imply that query $Q_{\mathrm{art\_students}}$ returns a
complete answer, we proceed according to Theorem~\ref{thm-characterisation:TC:QC:inc}
as follows: 
\begin{enumerate}
\item We take the body of the query $Q_{\mathrm{art\_students}}$ as a prototypical
test database: $L=\set{\student(n,c,s),\class(s,c,f,p)}$.
\item We apply the functions $T_{C_{1}}$ and $T_{C_{2}}$ to $L$ to generate
the minimum information that can be found in any available database
if the TC statements are satisfied: $T_{C_{1}}(L)=\set{\class(s,c,f,p)}$
and $T_{C_{2}}(L)=\set{\student(n,\NULL,s)}$.
\item We evaluate $Q_{\mathrm{art\_students}}$ over $T_{C_{1}}(L)\cup T_{C_{2}}(L)$.
The result is $\set{(n)}$.
\end{enumerate}
The tuple $(n)$ is exactly the distinguished variable of $Q_{\mathrm{art\_students}}$.
Therefore, we conclude that $C_{1}$ and $C_{2}$ entail query completeness
under certain answer semantics. 
\end{example}
We will discuss the complexity of reasoning in detail in Section~\ref{nulls:sub:complexity}.
At this point we already remark that the reasoning is in NP for relational
conjunctive queries, since all that needs to be done is query evaluation,
first of the TC rules in order to calculate $T_{\C}(L)$, second of
$Q$, in order to check whether $\tpl x\in Q(T_{\C}(L))$. Also, for
relational conjunctive queries without self-joins the reasoning can
be done in polynomial time.

So far we have assumed that nulls in the available database are treated
as Codd nulls and that queries are evaluated under certain answer
semantics. Existing DBMSs, however, implement the SQL semantics of
nulls, which is more restrictive, as it does not allow for joins involving
nulls, and thus leads to fewer answers. In the following we will show
that SQL semantics gives us the same results as certain answer semantics
for a query $Q$, if $Q$ is complete and minimal.

In analogy to \quotes{$\models_{\IF}$}, we define for an IDB with
incomplete facts $\pdb=(\di,\da)$ the satisfaction of query completeness
as follows: 
\[
\pdb\models_{\IF,\sql}\compls Q\mbox{\ \ \ \ if and only if\ \ \ \ }Q^{s}(\di)=\nonulls{Q_{\sql}^{s}(\da)}
\]
Moreover, we write $\C\models_{\IF,\sql}\compls Q$ if and only if
$\pdb\models_{\IF,\sql}\compls Q$ for all IDBs where $\pdb\models\C$.
Intuitively, \quotes{$\models_{\IF,\sql}$} is similar to \quotes{$\models_{\IF}$},
with the difference that queries over $\da$ are evaluated as by an
SQL database system.

We show that query completeness for this new semantics can be checked
in a manner analogous to the one for certain answer semantics in Theorem~\ref{thm-characterisation:TC:QC:inc}.
The proof is largely similar.
\begin{lem}
\label{lem-characterization:of:inc:sql} Let $\query{Q(\tpl x)}L$
be a conjunctive query and $\C$ be a set of table completeness statements.
Then 
\[
\C\ensuremath{\models_{\IF,\sql}\compls Q}\quad\mbox{iff}\quad\tpl x\in Q_{\sql}(T_{\C}(L)).
\]

\end{lem}
Now, suppose that the conjunctive query $Q$ is minimal (cf.~\cite{ChandraM77}).
Then $Q$ returns a result over $T_{\C}(L)$ only if each atom from
$L$ has an indicator in $T_{\C}(L)$. The next lemma shows that it
does not matter whether the nulls in $T_{\C}(L)$ are interpreted
as SQL or as Codd nulls.
\begin{lem}
\label{lem-codd:sql:equivalence:for:minimal:queries} Let $\C$ be
a set of TC statements and $\query{Q(\tpl x)}L$ be a minimal conjunctive
query. Then 
\[
\tpl x\in Q_{\sql}(T_{\C}(L))\quad\mbox{iff}\quad\tpl x\in Q_{\cert}(T_{\C}(L)).
\]

\end{lem}
Combining Theorem~\ref{thm-characterisation:TC:QC:inc} and Lemmas~\ref{lem-characterization:of:inc:sql}
and~\ref{lem-codd:sql:equivalence:for:minimal:queries} we conclude
that for minimal conjunctive queries that are known to be complete,
it does not matter whether one evaluates them under certain answer
or under SQL semantics.
\begin{thm}
Let $\pdb=(\di,\da)$ be an incomplete database with incomplete facts,
let $\C$ be a set of TC statements, and let $\query{Q(\tpl x)}L$
be a minimal conjunctive query. If $\C\models_{\IF}\compls Q$ and
if $\pdb\models\C$ then $Q_{\cert}^{s}(\da)=\nonulls{Q_{\sql}^{s}(\da)}$. \end{thm}
\begin{proof}
Let $Q$ be a minimal query. If $\C\models_{\IF}\compls Q$, then
$\tpl x\in Q_{\cert}(T_{\C}(L))$ by Theorem $\ref{thm-characterisation:TC:QC:inc}$
which implies $\tpl x\in Q_{\sql}(T_{\C}(L))$ by Lemma $\ref{lem-codd:sql:equivalence:for:minimal:queries}$,
from which we conclude $\C\models_{\IF,\sql}\compls Q$ by Lemma $\ref{lem-characterization:of:inc:sql}$.

As a consequence, for any IDB $\pdb=(\di,\da)$ such that $\pdb\models\C$
we have $Q_{\cert}^{s}(\da)=Q^{s}(\di)=\nonulls{Q_{\sql}^{s}(\da)}$.
\end{proof}
It follows that for complete queries we also get a complete query
result when evaluating them over standard SQL databases.
\begin{example}
Consider again the query $Q_{\mathrm{art\_students}}$ from Example~\ref{ex:query:completeness:arts:students},
where $\query{Q_{\mathrm{art\_students}}(n)}{\student(n,c,s),\class(s,c,f,\mathrm{'arts'})}$,
and the TC statements $C_{1}$ and $C_{2}$ from Example~\ref{ex:prototypical:reasoning}
that entailed query completeness over IDBs with incomplete facts.
Since $Q_{\mathrm{art\_students}}$ has no self-joins it is clearly
minimal, and hence over the available database of any IDB that satisfies
$C_{1}$ and $C_{2}$ we can evaluate it under set semantics and will
get a complete query result.
\end{example}

\subsection{Restricted Facts}

\label{subsec:reasoning:restricted:facts}

We now move to IDBs with restricted facts. Recall that in this case
a null in a fact indicates that an attribute is not applicable. Accordingly,
an IDB with restricted facts is a pair $(\di,\da)$ where both the
ideal and the available database may contain nulls, and where the
available is a subset of the ideal database $(\da\subseteq\di)$.

Again, we suppose that we are given a set of TC statements $\C$ and
a conjunctive query $\query{Q(\tpl x)}L$, which is to be evaluated
under set semantics. Similar to the case of incomplete facts, we say
that \emph{$\C$ entails $\compls Q$ over IDBs with restricted facts},
written 
\begin{equation}
\C\models_{\RF}\compls Q,\label{eqn-TC:QC:Par:Facts:sets}
\end{equation}
 iff for every such IDB $\pdb$ we have that 
\[
\pdb\models\C\quad\mbox{implies}\quad\pdb\models_{\RF}\compls Q.
\]

We will derive a characterization of (\ref{eqn-TC:QC:Par:Facts:sets})
that can be effectively checked. We reuse the function $T_{\C}$ defined
in Equation~(\ref{eq:T:C:Operator}), for which now the following
properties hold (see also Proposition \ref{prop-properties:of:TC:transformation-incF}).
\begin{prop}
Let $D$ be an instance that may contain nulls and let $\D_{1}=(D\cup T_{\C}(D),T_{\C}(D))$.
Then 
\begin{enumerate}
\item $\D_{1}$ is an IDB with restricted facts; 
\item $\D_{1}\models\C$. 
\end{enumerate}
Moreover, if $D'$ is another instance such that $(D\cup D',D')$
is an IDB with restricted facts that satisfies $\C$, then $D'$ dominates
$T_{\C}(D)$. 
\end{prop}
In contrast to databases with incomplete facts, nulls can now appear
in the output of queries over the ideal database, and therefore must
not be ignored in query answers over the available database. Recent
results in~\cite{Franconi:Et:Al-SQL:Nulls-AMW} imply that for queries
over databases with restricted facts, evaluation according to SQL's
semantics of nulls returns correct results.

The characterization of completeness entailment is different now because
$Q$'s body $L$ is no more a prototypical instance for $Q$ to retrieve
an answer $\tpl x$. Since the ideal database may now contain nulls,
we must consider the case that variables in $L$ are mapped to~$\NULL$
when $Q$ is evaluated over $\di$.

We first present a result for \emph{boolean\/} queries, that is,
for queries where the tuple of distinguished variables $\tpl x$ is
empty, and for \emph{linear\/} (or self-join free) queries, that
is, queries where no relation symbol occurs more than once.

A variable $y$ in a query $\query{Q(\tpl x)}L$ is a \emph{singleton\/}
variable, if it appears only once in $L$. Recall that only singleton
variables can be mapped to $\NULL$ when evaluating $Q$ under SQL
semantics. Let $L^{\NULL}$ and $\tpl x^{\NULL}$ be obtained from
$L$ and $\tpl x$, respectively, by replacing all singleton variables
with $\NULL$.
\begin{thm}
\label{thm:characterization:set:sem:for:boolean:and:linear:queries}
Let $\query{Q(\tpl x)}L$ be a boolean or linear conjunctive query
and $\C$ be a set of table completeness statements. Then 
\[
\C\models_{\RF}\compls Q\quad\mbox{iff}\quad\tpl x^{\NULL}\in Q_{\sql}(T_{\C}(L^{\NULL})).
\]

\end{thm}
The theorem reduces completeness reasoning in the cases above to conjunctive
query evaluation. We conclude that deciding TC-QC entailment wrt.\
databases with restricted facts is in PTIME for linear and NP-complete
for arbitrary boolean conjunctive queries.

For general conjunctive queries, which may have distinguished variables,
evaluating $Q$ over a single test database obtained from $L$ is
not enough. We can show, however, that it is sufficient to consider
all cases where singleton variables in $L$ are either null or not.
A \emph{null version\/} of $L$ is a condition obtained from $L$
by replacing some singleton variables with $\NULL$. Any valuation
$v$ for $L$ that replaces some singleton variables of $L$ with
$\NULL$ and is the identity otherwise leads to a null version $vL$
of $L$.
\begin{thm}
\label{thm:characterization:tc:qc:restricted:facts} Let $\query{Q(\tpl x)}L$
be a conjunctive query. Then the following are equivalent: 
\begin{itemize}
\item $\C\models_{\RF}\compls Q$; 
\item $v\tpl x\in Q_{\sql}(T_{\C}(vL))$, for every null version $vL$ of
$L$. 
\end{itemize}
\end{thm}
The theorem says that instead of just one prototypical case, we have
to consider several now, because query evaluation for databases with
nulls is more complicated: while the introduction of nulls makes the
satisfaction of TC statements and the query evaluation more difficult,
it also creates more possibilities to retrieve null as a result (see~\cite{containment:dbs:with:nulls}).

The above characterisation can be checked by a $\piptwo$ algorithm:
in order to verify that containment does not hold, it suffices to
guess one null version $vL$ and then show that $v\tpl x$ is not
in $Q(T_{\C}(vL))$, which is an NP task. 

\subsection{Ambiguous Nulls}

\label{sec:one:ambiguous:null} So far we have assumed that nulls
have one of two possible meanings, standing for unknown or for non-existing
values. In this section we discuss completeness reasoning in the presence
of one syntactic null value, which can have three possible meanings,
the previous two plus indeterminacy as to which of those two applies.
This is the typical usage of nulls in SQL.

We model IDBs for this case as pairs $\pdb=(\di,\da)$, where both
instances, $\di$, $\da$, may contain $\NULL$ and each tuple in
$\da$ is dominated by a tuple in $\di$. We assume that queries are
evaluated as in SQL, since we cannot tell which nulls are Codd-nulls
and which not. For a query $Q$ and $\ast\in\set{s,b}$ we define
\begin{equation}
\pdb\models_{\AN}\textit{Compl}^{\ast}(Q)\quad\mbox{iff}\quad Q_{\sql}^{\ast}(\di)=Q_{\sql}^{\ast}(\da).\label{}
\end{equation}

Different from the case where nulls stand for unknown values, we may
not drop nulls in the query result over the available database, because
they might carry information (absence of a value).

We observe that without further restrictions on the IDBs, for many
queries there is no way to conclude query completeness from table
completeness.
\begin{prop}
\label{prop:keys:needed:for:bag:semantics:and:ambiguous:nulls:in:set:semantics}
There exists an IDB $\D$ with ambiguous nulls and a query $Q$, such
that $\D$ satisfies any set of TC statements but $\D$ does not satisfy
$\compls Q$. \end{prop}
\begin{proof}
Let $\D$ be with $\di=\set{\student(\mathrm{Mary},\mathrm{2a},\schoolone)}$
and $\da=\set{\student(\mathrm{Mary},\mathrm{2a},\mathrm{Chester}),\student(\mathrm{Mary},\NULL,\schoolone)}$.
Clearly, $\D$ satisfies all possible TC statements, because every
fact from the ideal database is also in the available database. But
the query $\query{Q_{\mathrm{classes}}(c)}{\student(n,c,s)}$ is not
complete over $\pdb$, because $Q^{s}(\di)=\set{(2a)}$ while $Q^{s}(\da)=\set{(2a),(\NULL)}$.
\end{proof}
Inspecting $\pdb$ in the proof above more closely, we observe that
the two facts in $\da$ are dominated by the same fact in the ideal
database. Knowing that, we can consider the second fact in $\da$
as redundant: it does not add new information about Mary. This duplicate
information leads to the odd behaviour of $\D$ wrt completeness:
while all information from the ideal database is also in the available
database, $Q(\da)$ contains an additional fact with a null.

Sometimes, such duplicates occur naturally, e.g., when data from different
sources is integrated. In other scenarios, however, redundancies are
unlikely because objects are identified by keys, and only non-key
attributes may be unknown or non-applicable.

In a school database, it can happen that address or birth place of
a student are unknown. In contrast, it is hard to imagine that one
may want to store a fact $\student(\NULL,\NULL,\schoolone)$, saying
that there is a student with unknown name and class at the $\schoolone$.

Keys alone, however, are still not sufficient:
\begin{example}
Suppose we are given an incomplete database with\linebreak{}
 $\di=\{\student(\Mary,\twoA,\schoolone),\student(\Paul,\twoA,\schoolone)\}$
and\linebreak{}
 $\da=\set{\student(\Mary,\twoA,\schoolone),\student(Paul,\NULL,\schoolone)}$.
\linebreak{}
Observe that there are no redundant tuples in $\da$. The TC statement\linebreak{}
 $\tc{\student(n,c,s)}{\set 2}{\true}$, which says that all classes
from the ideal database are also in the available database, is satisfied
over this IDB. One might believe that over an IDB satisfying this
statement the query $Q_{\mathrm{classes}}$, defined above, is complete,
as it is the case for IDBs with incomplete facts or with restricted
facts. However, query evaluation returns that $Q_{\mathrm{classes}}^{s}(\di)=\set{(\twoA)}$
while $Q_{\mathrm{classes}}^{s}(\da)=\set{(\twoA),(\NULL)}$.
\end{example}
The problem with ambiguous nulls is that while all information needed
for computing a query result may be present in the available database,
it is not clear how to treat a null in the query answer. If it represents
an unknown value, we can discard it because the value will still be
there explicitly. But if it represents that no value exists, it should
also show up in the query result.

Therefore, we conclude that one should disambiguate the meaning of
null values. In the next section we propose how to do this in an SQL
database.

\section{Making Null Semantics Explicit}

\label{nulls:sub:make_nulls_explicit}

Nulls in an available database can express three different statements
about a value: absence, presence with the concrete value being unknown,
and indeterminacy which of the two applies. As seen in Section~\ref{sec:one:ambiguous:null},
this ambiguity makes reasoning impossible. To explicitly distinguish
between the three meanings of nulls in an SQL database, we present
an approach that adds an auxiliary boolean attribute to each attribute
that possibly has nulls as values.
\begin{example}
Consider relation $\student(\name,\code,\school)$. Imagine a student
John for whom the attribute $\code$ is null because John attends
a class, but the information was not entered into the database yet.
Imagine another student Mary for whom $\code$ is null because Mary
is an external student and does not attend any class. Imagine a third
student Paul for whom $\code$ is null because it is unknown whether
or not he attends a class. We mark the different meanings of nulls
by symbols $\bot_{\uk}$ (unknown but existing value), $\bot_{\na}$
(not applicable value) and $\bot_{\bot}$ (indeterminacy), but remark
that in practice, in an SQL database, all three cases would be expressed
using syntactically identical null values.

We can distinguish them, however, if we add a boolean attribute $\hascode$.
For John, the value of $\hascode$ would be $\true$, expressing that
the tuple for John has a code value, which happens to be unknown,
indicated by the $\NULL$ for $\code$. For Mary, $\code$ would have
the value $\false$, expressing that the attribute $\code$ is not
applicable. For Paul, the $\hascode$ attribute itself would be $\NULL$,
expressing that nothing is known about the actual value. Table~3
shows a $\student$ instance with explicit types of null, on the left
using three nulls, on the right with a single null and the auxiliary
attribute. 
\end{example}
\begin{table}
\begin{centering}
\begin{tabular}{|ccc|c|cccc|}
\cline{1-3} \cline{5-8} 
\multicolumn{3}{|c|}{student} & \multicolumn{1}{c|}{} & \multicolumn{4}{c|}{student}\tabularnewline
name & \ldots{} & code  &  & name  & \ldots{} & hasCode  & code\tabularnewline
\cline{1-3} \cline{5-8} 
\emph{Sara } &  & \emph{2a}  &  & \emph{Sara } &  & $\true$  & \emph{2a}\tabularnewline
\emph{John } &  & $\NULL_{\uk}$  &  & \emph{John } &  & $\true$  & $\NULL$\tabularnewline
\emph{Mary } &  & $\NULL_{\na}$  &  & \emph{Mary } &  & $\false$  & $\NULL$\tabularnewline
\emph{Paul } &  & $\NULL_{\NULL}$  &  & \emph{Paul } &  & $\NULL$  & $\NULL$\tabularnewline
\cline{1-3} \cline{5-8} 
\end{tabular}
\label{fig:making:nulls:explicit:example} 
\par\end{centering}

\caption{Making the semantics of nulls explicit}
\end{table}

In general, for an attribute $\attr$ where we want to disambiguate
null values, we introduce a boolean attribute $\hasattr$. We refer
to $\hasattr$ as the \emph{sign\/} of $\attr$, because it signals
whether a value exists for the attribute, no value exists, or whether
this is unknown.

Note that if $\hasattr$ is $\false$ or $\NULL$, then $\attr$ must
be $\NULL$. This can be enforced by an SQL check constraint.

As seen earlier, in general SQL semantics does not fully capture the
semantics of unknown nulls as it may miss some certain answers. We
will show in Theorem~\ref{thm-sql:implementation:of:combined:nulls},
that our encoding can be exploited to compute answer sets for complete
queries by joining attributes with nulls according to SQL semantics
and then using the signs to drop tuples with unknown and indeterminate
nulls.

\section{Reasoning for Different Nulls}

\label{nulls:sub:reasoning_different_nulls}


In the previous section we showed how to implement a syntactic distinction
of three different meanings of null values in SQL databases. In this
section we discuss how to reason with these three different nulls.

An instance $D$ with the three different kinds of nulls represents
an infinite set of instances $D'$ that can be obtained from $D$
by (i)~replacing all occurrences of $\NULL_{\uk}$ with concrete
values and (ii)~replacing all occurrences of $\NULL_{\NULL}$ with
concrete values or with $\NULL_{\na}$.

As usual, the set of \emph{certain answers} of a query $Q$ over $D$
consists of the tuples that are returned by $Q$ over all such $D'$
and is denoted as $Q_{\cert}(D)$.

It is easy to see that a tuple $\tpl d$ is in $Q_{\cert}(D)$ iff
the only nulls in $\tpl d$ are $\NULL_{\na}$ and there exists a
valuation $v$ such that (i) $\tpl d=v\tpl x$, (ii) $vL\subseteq D$,
(iii) $v$ does not map join variables to $\NULL_{\na}$ or $\NULL_{\NULL}$,
and (iv) no two occurrences of a join variable are mapped to different
occurrences of $\NULL_{\uk}$. Intuitively, this means that we have
to treat $\NULL_{\uk}$ as Codd null and the other nulls as SQL nulls.

We say that an incomplete database $\D=(\di,\da)$ contains \emph{partial
facts\/} if (i)~the facts in $\di$ may contain the null $\NULL_{\na}$,
(ii)~the facts in $\da$ may contain all three kinds of nulls, and
(iii)~each fact $R(\tpl d)\in\da$ is \emph{dominated} by a fact
$R(\tpl d')$ in $\di$ in the sense that for any position $p$ 
\vspace{5pt}
 \begin{compactitem} 

\item if $\tpl d[p]=\NULL_{\na}$, then also $\tpl d'[p]=\NULL_{\na}$, 

\item if $\tpl d[p]=\NULL_{\uk}$, then $\tpl d'[p]$ is a value
from the domain $\dom$, 

\item if $\tpl d[p]=\NULL_{\NULL}$, then $\tpl d'[p]$ is $\NULL_{\na}$
or in $\dom$, 

\item if $\tpl d[p]=d$ for a value $d\in\dom$, then also $\tpl d'[p]=d$.
\end{compactitem} \vspace{5pt}

We then say that a query is complete over a database $\D=(\di,\da)$
with partial facts, if $Q(\di)={Q_{\cert}(\da)}$, and write $\D\models_{\TN}\compl Q$.

Satisfaction of TC-statements is not affected by these changes, as
$\di$ contains only nulls $\NULL_{\na}$, which indicate restricted
facts that can be treated according to SQL semantics.


\begin{example}
Consider the available database $\da$ that contains the three facts
$\class(\schoolone,\mathrm{1a},\NULL_{\uk},\arts)$, $\class(\schoolone,\mathrm{2b},\NULL_{\na},\arts)$
and $\class(\schoolone,\mathrm{3c},\NULL_{\NULL},\arts)$. Also, consider
the query from Example~\ref{ex:difference:of:certain:answer:semantics}
that asks for all classes whose form teacher is also form teacher
of an arts class, written as $\query{Q(c_{1})}{\class(s_{1},c_{1},t,p_{1}),\class(s_{2},c_{2},t,\mathrm{'arts'})}$.

Then similar to before, the only tuple in $Q_{\cert}(\da)$ is $(1a)$,
because since the teacher of that class is unknown but existing, it
holds in any complete database that the class 1a has a teacher that
also teaches an arts class (1a again). The tuples 2b and 3c do not
show up in the result, because the former has no form teacher at all
$(\NULL_{\na})$, while the latter may or may not have a form teacher. 
\end{example}
A first result is that TC-QC entailment over IDBs with partial facts
is equivalent to entailment over IDBs with restricted facts:
\begin{thm}
\label{thm-characterization:TC:QC:three:nulls} Let $Q$ be a conjunctive
query and $\C$ be a set of TC statements. Then 
\[
\C\models_{\TN}\compls Q\quad\mbox{iff}\quad\C\models_{\RF}\compls Q.
\]
 \end{thm}
\begin{proof}
\OnlyIf Trivial, because IDBs with restricted facts are IDBs with
partial facts that contain only the null value $\NULL_{\na}$. 

\If Assume, $\C\not\models_{\TN}\compls Q$. Then there is an IDB
with partial facts $\D$ such that $\D\models\C$, but $\D\not\models_{\TN}\compls Q$.
We construct an IDB $\D_{0}$ with restricted facts that also satisfies
$\C$, but does not satisfy $\compls Q$. Let $\da_{0}=\da[\NULL_{\uk}/\NULL_{\na},\NULL_{\NULL}/\NULL_{\na}]$
be the variant of $\da$ where $\NULL_{\uk}$ and $\NULL_{\NULL}$
are replaced by $\NULL_{\na}$, and let $\di_{0}=\di\cup\da_{0}$.

The additional facts in $\di_{0}$ do not lead to violations of TC
statements, since they are dominated by facts in $\di$, thus, $\D_{0}\models\C$.
However, $Q(\da_{0})\incl Q(\da)$, since changing nulls to $\NULL_{\na}$
makes query evaluation more restrictive, and $Q(\di)\incl Q(\di_{0})$
due to monotonicity. Hence, $Q(\da_{0})\subsetneqq Q(\di_{0})$, that
is, $\D_{0}\not\models_{\RF}\compls Q$. 
\end{proof}
Also, we define the query evaluation $\nonullsmod{Q(D)}$ as $Q(D)$
without all tuples containing $\NULL_{\uk}$ or $\NULL_{\NULL}$.

Similar to a database with incomplete facts only, it holds that query
answering for minimal queries that are complete does not need to take
into account certain answer semantics but can safely evaluate the
query using standard SQL semantics:
\begin{thm}
\label{thm-sql:implementation:of:combined:nulls} Let $\D=(\di,\da)$
be an incomplete database with partial facts, $Q$ be a minimal conjunctive
query and $\C$ be a set of table completeness statements. If $\C\models_{\TN}\compls Q$
and $\D\models\C$ then $Q_{\cert}^{s}(\da)=\nonullsmod{Q_{\sql}^{s}(\da)}$. 
\end{thm}

\section{Queries under Bag Semantics}

\label{nulls:sub:bag:semantics}

Bag semantics is the default semantics of SQL queries, while set semantics
is enabled with the \texttt{DISTINCT} keyword. As the next example
shows, for relations without keys reasoning about query completeness
under bag semantics may not be meaningful.


\begin{example}
Consider the incomplete database with incomplete facts $\pdb=(\di,\da)$,
where $\di=\set{\student(\Mary,\twoA,\Chester)}$ and\linebreak[4]
$\da=\set{\student(\Mary,\twoA,\schoolone),\student(\Mary,\NULL,\schoolone)}$.\linebreak{}
 Since it is a priori not possible to distinguish whether the fact
containing $\bot$ is redundant, the boolean query $\query{Q()}{\student(n,c,s)}$
that is just counting the number of students is not complete, because
the redundant tuple in the available database leads to a miscount. 
\end{example}
As tuples with nulls representing unknown values can introduce redundancies,
we require that keys are declared for IDBs with incomplete facts,
with one ambiguous null or with partial facts. Only for incomplete
databases with restricted facts keys are not necessary, because there
the available database is always a subset of the ideal one and hence
no redundancies can appear.

Formally, for a relation $R$ with arity $n$, a \emph{key} is a subset
of the attribute positions $\set{1,\ldots,n}$. Without loss of generality
we assume that the key attributes are the first $k(R)$ attributes,
where $k$ is a function from relations to natural numbers. An instance
$D$ \emph{satisfies the key} of a relation $R$, if (i) no nulls
appear in the key positions of facts and (ii) no two facts have the
same key values, that is, if for all $R({\tpl d})$, $R({\tpl d}')\in D$
it holds that $\tpl d[1..k(R)]=\tpl d'[1..k(R)]$ implies $\tpl d=\tpl d'$,
where $\tpl d[1..k(R)]$ denotes the restriction of $\tpl d$ to the
positions $1..k(r)$.

Table completeness statements that do not talk about all key attributes
of a key are not useful for deciding the entailment of query completeness
under bag semantics, because, intuitively, they cannot assure that
the right multiplicity of information is in the available database.
We say that a TC statement $\tc{R(\tpl x)}PG$ is \emph{key-preserving},
if $\set{1..k(R)}\subseteq P$. In the following, we only consider
TC statements that are key-preserving.

We develop a characterization for TC-QC entailment that is similar
to the one for set semantics. However, now we need to ensure that
over a prototypical database not only query answers but also valuations
are preserved, because the same query answer tuple can be produced
by several valuations. So if a valuation is missing, the multiplicity
of a tuple in the result is incorrect. As a consequence, a set of
TC statements may entail completeness of a query $Q$ for set semantics,
but not for bag semantics.
\begin{example}
The relation $\result(\name,\subject,\grade)$ stores the language
courses that students take. Consider the query 
\[
\query{Q_{\mathrm{nr\_for\_french}}(n)}{\result(n,French,g)},\result(n,s,g'),
\]
which counts for each student that took French, how many courses he/she
attends in total. Under set semantics, $Q_{\mathrm{nr\_for\_french}}$
is complete if $\da$ contains all facts about French courses, which
is expressed by the TC statement $C_{\mathrm{french}}=$$\tc{\result(n,\mathrm{French})}{\set{1,2}}{\true}$.
To test completeness for set semantics, we apply $T_{C}$ to the query
body $L$, which results in $T_{C}(L)=\set{\result(n,\mathrm{French},\bot)}$,
since the first body atom is not constrained by $C_{\mathrm{french}}$.
Evaluating $Q_{\mathrm{nr\_for\_french}}$ over $T_{C}(L)$ returns
$(n)$, which shows set completeness.

But this does not entail that $Q_{\mathrm{nr\_for\_french}}$ is complete
under bag semantics. The IDB $(L,T_{C}(L)$ is a counterexample: it
satisfies $C$ and we can evaluate $Q_{nr\_for\_french}$ over $L$
two times, while over $T_{\C}(L)$ just once. If Paul takes French
and Spanish according to $\di$, it is clearly not sufficient to only
have the fact about French in $\da$ when we want to count how many
courses Paul takes.
\end{example}
We therefore modify the test criterion in Theorem~\ref{thm:characterization:set:sem:for:boolean:and:linear:queries}
in two ways.

For a query $\query{Q(\tpl x)}L$, the tuple $\tpl w$ of \emph{crucial
variables\/} consists of the variables that are in $\tpl x$ or occur
in key positions in $L$. For any two valuations $\alpha$ and $\beta$
that satisfy $L$ over a database $D$, we have that $\alpha$ and
$\beta$ are identical if they agree on $\tpl w$. Thus, the crucial
variables determine both, the answers of $Q$ and the multiplicities
with which they occur. We associate to $Q$ the query $\query{\bar{Q}(\tpl w)}L$
that has the same body as $Q$, but outputs all crucial variables.
Consequently, $Q$ is complete under set semantics if and only if
$\bar{Q}$ is complete under set semantics. The first modification
of the criterion will consist in testing $\bar{Q}$ instead of $Q$.

A direct implication of the first modification is that we need not
consider several null versions $vL$ of $L$ as in Theorem~\ref{thm:characterization:tc:qc:restricted:facts}.
The reason for doing so was that a null $\NULL=\alpha x$ in the output
of $Q$ over $vL$ could have its origin in an atom $vA$ in $vL$
such that $x$ does not occur in $A$, but another variable, say $y$
is instantiated to $\NULL$. Now, the query $\bar{Q}$ passes the
test for set completeness only if an atom in $L$ is mapped to an
atom with the same key values. Thus, a variable $x$ cannot be bound
to a null $\NULL=\gamma y$. Hence it suffices to consider just the
one version $L^{\bot}$ where all singleton variables are mapped to
null. By the same mapping, ${\tpl w}^{\bot}$ is obtained from $\tpl w$.

The second modification is due to the possibility that several TC
statements constrain one fact in $\di$ and thus $T_{\C}$ generates
several indicators. Since we assumed TC statements to be key-preserving,
these indicators all agree on their key positions. However, in some
non-key position one indicator may have a null while another one has
a non-null value. So, $T_{\C}(L^{\NULL})$ may not satisfy the keys.
This can be repaired by \quotes{chasing} $T_{\C}(L^{\NULL})$ (cf.~~\cite{foundations_of_dbs}).

The function $\chase$ takes a database $D$ with nulls as input and
merges any two $R$-facts $A'$, $A''$ that have the same key values
into one $R$-fact $A$ as follows: the value of $A$ at position
$p$, denoted $A[p]$ is $A'[p]$ if $A'[p]\neq\NULL$ and is $A''[p]$
otherwise. Clearly, if $\C$ is key-preserving and $D$ satisfies
the keys, then $\chase(T_{\C}(D))$ also satisfies the keys. Intuitively,
$\chase$ condenses information by applying the key constraints. Obviously,
$\chase$ runs in polynomial time.
\begin{example}
Let $\name$ be the key of the relation $\student$.\linebreak{}
Consider the database instance $D=\set{\student(\Mary,\twoA,\schoolone)}$,
and consider the set $\C=\set{C_{1},C_{2}}$ of key-preserving TC
statements where $C_{1}=\tc{\student(n,c,s)}{\set{1,2}}{\true}$ and
$C_{2}=\linebreak\tc{\student(n,c,s)}{\set{1,3}}{\true}$. Without
taking into account \linebreak{}
the key, the instance $T_{\C}(D)$ is $\{\student(\Mary,\twoA,\NULL),$\linebreak{}
$\student(\Mary,\NULL,\schoolone)\}$. The \chase function unifies
the two facts, therefore, $\chase(T_{\C}(D))=\set{\student(\Mary,\twoA,\schoolone)}$. 
\end{example}
We now are ready for our characterization of completeness entailment
under bag semantics, which is similar, but slightly more complicated
than the one in Theorem~\ref{thm:characterization:set:sem:for:boolean:and:linear:queries}.
\begin{thm}
\label{thm:characterization:entailment:bag:semantics} Let $\query{Q(\tpl x)}L$
be a conjunctive query and $\C$ be a set of key-preserving TC statements.
Then 
\[
\C\models_{\TN}\complb Q\quad\mbox{iff}\quad{\tpl w}^{\bot}\in\bar{Q}(\chase(T_{\C}(L^{\bot}))).
\]

\end{thm}
Since the criterion holds for incomplete databases with three different
nulls, it holds also for the special cases where only one type of
null values is present (restricted or incomplete facts).

Notably, it also holds for incomplete databases with one ambiguous
null, because when keys are present and TC statements guarantee that
all mappings are preserved, no additional nulls can show up in the
query result.

\section{Complexity of Reasoning}

\label{nulls:sub:complexity} 

We now discuss the complexity of inferring query completeness from
table completeness. 
We define $\tcqc_{\star}$ as the problem of deciding whether under
$\star$-semantic for all incomplete databases $\D$ it holds that
$\D\models\C$ implies that $\D\models\compl Q$, where both the query
and the TC statements are formulated using relational conjunctive
queries (that is, queries without comparisons). We will find that
for all cases considered in the paper, the complexity of reasoning
is between NP and $\piptwo$:
\begin{thm}
[Complexity Bounds] \label{thm-complexity:bounds} \mbox{} 
 
\begin{itemize}
\item $\tcqc_{\IF}^{s}$ is NP-complete; 
\item $\tcqc_{\RF}^{s}$ is NP-hard and in $\piptwo$; 
\item $\tcqc_{\TN}^{s}$ is NP-hard and in $\piptwo$; 
\item $\tcqc_{\TN}^{b}$ is NP-complete. 
\end{itemize}

\end{thm}
\begin{proof}
NP-hardness in all four cases can be shown by a reduction of containment
of Boolean conjunctive queries, which is known to be NP-complete~\cite{ChandraM77}.
We sketch the reduction for (1)--(3), the one for (4) being similar.
Suppose we want to check whether $\query{Q()}L$ is contained in $\query{Q'()}{L'}$.
Let $P$ be a new unary relation. Consider the query $\query{Q_{0}()}{P(a),L}$
and the TC statement $C_{0}=\tc{P(a)}{\set 1}{L'}$. Let $\C$ consist
of $C_{0}$ and the statement that $R$ is complete for every relation
$R$ in $L$. Then it follows from Theorems~\ref{thm-characterisation:TC:QC:inc},
\ref{thm:characterization:set:sem:for:boolean:and:linear:queries}
and \ref{thm-characterization:TC:QC:three:nulls} that $\C\models_{\ast}\compls Q$,
where $\ast\in\set{\IF,\RF,\TN}$, if and only if $P(a)\in T_{\C}(L)$,
$P(a)\in T_{\C}(L^{\NULL})$, and $P(a)\in T_{\C}(L^{\NULL})$, respectively.
The latter three conditions hold iff $P(a)=T_{C_{0}}(P(a),L)$, which
holds iff $Q$ is contained in $Q'$.


Problem 1 is in NP, because according to Theorem~\ref{thm-characterisation:TC:QC:inc}
to show that the entailment holds, it suffices to construct $T_{\C}(L)$
by guessing valuations that satisfy sufficiently many TC statements
in $\C$ over $L$, and to guess a valuation that satisfies $Q$ over
$T_{\C}(L)$ such that the tuple $\tpl x$ is returned.

Problem 2 is in $\piptwo$, because according to the characterization
in Theorem~\ref{thm:characterization:tc:qc:restricted:facts}, to
show that entailment does not hold, it suffices to guess one null
version $\gamma L$ of the body of $Q$ and show that $\gamma\tpl x$
is not in $Q(\chase(T_{\C}(\gamma L)))$, which is an NP task.

Problem 3 is in $\piptwo$ for the same reason.

Problem 4 is in NP, because we do not consider different nullversions
of $L$ but only one. The remaining argument is the same as for Problem
1, since one needs to show that $\tpl x^{\bot}$ is in $T_{\C}(L^{\bot})$,
which is an NP task. \end{proof}

Reasoning becomes easier for the special cases of linear queries,
that is, queries, in which no relation symbol occurs more than once
and boolean queries, that is, queries without output variables.
\begin{thm}
[Special Cases] Let $\ast\in\set{\IF,\RF,\TN}$. Then 
\begin{enumerate}
\item $\tcqc_{\ast}^{s}$\, and\, $\tcqc_{\ast}^{b}$\, are in PTIME
for linear queries;
\item $\tcqc_{\ast}^{s}$\, is NP-complete for boolean queries. 
\end{enumerate}
\end{thm}
\begin{proof}
Regarding Claim (i), the most critical case is $\ast=\RF$. For linear
queries under bag semantics, observe that the criterion in Theorem~\ref{thm:characterization:entailment:bag:semantics}
can be checked in polynomial time. First, there is only one choice
to map an atom in a query $Q_{C}$ to an atom in $L^{\NULL}$ (the
one with the same relation). Second, $\chase(T_{\C}(L^{\NULL}))$
can be computed in polynomial time. Lastly, the evaluation of $\bar{Q}$
over the chase result is in PTIME, because an atom in $\bar{Q}$ can
be mapped in only one way. Note that for linear queries under set
semantics, we only need to consider one null version $L^{\NULL}$
because a binding for an output term can only come from one position.

The lower bounds of Claim (ii) follow from Theorem~\ref{thm-complexity:bounds},
the upper bounds from Theorem~\ref{thm-complexity:bounds} for $\IF$,
and from Theorems~\ref{thm:characterization:set:sem:for:boolean:and:linear:queries}
and~\ref{thm-characterization:TC:QC:three:nulls} for $\RF$ and
$\TN$, since evaluation of conjunctive queries is in NP. 
\end{proof}
\begin{table}[t]
\begin{centering}
\begin{tabular}{|>{\centering}p{2.5cm}|c>{\centering}p{2.7cm}|}
\hline 
\multicolumn{1}{|>{\centering}p{2.5cm}|}{} & \multicolumn{2}{c|}{Query semantics}\tabularnewline
Incomplete database class  & set semantics  & bag semantics \& databases with keys\tabularnewline
\hline 
no nulls  & NP-complete  & NP-complete\tabularnewline
incomplete facts  & NP-complete  & NP-complete\tabularnewline
restricted facts  & NP-hard, in $\piptwo$  & NP-complete\tabularnewline
partial facts  & NP-hard, in $\piptwo$  & NP-complete\tabularnewline
\hline 
\end{tabular}
\par\end{centering}

\caption{Complexity of TC-QC entailment}

\label{fig:complexity:of:reasoning} 
\end{table}

In Table~\ref{fig:complexity:of:reasoning} we summarize our complexity
results for TC-QC entailment over databases with nulls and compare
them with the results for databases without nulls. Notably, if we
have keys then under bag semantics the complexity does not increase
with respect to databases without null values, while for the containment
problem for bag semantics not even decidability is known \cite{containment:bag:semantics:pods}. 

For queries under set semantics, it remains open whether the complexity
of reasoning increases from NP to $\piptwo$ for databases with restricted
facts and with 3 null values.

\section{Related Work}

Since the introduction of null values in relational databases~\cite{codd_null},
there has been a long debate about their semantics and the correct
implementation. In particular, the implementation of nulls in SQL
has led to wide criticism and numerous proposals for improvement (for
a survey, see~\cite{meyden:uebersicht:incompleteness:databases}).
Much work has been done on the querying of incomplete databases with
missing but existing values~\cite{reiter:evaluation:of:queries:for:nulls,abiteboul1991representation},
while only recently, Franconi and Tessaris showed that SQL correctly
implements null values that stand for inapplicable attributes~\cite{Franconi:Et:Al-SQL:Nulls-AMW}.
It was observed early on that different syntactic null values in databases
would allow to capture more information~\cite{codd:ambiguity:nulls:pods},
but these ideas did not reach application.

Fan and Geerts discussed incomplete data also in the form of missing,
but existing values \cite{Fan:Geerts-capturing_missing_tuples_and_values:pods:10},
which they represented by \emph{c-tables}~\cite{imielinski_lipski_representation_systems}.
However, their work is not directly comparable, because they work
in the setting of master data, where completeness follows from correspondence
with a complete master data source.

\section{Summary}

In this chapter we have extended the previous model by allowing incompleteness
in the form of null values. We have shown that the ambiguity of null
values as used in SQL is problematic, and that it is necessary to
syntactically differentiate between the different meanings.

We characterized completeness reasoning for null values that stand
for missing values, for nonapplicable values, and reasoning in the
case that both are present. 

While SQL's query evaluation is generally not correct for nulls that
represent missing values, we showed that for a minimal complete query
correct query answers can be calculated from the SQL query result
by dropping tuples with unknown and indeterminate nulls.

In the next chapter, we will discuss reasoning for geographic databases.

\chapter{Geographical Data}

\label{chap:osm}

Volunteered geographical information systems are gaining popularity.
The most established one is OpenStreetMap (OSM), but also classical
commercial map services such as Google Maps now allow users to take
part in the content creation. 

Assessing the quality of spatial information is essential for making
informed decisions based on the data, and particularly challenging
when the data is provided in a decentralized, crowd-based manner.
In this chapter, we show how information about the completeness of
features in certain regions can be used to annotate query answers
with completeness information. We provide a characterization of the
necessary reasoning and show that when taking into account the available
database, more completeness can be derived. OSM already contains some
completeness statements, which are originally intended for coordination
among the editors of the map. A contribution of this chapter is therefore
to show that these statements are not only useful for the producers
of the data but also for the consumers. 

Preliminary versions of the results up to Proposition \ref{lem:compl-simple-queries}
have been published at the BNCOD 2013 conference \cite{BNCOD-razniewski}.

\section{Introduction}

Storage and querying of geographic information poses additional requirements
that motivated the development of dedicated architectures and algorithms
for spatial data management. Recently, due to the increased availability
of GPS devices, volunteered geographical information systems have
quickly evolved, with OpenStreetMap (OSM) being the most prominent
one. Ongoing open public data initiatives that allow to integrate
government data also contribute. The level of detail of OpenStreetMap
is generally significantly higher than that of commercial solutions
such as Google Maps or Bing Maps, while its accuracy and completeness
are comparable.

OpenStreetMap allows to collect information about the world in remarkable
detail. This, together with the fact that the data is collected in
a voluntary, possibly not systematic manner, brings up the question
of the completeness of the OSM data. When using OSM, it is desirable
also to get metadata about the completeness of the presented data,
in order to properly understand its usefulness. 

Assessing completeness by comparison with other data is only possible,
if a more reliable data source for comparison exists, which is generally
not the case. Therefore, completeness can best be assessed by metadata
about the completeness of the data, that is produced in parallel to
the base data, and that can be compiled and shown to users. When providing
geographical data it is quite common to also provide metadata, e.g.,
using standards such as the FGDC metadata standard%
\footnote{http://www.fgdc.gov/metadata/geospatial-metadata-standards%
} show. However, little is known about how query answers can be annotated
with completeness information.

As an example, consider that a tourist wants to find hotels in some
town that are no further than 500 meters away from a park. Assume,
that, as shown in Figure \ref{fig:motivating-example}, the data about
hotels and parks is only complete in parts of the map. Then, the query
answer is only complete in the intersection of the areas where hotels
are complete and a zone 500 meters inside the area where spas are
complete (green in the figure), because outside, either hotels or
spas within 500 meters from a hotel could be missing from the database,
thus leading to missing query results.

\begin{figure}
\begin{centering}
\includegraphics[scale=0.5]{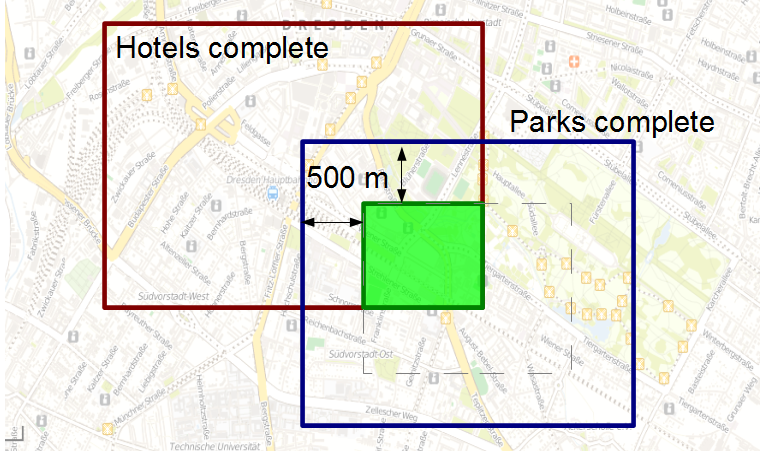} 
\par\end{centering}

\caption{Spatial query completeness analysis example. Assumed that hotels are
complete within the brown rectangle, and parks within the blue rectangle,
a query for hotels that have a park within 500 meters distance will
definitely return all answers that are located within the green rectangle.}

\label{fig:motivating-example} 
\end{figure}

Our contribution in this chapter is a methodology for reasoning about
the completeness queries over spatial data. In particular, we show
that metadata can allow elaborate conclusions about query completeness,
when one takes into account the data actually stored in the database.
We also show that metadata about completeness is already present to
a limited extent for OSM, and discuss practical challenges regarding
acquisition and usage of completeness metadata in the OSM project.

The structure of this chapter is as follows: In Section \ref{sec:geocompl-motivating-example-osm},
we present a sample scenario, in Section \ref{sec:geocompl-background}
we discuss spatial database systems, online map services, geographical
data completeness and OpenStreetMap. In Section \ref{sec:geocompl-formalization},
we give formalizations for expressing completeness over spatial databases.
In Section \ref{sec:geocompl-reasoning}, we present results for reasoning,
and discuss practical aspects in Section \ref{sec:geocompl-discussion}.
Section \ref{sec:geocompl-related work} contains related work. Preliminary
versions of some of the results contained in this chapter have been
published at the BNCOD 2013 conference \cite{BNCOD-razniewski}.

\section{Motivating Scenario: OpenStreetMap}

\label{sec:geocompl-motivating-example-osm} OpenStreetMap is a popular
volunteered geographical information system that allows access to
its base data to anyone. To coordinate their efforts, the creators
(usually called Mappers) of the data use a Wiki to record the completeness
of features in different areas. This information then allows to assess
the completeness of complex queries over the data.

As a particular use case, consider that a user Mary is planning vacations
in Abingdon, UK. Assume Mary is interested in finding a 3-star hotel
that is near a public park. Using the Overpass API, she could formulate
in XML the following query and execute it online over the OSM database%
\footnote{http://overpass-turbo.eu/%
}:\medskip{}

\begin{lstlisting}[basicstyle={\footnotesize\ttfamily}]
<query type="node">  
  <has-kv k="tourism" v="hotel"/>
  <has-kv k="stars" v="3"/>             
  <bbox-query e="7.25" n="50.8" s="50.7" w="7.1"/>
</query>   
<query type="node">     
  <around radius="500"/>     
  <has-kv k="leisure" v="park"/>     
  <bbox-query e="7.25" n="50.8" s="50.7" w="7.1"/>   
</query>   
<print/>
\end{lstlisting}

\medskip{}
The query could return as answer the hotels Moonshine Star, British
Rest and Holiday Inn, which in XML would be returned as follows:\medskip{}

\begin{lstlisting}[basicstyle={\footnotesize\ttfamily}]
<?xml version="1.0" encoding="UTF-8"?> 
<osm version="0.6" generator="Overpass API"> 
 <meta osm_base="2014-03-05T17:47:02Z"/>
  <node id="446099398" lat="48.9995855" lon="9.1475664">
    <tag k="tourism" v="hotel"/>
    <tag k="name" v="Moonshine Star"/>
    <tag k="stars" v="3"/>
    <tag k="restaurant" v="yes"/>
  </node>
  <node id="459972551" lat="48.9997612" lon="9.1483558">
    <tag k="amenity" v="hotel"/>
    <tag k="name" v="British Rest"/>
    <tag k="stars" v="3"/>
    <tag k="restaurant" v="yes"/>
  </node>
  <node id="459972551" lat="48.9997412" lon="9.1483658">
    <tag k="amenity" v="hotel"/>
    <tag k="name" v="Holiday Inn"/>
    <tag k="stars" v="3"/>
  </node>
</osm> 
\end{lstlisting}

\medskip{}
Before taking further steps in decision making, Mary is interested
to know whether this answer is trustworthy: Are these really all hotels
in Abingdon near a park? She therefore looks into the OSM Wiki page
of Abingdon and finds the completeness statements as shown in Figure
\ref{fig:compl-statements-abingdon}.

\begin{figure}
\begin{centering}
\includegraphics[scale=0.63]{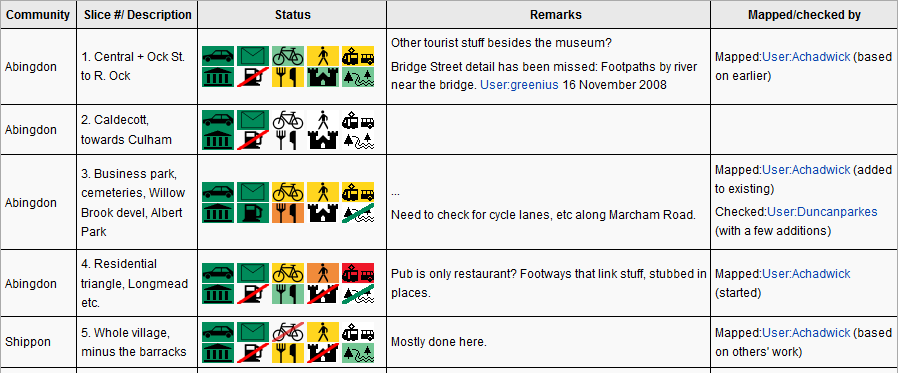}
\par\end{centering}

\caption{Extract from the OpenStreetMap-Wiki page for Abingdon. Source: http://wiki.openstreetmap.org/wiki/Abingdon}

\label{fig:compl-statements-abingdon}
\end{figure}
She also finds a legend for this table as shown in Figure \ref{fig:legend-osm-statements}
and a partitioning of Abingdon in districts as shown in Figure \ref{fig:partitioning-abingdon}.
To conclude in which parts of Abingdon the query is complete, she
has to watch for two things: First, she has to watch for those districts
in which hotels (pictogram: fork/knife) and parks (pictogram: trees/river)
are complete. But that is not all: She also has to watch for those
areas where parks are complete, but no parks are present in Abingdon.
Because those areas do not matter for the query result at all, independent
of whether they actually host hotels or not.

As another use case, consider emergency planning, where the planners
are interested to find all schools that are are within a certain radius
of a nuclear power plant. Querying the database again, he might miss
some information. Therefore, to assess in which areas the query answer
is complete, he not only has to watch for areas where schools and
power plants are complete, but also for areas where power plants are
complete and no power plants are present.

\begin{figure}
\begin{centering}
\includegraphics[scale=0.24]{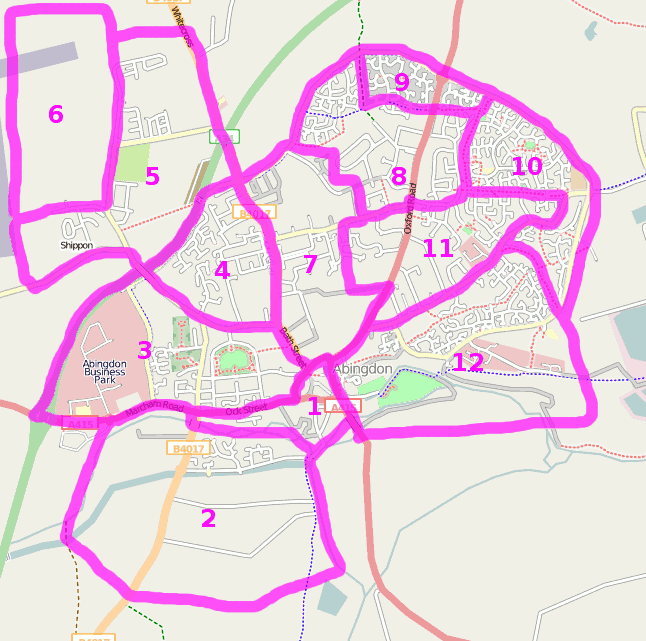}
\par\end{centering}

\caption{Partition of Abingdon made on the OSM wiki page. Source: http://wiki.openstreetmap.org/wiki/Abingdon}

\label{fig:partitioning-abingdon}
\end{figure}

\begin{figure}
\begin{centering}
\includegraphics[scale=0.65]{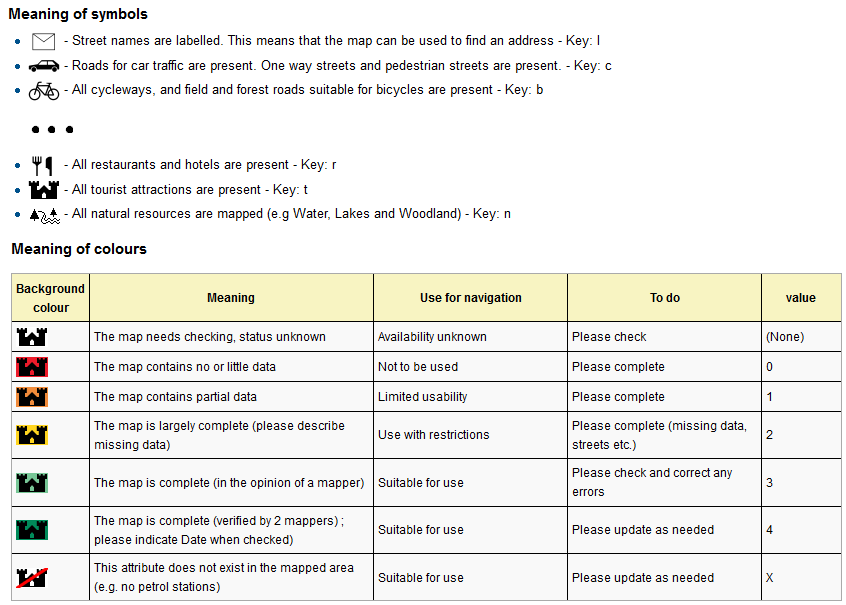}
\par\end{centering}

\caption{Legend for completeness statements as shown on the OpenStreetMap wiki
page. Source: http://wiki.openstreetmap.org/wiki/Template:En:Map\_status}
\label{fig:legend-osm-statements}
\end{figure}

\section{Background}

\label{sec:geocompl-background} In the following, we introduce spatial
database systems and online map services, the problem of geographical
data completeness and OpenStreetMap.

\subsection{Spatial Databases Systems and Online Map Services }

To facilitate storage and retrieval, geographic data is usually stored
in spatial databases. According to \cite{guting-vldb-intro-spatial-dbs},
spatial databases have three distinctive features. First, they are
database systems, thus classical relational/tree-shaped data can be
stored in them and retrieved via standard database query languages.
Second, they offer spatial data types, which are essential to describe
spatial objects. Third, they efficiently support spatial data types
via spatial indexes and spatial joins.

Online map services usually provide graphical access to spatial databases
and provide services for routing and address finding. There are several
online map services available, some of the most popular ones being
Google Maps, Bing Maps, MapQuest and OpenStreetMap. With the exception
of OSM, the data underlying those services is not freely accessible.
The most common uses of those services are routing (``Best path from
A to B?''), address retrieval (``Where is 2nd street?'') and business
retrieval (``Hotels in Miami''). While the query capabilities of
most online map services are currently still limited (one can usually
only search for strings and select categories), spatial databases
generally allow much more complex queries.
\begin{example}
Tourists could be interested in finding those hotels that are less
than 500 meters from a spa and 1 kilometer from the city center. Real
estate agents could be interested in properties that are larger than
1000 square meters and not more than 5 kilometers from the next town
with a school and a supermarket. Evacuation planners might want to
know which public facilities (schools, retirement homes, kindergartens,
etc.) are within a certain range around a chemical industry complex.
\end{example}

\subsection{Geographical Data Completeness}

Geographical data quality is important, as for instance recent media
coverage on Apple misguiding drivers into remote Australian desert
areas shows.%
\footnote{http://www.dailymail.co.uk/sciencetech/article-2245773/Drivers-stranded-Aussie-desert-Apple-glitch-Australian-police-warn-Apple-maps-kill.html%
} Since long there has been work on geographical data quality, however
it was mostly focusing on precision and accuracy \cite{Goodchild-book-spatial-data-quality},
which are fairly uniform among different features. Completeness in
contrast, is highly dependent on the type of feature. If metadata
about completeness is present, it is attractive to visualize it on
maps \cite{visualizing-spatial-dq}. Completeness is especially a
challenge when (1) databases are to capture continuously the current
state (as opposed to a database that stores a map for a fixed date)
because new features can appear, (2) databases are built up incrementally
and are accessible during build-up (as it is the case for OSM) and
(3) the level of detail that can be stored in the database is high
(as it is easier to be complete for all highways in a state than for
all post boxes).

There have been several attempts on assessing the completeness of
OpenStreetMap based on comparison with other data sources. In this
chapter, we take a different approach based on completeness metadata.

\subsection{OpenStreetMap}

OpenStreetMap (OSM) is a free, open, collaboratively edited map project.
Its organization is similar to that of Wikipedia. Its aim is to create
a map of the world. The map consists of features, where basic features
are either points, polygons or groups, and each feature has a primary
category, such as highway, amenity or similar. Then, each feature
can have an unrestricted set of key-value pairs. Though there are
no formal constraints on the key-value pairs, there are agreed standards
for each primary feature category.%
\footnote{http://wiki.openstreetmap.org/wiki/Map\_features%
} 

There have been some assessments of the completeness of OSM based
on comparison with other data sources, which showed that the road
map completeness is generally good \cite{OSM-completeness-analysis-england,OSM-completeness-analysis-2,quality-metrics-for-osm}.
Assessment based on comparison is however a method that is very limited
in general, as it relies on a data source which captures some aspects
equally good as OSM. Especially since due to the open key-value scheme,
the level of detail of OSM is not limited, comparison is not possible
for many aspects. Examples of the deep level of detail are the kind
of trash that trash bins accept or the opening hours of shops or the
kind of fuel used in public fire pits %
\footnote{http://wiki.openstreetmap.org/wiki/Tag:amenity\%3Dbbq%
} (these attributes are all agreed as useful by the OSM community).

While the most common usage of OSM is as online map service, it also
provides advanced querying capabilities, for instance via the Overpass
API web interface.%
\footnote{http://overpass-turbo.eu%
} Also, the OSM data, which is natively in XML, can be downloaded,
converted and loaded into classical SQL databases with geographical
extensions.

\section{Formalization}

\label{sec:geocompl-formalization} In the following, we formalize
spatial databases, queries with the spatial distance function, incompleteness
in databases, completeness statements for spatial databases and completeness
areas for spatial queries.

\subsection{Spatial Databases}

While OSM uses an extendable data format based on key-value pairs,
and stores its data natively in XML, this data can easily be transferred
into relational data, thus, in the following, we adopt a relational
database view in the following.

Similarly to classical relational databases, spatial databases consist
of sets of facts, here called features, which are formulated using
a fixed vocabulary, the database schema. In difference to classical
relational databases (see Sec.~\ref{sec:prelim:relational-databases}),
in a spatial database, each feature has one \emph{location} attribute.
For simplicity, we assume that these locations are only points. 

We assume a fixed set of feature names $\Sigma$, where each feature
name $F$ has a set of arbitrary attributes and one location attribute.
Then, a \emph{spatial database} is a finite set of facts over $\Sigma$
that may contain null values. Null values correspond to key/value
pairs that are not set for a given feature.\global\long\def\abgd{\mathrm{Abgd}}

\begin{example}
Consider the three Moonshine Star, British Rest and Holiday Inn from
above. Furthermore, assume that there are also 2 parks in the database.
Then, in a geographical database $D_{\abgd}$, this information would
be stored as follows:\label{ex:spatial-db}\medskip{}

{\scriptsize{}\hspace{-0.3cm}}%
\begin{tabular}{|cccc|c|ccc|}
\cline{1-4} \cline{6-8} 
\multicolumn{4}{|c|}{{\scriptsize{}Hotel}} &  & \multicolumn{3}{c|}{{\scriptsize{}Park}}\tabularnewline
{\scriptsize{}name} & {\scriptsize{}stars} & {\scriptsize{}rstnt} & {\scriptsize{}location} &  & {\scriptsize{}name} & {\scriptsize{}size} & {\scriptsize{}location}\tabularnewline
\cline{1-4} \cline{6-8} 
{\scriptsize{}Moonshine Star} & {\scriptsize{}3} & {\scriptsize{}yes} & {\scriptsize{}48.55:9.64} & {\scriptsize{}$\ \ \ \ $} & {\scriptsize{}Central Park} & {\scriptsize{}med} & {\scriptsize{}48.20:9.57}\tabularnewline
\cline{1-4} \cline{6-8} 
{\scriptsize{}British Rest} & {\scriptsize{}3} & {\scriptsize{}yes} & {\scriptsize{}48.12:9.58} &  & {\scriptsize{}King's Garden} & {\scriptsize{}small} & {\scriptsize{}48.49:9.61}\tabularnewline
\cline{1-4} \cline{6-8} 
{\scriptsize{}Holiday Inn} & {\scriptsize{}3} & {\scriptsize{}no} & {\scriptsize{}48.41:9.37} & \multicolumn{1}{c}{} &  &  & \multicolumn{1}{c}{}\tabularnewline
\cline{1-4} 
\end{tabular}{\scriptsize \par}

\medskip{}

Represented on a map, this information could look as in Figure \ref{fig:osm:spatial:db}.
\end{example}
\begin{figure}
\begin{centering}
\includegraphics[width=1\textwidth]{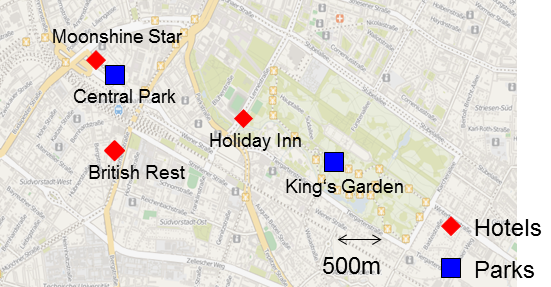}
\par\end{centering}

\caption{Visualization of the database $D_{\abgd}$ from Example \ref{ex:spatial-db}.}
\label{fig:osm:spatial:db}
\end{figure}

Spatial query languages allow the use of spatial functions. As we
assume that all spatial objects are points, only the spatial relation
$\dist(l_{1},l_{2})$, which describes the distance between locations
$l_{1}$ and $l_{2}$ is meaningful.

A \emph{simple spatial query} is written as $\query{Q(\tpl d,l)}{R(\tpl d,l)}$,
where $R$ is a relation, the terms $\tpl d$ are either constants
or variables and $l$ is the location attribute of $R$.

Over spatial databases, it is especially interesting to retrieve features
for which there exist other features (not) within a certain proximity.
To express such queries, we introduce the class of so called distance
queries, on which we will focus in the remainder of this chapter:

Intuitively, a \emph{distance query} asks for a feature for which
certain other features exist within a certain radius. Its shape resembles
that of a star, because the joins between atoms in the query appear
only between the first atom and other atoms. Formally, a distance
query with $n$ atoms is written as follows:\global\long\def\nicehotels{\mathrm{niceHotels}}
\global\long\def\abgd{\mathrm{Abgd}}

\begin{eqnarray}
 & \query{Q(\tpl d_{1},l{}_{1})}{R_{1}(\tpl d_{1},l_{1}),} & R_{2}(\tpl d_{2},l_{2}),\dist(l_{1},l_{2})<c_{2},\label{eq:star-shaped-query}\\
 &  & R_{3}(\tpl d_{3},l_{3}),\dist(l_{1},l_{3})<c_{3},\nonumber \\
 &  & \ldots\nonumber \\
 &  & R_{n}(\tpl d_{n},l_{n}),\dist(l_{1},l_{n})<c_{n}\nonumber 
\end{eqnarray}

where $l_{i}$ is the geometry attribute of the feature $R_{i}$,
and the $c_{i}$ are constants. Later, we will also discuss negated
atoms. Note that using the relations $'\neq'$ and $'='$ together
with $\dist$ does not make sense for a nearly continuous-valued attribute
such as location, and that the expression $'\dist>c'$ not make sense,
because in order to evaluate such a query, one would need to scan
the features in the whole world.\global\long\def\hotel{\mathit{hotel}}
\global\long\def\pub{\mathit{pub}}
\global\long\def\park{\mathit{park}}

\begin{example}
Consider again Mary's query that asked for 3-star hotels with a park
within 500 meters distance. As a distance query, it would be written
as follows:

{\footnotesize{}
\begin{eqnarray*}
 & \query{Q_{\nicehotels}(n,s,r,l_{\hotel})}{\hotel(n,s,r,l_{\hotel}),} & \park(n',s',l_{\park}),\dist(l_{\hotel},l_{\park})<500m
\end{eqnarray*}
}{\footnotesize \par}

A query that additionally also asks for pubs within a kilometer and
a train station within 1 kilometer would be written as follows:

\hspace{-0.2cm}{\footnotesize{}
\begin{eqnarray*}
 & \query{Q_{\mathrm{nicerHotel}}(n,s,r,l_{\hotel})}{\hotel(n,s,r,l_{\hotel}),} & \park(n',s',l_{\park}),\dist(l_{\hotel},l_{\park})<500mm\\
 &  & \pub(n'',l_{\pub}),\dist(l_{\hotel},l_{\pub}<1km\\
 &  & station(n''',l_{station}),\dist(l_{\hotel},l_{station})<1km
\end{eqnarray*}
}{\footnotesize \par}
\end{example}
\begin{figure}
\begin{centering}
\includegraphics[scale=0.65]{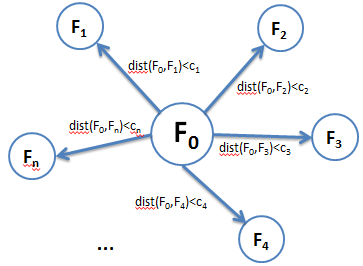}
\par\end{centering}

\caption{Distance query.}

\label{fig:distancequery}
\end{figure}

\subsection{Completeness Formalisms}

In the following, we formalize incomplete databases as in Section
\ref{sec:prelim:incomplete-databases}, extend table completeness
statements to feature completeness statements for spatial databases,
and show that for queries now their query completeness area becomes
relevant.

\paragraph{Incomplete Databases}

Online spatial databases that try to capture the world can hardly
contain all features of the world. As before, we model such incomplete
databases as pairs of an ideal database $\di$, which describes the
information that holds according to the real world, and an available
database $\da$, which contains the information that is actually stored
in the database. Again we assume that the stored information $\da$
is a subset of the information that holds in the real world $\di$.
\begin{example}
\label{ex:Irish-pubs}

Consider that the available database is $D_{\abgd}$ as shown in Figure
\ref{fig:osm:spatial:db}. It might be that in reality, there exists
another park, the Hyde Park, and another hotel the Best Marigold.
Thus the ideal database would also contain those two facts, and, represented
on a map, would look as shown in Figure \ref{fig-geocompl-idb}.
\end{example}
\begin{figure}
\begin{centering}
\includegraphics[width=1\textwidth]{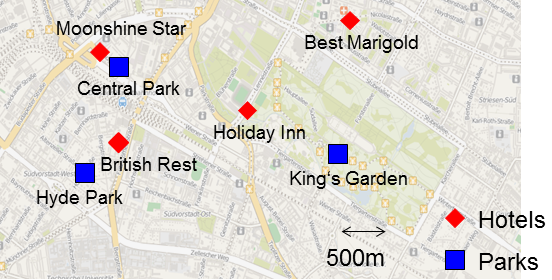}
\par\end{centering}

\caption{Map representation of the ideal database of Example \ref{ex:Irish-pubs}.}

\label{fig-geocompl-idb}
\end{figure}

\paragraph{Feature completeness statements}

Adapting the well-known table completeness statements (see Section
\ref{sec:prelim:table-completeness}), feature completeness statements
can be used to express that certain features are complete in a certain
area. 

Formally, a \emph{feature completeness statement }consists of a feature
name $R$, a set of selections $M$ on the attributes $\tpl a$ of
the relation $R$ and an area $A$. We write such a statement $F$
as $\compl{R,M,A}$. It has a corresponding simple query, which is
defined as $\query{Q_{F}(\tpl a,l)}{R(\tpl a,l),M}.$ An incomplete
database $(\di,\da)$ satisfies the statement, if $Q_{F}(\di)\subseteq\da$.
\begin{example}
\label{ex:osm:compl-statement}Consider a feature completeness statement
$c_{\hotel}$ expressing that hotels are complete in the area $A_{1}$,
and a statement $c_{\park}$ expressing that parks are complete in
the area $A_{2}$, where $A_{1}$ and $A_{2}$ overlap as shown in
Figure \ref{fig:osm:completeness-statements}, as follows:
\begin{eqnarray*}
c_{\hotel} & =\compl{\hotel(n,s,r,l);\emptyset;A_{1}}\\
c_{\park} & =\compl{\park(n,s,l};\emptyset;A_{2})
\end{eqnarray*}

\end{example}
Observe also that each green icon in Figure \ref{fig:compl-statements-abingdon}
actually is a completeness statement. For example, the first green
icon says that all roads are complete in the center of Abingdon.

\begin{figure}
\begin{centering}
\includegraphics[width=0.8\textwidth]{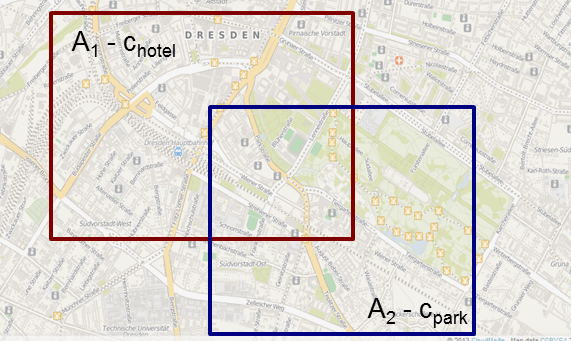}
\par\end{centering}

\caption{Areas $A_{1}$ and $A_{2}$ for the completeness statements $c_{\hotel}$
and $c_{\park}$ from Example \ref{ex:osm:compl-statement}.}

\label{fig:osm:completeness-statements}
\end{figure}

\paragraph{Query completeness area}

When querying an available database, one is interested in getting
all features that satisfy the query wrt. the ideal world. If data
is missing in the available database, then this cannot be guaranteed
everywhere. For example, when pubs are not complete for north district
then a query for all Irish pubs may be complete in the center but
not in the north district.

Given a set of feature completeness statements $\F$, an available
database $D$ and a query $Q$, the query completeness area of $Q$
is the set $S$ of all points such that it holds that for any ideal
database $\di$ with $(\di,D)$ satisfying $\F$ it holds that $Q(\di)=Q(D)$.

\medskip{}

\framebox{\begin{minipage}[t]{0.9\columnwidth}%
Reasoning Problem\\
\textbf{Input}: Set of feature completeness statements $\F$,\\
 $\ \ \ \ \ $Database $D$, query $Q$

\textbf{Output}: Completeness area of $Q$%
\end{minipage}}

\medskip{}

It is clear that the completeness area depends on the completeness
statements. We remind that also the database $D$ has a crucial influence.
As discussed in the motivating example, areas where a constraining
feature is not present, but complete according to the completeness
statements, also belong to the completeness area.\global\long\def\nicehotels{\mathrm{niceHotels}}
\global\long\def\abgd{\mathrm{Abgd}}

\begin{example}
Consider the completeness statement $c_{\hotel}$ and $c_{\park}$
from above, the database instance $D_{\abgd}$ from Example \ref{fig:osm:spatial:db},
and consider Mary's query $Q_{\nicehotels}$ from Example TODO. Then
the completeness area for this query would be the green area shown
in Figure \ref{fig:osm:compl-area}. The upper left area is complete,
because due to $c_{\hotel}$, hotels are complete there, and the only
hotels in that area in $D_{\abgd}$ are the Moonshine Star and the
British Rest, where the one is an answer and for the other it is not
known. The green area on the lower right is complete, because parks
are complete in the surrounding, but there are no parks nearby, so
there cannot be any hotels that are answers to the query.

\label{example-completeness-area}

\begin{figure}
\begin{centering}
\includegraphics[width=1\textwidth]{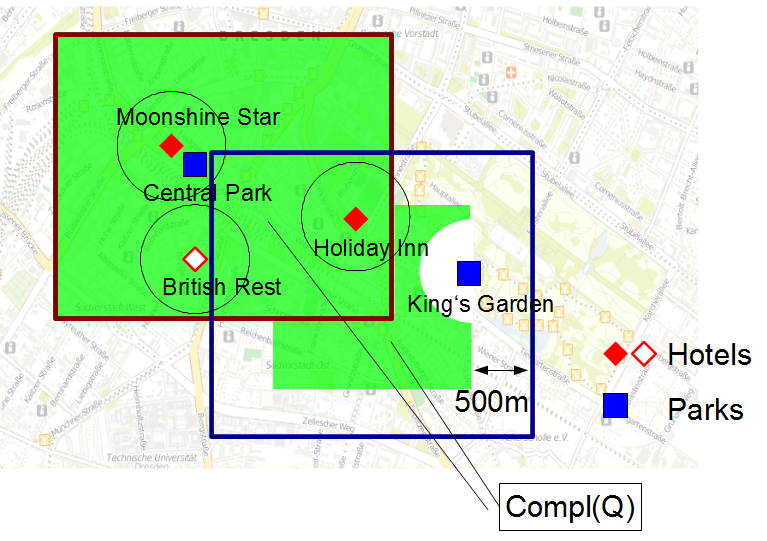}
\par\end{centering}

\caption{Completeness area for the query $Q_{\nicehotels}$ based on the completeness
statement $c_{\hotel}$ and $c_{\park}$ and the database $D_{\abgd}$
as discussed in Example \ref{example-completeness-area}. }
\label{fig:osm:compl-area}
\end{figure}

\end{example}

\section{Completeness Assessment}

\label{sec:geocompl-reasoning}

In the following, we show how the completeness area of a query can
be computed. We start with simple queries and then show how the computation
for distance queries can be reduced to the one of simple queries.
Lastly, we discuss the complexity of completeness assessment.
\begin{table}
\begin{centering}
{\small{}}%
\begin{tabular}{|c|>{\centering}p{7cm}|}
\hline 
{\small{}Symbol} & {\small{}Meaning}\tabularnewline
\hline 
\hline 
{\small{}$Q$} & {\small{}Query}\tabularnewline
\hline 
{\small{}$\F$} & {\small{}Set of feature completeness statements}\tabularnewline
\hline 
{\small{}$\da$} & {\small{}Available database instance}\tabularnewline
\hline 
{\small{}$\ca$} & {\small{}completeness area of a query}\tabularnewline
\hline 
{\small{}$\oor$} & {\small{}out of range area}\tabularnewline
\hline 
{\small{}$\coor$} & {\small{}complete and out of range area}\tabularnewline
\hline 
{\small{}$\cert_{Q,\F,\da}$} & {\small{}certain answers to $Q$ wrt. $\F$ and $\da$}\tabularnewline
\hline 
{\small{}$\poss_{Q,\F,\da}$} & {\small{}possible answers to $Q$wrt. $\F$ and $\da$}\tabularnewline
\hline 
{\small{}$\imposs_{Q,\F,\da}$} & {\small{}impossible answers to $Q$wrt. $\F$ and $\da$}\tabularnewline
\hline 
{\small{}$\notins_{F,c,\da}$} & {\small{}area where no $F$-feature is within distance $c$}\tabularnewline
\hline 
{\small{}$\complins_{F,c,\da}$} & {\small{}area where $F$-features are complete within distance $c$}\tabularnewline
\hline 
\end{tabular}
\par\end{centering}{\small \par}

\caption{Notation table}
\end{table}

\subsection{Assessment for Simple Queries}

Simple queries are queries that do not contain any joins, but only
select features with certain attribute values. We will use them as
building blocks for the more complex distance queries. For finding
the completeness area of a simple query, one only needs to take the
union of the areas of all completeness statements that capture the
queried features:
\begin{prop}
Let $\F$ be a set of FC statements and $\query{Q(\tpl d,l)}{R(\tpl d,l)}$
be a simple query. Then $\acompl$, the completeness area of $Q$
wrt $\F$ is computed as follows:\label{lem:compl-simple-queries}
\[
\acompl=\ \bigcup\,\bigset{A_{i}\bigmid\mbox{\ensuremath{Q\subseteq Q_{F_{i}}}}\ \wedge\ F_{i}\in\F}.
\]

\end{prop}
The containment checks needed to compute $\ca$ are straightforward:
In each containment, we have to check whether the selection by $Q$
is the same or more specific than the selection by $Q_{F}$, which
is a pairwise comparison for each attribute. Thus, for a fixed database
schema, the completeness check is linear in the size of the set $\F$
of feature completeness statements.\global\long\def\simple{\mathrm{simple}}

\begin{example}
Remember the completeness statement $c_{\hotel}$ from Example \ref{ex:osm:compl-statement},
which asserted completeness for hotels in an area $A_{1}$. Consider
furthermore a simple query $\query{Q_{\simple}(n,3,r,l)}{\hotel(n,3,r,l)}$
that asks for all hotels with 3 stars. Clearly, the completeness area
for $Q_{\simple}$ will contain $A_{1}$, because $Q_{\simple}$ is
contained in the query $Q_{c_{\hotel}}$ that asks for all hotels
regardless of their number of stars.
\end{example}
In the following, we use the completeness of simple queries as building
block for reasoning about the completeness of distance queries.

\subsection{Assessment for Distance Queries}

For distance queries, we have to take into account the completeness
area of each literal. In general, for a point to be in the completeness
area, the point has to be in the completeness areas for all the simple
queries that constitute the distance query. Furthermore, in specific
cases when looking at the database instance, completeness can also
be concluded even when only some features are complete.

As shown in \cite{BNCOD-razniewski}, for an arbitrary distance query
$Q$ as in Equation (\ref{eq:star-shaped-query}), one can introduce
a simple query $Q_{L_{i}}$ for each literal $L_{i}$ in $Q$, defined
as 
\[
\query{Q_{L_{i}}(g_{i})}{R_{i}(\tpl d_{i})}
\]

for $i=1,\ldots,n$, which we call a \emph{component query} of $Q$. 

The function $\shrink(A,d)$ that shrinks an area $A$ by a distance
$d$ is defined as one would expect as the set of all points within
$A$ that are at least $d$ apart from the border of $A$. Then, when
neglecting the actual state of the database, the completeness area
of a distance query is defined as follows:
\begin{prop}
[Schema-level completeness area]

Let $\F$ be a set of FC statements, $\query Q{L_{1},\ldots,L_{n}}$
be a distance query and $Q_{L_{1}},\ldots,Q_{L_{n}}$ the component
queries of $Q$. Then for any database $\da$ the completeness area
$\acompl$ of $Q$ wrt.\ $\F$ satisfies 
\[
\acompl=\complatom{L_{1}}\cap\shrink(\complatom{L_{2}},c_{2})\cap\ldots\cap\shrink(\complatom{L_{n}},c_{n}),
\]

\label{prop:instance-independent}
\end{prop}
As seen before in Lemma \ref{lem:compl-simple-queries}, the computation
of the completeness areas $\ca_{L_{1}}$ to $\ca_{L_{n}}$ for the
component queries is straightforward. 

As Example \ref{example-completeness-area} showed, wrt.\ a known
available database, the completeness area may however be larger, as
it also includes those areas where a queried feature is too far away
and also complete, so that the query becomes unsatisfiable in that
area.

Before proceeding to the characterization, additional terminology
is needed. 

With $\complins_{L,c}$ we denote the set of all points for which
it holds that the literal $L$ is complete within the distance $c$.
This area can be computed as 
\[
\complins_{L,c}=\shrink(CA_{L},c)
\]
The area $\notins_{L,c}$ denotes the set of all points $p$ for which
it holds that no feature satisfying $L$ is within the distance $c$
of $p$. This are can be computed as 
\[
\notins_{L,c}=\neg(\bigcup_{f\in\query{Q(l)}L}\buffer(f,c))
\]
 Using these two areas, we define the area $\coor_{L,c}$, which stands
for the set of points where features satisfying $L$ are both not
existent within the distance $c$ and also complete within the distance
$c$ as: 

\[
\coor_{L,c}=\complins_{L,c}\cap\notins_{L,c}
\]

This area defines additional parts of the completeness area of a query:
If the query asks for a feature within a certain distance, but for
a given point there is no such feature within that distance, and the
feature is also complete within that distance, then, even if some
feature satisfying the output atom $L_{1}$ of a distance query is
present, it can never satisfy the atom $L$. Thus, the query is complete
in that point as well.
\begin{example}
Consider again Figure \ref{fig:osm:compl-area}, and observe the green
area at the lower right. For each point in that area, it holds that
no park is within 500 meters in the available database, and also parks
are complete within 500 meters according to $c_{\park}$. Thus, these
points lie in the area $\coor_{L_{2},500m}$.
\end{example}
Also, a second conclusion can be made wrt. the database instance:
Wherever the output feature $L_{1}$ alone is complete, and there
is no feature present in the database that could potentially become
an answer, the query is complete as well.

Formally, consider a query $\query Q{L_{1},\ldots,L_{n}}$ and a database
instance $\da$. Then, all features in $\da$ that satisfy $\tpl d_{1}$,
that is, all tuples in the result to $Q_{L_{1}}$ can be grouped into
three disjoint categories:
\begin{enumerate}
\item Certain answers,
\item Impossible answers,
\item Possible answers.
\end{enumerate}
\emph{Certain answers} are those features, which already satisfy the
query in the current database instance, that is:
\[
f\in\cert_{Q,\F,D}\mbox{\ \ \ \ iff\ \ \ \ }f\in Q(\da)
\]

\emph{Impossible answers} are those features which are in $Q_{L_{1}}(\da)$,
but cannot be in the answer of the query because some atom $L_{i}$
with $i>1$ is unsatisfiable for them, that is: 
\[
f\in\imposs_{Q,F,D}\mbox{\ \ \ iff\ \ \ }\forall\di:(\di,\da)\models\F\mbox{\ \ it holds that\ \ }f\not\in Q(\di)
\]
To practically compute $\imposs_{Q,F,D}$, we can use the previously
introduced function $\coor$:
\begin{lem}
Let $Q$ be a query, $\F$ be a set of feature completeness statements
and $D$ be a database. Then for any feature $f\in Q_{L_{1}}(D)$
with location $l_{f}$:

\[
f\in\imposs_{Q,\F,D}\ \ \ \ \mbox{iff}\ \ \ \ l_{f}\in(\coor_{L_{2},c_{2}}\cup\ldots\cup\coor_{L_{n},c_{n}})
\]

\end{lem}
The intuitive meaning is that a feature is an impossible answer if
and only if for one of the literals in the query it holds that no
features that satisfy it are within the required range, but those
features are also complete there.

The remaining features in $Q_{L_{1}}(\da)$ that are neither certain
nor impossible answers are \emph{possible answers}. Possible answers
are characterized by the fact that currently they are not in the answer
to the query, but the completeness statements do not exclude the chance
that the features are answers over the ideal database. The possible
answers are computed as: 
\[
Q_{L_{1}}(D)\setminus(\cert_{Q,F,D}\cup\imposs_{Q,F,D})
\]

\begin{example}
In Figure \ref{fig:osm:compl-area}, the hotel Moonshine Star is a
certain answer, because the Central Park is located nearby. The hotel
British Rest is a possible answer, because there is no park nearby
shown in the database, but parks are not complete in the 500-meter-surrounding
of the hotel. The hotel Holiday Inn is an impossible answer, because
both there is no park in the 500-meter-surrounding and parks are are
complete in that surrounding.
\end{example}
We can now characterize the completeness are of a distance query as
follows:
\begin{thm}
Let $\F$ be a set of FC statements, $\query Q{L_{1},\ldots,L_{n}}$
be a distance query and $D$ be an available database. Furthermore,
let $\ca_{1}$ be the completeness area for the literal $L_{1}$.
Then the completeness area of $Q$ wrt. $\F$ and $D$ is

\[
\ca=\ca_{1}\cup\coor_{L_{2},c_{2}}\cup\ldots\cup\coor_{L_{n},c_{n}}\setminus\poss_{Q}
\]

\end{thm}
An implication of this theorem is that the completeness area may contain
incomplete points. Note that these points cannot be in any of the
$\coor$-areas, as such areas by definition cannot contain any possible
answers.

\subsection{Quantification}

So far, we have only discussed how to describe the completeness area
of a query. Obviously, we can quantify the proportion of the an area
of interest that is contained in the completeness area (e.g., 80\%
of Abingdon lie in the completeness area of the query). Since the
completeness area however can contain incomplete points (caused by
possible answers), an area which is 100\% contained in the completeness
area still satisfies query completeness only if the number of possible
answers in the area is zero. 

In general, for an area that is 100\% contained in the completeness
area of a query, we can give bounds for the completeness as percentage
of tuples from the ideal database that are already in the answer over
the available database. Using the relationship between certain and
possible answers, we can give bounds as follows:

\[
1\geq\mbox{Completeness}(Q,\F,\da)\geq\frac{\mid\cert_{Q,\F,\da}\mid}{\mid\cert_{Q,\F,\da}\mid+\ |\poss_{Q,\F,\da}|}
\]

These bounds however might be very wide. Since in the real world features
are not distributed uniformly, any conclusions beyond this bounds
are difficult.
\begin{example}
Consider again the database and completeness statements as shown in
Figure \ref{fig:osm:compl-area}. Then for the query $Q_{\nicehotels}$,
the completeness in the green area lies between 50\% and 100\%, because
there is one certain and one possible answer in that area.
\end{example}

\subsection{Distance Queries with Negation}

It may be interesting to ask queries that include negated literals.
For example, one could ask for the schools that \emph{do} \emph{not}
have a nuclear power plant within 10 kilometers distance. Formally,
a distance query with negation has a form as follows:
\begin{eqnarray*}
 & \query{Q(\tpl d_{1},l{}_{1})}{R_{1}(\tpl d_{1},l_{1}),} & R_{2}(\tpl d_{2},l_{2}),\dist(l_{1},l_{2})<c_{2},\\
 &  & \ldots\\
 &  & R_{i}(\tpl d_{i},l_{i}),\dist(l_{1},l_{i})<c_{i},\\
 &  & \neg R_{i+1}(\tpl d_{i+1},l_{i+1}),\dist(l_{1},l_{i+1})<c_{i+1},\\
 &  & \ldots\\
 &  & \neg R_{n}(\tpl d_{n},l_{n}),\dist(l_{1},l_{n})<c_{n}
\end{eqnarray*}

such that literals from $1$ to $i$ are positive, and from $i+1$
to $n$ are negated. Completeness in the instance-independent case
(Lemma \ref{lem:compl-simple-queries} and Proposition \ref{prop:instance-independent})
holds analogous. In the instance-dependent case, things are different:
\global\long\def\ir{\mathrm{IR}}

We introduce a function $\ir_{L,c,\da}$ for calculating the area
where a feature satisfying an atom $L$ is within a range $c$ in
the database $\da$ as

\[
\bigcup_{f\in Q_{L}(\da)}\buffer(l_{f},c)
\]

Certain and impossible answers are now also differently defined. Let
$Q^{+}$ be the positive part of $Q$. Then
\begin{itemize}
\item $\tilde{\cert}_{Q,\F,\da}=Q(\da)\cap\coor_{L_{i+1}}\cap\ldots\cap\coor_{L_{n}}$
\item $\tilde{\imposs}_{Q,\F,\da}=\imposs_{Q^{+}}\cup(Q_{L_{1}}\cap(\ir_{L{}_{j+1}}\cup\ldots\cup\ir_{L_{n}}))$
\item $\tilde{\poss}$ is defined as before as the remaining features in
$Q_{L_{1}}(\da)$ that are neither certain answers nor impossible
answers. 
\end{itemize}
Having this definitions, we can now define the completeness area of
a query with negation as follows:
\begin{thm}
Let $Q$ be a distance query with negation, $\F$ be a set of feature
completeness statements and $\da$ be a database instance. Then $\ca_{Q,\F,\da}$,
the completeness area of $Q$ wrt. $\F$ and $\da$, satisfies
\[
\ca_{Q,\F,\da}=\ca_{1}\cup\coor_{2}\cup\ldots\cup\coor_{j}\cup\ir_{j+1}\cup\ldots\cup\ir_{n}\setminus\tilde{\poss}_{Q}
\]

\end{thm}
The differences to the positive case are that now also those areas,
where a negated feature is present in the available database in the
surrounding, belong to the completeness area, and that the possible
answers are defined differently.

\subsection{Explanations for Incompleteness}

For a point where a query is not complete, it may be interesting to
know which kind of features can be missing. For a singular point within
a completeness area due to a possible answer, the answer is just this
possible answer. For other points, one has to take into account all
tuples that satisfy the query but do not satisfy the completeness
statements at this point.
\begin{example}
Consider two completeness expressing that hotels with 3 stars and
a dinner restaurant are complete, and that 4 star hotels with a full-day
restaurant are complete. Then there are four possible types of hotels
that could be missing: Hotels not with 3 stars and not with 4 stars,
hotels not with 3 stars and not with a full-day restaurant, hotels
not with an evening restaurant and not with 4 stars, and hotels with
neither an evening restaurant nor a full-day restaurant. Essentially,
as there are two ways to violate the first statement and two ways
to violate the second statement, there are 4 combinations that violate
both statements.
\end{example}
In general, for a feature class with $m$ attributes and $n$ completeness
statements for that class, there could be $m^{n}$ possible explanations.
That means, given that sufficiently many distinct values are used
for attributes in completeness statements, the explanations for incompleteness
will grow exponential. This may not be a problem in practice however,
because possibly only few attributes will be instantiated in completeness
statements. For some attributes, such as stars of a hotel, it makes
sense to instantiate them, but for many others, such as opening hours
or phone numbers it clearly does not make sense.

\subsection{Comparisons}

Using comparisons in completeness statements may be useful, for example
for saying that all hotels with at least 3 stars are complete.

In line with previous results (see Section \ref{sec:prelims:query-containment}),
when adding comparisons to the formalism, reasoning becomes more complex.
In particular, already for simple queries, the problem of deciding
whether a certain point is in the completeness area, becomes coNP
hard:
\begin{thm}
Let $Q$ be a simple query with comparisons, $\F$ be a set of feature
completeness statements and $\da$ be a database instance. Then deciding
whether a point $p$ is in $\ca_{Q,\F}$ is coNP-hard.\end{thm}
\begin{proof}
By reduction of the propositional tautology problem. \global\long\def\sig{\mathrm{sig}}

Consider a propositional tautology problem that asks whether a formula
$\phi=l_{1}\wedge l_{2}\wedge l_{3}\bigvee\ldots\bigvee l_{n-2}\wedge l_{n-1}\wedge l_{n}$
is a tautology, and assume that $\phi$ contains variables $v_{1}$
to $v_{m}$. Then this tautology problem can be reduced to a basic-query-completeness
problem as follows: First, one introduces a feature $R$ with $m$
arguments. Then, for each clause $l_{i}\wedge l_{i+1}\wedge l_{i+2}$
one introduces a completeness statement $\compl{R(v_{1},\ldots,v_{n});v_{i}=\sig(l_{i}),v_{i+1}=\sig(l_{i+1}),v_{i+2}=\sig(l_{i+2});A}$,
where $\sig(v_{i})$ returns $\true$ if $l_{i}$ is positive and
returns false otherwise, and by considering a query $\query{Q()}{R(x_{1},\ldots,x_{m})}.$

Clearly, any point in $A$ lies in the completeness area of $Q$ if
and only if $\phi$ is a tautology.
\end{proof}
Similarly as in the section before, this result may be little harmful
in practice, as in practice likely only few attributes of a feature
will be used for completeness statements, and only few completeness
statements will overlap (see also the discussion of the statements
in the following section).

\section{Discussion}

\label{sec:geocompl-discussion} In this section, we discuss various
practical considerations regarding the theory presented so far.

\paragraph{Pragmatics}

All queries used in this chapter are in some way asymmetric, as the
features used in constraining the output feature are more seldom than
the output feature. E.g., there are much more schools than than nuclear
power plants, or considerable more hotels than train stations. Possibly,
this will also hold for most practically used queries. If that is
the case, then in the instance reasoning, the condition that a constraining
feature is complete but out of range is likely to contribute significantly
to the completeness area of queries.

\paragraph{Language of Statements in OpenStreetMap}

The statements as used on the OSM Wiki (see Figure \ref{fig:legend-osm-statements})
do not use comparisons. Thus, the reasoning is in PTIME. The statements
in OSM are furthermore of an easy kind because they do not use constants
at all, but just express that one out of 12 feature classes is complete
in a certain area. Furthermore, the statements are also computationally
well-behaved in another way: The areas for which they are given do
not overlap, instead, the statements are always given for disjoint
areas (compared with stating that Irish pubs are complete in all Abingdon
and pubs are also complete in the center of Abingdon, which are spatially
and semantically overlapping statements).

\paragraph{Current Usage in OpenStreetMap}

So far, the use of completeness statements on the OSM wiki is sparse.
More concretely, out of 22,953 on 11th of June, 2013, only approximately
1,300 Wiki pages (\textasciitilde{}5\%) give completeness statements
(estimate based on number of pages that contain an image used in the
table).

Another limitation is that at the moment completeness statements are
only given for urban areas. This may change if completeness statements
become more frequently used.

Particular challenges are the dynamicity of the real-world in two
aspects: New features can arise that toggle previous completeness
statements incorrect, and features can disappear.

The first challenge can be addressed by regularly reviewing completeness
statement, and giving completeness guarantees only with time stamps
(``complete as of xx.yy.zzzz''). The second challenge goes beyond
the term of completeness, and instead asks also for correctness guarantees.
Mappers then not only would have to guarantee that all information
of the real world is captured in the database, but also the contrary.

In OSM, completeness statements come in 7 different levels, ranging
from unknown to completeness verified by two persons (see Figure \ref{fig:legend-osm-statements}).
In that figure the lower table also contains a row concerning the
implications on usage (\textquotedbl{}Use for navigation\textquotedbl{}).
Still, it remains hard know how to interpret the levels and to know
the implications on data usage.

\paragraph{Gamification}

Using games to achieve human computation tasks is a popular topic.
Games such as Google's Ingress%
\footnote{www.ingress.com%
} have shown that there is a considerable interest in geographical
augmented-reality games. Projects such as Urbanopoly \cite{urbanopoly}
show that this interest could in principle also be utilized for computation
of geographic information.

The general aims of a project for promoting completeness statement
usage would be twofold:
\begin{enumerate}
\item To obtain as many and as general completeness statements as possible
\item To ensure that the current statements are correct.
\end{enumerate}
To achieve both goals, one could introduce a game where users get
points for making correct statements, with the points being proportional
to the extend of the statements, thus covering the first goal. To
cover the second goal, the game would also reward the falsification
of completeness statements, and penalize the players that gave wrong
statements. Correctness of falsifications could be based on common
crowd-sourced consolidation techniques as discussed for instance in
\cite{human-computation-consolidation}.

The charm of such a method would be that by altering the reward function,
one could steer user efforts into topics of interest, e.g., by giving
more points for mapping efforts in areas with low completeness.

\section{Related Work}

\label{sec:geocompl-related work}

To the best of our knowledge, the only work on analyzing the completeness
of OpenStreetMap was done Mooney et al. \cite{quality-metrics-for-osm}
and Haklay and Ellul \cite{OSM-completeness-analysis-england,OSM-completeness-analysis-2}.
The former introduced general quality metrics for OSM, while the latter
analyzed the completeness of the road maps in England by comparing
them with government data sources.

Regarding metadata based completeness assessment of geographical data,
no work has been done so far.

\section{Summary}

In this chapter we have discussed how to assess the completeness of
spatial databases based on metadata. For the class of distance queries,
we have reduced completeness assessment to the assessment of simple
queries, combined with the consideration of possible answers. We have
also shown that in principle, giving explanations for incompleteness
and reasoning over statements with comparisons is coNP-hard.

The statements used in OpenStreetMap are of a simple kind however,
and give expectations that systems implementing our algorithms using
the OSM completeness statements will face little computational challenges.
On the other hand, the conceptual challenges regarding the maintenance
and meaning of completeness statements are more serious, and cannot
be answered only from the formal side.

We are currently working on an implementation of our theory %
\footnote{http://www.inf.unibz.it/\textasciitilde{}srazniewski/geoCompl/%
}, a screenshot of our system can be seen in Figure \ref{fig:demo}.
We use the Java Topoloy Suite (JTS) for the computation of the spatial
operations, and plan to use Leaflet scripts for deploying the demo
online.

\begin{figure}
\includegraphics[scale=0.38]{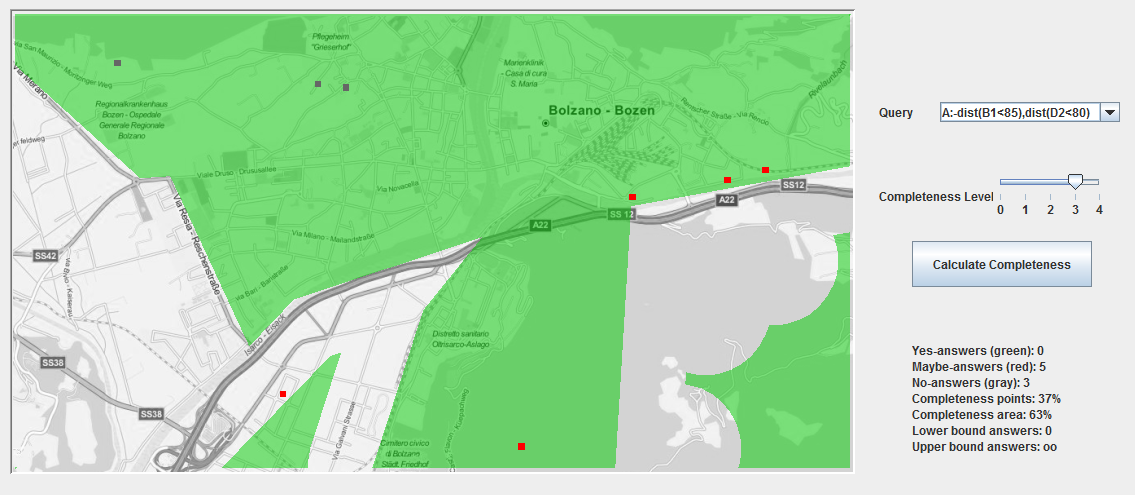}

\caption{Screenshot from our implementation of geographic completeness reasoning}

\label{fig:demo}

\end{figure}

\chapter{Linked Data}

\label{chap:lod}

With thousands of RDF data sources today available on the Web, covering
disparate and possibly overlapping knowledge domains, the problem
of providing high-level descriptions (in the form of metadata) of
their content becomes crucial. In this chapter we discuss reasoning
about the completeness of semantic web data sources. We show how the
previous theory can be adapted for RDF data sources, what peculiarities
the SPARQL query language offers and how completeness statements themselves
can be expressed in RDF. This chapter originated from the co-supervision
of the master thesis of Fariz Darari \cite{darari:master:thesis:2013}.
Subsequently, the results have been published at the International
Semantic Web Conference 2013 \cite{darari:ISWC:2013}.

This chapter discusses the foundation for the expression of completeness
statements about RDF data sources. The aim is to complement with \textit{qualitative}
descriptions about completeness the existing proposals like \void\ that
mainly deal with \textit{quantitative} descriptions. We develop a
formalism and show its feasibility. The second goal of this chapter
is to show how completeness statements can be useful for the semantic
web in practice. We believe that the results have both a theoretical
and practical impact. On the theoretical side, we provide a formalization
of completeness for RDF data sources and techniques to reason about
the completeness of query answers. From the practical side, completeness
statements can be easily embedded in current descriptions of data
sources and thus readily used. The results presented in this chapter
have been implemented by Darari in a demo system called CORNER.

\noindent \textbf{Outline.} The chapter is organized as follows. Section
\ref{sec:introduction} provides an introduction to the Semantic Web
and the challenges wrt.\ completeness. Section~\ref{sec:scenario}
discusses a real world scenario and provides a high level overview
of the completeness framework. Section~\ref{sec:formal-framework}
after providing some background introduces a formalization of the
completeness problem for RDF data sources. This section also describes
how completeness statements can be represented in RDF. In Section~\ref{sec:compl-querying-single}
we discuss how completeness statements can be used in query answering
when considering a single data source at a time. In Section~\ref{sec:discussion}
we discuss some aspects of the proposed framework, and in Section
\ref{sec:rdf:relatedwork} we discuss related work.

\section{Background}

\label{sec:introduction} The Resource Description Framework (RDF)~\cite{W3C04}
is the standard data model for the publishing and interlinking of
data on the Web. It enables the making of \textit{statements} about
(Web) resources in the form of triples including a \textit{subject},
a \textit{predicate} and an \textit{object}. Ontology languages such
as RDF Schema (RDFS) and OWL provide the necessary underpinning for
the creation of vocabularies to structure knowledge domains. Friend-of-a-Friend
(FOAF), Schema.RDFS.org and Dublin Core (DC) are a few examples of
such vocabularies. 
RDF is now a reality; efforts like the Linked Open Data project~\cite{heath2011}
give a glimpse of the magnitude of RDF data today available online.
The common path to access such huge amount of structured data is via
SPARQL endpoints, that is, network locations that can be queried upon
by using the SPARQL query language~\cite{W3C-SPARQL1.1}.


With thousands of RDF data sources covering possibly overlapping knowledge
domains the problem of providing high-level descriptions (in the form
of metadata) of their content becomes crucial. Such descriptions will
connect data publishers and consumers; publishers will advertise ``what''
there is inside a data source so that specialized applications can
be created for data source discovering, cataloging, selection and
so forth. Proposals like the \void~\cite{W3C11} vocabulary touched
this aspect. With \void\ it is possible to provide statistics about
how many instances a particular \textit{class} has, info about its
SPARQL endpoint and links with other data sources, among the other
things. However, \void\ mainly focuses on providing \textit{quantitative}
information. 
We claim that toward comprehensive descriptions of data sources \textit{qualitative}
information is crucial.

\section{Motivating Scenario}

\label{sec:scenario} In this section we motivate the need of formalizing
and expressing completeness statements in a machine-readable way.
Moreover we show how completeness statement are useful for query answering.
We start our discussion with a real data source available on the Web.
Figure \ref{fig:imdb} shows a screenshot taken from the IMDB web
site. The page is about the movie Reservoir Dogs; in particular it
lists the cast and crew of the movie. For instance, it says that Tarantino
was not only the director and writer of the movie but also the character
Mr. Brown. As it can be noted, the data source includes a ``completeness
statement'', which says that the page is \textit{complete for all
cast and crew members of the movie}. The availability of such statement
increases the potential value of the data source. In particular, users
that were looking for information about the cast of this movie and
found this page can prefer it to other pages since, assuming the truth
of the statement, all they need is here. 
\begin{figure}[!h]
\centering \includegraphics[width=0.9\textwidth]{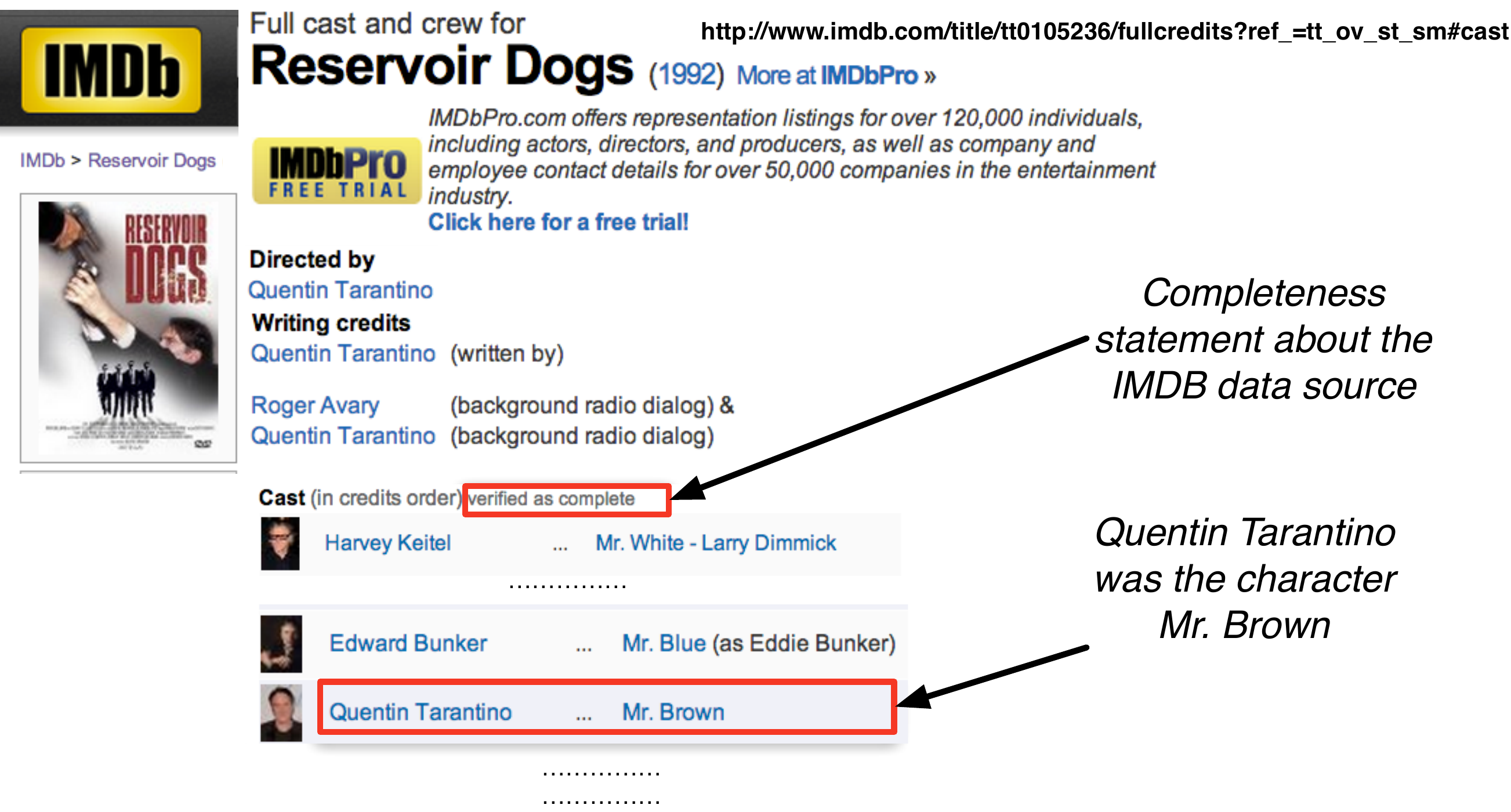}
\caption{A \textit{completeness statement} in IMDB.}

\label{fig:imdb} 
\end{figure}

The problem with such kind of statements, expressed in natural language,
is that they cannot be automatically processed, thus hindering their
applicability, for instance, in query answering. Indeed, the interpretation
of the statement ``verified as complete'' is left to the user. On
the other hand, a reasoning and querying engine when requested to
provide information about the cast and crew members of Reservoir Dogs
could have leveraged such statement and inform the user about the
completeness of the results.

\smallskip{}
\textbf{Machine readable statements.} In the RDF and linked data context
with generally incomplete and possibly overlapping data sources and
where \textit{``anyone can say anything about any topic and publish
it anywhere''} having the possibility to express completeness statements
becomes an essential aspect. The machine-readable nature of RDF enables
to deal with the problems discussed in the example about IMDB; completeness
statements can be represented in RDF. As an example, the high-level
description of a data source like DBpedia could include, for instance,
the fact that it is complete for all of Quentin Tarantino's movies.
Figure \ref{fig:graph-compl} shows how the data source DBPedia can
be complemented with completeness statements expressed in our formalism.
Here we give a high level presentation of the completeness framework;
details on the theoretical framework supporting it are given in Section~\ref{sec:formal-framework}.
\begin{figure}[!h]
\centering \includegraphics[width=0.9\textwidth]{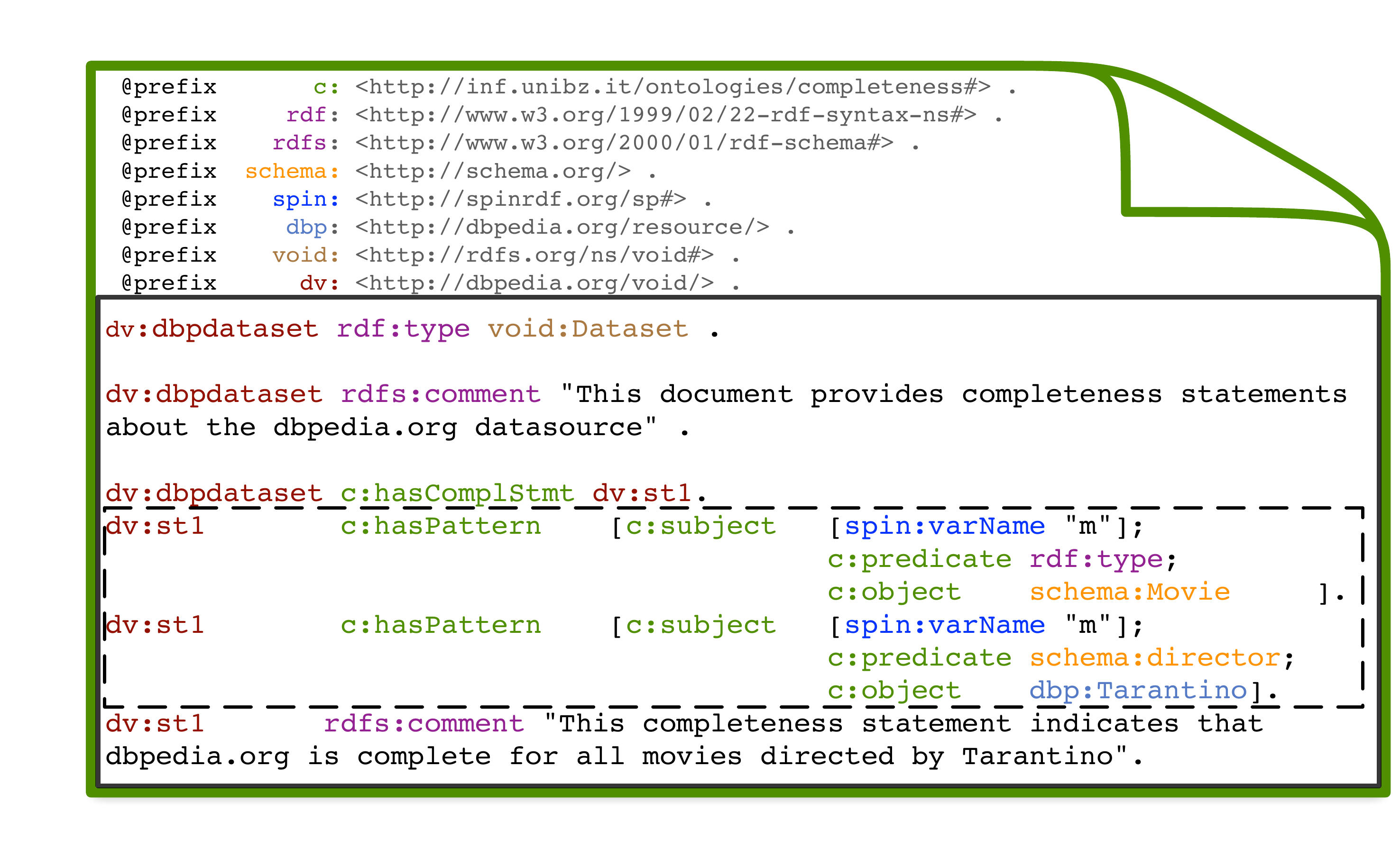}
\caption{Completeness statements about \texttt{dbpedia.org}}

\label{fig:graph-compl} 
\end{figure}

\noindent A simple statement can be thought of as a SPARQL Basic Graph
Pattern (BGP). The BGP{\small{} }\texttt{\small{}(?m rdf:type schema:Movie).(?m
schema:director}~\linebreak{}
\texttt{\small{} dbp:Tarantino)}, for instance, expresses the fact
that \texttt{dbpedia.org} is complete for all movies directed by Tarantino.
In the figure, this information is represented by using an ad-hoc
completeness vocabulary (see Section~\ref{subsec:representationInRDF})
with some properties taken from the SPIN%
\footnote{http://spinrdf.org/sp.html\#sp-variables.%
} vocabulary. For instance, the $\texttt{compl:hasPattern}$ links
a completeness statement with a pattern.

\smallskip{}
\textbf{Query Completeness.} The availability of completeness statements
about data sources is useful in different tasks, including data integration,
data source discovery and query answering. In this chapter we will
focus on how to leverage completeness statements for query answering.
The research question we address is how to assess whether available
data sources with different degree of completeness can ensure the
completeness of query answers. 
Consider the scenario depicted in Figure \ref{fig:example-query}
where the data sources DBpedia and LinkedMDB are described in terms
of their completeness. The Web user Syd wants to pose the query $Q$
to the SPARQL endpoints of these two data sources asking for \textit{all
movies directed by Tarantino in which Tarantino also starred}. By
leveraging the completeness statements, the query engines at the two
endpoints could tell Syd whether the answer to his query is complete
or not. For instance, although DBPedia is complete for all of Tarantino's
movies (see Figure \ref{fig:graph-compl}) nothing can be said about
his participation as an actor in these movies (which is required in
the query). Indeed, at the time of writing this chapter, DBPedia is
actually incomplete; this is because in the description of the movie
Reservoir Dogs the fact is missing that Tarantino was the character
Mr. Brown (and from Figure \ref{fig:imdb} we know that this is the
case). On the other hand, LinkedMDB, the RDF counterpart of IMDB,
can provide a complete answer. Indeed, with our framework it is possible
to express in RDF the completeness statement available in natural
language in Figure \ref{fig:imdb}. This statement has then been used
by the CORNER reasoning engine, implementing our formal framework,
to state the completeness of the query. 
\begin{figure}[!h]
\centering \includegraphics[width=1\textwidth]{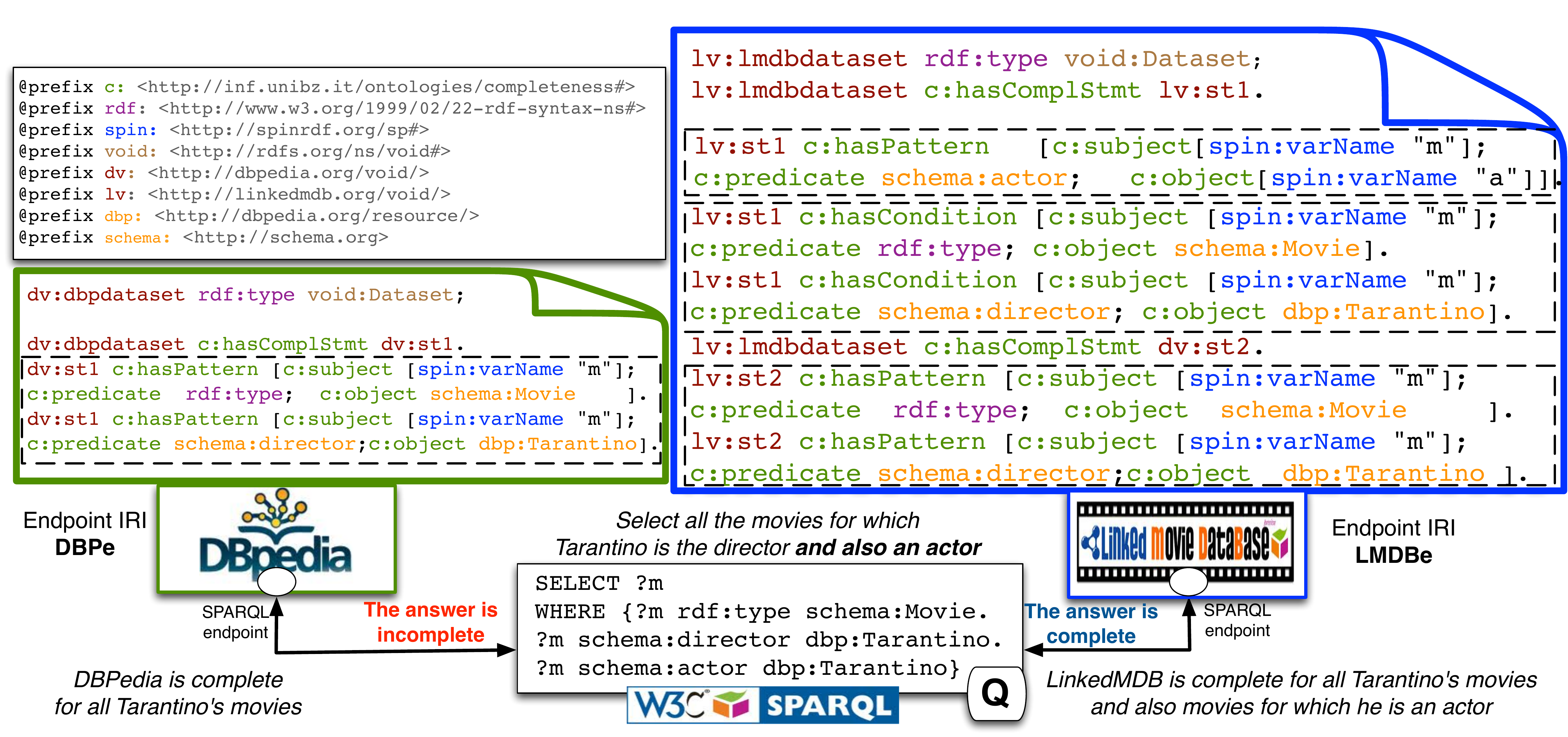}
\caption{Completeness statements and query answering.}

\label{fig:example-query} 
\end{figure}

In this specific case, LinkedMDB can guarantee the completeness of
the query answer because it contains all the actors in Tarantino's
movies (represented by the statement \texttt{\small{}lv:st1}) in addition
to the\linebreak{}
 Tarantino's movies themselves (represented by the statement \texttt{\small{}lv:st2}). 

Note that the statement \texttt{\small{}lv:st1} includes two parts:
(i) the pattern, which is expressed via the BGP \texttt{\small{}(?m,
schema:actor, ?a)} and (ii) the conditions, that is, the BGP \texttt{\small{}(?m,
rdf:type, schema:Movie).(?m, schema:director, dbp:Tarantino)}. Indeed,
a completeness statement allows one to say that a certain part (i.e.,
with respect to some conditions) of data is complete, or in other
words, it can be used to state that a data source contains all triples
in a pattern $P_{1}$ that satisfy a condition $P_{2}$. The detailed
explanation and the semantics of completeness statements can be found
in Section~\ref{sec:formal-framework}.

\section{Framework for RDF Data}

\label{sec:formal-framework} In the following, we introduce the RDF
data format and the SPARQL query language, and show how the previous
notions for talking about data completeness can be extended to the
new setting.

\smallskip{}
\textbf{RDF and SPARQL.} \label{sec:background} We assume that there
are three pairwise disjoint infinite sets $I$ (\emph{IRIs}), $L$
(\emph{literals}) and $V$~(\emph{variables}). We collectively refer
to IRIs and literals as \emph{RDF terms} or simply \emph{terms}. A
tuple $(s,p,o)\in I\times I\times(I\cup L)$ is called an \emph{RDF
triple} (or a \emph{triple}), where $s$ is the \emph{subject}, $p$
the \emph{predicate} and $o$ the \emph{object} of the triple. An
\emph{RDF graph} or \emph{data source} consists of a finite set of
triples~\cite{W3C04}. For simplicity, we omit namespaces for the
abstract representation of RDF graphs.

\noindent The standard query language for RDF is SPARQL. The basic
building blocks of a SPARQL query are \emph{triple patterns,} which
resemble RDF triples, except that in each position also variables
are allowed. SPARQL queries include \emph{basic graph patterns} (BGP),
built using the $\AND$ operator, and further operators, including
$\OPT$, $\FILTER$, $\UNION$ and so forth. In this paper we consider
the operators $\AND$ and $\OPT$. Moreover, we also consider the
result modifier $\DISTINCT$. Evaluating a graph pattern $P$ over
an RDF graph $G$ results in a set of mappings $\mu$ from the variables
in $P$ to terms, denoted as $\eval PG$. Further information about
SPARQL can be found in \cite{perez2009semantics}.

SPARQL queries come as $\SELECT$, $\ASK$, or $\CONSTRUCT$ queries.
A $\SELECT$ query has the abstract form $(W,P)$, where $P$ is a
graph pattern and $W$ is a subset of the variables in $P$. A $\SELECT$
query $Q=(W,P)$ is evaluated over a graph $G$ by restricting the
mappings in $\eval PG$ to the variables in $W$. The result is denoted
as $\eval QG$. Syntactically, an $\ASK$ query is a special case
of a $\SELECT$ query where $W$ is empty. For an $\ASK$ query $Q$,
we write also $\eval QG=\true$ if $\eval QG\neq\eset$, and $\eval QG=\false$
otherwise. A $\CONSTRUCT$ query has the abstract form $(P_{1},P_{2})$,
where $P_{1}$ is a BGP and $P_{2}$ is a graph pattern. In this paper,
we only use $\CONSTRUCT$ queries where also $P_{2}$ is a BGP. The
result of evaluating $Q=(P_{1},P_{2})$ over $G$ is the graph $\eval QG$,
that is obtained by instantiating the pattern $P_{1}$ with all the
mappings in $\eval{P_{2}}G$.

Later on, we will distinguish between three classes of queries: (i)
Basic queries, that is, queries $\sparqlquery WP$ where $P$ is a
BGP and which return bags of mappings (as it is the default in SPARQL),
(ii) $\DISTINCT$ queries, that is, queries $\sparqlquery WP^{d}$
where $P$ is a BGP and which return sets of mappings, and (iii) $\OPT$
queries, that is, queries $\sparqlquery WP$ without projection ($W=\var(P)$)
and $P$ is a graph pattern with $\OPT$.

\subsection{Completeness Statements and Query Completeness}

\label{sec:formal-definitions} We are interested in formalizing when
a query is complete over a potentially incomplete data source and
in describing which parts of such a source are complete. When talking
about the completeness of a source, one implicitly compares the information
\emph{available} in the source with what holds in the world and therefore
should \emph{ideally} be also present in the source. As before, we
only consider sources that may miss information, but do not contain
wrong information. 
\begin{defn}
[Incomplete Data Source] We identify data sources with RDF graphs.
Then, adapting the previous notion of incomplete databases, we define
an incomplete data source as a pair $\pg=(\Ga,\Gi)$ of two graphs,
where $\Ga\subseteq\Gi$. We call $\Ga$ the \emph{available} graph
and $\Gi$ the \emph{ideal} graph. \label{def:incomplete-ds} 
\end{defn}

\begin{example}
[Incomplete Data Source] \label{ex:incompl_data_source} Consider
the DBpedia data source and suppose that the only movies directed
by Tarantino are Reservoir Dogs, Pulp Fiction, and Kill Bill, and
that Tarantino was starred exactly in the movies Desperado, Reservoir
Dogs, and Pulp Fiction. For the sake of example, suppose also that
the fact that he was starred in Reservoir Dogs is missing in DBpedia%
\footnote{as it was the case on 7 May 2013%
}. Using Definition~\ref{def:incomplete-ds}, we can formalize the
incompleteness of the DBpedia data source $\pg_{\dbp}$ as: 
\begin{align*}
\Ga_{\dbp}=\; & \{\triple{reservoirDogs}{director}{tarantino},\!\triple{pulpFiction}{director}{tarantino},\\[-2pt]
 & \triple{killBill}{director}{tarantino},\triple{desperado}{actor}{tarantino},\\[-2pt]
 & \triple{pulpFiction}{actor}{tarantino},\triple{desperado}{type}{Movie},\\[-2pt]
 & \triple{reservoirDogs}{type}{Movie},\triple{pulpFiction}{type}{Movie},\\
 & \triple{killBill}{type}{Movie}\}\\
\Gi_{\dbp}=\; & \Ga_{dbp}\cup\set{\triple{reservoirDogs}{actor}{tarantino}}
\end{align*}
 
\end{example}

We now introduce \emph{completeness statements}, which are used to
denote the partial completeness of a data source, that is, they describe
for which parts the ideal and available graph coincide. 
\begin{defn}
[Completeness Statement] A completeness statement\linebreak{}
 $\defineGC{P_{1}}{P_{2}}$ consists of a non-empty BGP $P_{1}$ and
a BGP $P_{2}$. We call $P_{1}$ the \emph{pattern} and $P_{2}$ the
\emph{condition} of the completeness statement. \label{def:compl-stat} 
\end{defn}
For example, we express that a source is complete for all pairs of
triples that say \quotes{\emph{$?m$ is a movie and $?m$ is directed
by Tarantino}} using the statement 
\begin{equation}
\gc_{\mathit{dir}}=\defineGC{\triple{?m}{type}{Movie},\triple{?m}{director}{tarantino}}{\emptyset},\label{eq:1}
\end{equation}
 whose pattern matches all such pairs and whose condition is empty.
To express that a source is complete for all triples about acting
in movies directed by Tarantino, we use 
\begin{equation}
\gc_{\mathit{act}}=\defineGC{\triple{?m}{actor}{?a}}{\triple{?m}{director}{tarantino},\triple{?m}{type}{Movie}},\label{eq:completeness-statement-act}
\end{equation}
 whose pattern matches triples about acting and the condition restricts
the acting to movies directed by Tarantino.

We define the satisfaction of completeness statements over incomplete
data sources analogous to the one for table completeness statements
over incomplete databases (Definition \ref{def:satisfaction-of-tc}).
To a statement $C=\defineGC{P_{1}}{P_{2}}$, we associate the \CONSTRUCT\ query
$Q_{C}=(P_{1},P_{1}\cup P_{2})$. Note that, given a graph $G$, the
query $Q_{C}$ returns those instantiations of the pattern $P_{1}$
that are present in $G$ together with an instantiation of the condition.
For example, the query $Q_{\gc_{\mathit{act}}}$ returns all the actings
in Tarantino movies in $G$. 
\begin{defn}
[Satisfaction of Completeness Statements]For an incomplete data source
$\pg=(\Ga,\Gi)$, the statement $C$ is satisfied by $\pg$, written
$\pg\models C$, if $\eval{Q_{C}}{\Gi}\subseteq\Ga$ holds.\end{defn}
\begin{example}
To see that the statement $\gc_{\mathit{dir}}$ is satisfied by $\pg_{\dbp}$,
observe that the query $Q_{\gc_{\mathit{dir}}}$ returns over $\Gi_{\dbp}$
all three movie triples in $\id G_{\dbp}$, and that all these triples
are also in $\Ga_{\dbp}$. However, $\gc_{\mathit{act}}$ is \emph{not}
satisfied by $\pg_{\dbp}$, because $Q_{\gc_{\mathit{act}}}$ returns
over $\Gi_{\dbp}$ the triple $\triple{reservoirDogs}{actor}{tarantino}$,
which is not in $\Ga_{\dbp}$.
\end{example}
Observe that the completeness statements defined here go syntactically
beyond the table completeness statements introduced in Section \ref{sec:prelim:table-completeness},
as they allow more than one atom in the head of the statement. However,
these statements can easily be translated to a set of linearly many
statements with only one atom in the head as follows: 
\begin{prop}
Consider a completeness statement $C=\defineGC{P_{1}}{P_{2}}$ with
$P_{1}=t_{1},\ldots,t_{n}$. Then any incomplete data source $\pg$
satisfies $C$ if and only it satisfies the following statements:
\begin{eqnarray*}
 & \defineGC{t_{1}}{P_{1}\setminus\set{t_{1}},P_{2}}\\
 & \ldots\\
 & \defineGC{t_{n}}{P_{1}\setminus\set{t_{n}},P_{2}}.
\end{eqnarray*}

\end{prop}
When querying a potentially incomplete data source, we would like
to know whether at least the answer to our query is complete. For
instance, when querying DBpedia for movies starring Tarantino, it
would be interesting to know whether we really get all such movies,
that is, whether our query is complete over DBpedia. We next formalize
query completeness with respect to incomplete data sources.
\begin{defn}
[Query Completeness] Let $Q$ be a \SELECT\ query. To express that
\emph{$Q$ is complete}, we write $\buildQC Q$. An incomplete data
source $\calG=(\Ga,\Gi)$ satisfies the expression $\buildQC Q$,
if $Q$ returns the same result over $\Ga$ as it does over $\Gi$,
that is $\evalsetbag Q{\Ga}=\evalsetbag Q{\Gi}$. In this case we
write $\pg\models\buildQC Q$. \label{def:query-compl} \end{defn}
\begin{example}
[Query Completeness] \label{ex:query_compl} Consider the incomplete
data source $\pg_{\dbp}$ and the two queries $Q_{dir}$, asking for
all movies directed by Tarantino, and $Q_{dir+act}$, asking for all
movies, both directed by and starring Tarantino: 
\begin{align*}
Q_{dir}= & \setquery{?m}{\triple{?m}{type}{Movie},\triple{?m}{director}{tarantino}}\\
Q_{dir+act}= & (\{?m\}\{\triple{?m}{type}{Movie},\triple{?m}{director}{tarantino},\\
 & \triple{?m}{actor}{tarantino}\}
\end{align*}
Then, it holds that $Q_{dir}$ is complete over $\pg_{\dbp}$ and
$Q_{dir+act}$ is not. Later on, we show how to deduce query completeness
from completeness statements. \label{ex:query-compl}
\end{example}

\subsection{RDF Representation of Completeness Statements}

\label{subsec:representationInRDF} Practically, completeness statements
should be compliant with the existing ways of giving metadata about
data sources, for instance, by enriching the VoID description~\cite{W3C11}.
Therefore, it is essential to express completeness statements in RDF
itself. Suppose we want to express that LinkedMDB satisfies the statement:
\[
\gc_{\mathit{act}}=\defineGC{\triple{?m}{actor}{?a}}{\triple{?m}{type}{Movie},\triple{?m}{director}{tarantino}}.
\]
Then, we need vocabulary to say that this is a statement about LinkedMDB,
which triple patterns make up its pattern, and which its condition.
We also need a vocabulary to represent the constituents of the triple
patterns, namely subject, predicate, and object of a pattern. Therefore,
we introduce the property names whose meaning is intuitive: 
\begin{align*}
 & \VhasComplStmt,\ \VhasPattern,\ \VhasCondition,\ \Vsubject,\ \Vpredicate,\ \Vobject.
\end{align*}
If the constituent of a triple pattern is a term (an IRI or a literal),
then it can be specified directly in RDF. Since this is not possible
for variables, we represent a variable by a resource that has a literal
value for the property \VvarName. Now, we can represent $\gc_{\mathit{act}}$
in RDF as the resource \texttt{lv:st1} described in Figure~\ref{fig:example-query}.

More generally, consider a completeness statement $\defineGC{P_{1}}{P_{2}}$,
where $P_{1}=\set{t_{1},\ldots,t_{n}}$ and $P_{2}=\set{t_{n+1},\ldots,t_{m}}$
and each $t_{i}$, $1\leq i\leq m$, is a triple pattern. Then the
statement is represented using a resource for the statement and a
resource for each of the $t_{i}$ that is linked to the statement
resource by the property \VhasPattern\ or \VhasCondition, respectively.
The constituents of each $t_{i}$ are linked to $t_{i}$'s resource
in the same way via \Vsubject, \Vpredicate, and \Vobject. All resources
can be either IRIs or blank nodes.

\section{Completeness Reasoning over a Single Data Source}

\label{sec:compl-querying-single} In this section, we show how completeness
statements can be used to judge whether a query will return a complete
answer. We first focus on completeness statements that hold on a \emph{single}
data source, while completeness statements in the federated setting
are discussed in Section~\ref{sec:querying-federated}. 

\medskip{}
\textbf{Problem Definition.} Let $\gcs$ be a set of completeness
statements and $Q$ be a \SELECT\ query. We say that \emph{$\gcs$
entails the completeness of $Q$}, written $\gcs\models\buildQC Q$,
if any incomplete data source that satisfies $\gcs$ also satisfies
$\buildQC Q$.
\begin{example}
Consider $\gc_{\mathit{dir}}$ from (\ref{eq:1}). Whenever an incomplete
data source $\G$ satisfies $\gc_{\mathit{dir}}$, then $\Ga$ contains
all triples about movies directed by Tarantino, which is exactly the
information needed to answer query $Q_{\mathit{dir}}$ from Example
\ref{ex:query-compl}. Thus, $\set{\gc_{\mathit{dir}}}\models\buildQC{Q_{\mathit{dir}}}$.
This may not be enough to completely answer $Q_{\mathit{dir+act}}$,
thus $\set{\gc_{\mathit{dir}}}\not\models\buildQC{Q_{\mathit{dir+act}}}$.
We will now see how this intuitive reasoning can be formalized. 
\end{example}

\subsection{Completeness Entailment for Basic Queries}

\noindent \label{sec:compl-querying-plain} In difference to the characterization
for TC-QC entailment that is shown in Theorem \ref{thm:reduction_lc-qc_to_containment},
we now use characterization for completeness entailment that is similar
to the one in Theorem \ref{lem-characterization:of:inc:sql}.

To characterize completeness entailment, we use the fact that completeness
statements have a correspondence in \CONSTRUCT\
queries. We reuse the operator $\tcop$ from Definition \ref{def:tc-operator},
but define it now using the \CONSTRUCT\
queries as a mapping from graphs to graphs. Let $\C$ be a set of
completeness statements. Then 
\[
T_{\C}(G)=\bigcup_{C\in\C}Q_{C}(G).
\]
 Notice that for any data source $G$, the pair $(T_{\C}(G),G)$ is
an incomplete data source satisfying $\C$ and $T_{\C}(G)$ is the
smallest set (wrt.\ set inclusion) for which this holds wrt.\ the
ideal data source $G$.
\begin{example}
[Completeness Entailment] Consider the set of completeness statements
$\gcs_{\mathit{dir,act}}=\set{\gc_{\mathit{dir}},\gc_{\mathit{act}}}$
and the query $Q_{\mathit{dir+act}}$. Recall that the query has the
form $Q_{\mathit{dir+act}}=\sparqlquery{\set{m?}}{P_{\mathit{dir+act}}}$,
where 
\[
P_{\mathit{dir+\! act}}\!=\!\set{\triple{?m}{type}{Movie},\!\triple{?m}{director}{tarantino},\!\triple{?m}{actor}{tarantino}}.
\]
We want to check whether these statements entail the completeness
of $Q_{\mathit{dir+act}}$, that is, whether $\gcs_{\mathit{dir,act}}\models\buildQC{Q_{\mathit{dir+act}}}$
holds.
\end{example}

\begin{example}
[Completeness Entailment Checking] Suppose that $\pg=(\Ga,\Gi)$
satisfies $\gcs_{\mathit{dir,act}}$. Suppose also that $Q_{\mathit{dir+act}}$
returns a mapping $\mu=\set{?m\mapsto m'}$ over $\Gi$ for some term~$m'$.
Then $\Gi$ contains $\mu P_{\mathit{dir+act}}$, the instantiation
by $\mu$ of the BGP of our query, consisting of the three triples
$\triple{m'}{type}{Movie}$, $\triple{m'}{director}{tarantino}$,
and $\triple{m'}{actor}{tarantino}$.

The \CONSTRUCT\ query $Q_{\gc_{\mathit{dir}}}$, corresponding to
our first completeness statement, returns over $\mu P_{\mathit{dir+act}}$
the two triples $\triple{m'}{type}{Movie}$ and $\triple{m'}{director}{tarantino}$,
while the \CONSTRUCT\ query $Q_{\gc_{\mathit{act}}}$, corresponding
to the second completeness statement, returns the triple $\triple{m'}{actor}{tarantino}$.
Thus, all triples in $\mu P_{\mathit{dir+act}}$ have been reconstructed
by $T_{\gcs_{\mathit{dir,act}}}$ from $\mu P_{\mathit{dir+act}}$.

Now, we have $\mu P_{\mathit{dir+act}}=T_{\gcs_{\mathit{dir,act}}}(P_{\mathit{dir+act}})\subseteq T_{\gcs_{\mathit{dir,act}}}(\Gi)\subseteq\Ga$,
where the last inclusion holds due to $\pg\models\gcs_{\mathit{dir,act}}$.
Therefore, our query $Q_{\mathit{dir+act}}$ returns the mapping $\mu$
also over~$\Ga$. Since $\mu$ and $\pg$ were arbitrary, this shows
that $\gcs_{\mathit{dir,act}}\models\buildQC{Q_{\mathit{dir+act}}}$
holds. \label{ex:eval} 
\end{example}
In summary, in Example~\ref{ex:eval} we have reasoned about a set
of completeness statements $\gcs$ and a query $Q=\sparqlquery WP$.
We have considered a generic mapping $\mu$, defined on the variables
of $P$, and applied it to $P$, thus obtaining a graph $\mu P$.
Then we have verified that $\mu P=\TC(\mu P)$. From this, we could
conclude that for every incomplete data source $\pg=(\Ga,\Gi)$ we
have that $\eval Q{\Ga}=\eval Q{\Gi}$. Next, we make this approach
formal. 
\begin{defn}
[Prototypical Graph] Let $\sparqlquery WP$ be a query. The \emph{freeze
mapping} $\frozenID$ is defined as mapping each variable~$v$ in
$P$ to a new IRI $\frozen v$. Instantiating the graph pattern $P$
with $\frozenID$ yields the RDF graph $\frozen P:=\frozenID\, P$,
which we call the prototypical graph of $P$. \label{def:prot-graph} 
\end{defn}
\noindent Now we can generalize the intuitive reasoning from above
to a generic completeness check, analogous to Theorem \ref{thm:reduction_lc-qc_to_containment}:

\global\long\def\restr#1#2{#1_{|#2}}
 
\begin{thm}
[Completeness of Basic Queries]\label{thm:basic-queries} Let $\gcs$
be a set of completeness statements and let $Q=\sparqlquery WP$ be
a basic query. Then
\[
\gcs\models\buildQC Q\qquad\mbox{if and only if}\qquad\frozen P=\TC(\frozen P).
\]

\label{th-compl-query} \end{thm}
\begin{proof}
$\proofr$ If $\frozen P\neq T_{\C}(\frozen P)$, then the pair $(T_{\C}(G),G)$
is a counterexample for the entailment. It satisfies $\C$, but does
not satisfy $\buildQC Q$ because the mapping $\restr{\frozenID}{\var P}$,
the restriction of the frozen identity $\frozenID$ to the variables
of $P$, cannot be retrieved by $Q$ over the available graph $T_{\C}(\frozen P)$.

$\proofl$ If all triples of the pattern $\frozen P$ are preserved
by $\TC$, then this serves as a proof that in any incomplete data
source all triples that are used to compute a mapping in the ideal
graph are also present in the available graph.
\end{proof}



\global\long\def\osp{\;}
\global\long\def\osps#1{{\osp{#1}\osp}}

\global\long\def\P{{\cal P}}

\subsection{Queries with \DISTINCT}

In SPARQL, answers to basic queries may contain duplicates, that is,
they are evaluated according to bag semantics. The use of the \DISTINCT\ keyword
eliminates duplicates, thus corresponding to query evaluation under
set semantics. For a query $Q$ involving \DISTINCT, the difference
to the characterization in Theorem~\ref{thm:basic-queries} is that
instead of retrieving the full pattern $\tilde{P}$ after applying
$T_{\C}$, we only check whether sufficient parts of $\tilde{P}$
are preserved that still allow to retrieve the identity mapping on
the distinguished variables of $Q$.

\subsection{Completeness of Queries with the $\OPT$ Operator}

One interesting feature where SPARQL goes beyond SQL is the $\OPT$
(\quotes{optional}) operator. With $\OPT$ one can specify that
parts of a query are only evaluated if an evaluation is possible,
similarly to an outer join in SQL. For example, when querying for
movies, one can also ask for the prizes they won, if any. The $\OPT$
operator is used substantially in practice~\cite{PicalausaV11}. 

Intuitively, the mappings for a pattern $P_{1}\osps{\OPT}P_{2}$ are
computed as the union of all the bindings of $P_{1}$ together with
bindings for $P_{2}$ that are valid extensions, and including those
bindings of $P_{1}$ that have no binding for $P_{2}$ that is a valid
extension. For a formal definition of the semantics of queries with
the $\OPT$ operator, see~\cite{LetelierPPS12}. 

Completeness entailment for queries with $\OPT$ differs from that
of queries without:
\begin{example}
[Completeness with $\OPT$] Consider the following query 
\[
Q_{\mathit{maw}}=\queryOPT{\triple{?m}{type}{Movie}}{\triple{?m}{award}{?aw}}
\]
which asks for all movies and if available, also their awards. Consider
also
\[
C_{\mathit{aw}}=\defineGC{\triple{?m}{type}{Movie},\triple{?m}{award}{?aw}}{\emptyset}
\]
a completeness statement that expresses that all movies that have
an award are complete and all awards of movies are complete. If the
query $Q_{\mathit{maw}}$ used \AND instead of $\OPT$, then its
completeness could be entailed by $C_{\mathit{aw}}$. However with
$\OPT$ in $Q_{\mathit{maw}}$, more completeness is required: Also
those movies have to be complete that do not have an award. Thus,
$C_{\mathit{aw}}$ alone does not entail the completeness of $Q_{\mathit{maw}}$.
\end{example}
Graph patterns with $\OPT$ have a hierarchical structure that can
be made explicit by so-called pattern trees. A pattern tree $\T$
is a pair $(T,\P)$, where (i) $T=(N,E,r)$ is a tree with node set
$N$, edge set $E$, and root $r\in N$, and (ii) $\P$ is a labeling
function that associates to each node $n\in N$ a BGP $\P(n)$. We
construct for each triple pattern $P$ a corresponding pattern tree
$\T$.
\begin{example}
Consider a pattern $((P_{1}\OPT P_{2})\OPT(P_{3}\OPT P_{4}))$, where
$P_{1}$ to $P_{4}$ are BGPs. Its corresponding pattern tree would
have a root node labeled with $P_{1}$, two child nodes labeled with
$P_{2}$ and $P_{3}$, respectively, and the $P_{3}$ node would have
another child labeled with $P_{4}$. 
\end{example}
To this end, we first rewrite $P$ in such a way that $P$ consists
of BGPs connected with $\OPT$. For instance, $(t_{1}\,\OPT\, t_{2})\,\AND\, t_{3}$
would be equivalently rewritten as $(t_{1}\,\AND\, t_{3})\,\OPT\, t_{2}$.
If $P$ has this form, then we construct a pattern tree for $P$ as
follows. (i) If $P$ is a BGP, then the pattern tree of $P$ consists
of a single node, say $n$, which is the root. Moreover, we define
$\P(n)=P$. (ii) Suppose that $P=P_{1}\OPT P_{2}$. Suppose also that
we have constructed the pattern trees $\T_{1}$, $\T_{2}$ for $P_{1}$,
$P_{2}$, respectively, where $\T_{i}=(T_{i},\P_{i})$ and $T_{i}=(N_{i},E_{i},r_{i})$
for $i=1,2$. Suppose as well that $N_{1}$ and $N_{2}$ are disjoint.
Then we construct the tree $\T$ for $P$ by making the root $r_{2}$
of $T_{2}$ a child of the root $r_{1}$ of $T_{1}$ and defining
nodes, edges and labeling function accordingly.

Similarly to patterns, one can define how to evaluate pattern trees
over graphs, which leads to the notion of equivalence of pattern trees.
The evaluation is such that a pattern and the corresponding pattern
tree are equivalent in the sense that they give always rise to the
same sets of mappings. In addition, one can translate every pattern
tree in linear time into an equivalent pattern.

If one uses $\OPT$ without restrictions, unintuitive queries may
result. Pérez et al.\ have introduced the class of so-called well-designed
graph patterns that that avoid anomalies that may otherwise occur
\cite{perez2009semantics}. Well-designedness of a pattern $P$ is
defined in terms of the pattern tree $\T_{P}$. A pattern tree $\T$
is well-designed if all occurrences of all variables are connected
in the following sense: if there are nodes $n_{1}$, $n_{2}$ in $\T$
such that the variable $v$ occurs both in $\P(n_{1})$ and $\P(n_{2})$,
then for all the nodes $n$ on the path from $n_{1}$ to $n_{2}$
in $\T$ it must the case that $v$ occurs in $\P(n)$. We restrict
ourselves in the following to $\OPT$ queries with well-designed patterns,
which we call well-designed queries.

To formulate our characterization of completeness, we have to introduce
a normal form for pattern trees that frees the tree from redundant
triples. A triple $t$ in the pattern $\P(n)$ of some node $n$ of
$\T$ is \emph{redundant} if every variable in $t$ occurs also in
the pattern of an ancestor of $n$. Consider for example the pattern
$P_{ex}=\triple{?x}p{?y}\osps{\OPT}\triple{?x}r{?y}$. Its pattern
tree $\T_{ex}$ consists of two nodes, the root with the first triple,
and a child of the root with the second triple. Intuitively, the second
triple in the pattern is useless, since a mapping satisfies the pattern
if and only if it satisfies the first triple. Since all the variables
in the optional second triple occur already in the mandatory first
triple, no new variable bindings will result from the second triple.

From any well-designed pattern tree $\T$, one can eliminate in polynomial
time all redundant triples \cite{LetelierPPS12}. This may result,
however, in a tree that is no more well-designed. Letelier et al.~\cite{LetelierPPS12}
have shown that for every pattern tree $\T$ one can construct in
polynomial time an equivalent well-designed pattern tree $\T^{\mathit{NR}}$
without redundant triples, which is called the NR-normal form of $\T$.
The NR-normal form of $\T_{ex}$ above consists only of the root,
labeled with the triple $\triple{?x}p{?y}$.

For every node $n$ in $\T$ we define the branch pattern $P_{n}$
of $n$ as the union of the labels of all nodes on the path from $n$
to the root of $\T$. Then the \emph{branch query} $Q_{n}$ of $n$
has the form $\sparqlquery{W_{n}}{P_{n}}$, where $W_{n}=\var(P_{n})$.
\begin{thm}
[Completeness of $\OPT$-Queries] \label{thm:opt-queries} Let $\gcs$
be a set of completeness statements. Let $Q={\sparqlquery WP}$ be
a well-designed $\OPT$-query and $\T$ be an equivalent pattern tree
in NR-normal form. Then 
\[
\calC\models\buildQC Q\quad\ \ \mbox{iff}\quad\ \ \calC\models\buildQC{Q_{n}}\quad\textrm{for all branch queries \ensuremath{Q_{n}}\ of \ensuremath{\T}.}
\]
\end{thm}
\begin{proof}
\textquotedbl{}$\Rightarrow$\textquotedbl{}: By contradiction. Assume
the completeness of some branch query $Q_{n}$ of $\T$ is not entailed
by $\C$. Then, there must exist an incomplete graph $\G$ where $\eval{Q_{n}}{\id G}\neq\eval{Q_{n}}{\av G}$.
By construction of $Q_{n}$, every answer valuation $v$ that leads
to an answer $\mu_{n}$ to $Q_{n}$ over $\G$ is also a valuation
for $Q$, and leads either to the same mapping $\mu_{n}$ or to a
mapping $\mu$ that contains $\mu_{n}$. In both cases, if the valuation
$v$ is not satisfying for $Q_{n}$ over $\av G$, either $Q$ misses
a multiplicity of the same mapping $\mu_{n}$ over $\av G$, or $Q$
misses a multiplicity of the more general mapping $\mu$, and thus,
$Q$ is incomplete over $\G$ as well.

$"\Leftarrow"$: By contradiction. Assume that $\calC\not\models\buildQC Q$.
We have to show that there exists a branch query $Q_{n}$ of $Q$
such that $\calC\not\models\buildQC{Q_{n}}$. Since $\calC\not\models\buildQC Q$,
there must exist an incomplete data source $\G=(\Ga,\Gi)$ such that
some mapping $\mu$ is in $\eval Q{\Gi}$ but not in $\eval Q{\Ga}$.
By the semantics of $\OPT$ queries, $\mu$ must be a mapping of a
subtree of the pattern tree $\T$ for $Q$ that includes the root
of $\T$. Since $\mu$ is not satisfied over $\Ga$, there must be
at least one node $n$ in this subtree such that the triple $\mu n$
is not in $\Ga$. But then, the branch query $Q_{n}$ is not complete
over $(\Ga,\Gi)$ either, thus showing that $C$ does not entail completeness
of all branch queries of $Q$.
\end{proof}
Note that the proof above discusses only $\OPT$ queries without $\SELECT$.
For queries with $\SELECT$, the argument has to be extended to multiplicities
of mappings in the query result, but the technique remains the same. 

The theorem above allows to reduce completeness checking for an $\OPT$
query to linearly many completeness checks for basic queries.

\subsection{Completeness Entailment under RDFS Semantics}

\label{sec:compl-querying-rdfs}RDFS (RDF Schema) is a simple ontology
language that is widely used for RDF data~\cite{W3C04b}. RDFS information
can allow additional inference about data and needs to be taken into
account during completeness entailment: 
\begin{example}
[RDF vs.\ RDFS] Consider we are interested in the completeness of
the query

\[
Q_{\dir}=\sparqlquery{\set{?m}}{\set{\triple{?m}{director}{tarantino}}}
\]

asking for all objects that were directed by Tarantino, and consider
the completeness statement

\[
\gc_{\tn}=\defineGC{\triple{?m}{director}{tarantino},\triple{?m}{type}{Movie}}{\emptyset}
\]
that tells that all Tarantino movies are complete. A priori, we cannot
conclude that $C_{\tn}$ entails the completeness of $Q_{\dir}$,
because other pieces that Tarantino directed could be missing. If
however we consider the RDFS statement $\triple{director}{domain}{Movie}$
that tells that all pieces that have a director are movies, then completeness
of all Tarantino movies implies completeness of all pieces that Tarantino
directed, because there can be no other pieces than movies.

Or, consider the query $Q_{\film}=\sparqlquery{\set{?m}}{\set{\triple{?m}{type}{film}}}$,
asking for all films, and the completeness statement $\gc_{\movie}=\defineGC{\triple{?m}{type}{movie}}{\emptyset}$
saying that we are complete for all movies. A priori, we cannot conclude
that $\gc_{\movie}$ entails the completeness of $Q_{\film}$, because
we do not know about the relationship between films and movies. When
considering the RDFS statements $\triple{film}{subclass}{movie}$
and $\triple{movie}{subclass}{film}$ saying that all movies and films
are equivalent, we can conclude that $\{\gc_{\movie}\}\models\buildQC{Q_{\film}}$.
\label{ex:rdfs} 
\end{example}
The intuitive reasoning from above has to be taken into account when
reasoning about query completeness.

In the following, we rely on $\rho$DF, which is a formalization of
the core of RDFS~\cite{MunozPG09}. The vocabulary of $\rho$DF contains
the terms
\[
\mathit{subproperty,subclass,domain,range},\mathit{type}
\]
A \emph{schema graph} $S$ is a set of triples built using any of
the $\rho$DF terms, except $\mathit{type}$, as predicates.

We assume that schema information is not lost in incomplete data sources.
Hence, for incomplete data sources it is possible to extract their
$\rho$DF schema into a separate schema graph. The \emph{closure of
a graph $G$} wrt.\ a schema graph $S$ is the set of all triples
that are entailed. We denote this closure by $\cl_{S}(G)$. The computation
of this closure can be reduced to the computation of the closure of
a single graph that contains both schema and non-schema triples as
$\cl_{S}(G)=\cl(S\cup G)$. 
We now say that a set $\C$ of completeness statements \emph{entails}
the completeness of a query $Q$ wrt.\ a $\rho$DF schema graph $S$,
if for all incomplete data sources $(\Ga,\Gi)$ it holds that if $(\cl_{S}(\Ga),\cl_{S}(\Gi))$
satisfies $\C$ then it also satisfies $\buildQC Q$.


\begin{example}
[Completeness Reasoning under RDFS] Consider again the query $Q_{\film}$,
the schema graph $S$=$\{\triple{film}{subclass}{movie},$\linebreak{}
$\triple{movie}{subclass}{film}\}$ and the completeness statement
$\gc_{\movie}$ in Example~\ref{ex:rdfs}. Assume that the query
$Q_{\film}$ returns a mapping $\set{?m\mapsto m'}$ for some term
$m'$ over the ideal graph \Gi\ of an incomplete data source $\pg=(\Ga,\Gi)$
that satisfies $\gc_{\movie}$. Then, the triple $\triple{m'}{type}{film}$
must be in $\Gi$. Because of the schema, the triple $\triple{m'}{type}{movie}$
is then entailed (and thus in the closure $\cl_{S}(\Gi)$). As before,
we can now use the completeness statement $\gc_{\movie}$ to infer
that the triple $\triple{m'}{type}{movie}$ must also be in $\Ga$.
Again, the triple $\triple{m'}{type}{film}$ is then entailed from
the triple $\triple{m'}{type}{movie}$ that is in \Ga\ because of
the schema. Thus, $Q_{\film}$ then also returns the mapping $\set{?m\mapsto m'}$
over $\Ga$. Because of the prototypical nature of $m'$ and $(\Ga,\Gi)$,
the completeness statement entails query completeness in general.
\end{example}
Therefore, the main difference to the previous entailment procedures
is that the closure is computed to obtain entailed triples before
and after the completeness operator $T_{\C}$ is applied. For a set
of completeness statements $\C$ and a schema graph $S$, let $T_{\C}^{S}$
denote the function composition $\cl_{S}\circ T_{\C}\circ\cl_{S}$.
Then the following holds.
\begin{thm}
[Completeness under RDFS]Let $\gcs$ be a set of completeness statements,
$Q=\sparqlquery WP$ a basic query, and $S$ a schema graph. Then
\[
\gcs\models_{S}\buildQC Q\qquad\mbox{if and only if}\qquad\frozen P\subseteq T_{\C}^{S}(\frozen P).
\]
\end{thm}
\begin{proof}
\footnote{A similar proof can be found in \cite{darari:master:thesis:2013}%
} $\proofr$   If $\frozen P\not\subseteq T_{\C}^{S}(\frozen P)$,
then  the incomplete data source $(T_{\C}^{S}(\frozen P),\cl_{S}(\frozen P))$
is a counterexample for the entailment.  It satisfies $\C$ wrt.\
the schema $S$, but does not satisfy $\Compl Q$ because the identity
mapping $\frozenID$, which can be retrieved over the closure of the
ideal graph $\cl_{S}(\frozen P)$ cannot be retrieved  by $P$ over
the available graph $T_{\C}^{S}(\frozen P)$.

$\proofl$ Assume $\frozen P\subseteq T_{\C}^{S}(\frozen P)$. We
show that for an incomplete data source $\G=(\cl_{S}(\Ga),\cl_{S}(\Gi))$
such that $\G\models\C$, it holds that $\G\models\Compl Q$.  By
definition, $\G\models\Compl Q$ if  $\eval Q{\cl_{S}(\Ga)}=\eval Q{\cl_{S}(\Gi)}$.
By the semantics of $\SELECT$ queries, it is sufficient to prove
that  $\eval P{\cl_{S}(\Ga)}=\eval P{\cl_{S}(\Gi)}$.  Note that  $\eval P{\cl_{S}(\Ga)}\subseteq\eval P{\cl_{S}(\Gi)}$
immediately follows from the monotonicity of $P$ and the fact that
$\cl_{S}(\Ga)\subseteq\cl_{S}(\Gi)$.

As for $\eval P{\cl_{S}(\Ga)}\supseteq\eval P{\cl_{S}(\Gi)}$, suppose
that there is some mapping $\mu$ in $\eval P{\cl_{S}(\Gi)}$. Then,
$\mu P\subseteq\cl_{S}(\Gi)$ and because of the monotonicity of $T_{\C}^{S}$,
it holds that $T_{\C}^{S}(\,\mu P)\subseteq T_{\C}^{S}(\Gi)$.  By
the defintion of satisfaction of completeness statements wrt.\ RDFS,
$(\cl_{S}(\Ga),\cl_{S}(\Gi))\models\C$ implies that $\TC(\cl_{S}(\Gi))\subseteq\cl_{S}(\Ga)$.
By applying the closure $\cl_{S}$ once again on both sides, we find
that $\TC^{S}(\Gi)\subseteq\cl_{S}(\Ga)$. Composing this two inclusions,
we find that the following inclusion holds: $T_{\C}^{S}(\,\mu P)\subseteq T_{\C}^{S}(\Gi)\subseteq\cl_{S}(\Ga)$.

Because we have assumed that $\frozen P\subseteq T_{\C}^{S}(\frozen P)$,
it follows that $\mu\,\frozenID^{-1}\frozen P\subseteq\TC^{S}(\mu\,\frozenID^{-1}\frozen P)$.
 Since $\mu\,\frozenID^{-1}\frozen P=\mu P$ and  $\TC^{S}(\mu\,\frozenID^{-1}\frozen P)=\TC^{S}(\mu P)$,
 this means that $\mu P\subseteq\cl_{S}(\Ga)$.  Consequently, $\mu$
is also in $\eval P{\cl_{S}(\Ga)}$. Thus,  $\eval P{\cl_{S}(\Ga)}$
is in $\eval P{\cl_{S}(\Gi)}$ and thus $\eval P{\cl_{S}(\Ga)}=\eval P{\cl_{S}(\Gi)}$.
\end{proof}
As the computation of the closure can be done in polynomial time,
reasoning wrt.\ RDFS has the same complexity as reasoning for basic
queries.

\label{theo:completeness-rdfs} 

\global\long\def\fl#1{#1^{\mathit{\, fl}}}
 \global\long\def\fedmodels{\models_{\mathit{fed}}}

\section{Completeness over Federated Data Sources}

\label{sec:querying-federated} Data on the Web is intrinsically distributed.
Hence, the single-source query mechanism provided by SPARQL has been
extended to deal with multiple data sources. In particular, the recent
SPARQL 1.1 specification introduces the notion of query \textit{federation}~\cite{W3C13}.
A federated query is a SPARQL query that is evaluated across several
data sources, the SPARQL endpoints of which can be specified in the
query. 

So far, we have studied the problem of querying a \textit{single}
data source augmented with completeness statements. The federated
scenario calls for an extension of the completeness framework discussed
in Section~\ref{sec:compl-querying-single}. Indeed, the completeness
statements available about each data source involved in the evaluation
of a federated query must be considered to check the completeness
of the federated query. The aim of this section is to discuss this
aspect and present an approach to check whether the completeness of
a non-federated query (i.e., a query without \SERVICE\ operators)
can be ensured with respect to the completeness statements on each
data source. We also study the problem of rewriting a non-federated
query into a federated version in the case in which the query is complete.



\medskip{}
\textbf{Federated SPARQL Queries.} Before discussing existing results
on reasoning in the federated case, we formalize the notion of federated
SPARQL queries. A federated query is a SPARQL query executed over
a \emph{federated graph}. Formally speaking, a federated graph is
a family of RDF graphs $\gs=(\g_{j})_{j\in J}$ where $J$ is a set
of IRIs. A federated SPARQL query (as for the case of a non-federated
query) can be a \SELECT or an \ASK query~\cite{ArenasGP10}. In
what follows, we focus on the conjunctive fragment (i.e., the \AND
fragment) of SPARQL with the inclusion of the \SERVICE operator.
Non-federated SPARQL queries are evaluated over graphs. In the federated
scenario, queries are evaluated over a pair $(i,\gs)$, where the
first component is an IRI associated to the initial SPARQL endpoint,
and the second component is a federated graph. The semantics of graph
patterns with \AND and \SERVICE operators is defined as follows:
\begin{align*}
\eval t{(i,\gs)} & =\eval t{G_{i}}\\
\eval{P_{1}\osps{\AND}P_{2}}{(i,\gs)} & =\eval{P_{1}}{(i,\gs)}\Join\eval{P_{2}}{(i,\gs)}\\
\eval{\querySERVICE jP}{(i,\gs)} & =\eval P{(j,\gs)}
\end{align*}
 where $t$ ranges over all triple patterns and $P$, $P_{1}$, $P_{2}$
range over all graph patterns with \AND and \SERVICE operators.
We denote federated queries as $\bar{Q}$.

\textbf{Federated Completeness Reasoning.} Darari et al.\ have shown
\cite{darari:ISWC:2013} how to extend completeness reasoning to the
federated setting. They extended completeness statements with data
source indices, and defined query completeness of a non-federated
query as completeness wrt. the union of the ideal graphs of all data
sources. The main result of this work is that if a non-federated query
is complete over a set of datasources, then there exist a federated
version of the query such that each triple is evaluated over only
exactly one data source, and the query still returns the complete
result. 
\begin{example}
[Federated Data Sources]Consider the two data sources shown in Figure
\ref{fig:example-query} plus an additional data source named FB (=
Facebook) with the completeness statement
\[
C_{\mathit{fb}}={\buildGC{\triple{?m}{likes}{?l}}{\triple{?m}{type}{Movie},\triple{?m}{director}{tarantino}}}
\]
 and the query
\[
Q_{\mathit{fb}}=\setquery{?m,?l}{\triple{?m}{type}{Movie},\triple{?m}{director}{tarantino},\triple{?m}{likes}{?l}}
\]
that asks for the number of \textit{likes} of Tarantino's movies.

This query is complete over the three data sources, whose endpoints
are reachable at the IRIs \texttt{\small{}DBPe}{\small{}, }\texttt{\small{}LMDBe}{\small{}
and }\texttt{\small{}FBe,} because Facebook is complete for likes
of Tarantino movies and IMDB is complete for all Tarantino movies
and all directing of Tarantino. Since the query is complete, we can
compute a federated version $Q_{\mathit{fb}}$, which in this case
is $(\{?m,?l\},\{\querySERVICEBraces{\texttt{LMDBe}}{\triple{?m}{type}{Movie},\triple{?m}{director}{tarantino}}$\texttt{
AND} ${\querySERVICEBraces{\texttt{FBe}}{\triple{?m}{likes}{?l}}}$,
which returns a complete answer already.
\end{example}
Note that the results for the federated case as presented in \cite{darari:ISWC:2013}
only work as long as there are no comparisons in the completeness
statements. If the completeness statements may contain comparisons,
then it can be the case that only a combination of data sources together
ensures completeness, e.g.\ if one data source is complete for movies
before 1980 and the other for movies in or after 1980.

\section{Discussion}

\label{sec:discussion} We now discuss some aspects underlying the
completeness framework.

\paragraph{Availability of Completeness Metadata}

At the core of the proposed framework lies the availability of completeness
statements. We have discussed in Section~\ref{sec:scenario} how
existing data sources like IMDB already incorporate such statements
(Figure~\ref{fig:imdb}) and how they can be made machine-readable
with our framework. The availability of completeness statements rests
on the assumption that a domain ``expert'' has the necessary background
knowledge to provide such statements.

\noindent We believe that is in the interest of data providers to
annotate their data sources with completeness statements in order
to increase their value. Indeed, users can be more inclined to prefer
data sources including ``completeness marks'' to other data sources.
Moreover, in the era of crowdsourcing the availability of independent
``ratings'' from users regarding the completeness of data can also
contribute, in a bottom up manner, to the description of the completeness
of data sources. For instance, when looking up information about Stanley
Kubrick in DBpedia, as a by-product users can provide feedback as
to whether all of Kubrick's movies are present.

\paragraph{Complexity}

All completeness checks presented in this chapter are NP-complete.
The hardness holds because of the classical complexity of conjunctive
query containment; the NP upper bound follows because all completeness
checks require conjunctive query evaluation at their core. In practice,
we expect these checks to be fast, since queries and completeness
statements are likely to be small. After all, this is the same complexity
as the one of query evaluation and query optimization of basic queries,
as implemented in practical database management systems. All theorems
in this paper characterize completeness entailment using the transformation
$T_{\C}$ that is based on \CONSTRUCT queries. Thus, the completeness
checks can be straightforwardly implemented and can make use of existing
query evaluation techniques.

\paragraph{Vocabulary Heterogeneity}

In practice, a query may use a vocabulary different from that of some
data sources. In this work, we assume the presence of a global schema.
Indeed, one could use the \texttt{schema.org} vocabulary for queries,
since it has already been mapped to other vocabularies (e.g., DBpedia).

\paragraph{Implementation}

To show the feasibility of this proposal, Darari developed the CORNER
system, which implements the completeness entailment procedure for
basic and \DISTINCT\ queries with $\rho$DF%
\footnote{\noindent http://rdfcorner.wordpress.com%
}.

\section{Related Work}

\label{sec:rdf:relatedwork}

Fürber and Hepp~\cite{FurberH10} investigated data quality problems
for RDF data originating from relational databases. Wang et al.~\cite{wang2005}
focused on data cleansing while Stoilos et al.~\cite{StoilosGH10}
on incompleteness of reasoning tasks. The problem of assessing completeness
of linked data sources is discussed by Harth and Speiser~\cite{HarthS12};
here, completeness is defined in terms of \textit{authoritativeness}
of data sources, which is a purely syntactic property. Hartig et al.~\cite{HartigBF09}
discuss an approach to get more complete results of SPARQL queries
over the Web of Linked Data. Their approach is based on traversing
RDF links to discover relevant data during query execution. Still,
the completeness of query answers cannot be guaranteed. 

\noindent Indeed, the semantics of completeness is crucial also for
RDF data sources distributed on the Web, where each data source is
generally considered incomplete. To the best of our knowledge, the
problem of formalizing the semantics of RDF data sources in terms
of their completeness is open. Also from the more pragmatic point
of view, there exist no comprehensive solutions enabling the characterization
of data source in terms of completeness. As an example, with \void\ it
is not possible to express the fact that, for instance, the data source
IMDB is \textit{complete for all movies directed by Tarantino}. Having
the possibility to provide in a declaratively and machine-readable
way (in RDF) such kind of completeness statements paves the way toward
a new generation of services for retrieving and consuming data. In
this latter respect, the semantics of completeness statements interpreted
by a reasoning engine can guarantee the completeness of query answering.
We present a comprehensive application scenario in Section~\ref{sec:scenario}. 

The RDF data itself is based on the open-world assumption, implying
that in general the information is incomplete \cite{W3C04}. One approach
to deal with this incompleteness was proposed by Nikolaou and Koubarakis
\cite{nikolaou:2012}, in which they developed $RDF^{i}$, an extension
to RDF that can represent incomplete values by means of \emph{e-literals}.
The \emph{e-literals} behave like existentially quantified variables
in first-order logic, and are constrained by a global constraint.
Global constraints can in general be quantifier-free formulae of some
first-order constraint language. Both constitute syntactic devices
for the representation of incomplete information, called $RDF^{i}$
databases. They have also extended the standard SPARQL in order to
allow expressions of a first-order constraint language as FILTER expressions
and be able to pose queries that ask for certain answers over $RDF^{i}$
databases. However, incompleteness with respect to missing records
in RDF data was not covered by their approach.

In \cite{HarthS12}, Harth and Speiser observe that the notion of
completeness of sources can be defined based on authority. They study
three completeness classes and their interrelationships, for triple
patterns and conjunctive queries: one that considers the whole web,
one that regards documents in the surrounding of sources derived from
the query and one that considers documents according to the query
execution. Their work is orthogonal to ours in the sense that we define
completeness of sources based on their semantic structure. Also, their
work does not concern the OPT fragment of SPARQL queries and the RDFS
schema that may underlie RDF data. The partial completeness of RDF
data is not considered in their work either.

Recently, Patel-Schneider and Franconi presented an approach for integrity
constraints in an ontology setting \cite{patelschneier-franconi}.
 It is to completely specify certain concepts and roles, making them
analogous to database tables. On these concepts and roles, which are
called DBox, axioms act like integrity constraints. Moreover, the
answers returned by queries for DBox are complete. However, their
work was focused on the integrity constraint part, not on the query
answering part. Additionally, they did not cover the partial completeness
of concepts and roles, i.e., to specify that only certain parts of
concepts and roles are complete.

\section{Summary}

\label{sec:conclusion} RDF and SPARQL are recent technologies enabling
and alleviating the publication and exchange of structured information
on the semantic web. The availability of distributed and potentially
overlapping RDF data sources calls for mechanisms to provide qualitative
characterizations of their content. In this chapter, we have transferred
previous results for relational databases to the semantic web. We
have shown that although completeness information is present on the
web in some available data sources (e.g., IMDB discussed in Section~\ref{sec:scenario})
it is neither formally represented nor automatically processed. We
have adapted the relational framework for the declarative specification
of completeness statements to RDF data sources and underlined how
the framework can complement existing initiatives like \void. As
particularities, we studied the reasoning wrt. RDF schema and for
SPARQL queries containing the $\OPT$ keyword.

\chapter{Verifying Completeness over Processes}

\label{chap:bpm}

In many applications, data is managed via well documented processes.
If information about such processes exists, one can draw conclusions
about completeness as well. In this chapter, we present a formalization
of so-called \emph{quality-aware processes} that create data in the
real world and store it in the company's information system possibly
at a later point. We then show how one can check the completeness
of database queries in a certain state of the process or after the
execution of a sequence of actions, by leveraging on query containment,
a well-studied problem in database theory. Finally, we show how the
results can be extended to the more expressive formalism of colored
Petri nets. Besides Section \ref{sec:processes:petri:nets}, all results
in this chapter are contained in a conference paper by Razniewski
et al., published at the BPM 2013 conference \cite{razniewski:BPM2013},
or in the extended version available at Arxiv.org \cite{razniewski:BPM2013:Arxiv:extended}.

This chapter is divided as follows. In Section \ref{sec:processes:background},
we discuss necessary background information about processes. In Section
\ref{sec:processes-related-work}, we discuss related work on data
quality verification over processes. In Section \ref{sec:processes:samplescenario},
we discuss the scenario of the school enrollment data in the province
of Bozen/Bolzano in detail. In Section \ref{sec:processes:formalization},
we discuss our formal approach, introducing quality-aware transition
systems, process activity annotations used to capture the semantics
of activities that interact with the real world and with an information
system, and properties of query completeness over such systems. In
Section \ref{sec:processes:verification}, we discuss how query completeness
can be verified over such systems at design time, at runtime, how
query completeness can be refined and what the complexity of deciding
query completeness is. We conclude with a discussion of extensions
to (colored) Petri Nets in Section \ref{sec:processes:petri:nets}.

\section{Motivation and Background}

\label{sec:processes:background}

In the previous chapters we have discussed how reasoning over completeness
statements can be performed. In this chapter, we discuss how the same
statements, query completeness, can be verified over business process
descriptions.

In many businesses, data creation and access follow formalized procedures.
Strategic decisions are taken inside a company by relying on statistics
and business indicators such as KPIs. Obviously, this information
is useful only if it is reliable, and reliability, in turn, is strictly
related to quality and, more specifically, to completeness.

Consider for example the school information system of the autonomous
province of Bozen/Bolzano in Italy. Such an information system stores
data about schools, enrollments, students and teachers. When statistics
are computed for the enrollments in a given school, e.g., to decide
the amount of teachers needed for the following academic year, it
is of utmost importance that the involved data are complete, i.e.,
that the required information stored in the information system is
aligned with reality.

Completeness of data is a key issue also in the context of auditing.
When a company is evaluated to check whether its way of conducting
business is in accordance to the law and to audit assurance standards,
part of the external audit is dedicated to the analysis of the actual
data. If such data are incomplete w.r.t.~the queries issued during
the audit, then the obtained answers do not properly reflect the company's
behaviour.

A common source of data incompleteness in business processes is constituted
by delays between real-world events and their recording in an information
system. This holds in particular for scenarios where processes are
carried out partially without support of the information system. E.g.,
many legal events are considered valid as soon as they are signed
on a sheet of paper, but their recording in the information system
could happen much later in time. Consider again the example of the
school information system, in particular the enrollment of pupils
in schools. Parents enroll their children at the individual schools,
and the enrollment is valid as soon as both the parents and the school
director sign the enrollment form. However, the school secretary may
record the information from the sheets only later in the local database
of the school, and even later submit all the enrollment information
to the central school administration, which needs it to plan the assignment
of teachers to schools, and other management tasks.

\paragraph{Relation to previous chapters}

The model of completeness (correspondence between query results over
an ideal and an available database) is the same as in the previous
chapters. While the reasoning problem of completeness in a state of
a process is different, the techniques to solve these problems (query
containment) are the same as before. While Chapters \ref{chap:general-reasoning}
and \ref{chap:nulls} left the question of where completeness information
could come from largely open, this chapter gives an answer for scenarios,
where data creation and manipulation follows formalized processes.

\section{Example Scenario}

\label{sec:processes:samplescenario}

Consider the example of the enrollment to schools in the province
of Bolzano. Parents can submit enrollment requests for their child
to any school they want until the 1st of March. Schools then decide
which pupils to accept, and parents have to choose one of the schools
in which their child is accepted. Since in May the school administration
wants to start planning the allocation of teachers to schools and
take further decisions (such as the opening and closing of school
branches and schools) they require the schools to process the enrollments
and to enter them in the central school information system before
the 15th of April.

A particular feature of this process is that it is partly carried
out with pen and paper, and partly in front of a computer, interacting
with an underlying school information system. Consequently, the information
system does often not contain all the information that hold in the
real world, and is therefore incomplete. E.g., while an enrollment
is legally already valid when the enrollment sheet is signed, this
information is visible in the information system only when the secretary
enters it into a computerized form.

A BPMN diagram sketching the main phases of this process is shown
in Figure \ref{figure:enrolment-BPMN-advanced}, while a simple UML
diagram of (a fragment of) the school domain is reported in Figure
\ref{figure:ER-diagram-school-world}. These diagrams abstractly summarize
the school domain from the point of view of the central administration.
Concretely, each school implements a specific, local version of the
enrollment process, relying on its own domain conceptual model. The
data collected on a per-school basis are then transferred into a central
information system managed by the central administration, which refines
the conceptual model of Figure \ref{figure:ER-diagram-school-world}.
In the following, we will assume that such an information system represents
information about children and the class they belong to by means of
a $pupil(pname,class,sname)$ relation, where $pname$ is the name
of an enrolled child, $class$ is the class to which the pupil belongs,
and $sname$ is the name of the corresponding school. 
\begin{figure}[t]
\centering \includegraphics{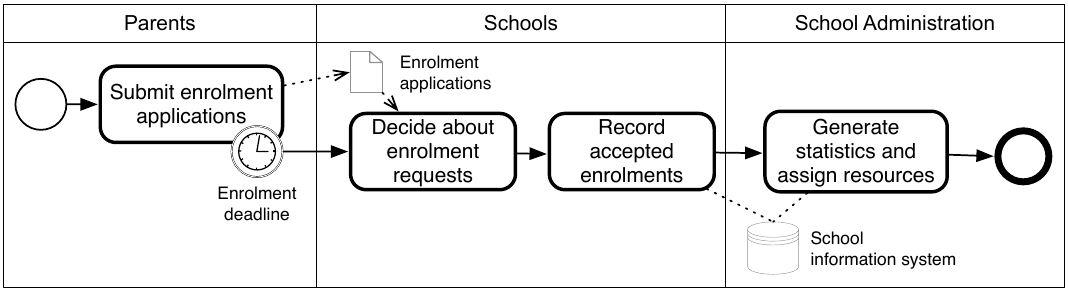}\caption{BPMN diagram of the main phases of the school enrollment process}

\label{figure:enrolment-BPMN-advanced} 
\end{figure}

\begin{figure}[t]
\centering \includegraphics{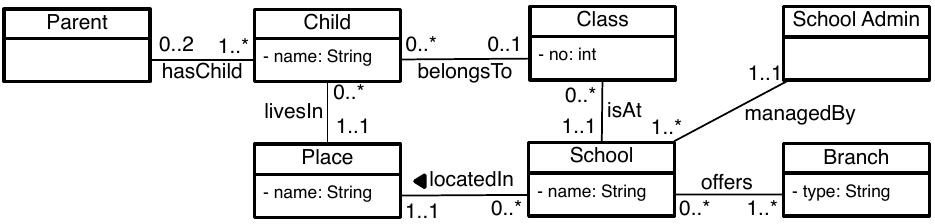}
\caption{UML diagram capturing a fragment of the school domain}

\label{figure:ER-diagram-school-world} 
\end{figure}

When using the statistics about the enrollments as compiled in the
beginning of May, the school administration is highly interested in
having correct statistical information, which in turn requires that
the underlying data about the enrollments must be complete. Since
the data is generated during the enrollment process, this gives rise
to several questions about such a process. The first question is whether
the process is generally designed correctly, that is, whether the
enrollments present in the information system are really complete
at the time they publish their statistics, or whether it is still
possible to submit valid enrollments by the time the statistics are
published. We call this problem the \emph{design-time verification}.

A second question is to find out whether the number of enrollments
in a certain school branch is already complete before the 15th of
April, that is, when the schools are still allowed to submit enrollments
(i.e., when there are school that still have not completed the second
activity in the school lane of Figure \ref{figure:enrolment-BPMN-advanced}),
which could be the case when some schools submitted all their enrollments
but others did not. In specific cases the number can be complete already,
when the schools that submitted their data are all the schools that
offer the branch. We call this problem the \emph{run-time verification}.

A third question is to learn on a finer-grained level about the completeness
of statistics, when they are not generally complete. When a statistic
consists not of a single number but of a set of values (e.g., enrollments
per school), it is interesting to know for which schools the number
is already complete and for which not. We call this the \emph{dimension
analysis}.

\section{Formalization}

\label{sec:processes:formalization}

We want to formalize processes such as the one in Figure \ref{figure:enrolment-BPMN-advanced},
which operate both over data in the real-world (pen\&paper) and record
information about the real world in an information system. We therefore
first introduce ideal databases and available databases to model the
state of the real world and the information system, and show then
how transition systems, which represent possible process executions,
can be annotated with effects for interacting with the ideal or the
available database.

\subsection{Ideal and Available Databases}

As in Chapter \ref{chap:general-reasoning}, we assume an ordered,
dense set of constants $\dom$ and a fixed set $\Sigma$ of relations.
A database instance is a finite set of facts in $\Sigma$ over $\dom$.
As there exists both the real world and the information system, in
the following we model this with two databases: $\id D$ called the
ideal database, which describes the information that holds in the
real world, and $\av D$, called the available database, which captures
the information that is stored in the information system. We assume
that the stored (available) information is always a subset of the
real-world (ideal) information. Thus, processes actually operate over
pairs $(\id D,\av D)$ of an ideal database and available database.
In the following, we will focus on processes that create data in the
real world and copy parts of the data into the information system,
possibly delayed. 

\label{example:incomplete-db} Consider that in the real world, there
are the two pupils John and Mary enrolled in the classes 2 and 4 at
the Hofer School, while the school has so far only processed the enrollment
of John in their IT system. Additionally it holds in the real world
that John and Alice live in Bolzano and Bob lives in the city of Merano.
The ideal database $\id D$ would then be 
\begin{eqnarray*}
 & \{ & \pupil(\John,2,\Hoferschool),\pupil(\Mary,4,\Hoferschool)\\
 &  & \livesin(\John,\Bolzano),\livesin(\Bob,\Merano),\livesin(Alice,\Bolzano)\}
\end{eqnarray*}
 while the available database would be 
\[
\{\pupil(\John,2,\Hoferschool)\}.
\]
Where it is not clear from the context, we annotate atoms with the
database they belong to, so, e.g., $\av{\pupil}(\John,4,\Hoferschool)$
means that this fact is stored in the available database.

\subsection{Query Completeness}

For planning purposes, the school administration is interested in
figures such as the number of pupils per class, school, profile, etc.
Such figures can be extracted from relational databases via SQL queries
using the COUNT keyword. In an SQL database that contains a table
$\mathit{pupil(name,class,school)}$, a query asking for the number
of students per school would be written as: 
\begin{align}
\begin{split} & \mbox{SELECT school, COUNT(*) as pupils\_nr}\\[-0.8ex]
 & \mbox{FROM pupil}\\[-0.8ex]
 & \mbox{GROUP BY school.}\\[-4ex]
\\
\end{split}
\label{eq:SQL-query-bpm}
\end{align}
 As discussed earlier, conjunctive queries formalize SQL queries.
A \emph{conjunctive query} $Q$ is an expression of the form $\query{Q(\tpl x)}{A_{1},\ldots,A_{n},M}$,
where $\tpl x$ are called the distinguished variables in the head
of the query, $A_{1}$ to $A_{n}$ the atoms in the body of the query,
and $M$ is a set of built-in comparisons \cite{foundations_of_dbs}.
We denote the set of all variables that appear in a query $Q$ by
$\var(Q)$. Common subclasses of conjunctive queries are linear conjunctive
queries, that is, they do not contain a relational symbol twice, and
relational conjunctive queries, that is, queries that do not use comparison
predicates. Conjunctive queries allow to formalize all single-block
SQL queries, i.e., queries of the form ``SELECT $\ldots$ FROM $\ldots$
WHERE $\ldots$''. As a conjunctive query, the SQL query (\ref{eq:SQL-query-bpm})
above would be written as: 
\begin{equation}
\query{Q_{p/s}(\mathit{schoolname},\mathit{\Count(name}))}{\pupil(\mathit{name,class,schoolname})}
\end{equation}

The formalization of query completeness over a pair of an ideal database
and an available database is as before: Intuitively, if query completeness
can be guaranteed, then this means that the query over the generally
incomplete available database gives the same answer as it would give
w.r.t.~the information that holds in the ideal database. Query completeness
is the key property that we are interested in verifying.

A pair of databases $(\id D,\av D)$ satisfies\emph{ query completeness}
of a query $Q$, if $Q(\id D)=Q(\av D)$ holds. We then write $(\id D,\av D)\models\compl Q$. 
\begin{example}
Consider the pair of databases $(\id D,\av D)$ from Example \ref{example:incomplete-db}
and the query $Q_{p/s}$ from above (2). Then, $\compl{Q_{p/s}}$
does not hold over $(\id D,\av D)$ because $Q(\id D)=\{(\Hoferschool,2)\}$
but $Q(\av D)=\{(\Hoferschool,1)\}$. A query for pupils in class
2, $\query{Q_{class2}(n)}{\pupil(n,2,s)}$, would be complete, because
$Q(\id D)=Q(\av D)=\{\John\}$. 

\end{example}

\subsection{Real-world Effects and Copy Effects}

\label{sec:effects} We want to formalize the real-world effect of
an enrollment action at the Hofer School, where in principle, every
pupil that has submitted an enrollment request before, is allowed
to enroll in the real world. We can formalize this using the following
implication: 

\[
\id{\pupil}(n,c,\Hoferschool)\leftsquigarrow\id{\req}(n,\Hoferschool)
\]
which should mean that whenever someone is a pupil at the Hofer school
now, he has submitted an enrollment request before. Also, we want
to formalize copy effects, for example where all pupils in classes
greater than 3 are stored in the database. This can be written with
the following implication:

\[
\id{\pupil}(n,c,s),c>3\rightarrow\av{\pupil}(n,c,s)
\]
which means that whenever someone is a pupil in a class with level
greater than three in the real world, then this fact is also stored
in the available database.

For annotating processes with information about data creation and
manipulation in the ideal database $\id D$ and in the available database
$\av D$, we use real-world effects and copy effects as annotations.
While their syntax is the same, their semantics is different. Formally,
a \emph{real-world effect} $r$ or a \emph{copy effect} $c$ is a
tuple $(R(\tpl x,\tpl y),G(\tpl x,\tpl z))$, where $R(\tpl x,\tpl y)$
is an atom, $G$ is a set of atoms and built-in comparisons and $\tpl x$,
$\tpl y$ and $\tpl z$ are sets of distinct variables. We call $G$
the \emph{guard} of the effect. The effects $r$ and $c$ can be written
as follows: 
\begin{eqnarray*}
 &  & r:\ \id R(\tpl x,\tpl y)\leftsquigarrow\exists\tpl z\!:\ \id G(\tpl x,\tpl z)\\
 &  & c:\ \id R(\tpl x,\tpl y),\id G(\tpl x,\tpl z)\rightarrow\av R(\tpl x,\tpl y)
\end{eqnarray*}

Real-world effects can have variables $\tpl y$ on the left side that
do not occur in the condition. These variables are not restricted
and thus allow to introduce new values.

A pair of real-world databases $(\id{D_{1}},\id{D_{2}})$ \emph{conforms}
to a real-world effect \linebreak{}
$\id R(\tpl x,\tpl y)\leftsquigarrow\exists\tpl z:\ \id G(\tpl x,\tpl z)$,
if for all facts $\id R(\tpl c_{1},\tpl c_{2})$ that are in $\id{D_{2}}$
but not in $\id{D_{1}}$ it holds that there exists a tuple of constants
$\tpl c_{3}$ such that the guard $\id G(\tpl c_{1},\tpl c_{3})$
is in $\id{D_{1}}$. The pair of databases conforms to a set of real-world
effects, if each fact in $\id{D_{2}}\setminus\id{D_{1}}$ conforms
to at least one real-word effect.

If for a real-world effect there does not exist any pair of databases
$(D_{1},D_{2})$ with $D_{2}\setminus D_{1}\neq\emptyset$ that conforms
to the effect, the effect is called \emph{useless}. In the following
we only consider real-world effects that are not useless.

The function $\copyy_{c}$ for a copy effect $c=\id R(\tpl x,\tpl y),\id G(\tpl x,\tpl z)\rightarrow\av R(\tpl x,\tpl y)$
over an ideal database $\id D$ returns the corresponding R-facts
for all the tuples that are in the answer of the query $\query{P_{c}(\tpl x,\tpl y)}{\id R(\tpl x,\tpl y),\id G(\tpl x,\tpl z)}$
over $\id D$. For a set of copy effects $\CE$, the function $\mathrm{copy_{\CE}}$
is defined by taking the union of the results of the individual copy
functions. 
\begin{example}
Consider a real-world effect $r$ that allows to introduce persons
living in Merano as pupils in classes higher than 3 in the real world,
that is, $r=\id{\pupil}(n,c,s)\leftsquigarrow c>3,\livesin(n,\Merano)$
and a pair of ideal databases using the database $\id D$ from Example
\ref{example:incomplete-db} that is defined as 
 $(\id D,\id D\cup\{\id{\pupil}(\Bob,4,\Hoferschool)\}$. Then this
pair conforms to the real-world effect $r$, because the guard of
the only new fact $\id{\pupil}(\Bob,4,\Hoferschool)$ evaluates to
true: Bob lives in Merano and his class level is greater than 3. The
pair $(\id D,\id D\cup\{\id{\pupil}(\Alice,1,\Hoferschool)\}$ does
not conform to $r$, because Alice does not live in Merano, and also
because the class level is not greater than 3.

For the copy effect $c=\id{\pupil}(n,c,s),c>3\rightarrow\av{\pupil}(n,c,s)$,
which copies all pupils in classes greater equal 3, its output over
the ideal database in Example \ref{example:incomplete-db} would be
$\{\av{\pupil}(\Mary,4,\Hoferschool)\}$.
\end{example}

\subsection{Quality-Aware Transition Systems}

To capture the execution semantics of \emph{quality-aware processes},
we resort to (suitably annotated) labeled transition systems, a common
way to describe the semantics of concurrent processes by interleaving
\cite{BaKG08}. This makes our approach applicable for virtually every
business process modeling language equipped with a formal underlying
transition semantics (such as Petri nets or, directly, transition
systems).

Formally, a \emph{(labeled) transition system} $T$ is a tuple $T=(S,s_{0},A,E)$,
where $S$ is a set of states, $s_{0}\in S$ is the initial state,
$A$ is a set of names of actions and $E\subseteq S\times A\times S$
is a set of edges labeled by actions from $A$. In the following,
we will annotate the actions of the transition systems with effects
that describe interaction with the real-world and the information
system. In particular, we introduce \emph{quality-aware transition
systems} (QATS) to capture the execution semantics of processes that
change data both in the ideal database and in the available database.

Formally, a \emph{quality-aware transition system} $\qats$ is a tuple
$\qats=(T,\re,\ce)$, where $T$ is a transition system and $\re$
and $\ce$ are functions from $A$ into the sets of all real-world
effects and copy effects, which in turn obey to the syntax and semantics
defined in Sec.~\ref{sec:effects}. Note that transition systems
and hence also QATS may contain cycles.
\begin{example}
Let us consider two specific schools, the Hofer School and the Da
Vinci School, and a (simplified version) of their enrollment process,
depicted in BPMN in Figure \ref{fig:school-processes} (left) (in
parenthesis, we introduce compact names for the activities, which
will be used throughout the example). As we will see, while the two
processes are independent from each other from the control-flow point
of view (i.e., they run in parallel), they eventually write information
into the same table of the central information system.

Let us first consider the Hofer School. In the first step, the requests
are processed with pen and paper, deciding which requests are accepted
and, for those, adding the signature of the school director and finalizing
other bureaucratic issues. By using relation $\id{\req}(n,\Hoferschool)$
to model the fact that a child named $n$ requests to be enrolled
at Hofer, and $\id{\pupil}(n,1,\Hoferschool)$ to model that she is
actually enrolled, the activity \textsf{pH} is a real-world activity
that can be annotated with the real-world effect $\id{\pupil}(n,1,\Hoferschool)\leftsquigarrow\id{\req}(n,\Hoferschool)$.
In the second step, the information about enrolled pupils is transferred
to the central information system by copying all real-world enrollments
of the Hofer school. More specifically, the activity \textsf{rH} can
be annotated with the copy effect\linebreak{}
 $\id{\pupil}(n,1,\Hoferschool)\rightarrow\av{\pupil}(n,1,\Hoferschool)$.

Let us now focus on the Da Vinci School. Depending on the amount of
incoming requests, the school decides whether to directly process
the enrollments, or to do an entrance test for obtaining a ranking.
In the first case (activity \textsf{pD}), the activity mirrors that
of the Hofer school, and is annotated with the real-world effect $\id{\pupil}(n,1,\Davinci)\leftsquigarrow\id{\req}(n,\Davinci)$.
As for the test, the activity \textsf{tD} can be annotated with a
real-world effect that makes it possible to enroll only those children
who passed the test: $\id{\pupil}(n,1,\Davinci)\leftsquigarrow\id{\req}(n,\Davinci),$\linebreak{}
$\id{\test}(n,mark),mark{\ \geq\ }6$. 

Finally, the process terminates by properly transferring the information
about enrollments to the central administration, exactly as done for
the Hofer school. In particular, the activity \textsf{rD} is annotated
with the copy effect $\id{\pupil}(n,1,\Davinci)\rightarrow\av{\pupil}(n,1,\Davinci)$.
Notice that this effect feeds the same $\pupil$ relation of the central
information systems that is used by \textsf{rH}, but with a different
value for the third column (i.e., the school name).

Figure \ref{fig:school-qats} (right) shows the QATS formalizing the
execution semantics of the parallel composition of the two processes
(where activities are properly annotated with the previously discussed
effects). Circles drawn in orange with solid line represent execution
states where the information about pupils enrolled at the Hofer school
is complete. Circles in blue with double stroke represent execution
states where completeness holds for pupils enrolled at the Da Vinci
school. At the final, sink state information about the enrolled pupils
is complete for both schools.

\begin{figure}[t]
\centering \includegraphics[width=0.48\textwidth]{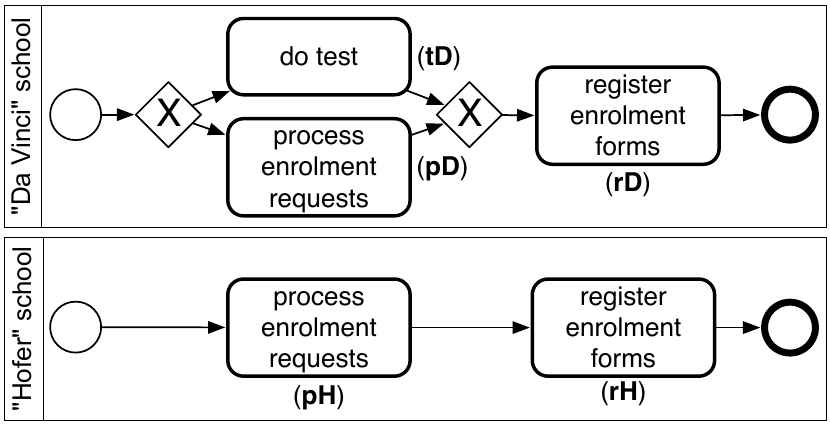}\includegraphics[width=0.48\textwidth]{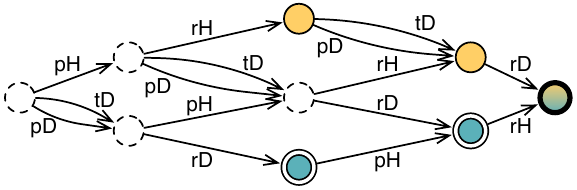}

\caption{BPMN enrollment process of two schools (left), and the corresponding
QATS (right)}

\label{schools-bpmn-and-qats} \label{fig:school-processes}\label{fig:school-qats}
\end{figure}

In Figure \ref{figure:enrolment-BPMN-as-QATS}, we have formalized the
school lane of the BPMN process from Figure \ref{figure:enrolment-BPMN-advanced}
as a QATS. The two actions correspond to the activities in the lane.
In Action 1, which corresponds to the acceptance of enrollment requests
by the school, the real-world effect $r_{1}$ allows to add new enrollments
into the real world. In Action 2, which corresponds to the insertion
of the enrollments into the database, the copy effect $c_{1}$ copies
all enrollments from the real world into the information system.

\begin{figure}[t]
\centering
\includegraphics[scale=1.3]{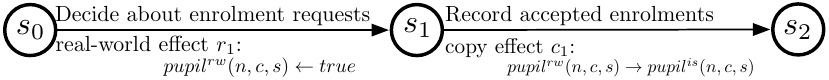}
\caption{The school lane from Figure \ref{figure:enrolment-BPMN-advanced} formalized
as QATS.}
\label{figure:enrolment-BPMN-as-QATS}
\end{figure}

\end{example}

\subsection{Paths and Action Sequences in QATSs}

Let $\qats=(T,\re,\ce)$ be a QATS. A \emph{path} $\pi$ in $\qats$
is a sequence $t_{1},\ldots,t_{n}$ of transitions such that $t_{i}=(s_{i-1},a_{i},s_{i})$
for all $i=1\ldots n$. An\emph{ action sequence} $\alpha$ is a sequence
$a_{1},\ldots,a_{m}$ of action names. Each path $\pi=t_{1},\ldots,t_{n}$
has also a \emph{corresponding} \emph{action sequence} $\alpha_{\pi}$
defined as $a_{1},\ldots,a_{n}$ . For a state $s$, the set $\aseq(s)$
is the set of the action sequences of all paths that end in $s$.

Next we consider the semantics of action sequences. A \emph{development}
of an action sequence $\alpha=a_{1},\ldots,a_{n}$ is a sequence $\id{D_{0}},\ldots,\id{D_{n}}$
of ideal databases such that each pair $\id{(D_{j}},\id{D_{j+1}})$
conforms to the effects $\re(\alpha_{j+1})$. Note that $\id{D_{0}}$
can be arbitrary. For each development $\id{D_{0}},\ldots,\id{D_{n}}$,
there exists a unique trace $\av{D_{0}},\ldots,\av{D_{n}}$, which
is a sequence of available databases $\av{D_{j}}$ defined as follows:
\[
\av{D_{j}}=\begin{cases}
\id{D_{j}} & \mbox{if }j=0\\
\av{D_{j-1}}\cup\mathrm{copy}_{\CE(t_{j})}(\id{D_{j}}) & \mbox{otherwise.}
\end{cases}
\]
 Note that $\av{D_{0}}=\id{D_{0}}$ does not introduce loss of generality
and is just a convention. To start with initially different databases,
one can just add an initial action that introduces data in all ideal
relations.

\subsection{Completeness over QATSs}

An action sequence $\alpha=a_{1},\ldots,a_{n}$ \emph{satisfies} query
completeness of a query $Q$, if for all developments of $\alpha$
it holds that $Q$ is complete over $\id{(D_{n}},\av{D_{n}})$, that
is, if $Q\id{(D_{n}})=Q(\av{D_{n}})$ holds. A path $P$ in a QATS
$\qats$ satisfies query completeness for $Q$, if its corresponding
action sequence satisfies it. A state $s$ in a QATS $\qats$ satisfies
$\compl Q$, if all action sequences in $\aseq(s)$ (the set of the
action sequences of all paths that end in $s$) satisfy $\compl Q$.
We then write $s\models\compl Q$. 

\begin{example}
Consider the QATS in Figure \ref{schools-bpmn-and-qats} (right) and
recall that the action $\textsf{pH}$ is annotated with the effect
$\id{\pupil}(n,1,\Hoferschool)\leftsquigarrow\ $\linebreak{}
$\id{\req}(n,\Hoferschool)$ for enrolling pupils in the real world,
and the action $\textsf{rH}$ with the copy effect $\id{\pupil}\!(n,1,\Hoferschool)\rightarrow$\linebreak{}
$\av{\pupil}\!(n,1,\Hoferschool)$. A path $\pi=((s_{0},\textsf{pH},s_{1}),$
$(s_{1},\textsf{rH},s_{2}))$ has the corresponding action sequence
$(\textsf{pH},\textsf{rH})$. Its models are all sequences $(\id{D_{0}},\id{D_{1}},\id{D_{2}})$
of ideal databases (developments), where $\id{D_{1}}$ may contain
additional pupil facts at the Hofer school w.r.t.~$\id{D_{0}}$ because
of the real-world effect of action $a_{1}$, and $\id{D_{2}}=\id{D_{1}}$.
Each such development has a uniquely defined trace $(\av{D_{0}},\av{D_{1}},\av{D_{2}})$
where $\av{D_{0}}=\id{D_{0}}$ by definition, $\av{D_{1}}=\av{D_{0}}$
because no copy effect is happening in action $a_{1}$, and $\av{D_{2}}=\av{D_{1}}\cup copy_{\ce(a_{1})}(\id{D_{1}})$,
which means that all pupil facts from Hofer school that hold in the
ideal database are copied into the information system due to the effect
of action $a_{1}$. Thus, the state $s_{2}$ satisfies $\compl{Q_{\textit{Hofer}}}$
for a query $\query{Q_{\textit{Hofer}}(n)}{\pupil(n,c,\Hoferschool)}$,
because in all models of the action sequence the ideal database pupils
at the Hofer school are copied into the available database by the
copy effect in action $\textsf{rH}$. 
\end{example}

\section{Verifying Completeness over Processes}

\label{sec:processes:verification}

In the following, we analyze how to check completeness in a state
of a QATS at design time, at runtime, and how to analyze the completeness
of an incomplete query in detail.

\subsection{Design-Time Verification}

When checking for query completeness at design time, we have to consider
all possible paths that lead to the state in which we want to check
completeness. We first analyze how to check completeness for a single
path, and then extend our results to sets of paths.

Given a query $\query{Q(\tpl z)}{R_{1}(\tpl d_{1}),\ldots,R_{n}(\tpl d_{n}),M},$
we say that a real-world effect $r$ is \emph{risky} w.r.t.~$Q$,
if there exists a pair of ideal databases $(\id{D_{1}},\id{D_{2}})$
that conforms to $r$ and where the query result changes, that is,
$Q(\id{D_{1}})\neq Q(\id{D_{2}})$. Intuitively, this means that ideal
database changes caused by $r$ can influence the query answer and
lead to incompleteness, if the changes are not copied into the available
database. 
\begin{prop}
[Risky Effects] Let $r$ be the real-world effect $R(\tpl x,\tpl y)\leftsquigarrow G_{1}(\tpl x,\tpl z_{1})$
and $Q$ be the query $\query Q{R_{1}(\tpl d_{1}),\ldots R_{n}(\tpl d_{n}),M}$.
Then $r$ is risky wrt.~$Q$ if and only if the following formula
is satisfiable: 
\[
G_{1}(\tpl x,\tpl z_{1})\ \wedge\ \bigl(\bigwedge_{i=1\ldots n}\! R_{i}(\tpl d_{i})\bigr)\ \wedge\ M\ \wedge\ \bigl(\bigvee_{R_{i}=R}(\tpl x,\tpl y)=\tpl d_{i}\bigr)
\]
 \end{prop}
\begin{proof}
\textquotedbl{}$\Leftarrow:$\textquotedbl{} If the formula is satisfied
for some valuation $\delta$, this valuation directly yields an example
showing that $r$ is risky wrt.\ $Q$ as follows: Suppose that the
disjunct is satisfied for some $i=k$. Then we can construct databases
$\id D_{1}$ and $\id D_{2}$ as $\id D_{1}=G_{1}(\delta\tpl x,\delta\tpl z_{1})\cup\set{\bigwedge_{i=1\ldots n,i\neq k}R_{i}(\delta\tpl d_{i})}$
and $\id D_{2}=\id D_{1}\cup\set{R_{k}(\delta\tpl d_{k})}$. Clearly,
$(\id D_{1},\id D_{2})$ satisfies the effect $r$ because for the
only additional fact $R_{k}(\delta\tpl d_{k})$ in $\id D_{2}$, the
condition $G_{1}$ is contained in $(\id D_{1})$. But $Q(\id D_{1})\neq Q(\id D_{2})$
because with the new fact, a new valuation for the query is possible
by mapping each atom to itself.

\textquotedbl{}$\Rightarrow:$\textquotedbl{} Holds by construction
of the formula, which checks whether it is possible for $R$-facts
to satisfy both $G_{1}$ and $Q$. Suppose $r$ is risky wrt.\ $Q$.
Then there exists a pair of databases $(\id D_{1},\id D_{2})$ that
satisfies $r$ and where $Q(\id D_{1})\neq Q(\id D_{2})$. Thus, all
new facts in $\id D_{2}$ must conform to $G_{1}$ and some facts
must also contribute to new evaluations of $Q$ that lead to $Q(\id D_{1})\neq Q(\id D_{2})$.
Thus, each such facts implies the existence of a satisfying assignment
for the formula. \end{proof}
\begin{example}
Consider the query $\query{Q(n)}{\pupil(n,c,s),\livesin(n,\Bolzano)}$
and the real-world effect $r_{1}=\pupil(n,c,s)\leftsquigarrow c=4$,
which allows to add new pupils in class 4 in the real world. Then
$r_{1}$ is risky w.r.t.~$Q$, because pupils in class 4 can potentially
also live in Bolzano. Note that without integrity constraints, actually
most updates to the same relation will be risky: if we do not have
keys in the database, a pupil could live both in Bolzano and Merano
and hence an effect $r_{2}=\pupil(n,c,s)\leftsquigarrow\livesin(n,\Merano)$
would be risky w.r.t.~$Q$, too. If there is a key defined over the
first attribute of $\livesin$, then $r_{2}$ would not be risky,
because adding pupils that live in Merano would not influence the
completeness of pupils that only live in Bolzano. 
\end{example}
We say that a real-world effect $r$ that is risky w.r.t.~a query
$Q$ is \emph{repaired} by a set of copy effects $\{c_{2},\ldots,c_{n}\}$,
if for any sequence of databases $(\id{D_{1}},\id{D_{2}})$ that conforms
to $r$ it holds that $Q(\id{D_{2}})=Q(\id{D_{1}}\cup\mathit{\mathit{copy}}{}_{c_{1}\ldots c_{n}}(\id{D_{2}}))$.
Intuitively, this means that whenever we introduce new facts via $r$
and apply the copy effects afterwards, all new facts that can change
the query result are also copied into the available database. 
\begin{prop}
[Repairing] Consider the query $\query Q{R_{1}(\tpl d_{1}),\ldots R_{n}(\tpl d_{n}),M}$,
let $\tpl v=\var(Q)$, a real-world effect $R(\tpl x,\tpl y)\leftsquigarrow G_{1}(\tpl x,\tpl z_{1})$
and a set of copy effects $\{c_{2},\ldots,c_{m}\}$. Then $r$ is
repaired by $\{c_{2},\ldots,c_{m}\}$ if and only if the following
formula is valid:\label{prop:bpm-main-theorem} {\footnotesize{}
\[
\forall\tpl x,\tpl y\!:\ \biggl(\Bigl(\exists\tpl z_{1},\tpl v\!:\ (G_{1}(\tpl x,\tpl z_{1})\wedge\bigwedge_{i=1\ldots n}\! R_{i}(\tpl d_{i})\wedge M\wedge\bigvee_{R_{i}=R}(\tpl x,\tpl y)=\tpl d_{i}\Bigr)\ \ \Rightarrow\ \ \bigvee_{j=2\ldots m}\exists\tpl z_{j}\!:\ G_{j}(\tpl x,\tpl z_{j})\biggr)
\]
}{\footnotesize \par}\end{prop}
\begin{proof}
\textquotedbl{}$\Leftarrow:$\textquotedbl{} Straightforward. If the
formula is valid, it implies that any fact $R(\tpl x)$ that is introduced
by the real-world effect $r$ and which can change the result of $Q$
also satisfies the condition of some copy effect and hence will be
copied.

\textquotedbl{}$\Rightarrow:$\textquotedbl{} Suppose the formula
is not valid. Then there exists a fact $R(\tpl x)$ which satisfies
the condition of the implication (so $R(\tpl x)$ can both conform
to $r$ and change the result of $Q$) but not the consequence (it
is not copied by any copy effect). Thus, we can create a pair $(\id D_{1},\id D_{2})$
of databases as before as $\id D_{1}=G_{1}(\tpl x,\tpl y)\cup\set{\bigwedge_{i=1\ldots n,i\neq k}R_{i}(\tpl d_{i})}$
and $\id D_{2}=\id D_{1}\cup\set{R_{k}(\tpl d_{k})}$ which proves
that $Q(\id D_{2})\neq Q(\id D_{1}\cup\mathit{copy}_{c_{1},\ldots c_{m}}(\id D_{2})$. 
\end{proof}
This implication can be translated into a problem of query containment
as follows: For a query $\query{Q(\tpl z)}{R_{1}(\tpl d_{1}),\ldots,R_{n}(\tpl d_{n})}$,
we define the atom-projection of $Q$ on the $i$-th atom as $\query{Q_{i}^{\pi}(\tpl x)}{R_{1}(\tpl d_{1}),\ldots,R_{n}(\tpl d_{n}),\tpl x=\tpl d_{i}}$.
Then, for a query $Q$ and a relation $R$, we define the $R$-projection
of $Q$, written $Q^{R}$, as the union of all the atom-projections
of atoms that use the relation symbol $R$, that is, $\bigcup_{R_{i}=R}Q_{i}^{\pi}$.
For a real-world effect $r=R(\tpl x,\tpl y)\leftsquigarrow G(\tpl x,\tpl z)$,
we define its associated query $P_{r}$ as $\query{P_{r}(\tpl x,\tpl y)}{R(\tpl x,\tpl y),G(\tpl x,\tpl z)}$. 
\begin{cor}
[Repairing and Query Containment] \label{cor:repairing-if-containment}Let
$Q$ be a query, $\alpha=a_{1},\ldots a_{n}$ be an action sequence,
$a_{i}$ be an action with a risky real-world effect $r$, and $\{c_{1},\ldots,c_{m}\}$
be the set of all copy effects of the actions $a_{i+1}\ldots a_{n}$.

Then $r$ is repaired, if and only if it holds that $P_{r}\cap Q^{R}\subseteq P_{c_{1}}\cup\ldots\cup P_{c_{m}}$. \end{cor}
\begin{proof}
Consider again the formula in Lemma \ref{prop:bpm-main-theorem}.
Then, the first conjunct on the lefthandside is the condition of the
real-world effect $r$, corresponding to $P_{r}$, the second conjunct
is the $R$-projection $Q^{R}$ of the query $Q$, and the third conjunct
is the intersection between $P_{r}$ and $Q^{R}$. The disjunction
on the righthandside corresponds to the union of the queries $P_{c_{1}}$
to $P_{c_{m}}$.
\end{proof}
Intuitively, the corollary says that a risky effect $r$ is repaired,
if all data that is introduced by $r$ that can potentially change
the result of the query $Q$ are guaranteed to be copied into the
information system database by the copy effects $c_{1}$ to $c_{n}$.

The corollary holds because of the direct correspondence between conjunctive
queries and relational calculus \cite{foundations_of_dbs}.

We arrive at a result for characterizing query completeness wrt.\ an
action sequence: 
\begin{lem}
[Action Sequence Completeness]Let $\alpha$ be an action sequence
and $Q$ be a query. Then $\alpha\models\compl Q$ if and only if
all risky effects in $\alpha$ are repaired.\label{lem:incompleteness-if-unhealed-rw-action}\end{lem}
\begin{proof}
``$\Leftarrow$'': Assume that all risky real-world effects in $\alpha$
are repaired in $\alpha$. Then by Lemma \ref{cor:repairing-if-containment}
any fact introduced by a real-world effect $r$ which can potentially
also influence the satisfaction of $\compl Q$ also satisfies the
condition of some later copy effect, and hence it is eventually copied
into some $\av{D_{j}}$ and hence it also appears in $\av{D_{n}}$,
which implies that $C$ is satisfied over $(\id{D_{n}},\av{D_{n}})$.

``$\Rightarrow$'': Assume the repairing does not hold for some
risky effect $r$ of an action $a_{i}\in\alpha$. Then by Lemma \ref{cor:repairing-if-containment},
since the containment does not hold, there exists a database $D$
with a fact $R(t)$ that is in $Q_{r}\cap Q^{R}(D)$ but not in $Q_{c_{i+1}}\cup\ldots\cup Q_{c_{n}}(D)$.
Then, we can create a development $\id{D_{0}},\ldots,\id{D_{n}}$
of $\alpha$ as $\id{D_{0}},\ldots,\id{D_{i-1}}=D\setminus\{R(t)\}$
and $\id{D_{i}},\ldots,\id{D_{n}}=D$. Its trace is $\av{D_{0}},\ldots,\av{D_{n}}=D\setminus\{R(t)\}$,
because since the containment does not hold, for none of the copy
effects in the following actions its guard evaluates to true for the
fact $R(t)$ and hence $R(t)$ is never copied into the available
database. But since $R(t)$ is in $Q^{R}(D)$, query completeness
for $Q$ is not satisfied over $(\id{D_{n}},\av{D_{n}})$ and hence
$\alpha\not\models\compl Q$. 
\end{proof}
Before discussing complexity results in Section \ref{sub:complexity},
we show that completeness entailment over action sequences and containment
of unions of queries have the same complexity. As discussed earlier,
common sublanguages of conjunctive queries are, e.g., queries without
arithmetic comparisons (so-called relational queries), or queries
without repeated relation symbols (so-called linear queries).

For a query language $\L$, we call $\ent{\L}$ the problem of deciding
whether an action sequence $\alpha$ entails completeness of a query
$Q$, where $Q$ and the real-world effects and the copy effects in
$\alpha$ are formulated in language $\L$. Also, we call $\UCont(\L,\L)$
the problem of deciding whether a query is contained in a union of
queries, where all are formulated in the language $\L$. 
\begin{thm}
Let $\L$ be a query languages. Then $\ent{\L}$ and $\UCont(\L,\L)$
can be reduced to each other in linear time.\label{thm:containment=00003Dcomplentailment}\end{thm}
\begin{proof}
``$\Rightarrow$'': Consider the characterization shown in Lemma
\ref{lem:incompleteness-if-unhealed-rw-action}. For a fixed action
sequence, the number of containment checks is the same as the number
of the real-world effects of the action sequence and thus linear.

``$\Leftarrow$'': Consider a containment problem $Q_{0}\subseteq Q_{1}\cup\ldots\cup Q_{n}$,
for queries formulated in a language $\L$. Then we can construct
a QATS $\qats=(S,s_{0},A,E,\re,\ce)$ over the schema of the queries
together with a new relation $R$ with the same arity as the queries
where $S=\{s_{0},s_{1},s_{2}\}$, $A=\{a_{1},a_{2}\},\re(a_{1})=\{\id R(\tpl x)\leftsquigarrow Q_{0}(\tpl x)\}$
and $\ce(a_{2})=\bigcup_{i=1\ldots n}\{Q_{i}(\tpl x)\rightarrow\av R(\tpl x)\}$.
Now, the action sequence $a_{1},a_{2}$ satisfies a query completeness
for a query $\query{Q'(\tpl x)}{R(\tpl x)}$ exactly if $Q_{0}$ is
contained in the union of the queries $Q_{1}$ to $Q_{n}$, because
only in this case the real-world effect at action $a_{1}$ cannot
introduce any facts into $\id{D_{1}}$ of a development of $a_{1},a_{2}$,
which are not copied into $\av{D_{2}}$ by one of the effects of the
action $a_{2}$. 
\end{proof}
We discuss the complexity of query containment and hence of completeness
entailment over action sequences more in detail in Section \ref{sub:complexity}.

So far, we have shown how query completeness over a path can be checked.
To verify completeness in a specific state, we have to consider all
paths to that state, which makes the analysis more difficult. We first
introduce a lemma that allows to remove repeated actions in an action
sequence: 
\begin{lem}
[Duplicate Removal]Let $\alpha=\alpha_{1},\tilde{a},\alpha_{2},\tilde{a},\alpha_{3}$
be an action sequence with $\tilde{a}$ as repeated action and let
$Q$ be a query. Then $\alpha$ satisfies $\compl Q$ if and only
if $\alpha'=\alpha_{1},\alpha_{2},\tilde{a},\alpha_{3}$ satisfies
$\compl Q$. \label{lem:repeated-transitions-can-be-ignored}\end{lem}
\begin{proof}
\textquotedbl{}$\Rightarrow$\textquotedbl{}: Suppose $\alpha$ satisfies
$\compl Q$. Then, by Proposition \ref{lem:incompleteness-if-unhealed-rw-action},
all risky real-world effects of the actions in $\alpha$ are repaired.
Let $a_{r}$ be an action in $\alpha$ that contains a risky real-world
effect $r$. Thus, there must exist a set of actions $A_{c}$ in $\alpha$
that follows $a_{r}$ and contains copy effects that repair $r$.
Suppose $A_{c}$ contains the first occurrence of $\tilde{a}$. Then,
this first occurrence of $\tilde{a}$ can also replaced by the second
occurrence of $\tilde{a}$ and then the modified set of actions also
appears after $a_{r}$ in $\alpha'$.

\textquotedbl{}$\Leftarrow$\textquotedbl{}: Suppose $\alpha'$ satisfies
$\compl Q$. Then, also $\alpha$ satisfies $\compl Q$ because adding
the action $\tilde{a}$ earlier cannot influence query completeness:
Since by assumption each risky real-world effect of the second occurrence
of $\tilde{a}$ is repaired by some set of actions $A_{c}$ that follows
$\tilde{a}$, the same set $A_{c}$ also repairs each risky real-world
effect of the first occurrence of $\tilde{a}$. 
\end{proof}
The lemma shows that our formalism can deal with cycles. While cycles
imply the existence of sequences of arbitrary length, the lemma shows
that we only need to consider sequences where each action occurs at
most once. Intuitively, it is sufficient to check each cycle only
once. 
\begin{rem}
If we consider a fixed start database, we cannot just drop all but
the last occurrence of an action. Consider e.g.\ a process consisting
of a sequence of three actions: a real-world effect $\id R(x)\leftsquigarrow\true$,
a copy effect $\id R(x)\rightarrow\av S(x)$ and again the real-world
effect $\id R(x)\leftsquigarrow\true$. Then if the start databases
are assumed to be empty, the first occurrence of $R(x)\leftsquigarrow\true$
cannot be dropped without changing the satisfaction of completeness
of a query $\query{Q(x)}{S(x)}$ in the end of the process. Still,
because we do not consider recursive queries, such dependencies would
presumably be finite.
\end{rem}
Based on the preceding lemma, we define the \emph{normal action sequence}
of a path $\pi$ as the action sequence of $\pi$ in which for all
repeated actions all but the last occurrence are removed.
\begin{prop}
[Normal Action Sequences]Let $\qats=(T,\re,\ce)$ be a QATS, $\Pi$
be the set of all paths of $\qats$ and $Q$ be a query. Then \end{prop}
\begin{enumerate}
\item for each path $\pi\in\Pi$, its normal action sequence has at most
the length $\mid\!\! A\!\!\mid$, 
\item there are at most $\Sigma_{k=1}^{\mid A\mid}\frac{\mid\! A\!\mid!}{(\mid\! A\!\mid-k)!}<(\mid\!\! A\!\!\mid+1)!$
different normal forms of paths, 
\item for each path $\pi\in\Pi$, it holds that $\pi\models\compl Q$ if
and only if $\alpha'$ satisfies $\compl Q$, where $\alpha'$ is
the normal action sequence of $\pi$.\end{enumerate}
\begin{proof}
The first two items hold because normal action sequences do not contain
actions twice. The third item holds because of Lemma \ref{lem:repeated-transitions-can-be-ignored},
which allows to remove all but the last occurrence of an action in
an action sequence without changing query completeness satisfaction.
\end{proof}
Before arriving at the main result, we need to show that deciding
whether a given normal action sequence can actually be realized by
a path is easy: 
\begin{prop}
Given a QATS $\qats$, a state $s$ and a normal action sequence $\alpha$.
Then, deciding whether there exists a path $\pi$ that has $\alpha$
as its normal action sequence and that ends in $s$ can be done in
polynomial time. \label{prop:checking-for-action-sequences-path-existence} \end{prop}
\begin{proof}
The reason for this proposition is that given a normal action sequence
$\alpha=a_{1},\ldots,a_{n}$, one just needs to calculate the states
reachable from $s_{0}$ via the concatenated expression 
\[
(a_{1},\ldots,a_{n})^{+},(a_{2},\ldots,a_{n})^{+},\ldots,(a_{n-1},a_{n})^{+},(a_{n})^{+}
\]
This expression stands exactly for all action sequences with $\alpha$
as normal sequence, because it allows repeated actions before their
last occurrence in $\alpha$. Calculating the states that are reachable
via this expression can be done in polynomial time, because the reachable
states $S_{n}^{\reach}$ can be calculated iteratively for each component
$(a_{i},\ldots,a_{n})^{+}$ as $S_{i}^{\reach}$ from the reachable
states $S_{i-1}^{\reach}$ until the previous component $(a_{i-1},\ldots,a_{n})^{+}$
by taking all states that are reachable from a state in $S_{i-1}^{\reach}$
via one or several actions in $\{a_{i},\ldots,a_{n}\}$, which can
be done with a linear-time graph traversal such as breadth-first or
depth-first search. Since there are only $n$ such components, the
overall algorithm works in polynomial time.
\end{proof}
Having shown that realization of a normal action sequence by a QATS
is in PTIME, we can prove the following main result:
\begin{thm}
Given a QATS $\qats$ and a query $Q$, both formulated in a query
language $\L$, checking ``$s\not\models\compl Q$?'' can be done
using a nondeterministic polynomial-time Turing machine with a $\ucont{\L}$-oracle.
\label{thm:final-characterization-design-time}\end{thm}
\begin{proof}
If $s\not\models\compl Q$, one can guess a normal action sequence
$\alpha$, check by Proposition \ref{prop:checking-for-action-sequences-path-existence}
in polynomial time that there exists a path $\pi$ from $s_{0}$ to
$s$ with $\alpha$ as normal action sequence, and by Theorem \ref{thm:containment=00003Dcomplentailment}
verify using the $\ucont{\L}$-oracle that $\alpha$ does not satisfy
$\compl Q$. 
\end{proof}
We discuss the complexity of this problem in Section~\ref{sub:complexity}

\subsection{Runtime Verification}

Similarly to the results about database instance reasoning in Section
\ref{sec:reasoning_with_instances}, more completeness can be derived
if the actual process instance is taken into account, that is, the
concrete activities that were carried out within a process. 

As an example, consider that the secretary in a large school can perform
two activities regarding the enrollments, either he/she can sign enrollment
applications (which means that the enrollments become legally valid),
or he/she can record the signed enrollments that are not yet recorded
in the database. For simplicity we assume that the secretary batches
the tasks and performs only one of the activities per day. A visualization
of this process is shown in Figure \ref{figure:runtime-verification-example-bpmn}.
Considering only the process we cannot draw any conclusions about
the completeness of the enrollment data, because if the secretary
chose the first activity, then data will be missing, however if the
secretary chose the second activity, then not. If however we have
the information that the secretary performed the second activity,
then we can conclude that the number of the currently valid enrollments
is also complete in the information system.

\begin{figure}[t]
\centering \includegraphics{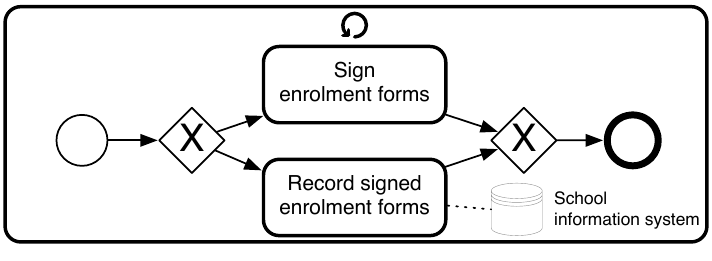}
\caption{Simplified BPMN process for the everyday activity of a secretary in
a school}

\label{figure:runtime-verification-example-bpmn} 
\end{figure}

Formally, a runtime verification problem consists of a path $\pi=t_{1},\ldots,t_{n}$
that was executed so far and a query $Q$. Again the problem is to
check whether completeness holds in the current state, that is, whether
all developments of $\pi$ satisfy $\compl Q$. Recall that we introduced
$\ent{\L}$ as the problem of deciding whether a path in a QATS formulated
in a language $\L$ satisfies completeness of a query formulated in
the same language $\L$.
\begin{cor}
The problems $\ent{\L}$ and $\ucont{\L}$ can be reduced to each
other in linear time. 
\end{cor}
The corollary follows directly from Theorem \ref{thm:containment=00003Dcomplentailment}
and the fact that a path satisfies completeness if and only if its
action sequence satisfies completeness.

Runtime verification becomes more complex when also the current, concrete
state of the available database is explicitly taken into account.
Given the current state $D$ of the database, the problem is then
to check whether all the developments of $\pi$ in which $\av{D_{n}}=D$
holds satisfy $\compl Q$. In this case repairing of all risky actions
is a sufficient but not a necessary condition for completeness:
\begin{example}
Consider a path $(s_{0},a_{1},s_{1}),(s_{1},a_{2},s_{2})$, where
action $a_{1}$ is annotated with the copy effect $\id{\req}(n,s)\rightarrow\av{\req}(n,s)$,
action $a_{2}$ with the real-world effect $\id{\pupil}(n,c,s)\leftsquigarrow\id{\req}(n,s)$,
a database $\av D_{2}$ that is empty, and consider a query $\query{Q(n)}{\pupil(n,c,s),}$ $\req(n,s)$.
Then, the query result over $\av D_{2}$ is empty. Since the relation
$\req$ was copied before, and is empty now, the query result over
any ideal database must be empty too, and therefore $\compl Q$ holds.
Note that this cannot be concluded with the techniques introduced
in this work, as the real-world effect of action $a_{2}$ is risky
and is not repaired. 
\end{example}
The complexity of runtime verification w.r.t.~a concrete database
instance is still open.

\subsection{Dimension Analysis}

When at a certain timepoint a query is not found to be complete, for
example because the deadline for the submissions of the enrollments
from the schools to the central school administration is not yet over,
it becomes interesting to know which parts of the answer are already
complete. 
\begin{example}
Consider that on the 10th of April, the schools ``Hofer'' and ``Da
Vinci'' have confirmed that they have already submitted all their
enrollments, while ``Max Valier'' and ``Gherdena'' have entered
some but not all enrollments, and other schools did not enter any
enrollments so far. Then the result of a query asking for the number
of pupils per school would look as in Figure \ref{figure:visualization-dimension-analysis}
(left table), which does not tell anything about the trustworthiness
of the result. If one includes the information from the process, one
could highlight that the data for the former two schools is already
complete, and that there can also be additional schools in the query
result which did not submit any data so far (see right table in Figure
\ref{figure:visualization-dimension-analysis}).\label{example:dimension-analysis} 
\end{example}
\begin{figure}[t]

\begin{centering}
\includegraphics[width=1\textwidth]{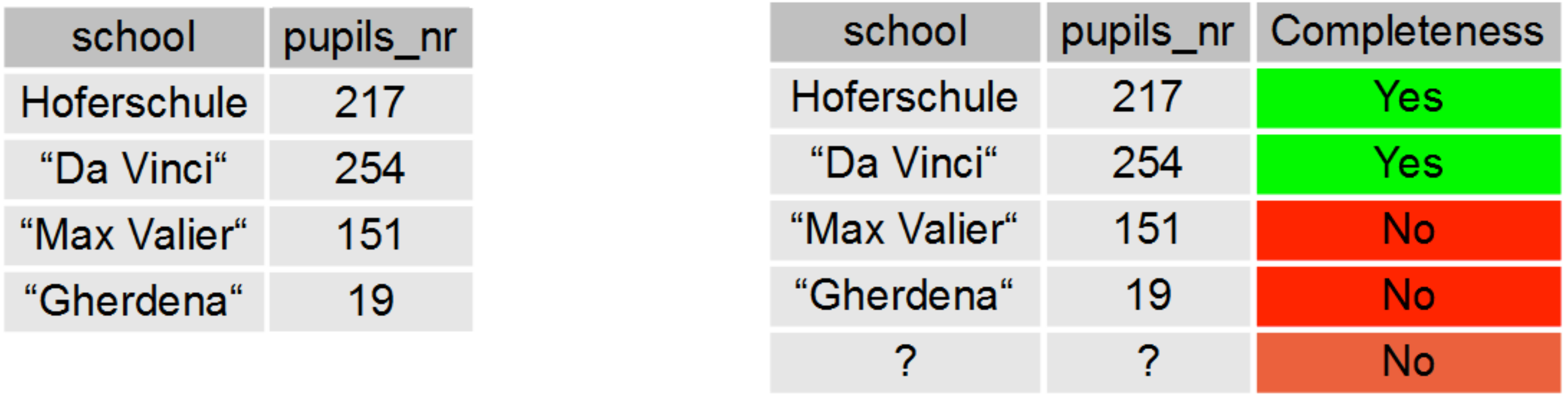} 
\par\end{centering}

\caption{Visualization of the dimension analysis of Example \ref{example:dimension-analysis}.}

\label{figure:visualization-dimension-analysis} 
\end{figure}

Formally, for a query $Q$ a dimension is a set of distinguished variables
of $Q$. Originally, dimension analysis was meant especially for the
arguments of a GROUP BY expression in a query, however it can also
be used with other distinguished variables of a query. Assume a query
$\query{Q(\tpl x)}{B(\tpl x,\tpl y)}$ cannot be guaranteed to be
complete in a specific state of a process. For a dimension $\tpl v\subseteq\tpl x$,
the analysis can be done as follows: 
\begin{enumerate}
\item Calculate the result of $\query{Q'(\tpl v)}{B(\tpl x,\tpl y)}$ over
$\av D$. 
\item For each tuple $\tpl c$ in $Q'(\av D)$, check whether $s,\av D\models\compl{Q[\tpl v/\tpl c]}$.
This tells whether the query is complete for the values $\tpl c$
of the dimension $V$. 
\item To check whether further values are possible, one has to guess a new
value $\tpl c_{new}$ for the dimension and show that $Q[\tpl v/\tpl c_{new}]$
is not complete in the current state. For the guess one has to consider
only the constants in the database plus a fixed set of new constants,
hence the number of possible guesses is polynomial for a fixed dimension
$\tpl v$.
\end{enumerate}
Step 2 corresponds to deciding for each tuple with a certain value
in $Q(\av D)$, whether it is complete or not (color red or green
in Figure \ref{figure:visualization-dimension-analysis}, right table),
Step 3 to deciding whether there can be additional values (bottom
row in Figure \ref{figure:visualization-dimension-analysis}, right
table).

\subsection{Complexity of Completeness Verification}

\label{sub:complexity}

In the previous sections we have seen that completeness verification
can be solved using query containment. Results on query containment
are already reported in Section \ref{sub:general-tc-qc-entailment}.
The results presented here follow from Theorems~\ref{thm:containment=00003Dcomplentailment}
and~\ref{thm:final-characterization-design-time}, and are summarized
in Figure~\ref{figure:complexity-results-table}. We distinguish
between the problem of runtime verification, which has the same complexity
as query containment, and design-time verification, which, in principle
requires to solve query containment exponentially often. Notable however
is that in most cases the complexity of runtime verification is not
higher than the one of design-time verification. 

The results on linear relational and linear conjunctive queries, i.e.,
conjunctive queries without selfjoins and without or with comparisons,
are borrowed from \cite{Razniewski:Nutt-Compl:of:Queries-VLDB11}.
The result on relational queries is reported in~\cite{Sagiv:Yannakakis-Containment-VLDB},
and that on conjunctive queries from~\cite{meyden-Complexity_querying_ordered_domains-pods}.
As for integrity constraints, the result for databases satisfying
finite domain constraints is reported in~\cite{Razniewski:Nutt-Compl:of:Queries-VLDB11}
and for databases satisfying keys and foreign keys in~\cite{rosati:2003:containment-keys-and-foreign-keys}.

\begin{figure}[t]
\renewcommand{\arraystretch}{1.3}

\begin{centering}
\begin{small} %
\begin{tabular}{|>{\centering}m{4cm}|>{\centering}m{3.5cm}>{\centering}m{2.9cm}|}
\hline 
Query/QATS language $\L$  & Runtime-verification:

Complexity of $\UCont(\L,\L)$ and $\ent{\L}$ ``$(\pi\models\compl Q)$''?  & Design-time verification:

Complexity of ``$s\models\compl Q$''?\tabularnewline
\hline 
Linear relational queries  & PTIME  & in $\coNP$\tabularnewline
Linear conjunctive queries  & $\coNP$-complete  & $\coNP$-complete\tabularnewline
 Relational queries  & NP-complete  & in $\Pi_{2}^{P}$\tabularnewline
Relational queries over databases with finite domains  & $\Pi_{2}^{P}$-complete  & $\Pi_{2}^{P}$-complete\tabularnewline
Conjunctive queries  & $\Pi_{2}^{P}$-complete  & $\Pi_{2}^{P}$-complete\tabularnewline
Relational queries over databases with keys and foreign keys  & in PSPACE  & in PSPACE\tabularnewline
\hline 
\end{tabular}\end{small} 
\par\end{centering}

\caption{Complexity of design-time and runtime verification for different query
languages.}

\label{figure:complexity-results-table} 
\end{figure}

\section{Extracting Transition Systems from Petri Nets}

\label{sec:processes:petri:nets}Transition systems are a very basic
formalism for describing the semantics of business processes. For
business processes itself, the quasi-standard for process models is
the business process modeling notation (BPMN). Large parts of BPMN
are well founded in coloured Petri nets. In turn, the models of coloured
Petri nets are represented by its reachability graph, which is a transition
system. It is therefore not surprising that the annotations of transition
systems with real-world and copy effects can also be expressed on
the level of coloured Petri nets.

Informally, coloured Petri nets (CPN) are systems in which tokens
of different types can move and interact according to defined actions~\cite{jensen:1987:coloured-petri-nets}.
An example of a CPN that is annotated with real-world and copy effects
is shown in Figure \ref{figure-example-QACPN}. In this CPN, there
exist three types of tokens: Time, persons and schools, which initially
populate the places on the left side in top-down order. The actions
can be executed whenever there are tokens in all the input places,
e.g., the action ``Enroll yourself'' can be executed whenever there
there is a time token in the first place, a person token in the second
place and a school token in the third place. 

When annotating Petri nets with real-world and copy effects, we can
now use the variables of the actions also in the effect specifications:
For instance, the real-world effect of the action ``Decide enrollments
of a school'' allows to create records \textit{enrolled(n,s)} for
pupils who expressed an enrollment desire at the school which performs
this action. 

Notice that while already coloured Petri nets also allow some representation
of data via the tokens, CPN annotated with real-world and copy effects
go clearly beyond this, because (1) copy actions are a kind of universally
quantified transitions, (2) they allow the introduction of new values,
and (3) allow the verification for arbitrary starting databases.

The semantics of Coloured Petri nets are their reachability graphs,
which, in turn, are transition systems. A common class of well-behaved
Petri nets are bound Petri nets. A Petri net is bound by a value $k$,
if the number of tokens in all reachable states is less or equal to
$k$. For $k$-bound Petri nets, their reachability graph is at most
exponential in $k$. Thus, completeness verification over coloured
Petri nets can be reduced to completeness verification over exponential
transition systems.

\begin{figure}[t]
\begin{centering}
\includegraphics[width=1\textwidth]{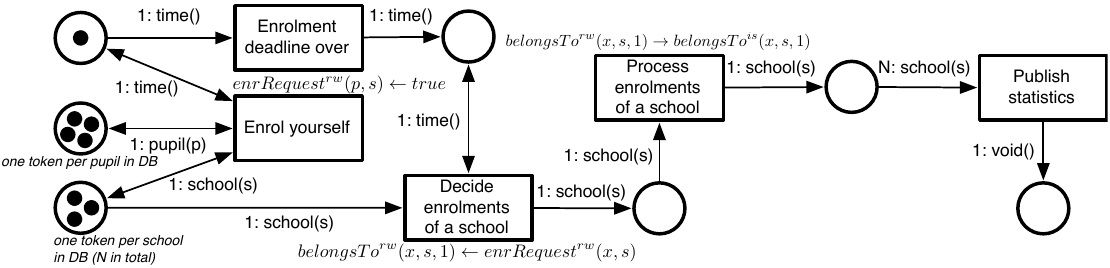} 
\par\end{centering}

\caption{Enrollment both from the perspective of schools and students modelled
in a CPN. A desirable property to check is that when the action ``Publish
statistics'' is executed, the data about enrolments is complete,
which in this example is the case.}

\label{figure-example-QACPN} 
\end{figure}

\section{Related Work}

\label{sec:processes-related-work}

In the BPM context, there have been attempts to model data quality
issues, like in \cite{bpmn-data-inaccuracy-modeling-weak,bpmn-data-quality:cappiello:caballero-weak,bagchi:2006:data-quality-and-business-process-modeling}.
However, these approaches mainly discussed general methodologies for
modeling data quality requirements in BPMN, but did not provide methods
to asses their fulfillment. In this chapter, we claim that process
formalizations are an essential source for learning about data completeness
and show how data completeness can be verified. In particular, our
contributions are (1) to introduce the idea of extracting information
about data completeness from processes manipulating the data, (2)
to formalize processes that can both interact with the real-world
and record information about the real-world in an information system,
and (3) to show how completeness can be verified over such processes,
both at design and at execution time.

Our approach leverages on two assumptions related to how the data
manipulation and the process control-flow are captured. From the data
point of view, we leverage on annotations that suitably mediate between
expressiveness and tractability. More specifically, we rely on annotations
modeling that new information of a given type is acquired in the real
world, or that some information present in the real world is stored
into the information system. We do not explicitly consider the evolution
of specific values for the data, as incorporating full-fledged data
without any restriction would immediately make our problem undecidable,
being simple reachability queries undecidable in such a rich setting
\cite{DDHV11,BCDDF11,BCDDM13}. From the control-flow point of view,
we are completely orthogonal to process specification languages. In
particular, we design our data completeness algorithms over (labeled)
transition systems, a well-established mathematical structure to represent
the execution traces that can by produced according to the control-flow
dependencies of the (business) process model of interest. Consequently,
our approach can in principle be applied to any process modeling language,
with the proviso of annotating the involved activities. We are in
particular interested in providing automated reasoning facilities
to answer whether a given query can be answered with complete information
given a target state or a sequence of activities.

\section{Summary }

In this chapter we have discussed that data completeness analysis
should take into account the processes that manipulate the data. In
particular, we have shown how process models can be annotated with
effects that create data in the real world and effects that copy data
from the real world into an information system. We have then shown
how one can verify the completeness of queries over transition systems
that represent the execution semantics of such processes. It was shown
that, similarly to the previous chapters, the problems here are closely
related to the problem of query containment, although now, it may
be the case that exponentially many containments have to be solved
for one completeness check. We also showed that completeness checking
is easier when the trace of the process is known.

We focused on the process execution semantics in terms of transition
systems. The results would allow the realization of a demonstration
system to annotate high-level business process specification languages
(such as BPMN or YAWL), extract the underlying quality-aware transition
systems, and apply the techniques here presented to check completeness.

\chapter{Discussion}

\label{chap:discussion}

In the previous chapters we have discussed how to analyze the completeness
of query answers over databases using metadata about the completeness
of parts of the data. A critical prerequisite for doing such analysis
is to actually obtain such completeness metadata, and to have reasons
to believe that this metadata is correct. Also, the practical implications
of the presented complexity results and the technical integration
of completeness reasoning into existing software landscapes are important
issues. In the following, we discuss these issues.

\paragraph{Obtaining Completeness Statements and Ensuring Statement Correctness}

In the setting of company databases, especially fast-changing transactional
databases that possibly get integrated into data warehouses regularly,
in order to have up-to-date completeness metadata, completeness statement
generation needs to be automated as much as possible. Where data creation
is done automatically (e.g., sensor data), it could be feasible to
also generate completeness statements automatically. Where data is
submitted manually (for instance, a human presses a \textquotedbl{}submit\textquotedbl{}
button), completeness statement generation should be bound to the
data submission. That is, whenever data is submitted, the user is
asked (or forced) to also make statements about the (in-)completeness
of the submitted data. This is crucial, because commonly the stakeholder
that will know most about the completeness of the data is the one
who submits the data.

In settings where the knowledge about completeness is not captured
directly at data creation, later attempts to get completeness statements
will require manual inspection and may be tricky, as often information
about data provenance is not maintained well. This may e.g.\ be a
problem when a database with completeness metadata is merged with
another database without such completeness information, e.g.\ after
an acquisition.

In the settings of crowd-based data such as OpenStreetMap or of integrated
data without any quality guarantees such as on the Semantic Web, there
is no way to enforce the generation of completeness metadata. Instead,
completeness metadata will need to be generated and maintained based
on mutual ratings and trust levels.

\paragraph{Practical Complexity}

While some theoretical complexity results presented in this thesis
may seem as if implementations could be very challenging, most discussed
schema-level reasoning problems are reduced to query containment,
which, for the languages discussed, is solved in existing DBMS every
day. We thus expect little runtime challenges when implementing schema
level reasoning using existing techniques for containment.

On the other hand, the results for reasoning wrt.\ a database instance
for which we showed a PTIME data complexity still pose a big challenge:
Even a linear data complexity may be not feasible for large data warehouses,
thus, for these results, more research is needed before they can be
implemented.

\paragraph{Set/bag disambiguation for TC-QC Reasoning}

A major complication in the presented results is the disambiguation
between set and bag semantics for queries. As Proposition \ref{prop:bag-and-set-normally-the-same}
however shows, $\tcqc$entailment reasoning for queries under bag
and set semantics only differs in cases where the queries under set
semantics not minimal, and can be synchronized again using query minimization.
As for all instances of $\tcqc^{s}(\L_{1},\L_{2})$ with $\L_{1}=\L_{2}$,
the complexity is the same as that of minimization of queries in $\L_{1}$,
the separate discussion of the reasoning for queries under set semantics
in Section \ref{sub:general-tc-qc-entailment} may give the wrong
impression that the problems are very different, while in fact in
all but one case they can be dealt with by the same algorithm.

\paragraph{Technical Integration}

The technical integration of completeness reasoning into data management
software remains an open problem. We can only conjecture that a completeness
reasoner component will require deep integration into the existing
data management software landscape.

First, the components that create completeness statements would need
to be integrated into software for creating and manipulation database
content, e.g.\ MS Access, SAP software, or custom-made web interfaces.

Second, the component that perform the actual reasoning should be
integrated with the DBMS, e.g.\ as a plugin, in order to allow the
execution of reasoning at the same time as query execution.

Third, the components for visualizing completeness information need
to be embedded into the software that is used to show query results,
such as management cockpits, web portals or business intelligence
tools such as Qlikview.

\paragraph{Impact}

Parts of the theory presented in this thesis have been implemented
by Savkovic \etal\  in a demonstration system called MAGIK \cite{ogi:demo:cikm,ogi:demo:vldb}.
The reasoning is MAGIK is performed by translating the reasoning problems
into logical programs, which are then solved by the DLV reasoner %
\footnote{http://www.dlvsystem.com/html/DLV\_User\_Manual.html%
}. MAGIK can also reasoning tasks that are not discussed in this thesis,
namely reasoning wrt.\ foreign keys and computing the most general
complete specializations or the least general complete generalization
of a query that are not complete.

\bibliographystyle{plain}
\bibliography{files/thesis,files/biblio_iswc}

\appendix

\chapter{Notation Table}

The listing below contains common notation used throughout this thesis.
Notation that is specific to a single chapter is not listed here.\medskip{}

\begin{longtable}{|c|>{\centering}p{9cm}|}
\hline 
Symbol & Meaning\tabularnewline
\hline 
$\Sigma$ & relational database schema\tabularnewline
\hline 
$A$ & relational atom\tabularnewline
\hline 
$\dom$ & domain, infinite set of constants\tabularnewline
\hline 
$D$ & database, set of ground atoms\tabularnewline
\hline 
$c$ & constant\tabularnewline
\hline 
$\tpl d$ & tuple of constants\tabularnewline
\hline 
$x,y,z,w$ & variables\tabularnewline
\hline 
$t$ & term, constant or variable\tabularnewline
\hline 
$R,(S,T)$ & relation names\tabularnewline
\hline 
$Q$ & conjunctive query\tabularnewline
\hline 
$B$ & body of a conjunctive query\tabularnewline
\hline 
$\tpl x$ & distinguished variables of a query\tabularnewline
\hline 
$\tpl y$ & nondistinguished variables of a query\tabularnewline
\hline 
$L$ & relational part of a body\tabularnewline
\hline 
$M$ & comparisons\tabularnewline
\hline 
$G$ & condition, set of atoms\tabularnewline
\hline 
$\cdot^{s}$ & set semantics\tabularnewline
\hline 
$\cdot^{b}$ & bag semantics\tabularnewline
\hline 
$\di$ & ideal database\tabularnewline
\hline 
$\da$ & available database\tabularnewline
\hline 
$\D$ & incomplete database, pair of an ideal and an available database\tabularnewline
\hline 
$C$ & table completeness statement\tabularnewline
\hline 
$Q_{C}$ & query associated to a table completeness statement\tabularnewline
\hline 
$v$ & valuation, mapping from variables into constants; sometimes also Greek
letters $(\sigma,\theta,\delta)$\tabularnewline
\hline 
$T_{C}$ & transformator function for a TC statement, maps databases into databases\tabularnewline
\hline 
$\compl Q$ & query completeness statement for a query $Q$\tabularnewline
\hline 
$\L$ & query language\tabularnewline
\hline 
$\Cont$ & containment problem of a query in a query\tabularnewline
\hline 
$\UCont$ & containment problem of a query in a union of queries\tabularnewline
\hline 
$\tctc$ & entailment problem of table completeness by table completeness\tabularnewline
\hline 
$\tcqc$ & entailment problem of table completeness by query completeness\tabularnewline
\hline 
$\tcqc$ & entailment problem of query completeness by query completeness\tabularnewline
\hline 
\end{longtable}

\section*{}
\end{document}